\documentclass[12pt,a4paper]{book}

\usepackage{amsthm}
\usepackage{amsmath}
\usepackage{amssymb}
\usepackage{graphicx}
\usepackage{float}
\usepackage[caption=false]{subfig}
\usepackage{cite}
\usepackage{color}
\definecolor{darkblue}{rgb}{0.2,0.2,0.6}
\usepackage[linktocpage,colorlinks=true,linkcolor=darkblue,citecolor=darkblue,urlcolor=darkblue]{hyperref}

\usepackage{setspace}
\usepackage[T1]{fontenc}
\usepackage{txfonts}
\usepackage[margin=3.0cm]{geometry}
\usepackage{cancel}
\usepackage{epigraph}
\usepackage[toc,page]{appendix}
\usepackage{comment}

\usepackage[compact]{titlesec}
\titlespacing{\section}{0pt}{*0}{*0}
\titlespacing{\subsection}{0pt}{*0}{*0}
\titlespacing{\subsubsection}{0pt}{*0}{*0}

\newcommand{\bit}{\begin{itemize}}
\newcommand{\eit}{\end{itemize}\par\noindent}
\newcommand{\ben}{\begin{enumerate}}
\newcommand{\een}{\end{enumerate}\par\noindent}
\newcommand{\beq}{\begin{equation}}
\newcommand{\eeq}{\end{equation}\par\noindent}
\newcommand{\beqa}{\begin{eqnarray}}
\newcommand{\eeqa}{\end{eqnarray}\par\noindent}
\newcommand{\beqn}{\begin{eqnarray}}
\newcommand{\eeqn}{\end{eqnarray}\par\noindent}

\newtheorem*{theorem*}{Theorem}
\newtheorem*{lemma*}{Lemma}
\newcommand{\Tr}{\text{Tr}}
\newcommand{\proofend}{\hfill\fbox\\\medskip }

\newtheorem{theorem}{Theorem}
\newtheorem{lemma}{Lemma}

\usepackage[table]{xcolor}
\usepackage{latexsym}
\usepackage[normalem]{ulem}

\providecommand{\U}[1]{\protect\rule{.1in}{.1in}}

\newtheorem{proposition}[theorem]{Proposition}

\newcommand{\idn}{\mathbb{I}}
\makeatletter
\def\hlinewd#1{%
  \noalign{\ifnum0=`}\fi\hrule \@height #1 \futurelet
   \reserved@a\@xhline}
\makeatother

\makeatletter
\renewcommand{\@chapapp}{}%

\makeatother

\begin{document}
\doublespacing
%
\begin{center}
{\Large {\bf Contextuality beyond the Kochen-Specker theorem}}
\vskip 0.70cm
{\bf {\em By}} 
\vskip -0.2cm
{\bf {\large Ravi Kunjwal}}
\vskip 0.0cm
{\bf {\large PHYS10201005002}}
\vskip 0.5cm
{\bf {\large The Institute of Mathematical Sciences, Chennai}}
\vskip 2.6cm
{\bf {\em {\large A thesis submitted to the
\vskip 0.05cm
Board of Studies in Physical Sciences
\vskip 0.05cm
In partial fulfillment of requirements
\vskip 0.05cm
For the Degree of 
}}}
\vskip 0.05cm
{\bf {\large DOCTOR OF PHILOSOPHY}}
\vskip 0.1cm
{\bf {\em of}}
\vskip 0.1cm
{\bf {\large HOMI BHABHA NATIONAL INSTITUTE}}
\vfill
\includegraphics[height=3.5cm, width=3.5cm]{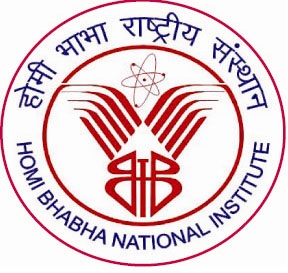}
\vfill
{\bf {\large April, 2016}}
\vfill
\end{center}

\newpage
\cleardoublepage
%
%
\centerline{{\bf{\LARGE Homi Bhabha National Institute}}}
\vskip 0.3cm
\centerline{{\bf {\large Recommendations of the Viva Voce Committee}}}
\vskip 0.3cm
As members of the Viva Voce Committee, we certify that we have read the
dissertation prepared by {\bf Ravi Kunjwal} entitled ``Contextuality 
beyond the Kochen-Specker theorem'' and
recommend that it may be accepted as fulfilling the thesis
requirement for the award of Degree of Doctor of Philosophy.
\vskip 0.5cm
\underline{\hspace{12.0cm}} Date:
\vskip -0.1cm 
Chair - Prof.~Rahul Sinha
\vskip 0.7cm
\underline{\hspace{12.0cm}} Date:
\vskip -0.1cm 
Guide/Convener - Prof.~Sibasish Ghosh
\vskip 0.7cm
\underline{\hspace{12.0cm}} Date:
\vskip -0.1cm 
Examiner - Prof.~Guruprasad Kar
\vskip 0.7cm

\underline{\hspace{12.0cm}} Date:
\vskip -0.1cm 
Member 1 - Prof.~R.~Shankar
\vskip 0.7cm
\underline{\hspace{12.0cm}} Date:
\vskip -0.1cm 
Member 2 - Prof.~Rajesh Ravindran
\vskip 2.0cm
\hspace{0.7cm} Final approval and acceptance of this thesis is
contingent upon the candidate's submission of the final copies of the
thesis to HBNI.
\vskip -0.2cm
\hspace{0.7cm} I hereby certify that I have read this thesis
prepared under my direction and recommend that it may be accepted as
fulfilling the thesis requirement.
\vskip 1.4cm 
{\bf Date:} 
\vskip 0.3cm 
{\bf Place:} IMSc, Chennai \hfill Guide: Prof.~Sibasish Ghosh \hspace{1.0cm}
\newpage
\cleardoublepage
\centerline{{\bf {\large STATEMENT BY AUTHOR}}}
\vskip 1.00cm
This dissertation has been submitted in partial fulfillment of
requirements for an advanced degree at Homi Bhabha National Institute
(HBNI) and is deposited in the Library to be made available to borrowers
under rules of the HBNI.
\vskip 0.6cm
Brief quotations from this dissertation are allowable without special
permission, provided that accurate acknowledgement of source is made.
Requests for permission for extended quotation from or reproduction of
this manuscript in whole or in part may be granted by the Competent
Authority of HBNI when in his or her judgement the proposed use of the
material is in the interests of scholarship. In all other instances,
however, permission must be obtained from the author.

\vskip 1.5cm


$~$\hspace{10.2cm} Ravi Kunjwal
\newpage
\cleardoublepage
~
\vskip 1.2cm
\centerline{{\bf{\large{DECLARATION}}}}
\vskip 1.2cm
I, hereby declare that the investigation presented in the thesis has been
carried out by me. The work is original and has not been submitted
earlier in whole or in part for a degree / diploma at this or any
other Institution / University.
\vskip 2.0cm
%
%
\rightline{Ravi Kunjwal \hspace{0.9cm}}
\newpage
\cleardoublepage
~
\vskip 1.0cm
\centerline{{\bf{\large List of Publications arising from the thesis}}}
\vskip 0.5cm
{\bf Journal:}
\begin{enumerate}
\item {\bf Kunjwal, R.} and Ghosh, S. (2014).
  {\em Minimal state-dependent proof of measurement contextuality for a qubit}. \\
  \href{http://journals.aps.org/pra/abstract/10.1103/PhysRevA.89.042118}{Physical Review A, {\bf 89}, 042118}.
  \href{http://arxiv.org/abs/1305.7009}{arXiv Preprint, 1305.7009}.
\item {\bf Kunjwal, R.}, Heunen, C., and Fritz, T. (2014).
  {\em Quantum realization of arbitrary joint measurability structures}. \\
  \href{http://journals.aps.org/pra/abstract/10.1103/PhysRevA.89.052126}{Physical
  Review A, {\bf 89}, 052126}.
  \href{http://arxiv.org/abs/1311.5948}{arXiv Preprint, 1311.5948}.

\item {\bf Kunjwal, R.} (2015).
{\em Fine's theorem, noncontextuality, and correlations in Specker's scenario}. \\
\href{http://journals.aps.org/pra/abstract/10.1103/PhysRevA.91.022108}
     {Physical Review A, {\bf91}, 022108}.
\href{http://arxiv.org/abs/1410.7760}{arXiv Preprint, 1410.7760}.

\item {\bf Kunjwal, R.} and Spekkens, R.W.
(2015).
{\em From the Kochen-Specker Theorem to Noncontextuality Inequalities without Assuming Determinism}. \\
\href{http://journals.aps.org/prl/abstract/10.1103/PhysRevLett.115.110403}
{Physical Review Letters, {\bf 115}, 110403}.
\href{http://arxiv.org/abs/1506.04150}{arXiv Preprint, 1506.04150}.

\item Mazurek, M.D., Pusey, M.F., {\bf Kunjwal, R.}, Resch, K.J., and Spekkens, R.W. (2015).
{\em An experimental test of noncontextuality without unphysical idealizations}.\\
\href{http://dx.doi.org/10.1038/ncomms11780}{Nat.~Commun.~7, 11780 doi: 10.1038/ncomms11780 (2016).}
\href{http://arxiv.org/abs/1505.06244}{arXiv Preprint, 1505.06244}.
\end{enumerate}

{\bf arXiv:}
\begin{enumerate}
 \item {\bf Kunjwal, R.} (2014).
 {\em A note on the joint measurability of POVMs and its implications for contextuality}.\\
 \href{http://arxiv.org/abs/1403.0470}{arXiv Preprint, 1403.0470}.
\end{enumerate}
\rightline{Ravi Kunjwal \hspace{0.9cm}}

\newpage
\cleardoublepage
\centerline{\bf{\large Contributed Talks \& Seminars}}
\vskip 0.5cm
{\bf Contributed talks at Conferences \& Workshops}
\begin{itemize}
 \item {\em A minimal state-dependent proof of measurement contextuality for a qubit} \\-- Ravi Kunjwal and Sibasish Ghosh.
 
 Contributed (short) talk (Parallel Session B) at the 13th Asian Quantum Information Science Conference \href{https://www.imsc.res.in/~aqis13/detailed_program.html}{(AQIS, 2013)}, IMSc Chennai, on August 27, 2013.
 
 (No conference proceedings were published.)

 \item {\em Fine's theorem, noncontextuality, and correlations in Specker's scenario} \\-- Ravi Kunjwal. 
 
 Contributed talk at \href{http://www.hri.res.in/~youqu15/youqu15/youqu15-prog.html}{Young Quantum 2015}, Harish-Chandra Research Institute (HRI), Allahabad, on February 24, 2015.
 
 (No conference proceedings were published.)

 \item {\em From the Kochen-Specker theorem to noncontextuality inequalities without assuming determinism} \\-- Ravi Kunjwal and Robert W.~Spekkens.
 
 Contributed talk at \href{http://www.cs.ox.ac.uk/qpl2015/schedule.pdf}{Quantum Physics and Logic 2015}, University of Oxford, on July 15, 2015. Video available at \url{https://goo.gl/nAvdW2}.
 
 (Not published in \href{http://eptcs.web.cse.unsw.edu.au/content.cgi?QPL2015}{Proceedings of the 12th International Workshop on Quantum Physics and Logic}, which only included 
 ``long original contributions''. My contributed talk was a ``short contribution'' to be published elsewhere, although a synopsis was included in the \href{http://www.cs.ox.ac.uk/qpl2015/preproceedings/50.pdf}{Preproceedings}.
 It was later published as a Letter in \href{http://journals.aps.org/prl/abstract/10.1103/PhysRevLett.115.110403}{Physical Review Letters, 115, 110403 (2015)}.)
 
 \item {\em Noncontextuality inequalities for Specker's compatibility scenario} \\-- Ravi Kunjwal and Robert W.~Spekkens
 
 Contributed Long Talk at \href{}{Quantum Physics and Logic 2016}, University of Strathclyde (Glasgow, Scotland), on June 8, 2016. Slides available at \url{http://qpl2016.cis.strath.ac.uk/pdfs/6ravi.pdf}.
 Video available at \url{https://goo.gl/WXdNHo}.
\end{itemize}

{\bf Seminars}
\begin{itemize}
 \item {\em Noncontextuality without determinism and admissible (in)compatibility relations: revisiting Specker's parable.}

Quantum Foundations Seminar at the Perimeter Institute for Theoretical Physics, Canada, on January 14, 2014. Video available at \url{http://pirsa.org/14010102}.

\end{itemize}
\vskip 2.0cm
\rightline{Ravi Kunjwal \hspace{0.9cm}}

\newpage
\cleardoublepage
\centerline{{\bf {\large DEDICATIONS}}}
\vskip 0.5cm
\centerline{\emph{To Ma, Papa, and Hina,}}
\centerline{\emph{for the years they have spent wondering what I have been up to all this while.}\footnote{And the years that they still may.}}
\centerline{\emph{This is it.}}

\newpage

\cleardoublepage
~
\vskip 1.0cm
\centerline{{\bf{\large ACKNOWLEDGEMENTS}}}
\vskip 0.5cm
\begin{quotation}
\centering
{\em The pursuit of science is at its best when
it is a part of a way of life.}

-- Alladi Ramakrishnan, January 3, 1962.
\end{quotation}

A lot of people and experiences and ideas go into the making of new ideas. This thesis is no exception and owes much to 
everyone who has -- by choice, chance, or circumstance -- played a role in bringing me towards its completion. 

Sibasish has been a friend and guide through all my PhD years at Matscience and I am grateful to him for taking out 
the time. The intellectual independence that I think I have been able to develop over these years has only been possible because 
as my PhD advisor, he trusted me enough to let me follow my instincts. My instincts led me to Rob, who in matters of my approach
to research has been my `spiritual' advisor. His passion for sorting out foundational issues in quantum theory 
is inspiring, so is his no-nonsense approach to the field. My visits to Perimeter Institute (PI) were possible only because he took a chance inviting an unknown grad student from India. 
This in turn would probably not have happened without Q+ Hangouts run by Matt Leifer and Daniel Burgarth,
where Rob gave a seminar which motivated me to write to him. Matt Leifer has been an invaluable mentor and I was glad to 
have interacted with him in PI. My thanks to all my co-authors who have played a key role in much of the 
work presented in this thesis: Sibasish Ghosh, Rob Spekkens, Tobias Fritz, Chris Heunen, Matt Pusey, Mike Mazurek,
and Kevin Resch. 

In Matscience, Prabha was always a constant source of support, academic and otherwise, particularly in the `tumultous'
starting years of my PhD when I did not know where I was headed in my research. Thank you for being the {\em sooper akka}
you have been, Madras first began to feel like home only because of you and Krishna. Rajarshi provoked me often enough into long
rants on academia in general and quantum foundations in particular and I can't say I haven't enjoyed being heard. Prof.~Simon
helped me ease into academia by having me on a joint project with Rajeev which led to my first published paper. The grilling he 
put one through in seminars or journal club talks did a lot to boost my confidence in seminars I have given elsewhere. 
Prof.~Mukunda -- 57 years my senior from St.~Stephen's, as I later discovered -- commended 
my very first journal club talk, which did much to motivate me as a fledgling graduate student. Such are his eclectic interests that 
conversations with him could range all the way from Heisenberg to Hemingway to Haydn. I am also grateful to Andreas Winter
for the encouragement and the conversations, academic and otherwise, during his visits to Matscience. Thanks are also due to 
Jonathan Barrett and Guruprasad Kar for their critical reading of the thesis and the resulting feedback.

I owe much to the wonderful people in the Matscience administrative office, whether it was Mr.~Vishnu Prasad's easy accessibility on 
matters of concern or Ms.~Indra's invaluable help in getting things in order for academic visits abroad. The general ambience of
goodwill and cooperation from the office made everything very smooth, whether it was the organization of ICA events, Open Day, or 
conferences. The facilities at IMSc, in particular the library and the new sports complex, 
have been fantastic. I spent many quiet hours in the library, poring over the wide collection we have, and I was also glad that
almost every book I suggested was procured by the library quite promptly. I also thank my friends and colleagues in IMSc, everyone from officemates to flatmates, for breaking the monotony now and then. Among others, let me thank
Subhadeep, Dibya, Archana, Arya, Rajeev, Somdeb, Sriluckshmy, KK, Madhusudhan, Prathik, Naveen, Prathamesh, and Madhushree. Let me also thank faculty members 
and visitors I spent many hours chatting with over food or coffee, including Vikram, Jam, Kamal, Simon (Kramer),
Sunder, Chandru, Vani, Murthy, Baskaran, Rahul Siddharthan, Ronojoy, Shankar, Hassan, and Gautam.
Outside of IMSc, I am grateful to many who made my time in Chennai fulfilling in many ways, such as the SSTCN folks I walked with along 
the beach over many nights: Arun, Akhila, John, Maya, Nishanth, Sowjanya, among others. Arun also invited me to visit him in Marudam and the 
few days I spent there were beautiful. I want to thank Sowjanya also for the cakes, kaleidoscopes, and walks in the Andamans.
Annual Chennai visits from `expats' like hummingbird-chaser Anusha and jazz aficionado Nitya were refreshing, and so were the 
endless hours spent talking physics and music with Vilasini or trekking under moonlight and loitering around Croc Bank
with her and Athira. For the many quiet hours I spent listening to its song, often calming down my chaotic mind,
let me also thank the sea that made Madras so special to me.

Much is also owed to the folks at PI, where I spent a few happy months in the fall semesters, gorging on the food in the Bistro as well as
the physics in the building. Thanks to Rob, Lucien, Tobias, Matt Pusey, Matt Leifer, Josh, Katja, Anirudh, Isaac, Huangjun, and Elie for 
the many interesting conversations over the lunch table; to Heidar, Yangang, Yasaman, Anton, Mansour, Marco, Scarlett, Gabriel, Natacha, Pablo,
Markus, Farbod, Dalimil, Damian, Vasudev, Ryszard, Jacob, and Richard, for the times sunny and snowy in Canada. Matt Pusey, in particular, 
has been great help in filtering out a lot of my misguided ideas through his probing questions and sharp interjections to 
careless claims I sometimes made. Tobias has been the go-to guy for any confusions mathematical and has helped me see a lot of things 
with more clarity. Special thanks to Chris Heunen who supported my
visit to Oxford during QPL 2015 as well as my visit to Edinburgh before QPL 2016. I do hope we get back to the projects we planned, now that my excuse of 
being busy writing a thesis is no longer valid.

The groundwork for my PhD was laid during my undergraduate years at St.~Stephen's College and I would like to acknowledge
the contribution of Bikram Phookun, Abhinav Gupta, and Vikram Vyas in motivating me to take up physics after college.
AG's quantum lectures, in particular, were my first introduction to the subject and I am glad that worked out well, because 
that first introduction -- if done badly -- could have turned me away from quantum matters (and possibly towards more 
astrophysical ones). Mathew Syriac, a year senior to me,
also played an important role in turning me towards matters of quantum foundations from a quantum information perspective,
although he now does cosmology!
Thanks to Prof.~Hari Dass who supported my visit to IISc for a summer project while I was still an undergraduate student. 
Let me also thank the College friends I did physics with (or not), and later met in places as far apart as Delhi, Cambridge,
Chennai, New York, and Waterloo during my PhD: Rahul, Raghav, Harshant, Rajita, Priyam, Ila, Rachel, Varun, Bahul, Aotula, 
Aruna, Philip. Doing physics in college, as opposed to engineering, was a decision I took because of JG, whether he intended 
it or not.
JG kindled in me a desire to follow in the footsteps of Galileo, his hero, so let me also thank him for motivating a
17-year-old me to take up physics. My school teachers such as Mr.~Raturi and Mr.~Kumra played a role in shaping my interests in physics 
and mathematics, and I am grateful to them for that.

Finally, my family has stood by me through all these years despite the ``unconventional'' career choice I made, becoming
a first-generation physicist. They had their apprehensions, of course, but I suppose I was just too stubborn to listen to their
warnings of how tough a PhD can be, how a lot of people start and don't finish and waste away years,
and how it doesn't pay much. 
This thesis is also a testimony to my mother's hardwork who brought up my sister and me with 
a tough love that is only so familiar to many Indian children. I am grateful to my father for his unstinting support 
in everything I have done, whether it was changing my mind to do physics instead of engineering after school or going on to do
a PhD after college instead of targeting the IAS. He trusted my judgment on all counts. Let me also thank my sister, who has a `real job'
but has refrained from dismissing my occupation, for her love and support.

If I have missed anyone, my apologies, but know that my gratitude is not chained to the words on these pages. Thank you.
\vskip 2.0cm
\rightline{Ravi Kunjwal \hspace{0.9cm}}

\newpage
\cleardoublepage
\centerline{\bf{\large List of changes suggested by the Thesis and Viva Voce Examiners}}
\vskip 0.5cm
\begin{enumerate}

 \item Chapter 1: On page 14, footnote 4 has been added to discuss the case of ontological models where the ontic state space is 
 not finite, noting Hardy's excess baggage theorem \cite{hardyexcess} and explaining why, for our purposes in
 this thesis, there is no loss of generality in presuming a finite ontic state space.

 \item Chapter 1: On page 23, the phrase `outcome determinism for projective (sharp) measurements in quantum theory' has been 
 replaced by `outcome determinism for projective (sharp) measurements in ontological models of quantum theory'
 to add clarity. Similarly, on page 28, the phrase `to prove ODSM in quantum theory' has been replaced by 
 `to prove ODSM in ontological models of quantum theory'.

 \item Chapter 2: On page 49, first paragraph, `Spekkens generalized notion of noncontextuality' replaced by 
 `Spekkens' generalized notion of noncontextuality'.

 \item Chapter 3: On page 70, second paragraph, the definition of a POVM is amended by correcting the 
 erroneously typed condition $\sum_{X_i\in\mathcal{F}_i}M_i(X_i)=I$.
 
 \item Updated the `List of Publications arising from the thesis', citing the journal version of:\\\\
Mazurek, M. D. et al. {\em An experimental test of noncontextuality without unphysical idealizations.}
\href{http://dx.doi.org/10.1038/ncomms11780}{Nat.~Commun.~7:11780 doi: 10.1038/ncomms11780 (2016).}
Same update to Ref.~[19] in the Synopsis and Ref.~\cite{exptl} in the Thesis.
 \item Updated the list of `Contributed Talks \& Seminars' with the recent contributed talk, `Noncontextuality
 inequalities for Specker's compatibility scenario', at QPL 2016.
\end{enumerate}
\vskip 0.3cm
\hspace{1.0cm}{\bf Date:} 05/08/2016 \hfill Candidate: Ravi Kunjwal \hspace{1.0cm}
\vskip 0.4cm 
\hspace{1.0cm}{\bf Place:} IMSc, Chennai \hfill Guide: Prof.~Sibasish Ghosh \hspace{0.5cm}
\newpage
\frontmatter
\pagestyle{plain}
\tableofcontents
\listoffigures
\listoftables

\mainmatter

\addcontentsline{toc}{chapter}{Synopsis}
\newpage
\chapter*{Synopsis}
The Kochen-Specker (KS) theorem \cite{SKS67} is a fundamental result in the foundations of quantum theory showing that it is impossible 
to accommodate the predictions of quantum theory within a framework in which outcomes of measurements are pre-determined
in a noncontextual manner. Failure of such a noncontextual model in accommodating quantum theory is often called \emph{contextuality}
in quantum information and quantum foundations. Along with Bell's theorem \cite{SBell64,SBell66,SBell76}, the Kochen-Specker theorem 
is one of the two major no-go theorems in quantum foundations. While Bell's theorem has proven to have wide-ranging implications for 
quantum information \cite{SBellreview}, the KS theorem has remained largely of foundational interest owing to implicit idealizations that make its
experimental testability a matter of controversy \cite{SMKC}. Recent work, though, has provided strong evidence that 
contextuality drives quantum-over-classical advantages in information processing and computation \cite{Smagic}.
This makes it all the more important to address problems with the experimental testability of the KS theorem.

Since we consider contextuality \emph{beyond} the KS theorem in this thesis, we will refer to the notion of noncontextuality due to Kochen
and Specker as \emph{KS-noncontextuality} and its failure demonstrated by the Kochen-Specker theorem as \emph{KS-contextuality}.
We adopt a generalized notion of (non)contextuality due to Spekkens \cite{SgenNC}, motivating noncontextuality as 
an expression of the Leibnizian idea of the \emph{identity of indiscernibles} \cite{SLeibniz} applicable to any operational theory. 
This notion of noncontextuality removes the unmotivated assumption of \emph{outcome determinism}
in the Kochen-Specker theorem, namely the assumption that the outcomes of measurements are fixed deterministically (and noncontextually)
for a physical system before a measurement is carried out and it is this value that the measurement reveals. Only the probability
of the measurement outcome is assumed to be fixed noncontextually for a physical system in the Spekkens framework.
This thesis thus considers contextuality beyond the Kochen-Specker theorem in two ways:
\begin{enumerate}
 
 \item The Kochen-Specker framework is applicable only to sharp (or projective) measurements in quantum theory. We
 consider questions of contextuality for unsharp (or nonprojective) measurements in quantum theory, as these are 
 the ones that are typically implemented in practice in any experiment because of inevitable noise in the implementation.
 This goes beyond the Kochen-Specker theorem in the sense of allowing nonprojective measurements, albeit still assuming 
 that the operational theory of interest is quantum theory. As we will show, these nonprojective (unsharp) measurements in quantum theory 
 exhibit (in)compatibility relations that are impossible for projective measurements, allowing for considerations of contextuality in
 scenarios not envisaged by the KS theorem.
  
 \item We then show how to extend the applicability of contextuality from quantum theory (for which Kochen-Specker theorem holds) to 
more general operational theories called generalized probabilistic theories (GPTs). This allows a treatment of contextuality that
does not presume a quantum model of the experiment and lays the groundwork necessary for applications of contextuality to 
device-independent quantum information processing. From a foundational viewpoint, this strengthens the Kochen-Specker theorem
by turning it into experimentally robust incarnations. This also allows tests of contextuality outside the ambit of the experimental
scenarios envisaged by the KS theorem.
\end{enumerate}

The following conclusions can be made on the basis of work presented in this thesis:
\begin{enumerate}
 \item Specker's scenario - the simplest one capable of admitting contextuality with respect to joint measurement contexts - allows a proof of contextuality \`a la Spekkens on a qubit with nonprojective
 measurements \cite{SKG}.
 \item Quantum theory allows arbitrary joint measurability structures when considering the most general quantum measurements, as opposed to the restricted possibilities offered by projective
 measurements \cite{SKHF}.
 \item Fine's theorem does not absolve one of the need to justify outcome determinism in noncontextual ontological models of quantum theory \cite{SFine}.
 \item It is possible to rule out noncontextuality for arbitrary operational theories rather than quantum theory alone, following our 
 operationalization of the KS theorem. We provide criteria for the same \cite{SKunjSpek,Sexptl}. 
 \item Specker's scenario also admits theory-independent criteria for deciding contextuality and leads to a generalization of such criteria to all $n$-cycle scenarios.
\end{enumerate}

A chapter-wise summary of the thesis follows:\\

{\bf Chapter 1} of the thesis is an introduction to concepts that will be used throughout the rest of the thesis.
We introduce the framework of operational theories and ontological models, followed by the definition of noncontextuality
due to Spekkens that will be used in this thesis. We then review the Kochen-Specker theorem and Bell's theorem and discuss the gap between
these two theorems from a foundational perspective as well as the perspective of applications in quantum information.

{\bf Chapter 2} takes a first look at a problem motivated by Ernst Specker's parable of the overprotective seer \cite{SSpe60,SLSW}. 
The question that we seek to answer is 
whether it is possible to exhibit contextuality for three quantum observables which are pairwise jointly measurable but not triplewise so.
This is the simplest admissible scenario that can exhibit contextuality of the KS-type, where the context is a joint measurement context.
Since it is impossible to realize the pairwise-but-not-triplewise joint measurability for three sharp (or projective) measurements, we are 
forced to consider unsharp measurements (POVMs or positive operator-valued measures) for this scenario. 
In their modern rendition of Specker's parable \cite{SLSW}, where Liang, Spekkens, and Wiseman pose this question, they conjectured 
that witnessing contextuality for this scenario would not be possible even if POVMs are considered. 
We take up this conjecture as a challenge and settle the question of witnessing contextuality in the affirmative, providing explicit constructions of the POVMs that 
achieve this. This is the first step in this thesis where we go beyond the KS theorem by considering contextuality for POVMs without assuming outcome determinism
for them. This chapter is based on work done with Sibasish Ghosh \cite{SKG}.

{\bf Chapter 3} examines the relationship between joint measurability of general quantum measurements and its implications for demonstrating contextuality with respect to 
joint measurement contexts. In particular,
a subtle issue regarding the type of joint measurability required in Specker's scenario is discussed and clarified in this chapter, paving the way for the results of Chapter 7.
This chapter is based on a note that has appeared on arXiv \cite{SRK}.

{\bf Chapter 4} considers a question that is raised in Chapter 2 regarding the admissibility of ``funny'' joint measurability relations in quantum theory -- those that are 
not achievable with projective measurements alone. Since the pairwise-but-not-triplewise joint measurability relation is admissible for POVMs in quantum theory,
it is natural to ask whether POVMs can also realize other, more complicated, joint measurability relations for more than three observables. We show that 
POVMs can, in fact, realize {\em any} joint measurability relation at all, providing a constructive proof of the same. This establishes the richness of joint measurability 
relations admissible in quantum theory and opens the door to asking whether these can be exploited for some information-theoretic tasks where POVMs have an edge over projective 
measurements. This chapter is based on work done in collaboration with Chris Heunen and Tobias Fritz \cite{SKHF}.

{\bf Chapter 5} engages with the results of Chapter 2 from the perspective of Fine's theorem \cite{SFine}, adapted to the case of KS-noncontextual models. We show that Fine's theorem
does not absolve one of the need to justify outcome determinism in considerations of noncontextuality. In particular, relaxing outcome determinism does not restrict the 
outcome indeterministic models to just the factorizable ones, unlike the case of Bell's theorem. This leads us to conclude that the problem of ruling out noncontextuality cannot be 
reduced to a marginal problem, unlike the problem of ruling out local causality. We seek to highlight this fundamental gap between local causality and noncontextuality in this chapter,
contrary to claims in the literature that seek to unify the mathematical treatment of the two hypotheses via a reduction to the marginal problem. This chapter is based on 
work published in Ref.~\cite{Sfinegen}.

{\bf Chapter 6} revisits the Kochen-Specker theorem and casts it in strictly operational terms (without requiring the validity of quantum theory), taking a further 
step beyond the Kochen-Specker theorem than Chapter 2 (which presumed the validity of quantum theory). We obtain a robust noncontextuality inequality that can be experimentally
tested to rule out noncontextual models of experiments. Our operational approach thus resolves the difficulty of experimentally testing contextuality by going beyond 
the Kochen-Specker paradigm. Indeed, the Kochen-Specker paradigm is recovered in an idealized limit of noiselessness in the experiment, one which is not achievable in practice.
We also outline an experimental test of noncontextuality that considers a simpler scenario than the one we envisage in this reformulation of the Kochen-Specker theorem.
The theoretical techniques involved in the realization of this experimental test will find use in any other experimental test of noncontextuality and we briefly mention this approach.
This chapter is based on two collaborations, one with Rob Spekkens \cite{SKunjSpek} and another with Mike Mazurek, Matt Pusey, Kevin Resch, and Rob Spekkens \cite{Sexptl}.

{\bf Chapter 7} returns to Specker's parable and obtains a robust generalization of the LSW inequality that was shown to be violated in Chapter 2. The key difference is that we relax the 
assumption that quantum theory correctly models the experiment in deriving our noncontextuality inequalities. This leads in a natural way to $n$-cycle noncontextuality inequalities 
that are robust to noise and generalize the known KS inequalities for these contextuality scenarios. While Chapter 6 can be seen as an operationalized version of the 
state-independent proofs of contextuality based on KS-uncolorability (such as the original one in Refs.~\cite{SKS67} and \cite{SCab1}), Chapter 7 provides an approach to operationalizing state-dependent proofs of contextuality
(such as the ones in Chapter 2 and in Ref.~\cite{SKCBS}). This chapter is based on unpublished joint work with Rob Spekkens, an earlier version of which can be found in a PIRSA seminar \cite{SRKtalk}.

{\bf Chapter 8} concludes with a discussion of open questions and problems with some existing claims in the literature on contextuality. We also indicate 
possible directions for future research.


\clearpage

\chapter{Introduction}
\setlength{\epigraphwidth}{\textwidth}
\epigraph{\it Because this position seems to arouse fierce controversy, let me stress our motivation: if quantum theory were 
not successful pragmatically, we would have no interest in its interpretation. It is precisely because of the enormous success 
of the QM mathematical formalism that it becomes crucially important to learn what that mathematics means. To find a rational 
physical interpretation of the QM formalism ought to be considered the top priority research problem of theoretical physics;
until this is accomplished, all other theoretical results can only be provisional and temporary[...] But our present QM formalism
is not purely epistemological; it is a peculiar mixture describing
in part realities of Nature, in part incomplete human information about Nature all scrambled up
by Heisenberg and Bohr into an omelette that nobody has seen how to unscramble. Yet we think
that the unscrambling is a prerequisite for any further advance in basic physical theory. For, if we
cannot separate the subjective and objective aspects of the formalism, we cannot know what we
are talking about; it is just that simple.
}{E.T. Jaynes, Probability in Quantum Theory (1996).}

Although this thesis doesn't unscramble the Jaynesian omelette, it does contribute to the project
by providing quantitative criteria for ruling out ``realities of Nature'' that are defined as ``noncontextual'' in a framework 
(dubbed the ``ontological models framework'') where probabilities represent ``incomplete human information'' just as they do, for example, 
in classical statistical mechanics. 
As Jaynes points out, even as fundamental a distinction as the 
one between \emph{reality} and one's \emph{knowledge of reality} is difficult to make without tying oneself up in knots over
what those two things correspond to in the quantum formalism. One could, of course, question whether this distinction is really fundamental,
but that already entails going beyond what we understand about probabilities without even talking about quantum theory.
Questioning the distinction between ``reality'' and our ``knowledge of reality''
is a project that is outside the scope of this thesis. I will take this distinction for granted in motivating the ideas here.\footnote{
Note to the philosophically inclined: I have, of course,
resisted the urge to define what one means by ``reality'' so far, philosophically speaking. For the purpose of this thesis, however, such a general definition is 
not required, as we will restrict our attention to ``reality'', or \emph{ontic states}, as defined in the ontological models
framework (which we will come to shortly). It is perhaps not the only way one could conceptualize reality and, indeed, I believe that there should be better
ways of doing it given the ``unnatural'' constraints (requiring a fine-tuning or conspiracy on Nature's part) various no-go theorems, particularly those of Bell and Kochen-Specker, 
put on the ontological models framework.}

Shorn of all interpretational baggage that various 
physicists may carry (and disagree about), the minimal facts that everyone 
agrees on about quantum theory are those concerning its operational predictions. These facts constitute 
\emph{operational quantum theory}, understood as a manual for how the three basic experimental procedures or \emph{operations} -- preparations,
transformations, and measurements -- on a system are to be performed in the laboratory and the probabilities with which various 
measurement outcomes may occur, specified by the Born rule. 

In order to make statements about ``reality'' and one's ``knowledge of reality'', we will make these ideas precise in the 
ontological models framework \cite{harriganspekkens}. We will then ask of the posited ``reality'' a particularly natural property:
noncontextuality. Roughly, noncontextuality is 
the idea that one should not posit distinctions in reality that can never make any observable difference to our experience of it: all such distinctions 
are superfluous, playing no explanatory role, and should therefore be eliminated.
Using the ontological models framework, we will then work out the constraints -- \emph{noncontextuality inequalities} -- that noncontextuality places 
on the observable statistics in an operational theory. We will find that it is impossible to make sense of certain 
quantum statistics in the framework of a noncontextual ontological model. Experimental criteria for ruling out noncontextual 
explanations of experimental data -- independent of any reference to quantum theory -- will then be explicated. 

We will see that maintaining noncontextuality in the ontological models framework while doing justice to experimental data is impossible, particularly if the data is in good agreement with
quantum predictions. This leaves us with two options:
either consider contextual ontological models as serious candidates for a viable foundation, or \footnote{And this is the option I am inclined towards.}
if one considers something akin to noncontextuality to be an essential feature of any putative foundation for quantum theory, 
one is led to reject the viability of the ontological models framework for this purpose. In the latter case, one is confronted with the challenge of providing an alternative 
to the ontological models framework where noncontextuality - appropriately formalized - is not in conflict with quantum theory.

Notwithstanding its unsuitability for the purpose of providing a foundation for quantum theory -- particularly if one insists on noncontextuality -- a purpose that the ontological models framework \emph{does} serve is to allow us to 
cleanly identify ways in which quantum theory may be deemed \emph{nonclassical}\footnote{By ``nonclassicality'', we roughly mean features not admissible in a pre-quantum 
or ``classical'' theory of physics and which may therefore be responsible for the advantages that quantum theory permits in quantum information 
and computation. One hopes to distill the essence of quantum theory by formalizing notions of nonclassicality and investigating
their role in quantum information applications.} and ways of testing 
this nonclassicality experimentally. These notions of nonclassicality (e.g. Bell nonlocality\cite{Bellreview}) play an essential role in many modern applications to quantum information theory.
Contextuality, in particular, is a strong form of such nonclassicality \cite{robsclassification} 
and there is mounting evidence of the crucial role it plays in quantum information and computation \cite{magic, magic2}.

The following sections in this chapter provide definitions of the concepts needed to carry the analysis forward in the rest of the thesis.

\section{Operational theories and Ontological models}

In this section, we recall the framework of operational theories and ontological models \cite{harriganspekkens,genNC} that will be essential to our discussion of 
noncontextuality.

{\bf Operational theory} --- An operational theory is specified by $(\mathcal{P},\mathcal{M},p)$, where $\mathcal{P}$ is the set of preparation procedures,
$\mathcal{M}$ is the set of measurement procedures, and $p(k|M,P)\in[0,1]$ denotes the probability that outcome $k\in\mathcal{K}_M$ occurs
on implementing measurement procedure $M\in\mathcal{M}$ following a preparation procedure $P\in\mathcal{P}$ on a system.

{\bf Ontological model} --- An ontological model $(\Lambda,\mu,\xi)$ of an operational theory $(\mathcal{P},\mathcal{M},p)$ posits an ontic state space $\Lambda$
such that a preparation procedure $P$ samples the ontic states $\lambda\in\Lambda$ according to a distribution over $\Lambda$, $\mu(\lambda|P)\in [0,1]$ ($\lambda\in\Lambda$) where 
$\sum_{\lambda\in\Lambda}\mu(\lambda|P)=1$, and the probability of occurrence of a measurement outcome $[k|M]$ ($M\in\mathcal{M}$ and $k\in\mathcal{K}_M$)
for any $\lambda\in\Lambda$ is specified by the \emph{response function}
$\xi(k|M,\lambda)\in[0,1]$, where $\sum_{k\in\mathcal{K}_M}\xi(k|M,\lambda)=1$.\footnote{For simplicity, 
we have pretended that the set of 
ontic states, $\Lambda$, is finite. However, to be completely general,
$\Lambda$ can be any measurable space with a $\sigma$-algebra $\Sigma$ and the ontological model then specifies a
probability measure $\mu$ over $\Lambda$, $\mu:\Sigma\rightarrow[0,1]$, a $\sigma$-additive
function such that $\mu(\Lambda)=1$. In this case, all summations over $\Lambda$ would become integrals. Indeed, if the operational theory is quantum theory, then a physically tenable ontological
model {\em necessarily} requires an ontic state space $\Lambda$ with infinitely many ontic states (cf.~Hardy's excess baggage theorem \cite{hardyexcess}).
A fully rigorous measure-theoretic approach to ontological models in the context of the
$\psi$-ontic \slash  $\psi$-epistemic debate is, for example, taken in Ref.~\cite{mattnumopus}. The results concerning 
contextuality in this thesis would survive any such generalization: this is 
because the operational predictions that will be of interest to us concern prepare-and-measure experiments with finite sets of 
preparations and finite sets of 
measurements having finite sets of outcomes. Under the assumption of measurement noncontextuality, 
it is always possible to imagine their measurement statistics
as arising from preparations that sample from a discrete set of ontic states, where each ontic state is simply an 
extremal assignment of probabilities to the various measurement outcomes. 
The total number of such extremal assignments is finite, hence a finite set of ontic
states would suffice to reproduce the operational predictions of interest. 
Even if one is working with a continuous measurable space $\Lambda$, 
the extremal assignments of probabilities to measurement outcomes can always be thought of as partitioning the 
space $\Lambda$ into a finite number of non-overlapping regions (except possibly measure zero overlaps). Each region in 
such a partition can be thought of as an ontic state in a new coarse-grained ontological model 
with a discrete ontic state space $\Lambda_{\rm discrete}$. The preparations can then
be presumed to sample from $\Lambda_{\rm discrete}$. See, for example, our operationalization of the KS theorem in 
Chapter 6 where 146 ontic states (cf.~page 142) suffice. The operational contradiction with noncontextuality arises only when
preparation noncontextuality, in addition to measurement noncontextuality, is assumed.
Therefore, in contrast to cases of interest to us: 1) Hardy's excess baggage theorem \cite{hardyexcess} requires an infinite number of preparations and measurements
to show that a $\Lambda$ that is countable and finite won't suffice,
and 2) Leifer's review \cite{mattnumopus} needs the most general $\Lambda$ to be able to accommodate ontological models of quantum theory
in which the wavefunction itself is regarded as an ontic state.}

The following condition of empirical adequacy prescribes how the operational theory and its ontological
model fit together:
\begin{equation}
 p(k|M,P)=\sum_{\lambda\in\Lambda}\xi(k|M,\lambda)\mu(\lambda|P).
\end{equation}
That is, the probability of a certain measurement outcome given an ontic state is averaged with respect to the distribution over ontic states sampled by the preparation procedure $P$
to yield the operational probabilities predicted by the operational theory. In other words, the causal account of a prepare-and-measure experiment is the following:
\begin{enumerate}
 \item Preparation procedure $P$ outputs a physical system in ontic state $\lambda$ with probability $\mu(\lambda|P)$.
 \item The physical system in ontic state $\lambda$ is then input to the measurement procedure $M$ which yields outcome $k$ with probability $\xi(k|M,\lambda)$. 
\end{enumerate}
Hence, $\lambda$ causally mediates between the preparation and measurement procedures. Coarse graining over it yields the operational statistics seen in an experiment.
Note that we have said nothing about the experimental accessibility of $\lambda$, i.e., $\lambda$ is not necessarily a ``hidden variable'', its status depends on the particular
ontological model that specifies $\Lambda$.

{\bf Operational equivalence:}\\
\emph{Preparation procedures} --- Two preparation procedures, $P$ and $P'$, are said to be operationally equivalent (denoted $P\simeq P'$)
if no subsequent measurement procedure $M\in\mathcal{M}$ (with outcome set $\mathcal{K}_M$) 
yields different statistics for them, i.e., 
\begin{equation}\label{prepopequiv}
 \forall [k|M]\quad(M\in\mathcal{M}, k\in\mathcal{K}_M): p(k|M,P)=p(k|M,P').
\end{equation}
That is, nothing in the predictions of the operational theory distinguishes the two preparation procedures and we say that they belong the same operational equivalence class 
of preparation procedures. We call such an equivalence class of preparation procedures a \emph{preparation}. 

Any parameters that can distinguish between preparation procedures 
in a given operational equivalence class (that is, procedures corresponding to the same preparation) constitute the \emph{preparation context}. As we will see, \emph{preparation noncontextuality}
then indicates the idea that the ontological representation of a preparation procedure should depend only on the operational equivalence class to which it belongs. In particular, this 
representation should be independent of any \emph{preparation context}.

{\it Tomographically complete sets of measurements} --- We assume that the operational theory specifies tomographically complete sets of measurements that can be used to identify any preparation, i.e.,
the statistics of measurements in a tomographically complete set, $\mathcal{M}_{\rm tomo}\subseteq\mathcal{M}$, on a preparation is enough to infer the statistics of any other measurement in $\mathcal{M}$
on that preparation. The minimum cardinality of a tomographically complete set of measurements is also specified by the operational theory. We will assume this cardinality to be finite
for a meaningful analysis of contextuality later on.

This property of tomographic completeness of a finite set of measurements allows one to verify the operational equivalence of preparation procedures by doing only measurements in $\mathcal{M}_{\rm tomo}$.
Operational equivalence of preparation procedures relative to $\mathcal{M}_{\rm tomo}$ implies their operational equivalence relative to the full set of measurements $\mathcal{M}$ in the operational theory.
That is, we can restate the operational equivalence $P\simeq P'$ as:
\begin{equation}
 \forall [k|M]\quad(M\in\mathcal{M}_{\rm tomo}, k\in\mathcal{K}_M): p(k|M,P)=p(k|M,P'),
\end{equation}
which implies Eq.(\ref{prepopequiv}).

Without this property of tomographic completeness, verification of such operational equivalence requires all (potentially infinite) possible measurements in $\mathcal{M}$ 
specified by the operational theory, rendering any experimental test of noncontextuality infeasible for reasons that will become clear when we define noncontextuality.

{\it Measurement procedures} --- Two measurement events, $[k|M]$ and $[k'|M']$ (where $M,M'\in\mathcal{M}$, $k\in\mathcal{K}_M$, $k'\in\mathcal{K}_{M'}$), are said to be operationally equivalent (denoted $[k|M]\simeq [k'|M']$)
if no preceding preparation procedure yields different statistics for them, i.e., 
\begin{equation}\label{effectopequiv}
\forall P\in\mathcal{P}: p(k|M,P)=p(k'|M',P).
\end{equation}
That is, nothing in the predictions of the operational theory distinguishes the two measurement events and we say that they belong to the same equivalence class of 
measurement events. We call such an equivalence class of measurement events an \emph{effect}. If $[k|M]\simeq[k'|M']$ for all measurement events in the two measurement procedures $M$ and $M'$,
we say that the two measurement procedures belong to the same equivalence class of measurement procedures ($M\simeq M'$). We call such an equivalence class of measurement procedures a \emph{measurement}.

Any parameters that can distinguish between measurement procedures (or measurement events) in a given operational equivalence class constitute the \emph{measurement context}. 
As we will see, \emph{measurement noncontextuality} then indicates the idea that the ontological representation of a measurement procedure should depend only on the operational equivalence class
to which it belongs. In particular, this representation should be independent of any \emph{measurement context}.

{\it Tomographically complete sets of preparations} --- We assume that the operational theory specifies tomographically complete sets of preparations that can be used to identify any measurement, i.e.,
the statistics of this measurement for any other preparation can be inferred from its statistics for this tomographically complete set of preparations, $\mathcal{P}_{\rm tomo}\subseteq\mathcal{P}$. 
The minimum cardinality of a tomographically complete set of preparations is also specified by the operational theory. We will assume this cardinality to be finite.

This property of tomographic completeness of a finite set of preparations allows one to verify the operational equivalence of measurement procedures by using only preparations in $\mathcal{P}_{\rm tomo}$.
Operational equivalence of measurement procedures relative to $\mathcal{P}_{\rm tomo}$ implies their operational equivalence relative to the full set of preparations $\mathcal{P}$ in the operational theory.
That is, we can restate the operational equivalence $[k|M]\simeq[k'|M']$ as
\begin{equation}
\forall P\in\mathcal{P}_{\rm tomo}: p(k|M,P)=p(k'|M',P),
\end{equation}
which implies Eq.(\ref{effectopequiv}).

Without this property of tomographic completeness, verification of such operational equivalence requires all (potentially infinite) possible preparations in $\mathcal{P}$ 
specified by the operational theory, 
again rendering any experimental test of noncontextuality infeasible for reasons that will become clear when we define noncontextuality.

\subsection{Operational quantum theory}

Operational quantum theory specifies $(\mathcal{P},\mathcal{M},p)$ as follows: $\mathcal{P}$ corresponds to the set of positive semidefinite density operators with unit 
trace, 
\begin{equation}
\mathcal{P}\equiv\left\{\rho_P\in\mathcal{B}(\mathcal{H}) \big| \rho_P\geq0, \Tr\rho_P=1\right\}, 
\end{equation}
where $\mathcal{B}(\mathcal{H})$ is the set of bounded linear operators on a complex separable Hilbert space $\mathcal{H}$.
$\mathcal{M}$ corresponds to the set of positive operator-valued measures (POVMs), each $M\in\mathcal{M}$ given by 
\begin{equation}
M\equiv\left\{E_{k|M}\in\mathcal{B}(\mathcal{H}) \Bigg| k\in\mathcal{K}_M, E_{k|M}\geq 0, \sum_kE_{k|M}=I\right\},
\end{equation}
where $I$ is the identity operator on the Hilbert space $\mathcal{H}$. 
Finally, we have  
\begin{equation}
 p(k|M,P)=\Tr(E_{k|M}\rho_P),
\end{equation}
the Born rule which associates outcome probabilities to the effects $E_{k|M}$ in a POVM $M\in\mathcal{M}$, given the density operator $\rho_P$.

{\bf Qubit example:} The simplest quantum system is defined on a two-dimensional Hilbert space, $\mathcal{H}\cong\mathbb{C}^2$, with quantum states and effects specified as
\begin{equation}
 \rho=\frac{1}{2}(I+\vec{\sigma}.\vec{n}),
\end{equation}
and
\begin{equation}
 E=\alpha I+\vec{\sigma}.\vec{a},
\end{equation}
respectively, where $\vec{\sigma}\equiv(\sigma_x,\sigma_y,\sigma_z)$ is the vector of qubit Pauli matrices $\sigma_x\equiv\left(\begin{smallmatrix} 0 & 1\\ 1 & 0\end{smallmatrix}\right)$,
$\sigma_y\equiv\left(\begin{smallmatrix} 0 & -i\\ i & 0\end{smallmatrix}\right)$, and $\sigma_z\equiv\left(\begin{smallmatrix} 1 & 0\\ 0 & -1\end{smallmatrix}\right)$, $\vec{n}\equiv(n_x,n_y,n_z)$
is a vector with $|\vec{n}|\equiv\sqrt{n^2_x+n^2_y+n^2_z}\leq 1$, $\alpha$ is a real number and $\vec{a}=(a_x,a_y,a_z)$ is a vector with real entries. 

\emph{Tomographically complete set of measurements} --- A tomographically complete set of measurements on a qubit has cardinality $3$, as is clear from the three parameters $n_x=\Tr (\sigma_x \rho)$,
$n_y=\Tr (\sigma_y \rho)$, and $n_z=\Tr (\sigma_z \rho)$ that fix $\rho$. Once $\rho$ is inferred from a tomographically complete set of measurements such as $\{\sigma_x,\sigma_y,\sigma_z\}$,
the outcome probabilities of any other measurement on $\rho$ can also be inferred.

\emph{Tomographically complete set of preparations} --- A tomographically complete set of preparations on a qubit has cardinality $4$, as is clear from the four parameters that fix 
qubit effect $E$: $\alpha=\Tr (\frac{I}{2} E)$,
$a_x=\Tr (\rho_x E) -\alpha$, $a_y=\Tr (\rho_y E) -\alpha$, $a_z=\Tr (\rho_z E) -\alpha$, where $\rho_x\equiv\frac{1}{2}(I+\sigma_x)$, $\rho_y\equiv\frac{1}{2}(I+\sigma_y)$, $\rho_z\equiv\frac{1}{2}(I+\sigma_z)$.
Once $E$ is inferred from a tomographically complete set of preparations such as $\{I/2,\rho_x,\rho_y,\rho_z\}$, its outcome probabilities for any other preparation on the qubit can also be inferred.

\section{Noncontextuality: an instance of Leibniz's identity of indiscernibles}

We now have the necessary tools to define noncontextuality \`a la Spekkens\cite{genNC}. But before we go into the definition itself, let us motivate the methodological principle that underlies the definition. Following Spekkens\cite{genNC},
we view noncontextuality as an instance of Leibniz's Principle of the Identity of Indiscernibles\cite{Leibniz}, that is, the \emph{physical} identity of \emph{operational} indiscernibles.
As a prescription for construction of a physical theory, this principle states: do not introduce distinctions between experimentally indistinguishable phenomena in your model of reality.
In other words, every distinction posited in one's physical theory should imply a difference in the operational predictions of the theory. Otherwise such distinctions 
are physically meaningless.

{\bf Preparation noncontextuality} is the following assumption on the ontological model of an operational theory:
\begin{equation}
P\simeq P'\Rightarrow \mu(\lambda|P)=\mu(\lambda|P')\quad\forall \lambda\in\Lambda.
\end{equation}

{\bf Measurement noncontextuality} is the assumption that
\begin{equation}
[k|M]\simeq [k'|M']\Rightarrow \xi(k|M,\lambda)=\xi(k'|M',\lambda)\quad\forall \lambda\in\Lambda.
\end{equation}
Note that this statement of measurement noncontextuality at the level of measurement events extends to the case of measurement procedures when we have one-to-one operational equivalence between their
constituent measurement events.

In both instances -- preparations and measurements -- the principle underlying noncontextuality is the same: no operational difference implies no ontological difference. That is, any \emph{context}
associated with the preparation or measurement procedure -- corresponding to differences within an operational equivalence class -- should not be relevant to the ontological representation of 
that procedure. Only the operational equivalence class should be relevant for the ontological representation, not differences of context, hence the term \emph{noncontextuality}. Indeed, one can 
also define
transformation noncontextuality \cite{genNC} which applies the same principle to transformations, but since we are only interested in prepare-and-measure experiments with no time evolution in 
between the preparation and measurement, we will not need transformation noncontextuality. Besides, a transformation can always be composed with a preparation to define a new effective preparation on
which the measurement is done, or the transformation can be composed with a measurement to define a new effective measurement which is done on the preparation. Hence, preparation and measurement
noncontextuality suffice for our considerations.

\section{The Kochen-Specker (KS) theorem}
\subsection{Traditional noncontextuality} 
Historically, Kochen and Specker \cite{KS67} first studied contextuality in quantum theory, leading to their eponymous theorem. We will refer to the notion of noncontextuality 
the KS theorem rules out as \emph{KS-noncontextuality}.
Following Ref.\cite{barrettkent} we first state the Kochen-Specker theorem in the following general terms before viewing it from the perspective
of the generalized notion of noncontextuality \cite{genNC}:
\begin{theorem}
Consider a map $V:\mathcal{K}\rightarrow\mathbb{R}$, where $\mathcal{K}$ is a set of Hermitian operators that act on an $n$-dimensional Hilbert 
space and $V(A)$ lies in the spectrum of $A$ for all $A\in\mathcal{K}$, satisfying the following conditions:
\begin{eqnarray}
V(A+B)&=&V(A)+V(B),\\
V(AB)&=&V(A)V(B),\quad\forall A,B\in\mathcal{K} \text{ such that } [A,B]=0.
\end{eqnarray}
Such a map $V$ is called a {\em KS-colouring} of $\mathcal{K}$. If $n>2$, then there exist {\em KS-uncolourable} sets, i.e., sets $\mathcal{K}$
for which no KS-colouring exists.
\end{theorem}
That the set of {\em all} Hermitian operators in $n>2$ is KS-uncolourable is a corollary of Gleason's theorem \cite{Gleason, Bell66}.
Kochen and Specker provided the first example of a finite KS-uncolourable set \cite{KS67} of 117 vectors in $n=3$.

Let us see what KS-colouring means for a set of mutually orthogonal projectors forming a resolution of the identity: $\{\Pi_1,\Pi_2,\dots,\Pi_N\}$,
such that $\Pi_i\Pi_j=\delta_{ij}\Pi_i$ and $\sum_{i=1}^N\Pi_i=I$. A KS-colouring $V$ would then require $V(\Pi_i\Pi_j)=V(\Pi_i)V(\Pi_j)$
and $V(I)=\sum_{i=1}^NV(\Pi_i)$. Now $V(O)=0$ and $V(I)=1$ since both the null operator $O$ and the identity operator $I$
take values in their spectrum. This in turn means 
$V(\Pi_i)V(\Pi_j)=V(\Pi_i\Pi_j)=V(O)=0$ for all $i\neq j$ and $\sum_{i=1}^NV(\Pi_i)=V(I)=1$. That is, $\{\Pi_1,\Pi_2,\dots,\Pi_N\}$
are assigned values in the spectrum $\{0,1\}$ such that exactly one of them is assigned 1 and the rest are assigned 0.

\subsection{KS-noncontextuality \`a la Spekkens} 
We will now show how the notion of KS-noncontextuality (exemplified by KS-colourings above) fits in the Spekkens framework \cite{genNC} and how this notion is 
generalized in this framework. It is the generalized notion due to Spekkens \cite{genNC} (defined in the previous section) that we refer to when we use the term 
``noncontextuality'', unless otherwise specified.

KS-noncontextuality supplements the assumption of measurement noncontextuality with an \emph{additional} requirement, namely, \emph{outcome determinism}. {\bf Outcome determinism} is the condition
that all the response functions in the ontological model are deterministic over the ontic state space, i.e. $\xi(k|M,\lambda)\in\{0,1\}$ for all $M\in\mathcal{M}$, $k\in\mathcal{K}_M$, $\lambda\in\Lambda$.
Besides the assumption of outcome determinism, the KS theorem is restricted to a particular \emph{type} of measurement context, namely, those contexts which correspond to projective measurements commuting with a given 
projective measurement (we will call such measurement contexts, ``commutative contexts''): if $[A,B]=0$ and $[A,C]=0$, then $B$ and $C$ provide two different contexts for the measurement of $A$, where $A,B,C$ are 
Hermitian operators and $B,C$ do not necessarily commute.
On the other hand, noncontextuality in the Spekkens framework treats \emph{any} distinction between procedures in the same operational equivalence class as a \emph{context}, not restricting itself
only to commutative contexts characteristic of the KS theorem. To be clear:

\begin{enumerate}
 \item KS-noncontextuality concerns measurement noncontextuality and outcome determinism applied to measurement contexts corresponding to Hermitian operators (projective measurements) commuting with a given 
 Hermitian operator (projective measurement). When considered in terms of projectors, which are also Hermitian operators, 
 the contexts are other projectors in the various orthonormal bases in which a given projector might appear.  
 
 \item Noncontextuality \`a la Spekkens abandons (i) the assumption of outcome determinism, (ii) the restriction to projective 
 measurements, and (iii) restriction to commutative contexts (that is, a `context' defined by the Hermitian operators that commute
 with a given Hermitian operator). This means that for ontological models of quantum theory, besides abandoning outcome determinism, one can now meaningfully speak of noncontextuality for POVMs (positive operator-valued measures) with respect to their 
 various contexts. For example, joint measurability of POVMs is a broader notion than their commutativity, the latter implying the former but not 
 conversely.\footnote{By commutativity of two POVMs we mean that each element of one POVM commutes with every element of the other POVM.}
 Instead of ``commutative contexts'' for PVMs (projection-valued measures)
 or projective measurements, we can now consider ``compatible contexts'' for POVMs, where by ``compatibility'' we refer to the notion of joint measurability that we will discuss at length in Chapter 2.
 
 \item Besides, KS-noncontextuality makes {\em no} reference to a notion of preparation noncontextuality such as the one defined by Spekkens \cite{genNC}. 
On the other hand, KS-noncontextuality {\em arises} in the Spekkens framework {\em via} the assumption of preparation noncontextuality applied to ontological models of operational {\em quantum} theory, in particular
projective measurements in quantum theory. It is possible to justify outcome determinism for projective (sharp) measurements in 
ontological models of quantum theory by appealing to preparation noncontextuality
and the quantum mechanical fact that such measurements can be made perfectly predictable on preparations corresponding to their eigenstates. In this way, the conditions required for the 
KS theorem can be recovered starting from preparation and measurement noncontextuality, along with the operational fact of perfect predictability of projective measurements on their eigenstates.

 \item Noncontextuality \`a la Spekkens also applies outside of quantum theory: because no commitment is made as to the representation of the preparations and measurements in the operational theory,
 the generalized notion of noncontextuality can be applied to any operational theory. This permits an understanding of contextuality in theory-independent terms, something not possible within the 
 framework of the Kochen-Specker theorem, which is really a no-go theorem for quantum theory and does not seek to make theory-independent claims.
\end{enumerate}

Before we proceed further, let us look at a simple example of the KS theorem in action.
\begin{figure}
 \includegraphics[scale=0.5]{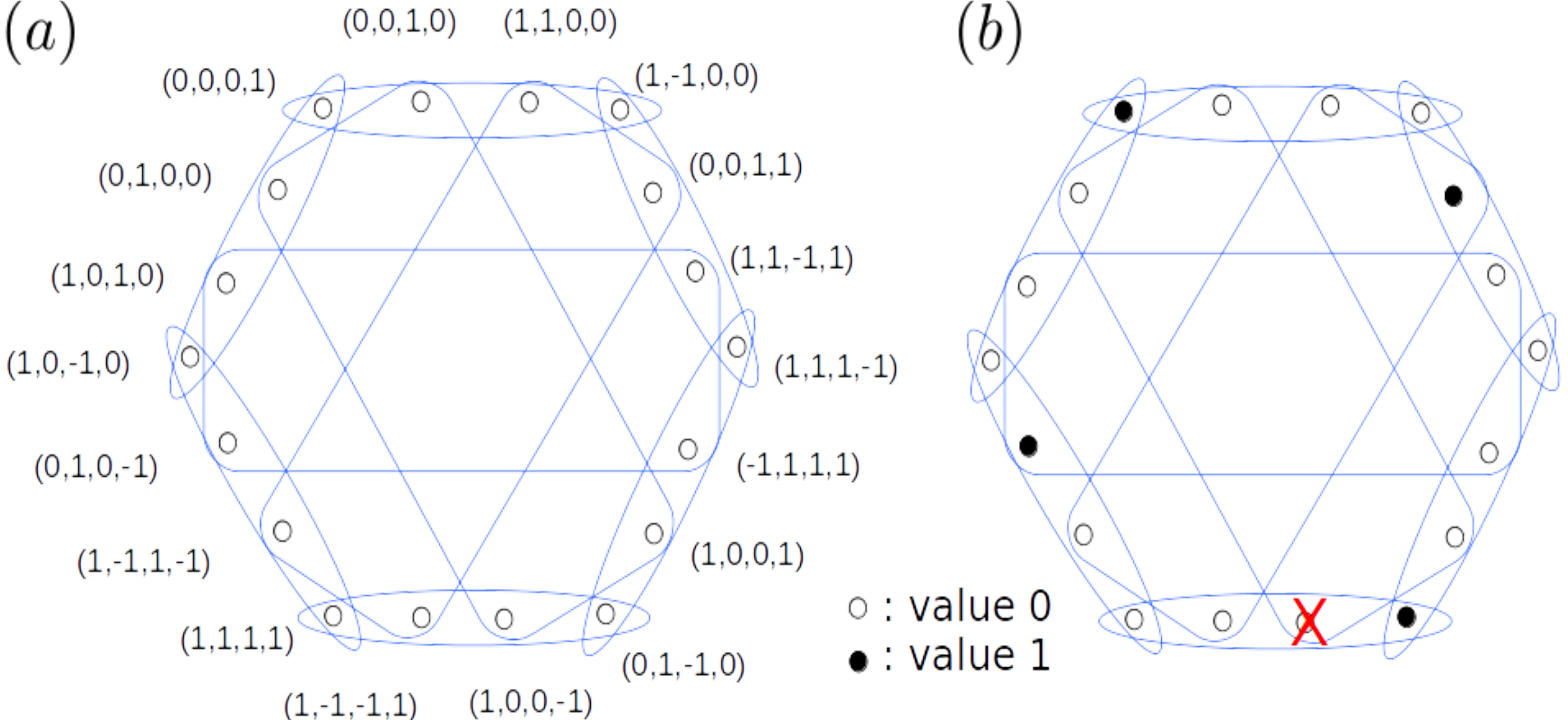}
 \caption{A simple example of the KS theorem in action, due to Cabello et al.~\cite{Cab1}. Figure (a) lists the vectors (and corresponding rank 1 projectors)
 associated with the nodes of the hypergraph, with the four nodes in each edge constituting an orthonormal basis. The 
 normalization factors are omitted to avoid clutter, but are presumed, so each vector should be appropriately normalized to a 
 unit vector.
 The black and white nodes in figure (b) denote value assignments $1$ and $0$ respectively.
 An attempt at such a value assignment is shown to fail. Indeed, \emph{any} such attempt will fail. }
\label{fig1}
\end{figure}
Figure \ref{fig1} shows such an example: the nodes in Fig.~\ref{fig1} denote measurement events in quantum theory, i.e. $18$ rays in a $4$-dimensional Hilbert space, each associated 
with a corresponding projector associated to a measurement event. The loops in Fig.~\ref{fig1} denote measurements, each consisting of four mutually exclusive and jointly exhaustive measurement events.
What this means is that the $4$ projectors in each loop correspond to an orthonormal basis in $\mathcal{H}\cong\mathbb{C}^4$. Thus, we have $18$ rays carved up into $9$ orthonormal bases of $4$ 
rays each. The (unnormalized) vectors associated with the rays are labelled in Fig.~\ref{fig1}(a). The normalization is omitted for clarity in the figure but we are imagining orthonormal bases in this example.

The operational equivalence between measurement events in Fig.~\ref{fig1} are implicit in the fact that every projector is shared between two orthonormal bases, each basis representing a \emph{context}
for the projector. The measurement event consists of a specification of the projector along with the orthonormal basis it's considered to be a part of, and we know that regardless of which 
orthonormal basis a projector appears in, the probability of occurrence of a measurement event associated with it is the same for any quantum state. 
This latter fact denotes the operational equivalence of the two measurement events corresponding to the same projector. 

Measurement noncontextuality now requires that in an ontological model
the physical state of the system $\lambda$ specifies the \emph{probabilities} of occurrence of measurement events independent of their context, i.e.,
for operationally equivalent measurement events, the same probabilities are assigned to them by $\lambda$. This condition of measurement noncontextuality on its own can be satisfied for 
the measurement events depicted in Fig.~\ref{fig1}. It just translates to being able to assign probabilities to the nodes in Fig.~\ref{fig1} in such a way that they add up to $1$ for each loop.
However, the Kochen-Specker contradiction arises when an additional requirement, besides measurement noncontextuality, is made: that the probabilities assigned by $\lambda$ be $\{0,1\}$-valued
or deterministic. Thus, instead of a measurement noncontextual assignment of {\em probabilities}, we now require a measurement noncontextual assignment of $\{0,1\}$ {\em values} by $\lambda$
to the measurement events in Fig.~\ref{fig1}. As depicted in Fig.~\ref{fig1}(b), any such attempt at a measurement noncontextual value-assignment -- a {\em KS-colouring} -- fails and therefore there cannot exist a 
$\lambda$ which makes such assignments of values, proving the Kochen-Specker theorem. To see why any such attempt would fail, it suffices to note the following: the fact that there are $9$ bases
means that the total number of projectors assigned value $1$ by $\lambda$ should be {\em odd}, since measurement in each basis must result in exactly one projector that occurs for a given $\lambda$; but since each projector
appears in two bases, each of those assigned value $1$ will appear in two bases, requiring an {\em even} number of projectors assigned value $1$, leading to the Kochen-Specker contradiction.
Such a proof of the KS theorem is called a {\em KS-uncolourability} proof.

The KS theorem therefore forces a choice between abandoning:
\begin{enumerate}
 \item measurement noncontextuality, or
 \item outcome determinism, or 
 \item both, measurement noncontextuality and outcome determinism.
\end{enumerate}

Traditionally, \emph{outcome determinism} has been taken for granted in ontological models of quantum theory and the conclusion then to be derived from the Kochen-Specker theorem 
is measurement contextuality (or, simply, ``contextuality'' in the Kochen-Specker framework). Alternatively, one may preserve measurement noncontextuality at the expense of outcome determinism
and conclude that any ontological model of quantum theory must admit ``intrinsic randomness'' or outcome indeterminism.

As an argument against the principle of noncontextuality due to Spekkens\cite{genNC}, however, the Kochen-Specker theorem fails: it's always possible to 
salvage measurement noncontextuality by abandoning outcome determinism. How then does the Spekkens approach recover or operationalize the Kochen-Specker theorem?
This is where preparation noncontextuality enters the picture: using the operational fact that projectors exhibit 
\emph{perfect predictability} when measured on the corresponding eigenstate, and the assumption of preparation noncontextuality, it is possible to infer outcome determinism for sharp (projective) measurements, or ``ODSM'',
in quantum theory. We sketch the argument below, based on Ref.~\cite{genNC}.

\begin{theorem}
Preparation noncontextuality (PNC) implies outcome determinism for sharp measurements (ODSM) in ontological models of quantum theory.
\end{theorem}

\begin{proof}
Given a set of rank-1 projectors $M\equiv\{\Pi_i\}$ constituting a PVM, i.e. $\Pi_i\Pi_j=\delta_{ij}\Pi_i$ and $\sum_i\Pi_i=I$, we can write down the corresponding set of pure state density operators $\{\rho_i\}$,
where $\rho_i=\Pi_i$ for all $i$. In the ontological model, the measurement events $\Pi_i$ are represented by the corresponding response functions $\xi(\Pi_i|\lambda)\in[0,1]$ and 
the preparations $\rho_i$ are represented by the corresponding distributions $\mu_i(\lambda)\geq0$, where $\sum_i\xi(\Pi_i|\lambda)=1$ for all $\lambda\in\Lambda$ and $\sum_{\lambda\in\Lambda}\mu_i(\lambda)=1$ for all $i$.
We first prove a lemma that will be used later:

\begin{lemma}\label{nonoverlap}
 If two preparations $P$ and $P'$ are distinguishable with certainty in a single-shot measurement, then the distributions representing them in the ontological model should be non-overlapping, i.e.
 \begin{equation}
  \mu(\lambda|P)\mu(\lambda|P')=0 \quad \forall \lambda\in\Lambda.
 \end{equation}
\end{lemma}
\begin{proof}
For $P$ and $P'$ to be distinguishable with certainty in a single-shot measurement, there must exist a measurement event $E$ in the operational theory such that $p(E|P)=1$ and $p(E|P')=0$.
The occurrence of $E$ identifies preparation $P$ and its non-occurrence identifies preparation $P'$. In the ontological model, this means
\begin{equation}
 \sum_{\lambda\in\Lambda}\xi(E|\lambda)\mu(\lambda|P)=1 \text{ and } \sum_{\lambda\in\Lambda}\xi(E|\lambda)\mu(\lambda|P')=0,
\end{equation}
so that
\begin{equation}
 \xi(E|\lambda)=\begin{cases}
                 1, &\text{ for all }\lambda\in\Lambda_P\\
                 0, &\text{ for all }\lambda\in\Lambda_{P'},
                \end{cases}
\end{equation}
where $\Lambda_P\equiv\{\lambda\in\Lambda|\mu(\lambda|P)>0\}$ and $\Lambda_{P'}\equiv\{\lambda\in\Lambda|\mu(\lambda|P')>0\}$. This implies that $\Lambda_P\cap\Lambda_{P'}=\phi$, because otherwise
there exist $\lambda\in\Lambda_P\cap\Lambda_{P'}$ such that $\xi(E|\lambda)=0$ and $\xi(E|\lambda)=1$, a contradiction. $\Lambda_P\cap\Lambda_{P'}=\phi$ implies that 
$\mu(\lambda|P)=0$ for all $\lambda\in\Lambda\backslash\Lambda_P\supseteq \Lambda_{P'}$ and 
$\mu(\lambda|P')=0$ for all $\lambda\in\Lambda\backslash\Lambda_{P'}\supseteq \Lambda_P$. We then have
\begin{equation}
\mu(\lambda|P)\mu(\lambda|P')=0 \quad\forall \lambda\in\Lambda,
\end{equation}
which is what we set out to prove.
\end{proof}

For any two orthogonal rank-1 density operators $\rho_k$ and $\rho_{k'}$ taken from the set $\{\rho_i\}$,
it is possible to distinguish them with certainty in a single-shot measurement using either the PVM $\{\Pi_k,I-\Pi_k\}$ or the PVM $\{\Pi_{k'},I-\Pi_{k'}\}$, since we have
\begin{equation}
 \Tr\rho_k\rho_{k'}=\delta_{k,k'}.
\end{equation}

From Lemma \ref{nonoverlap}, we can therefore conclude that 
\begin{equation}
\mu_k(\lambda)\mu_{k'}(\lambda)=0 \quad\forall \lambda\in\Lambda.
\end{equation}
This implies that $\Lambda_k\cap\Lambda_{k'}=\phi$ for all $k\neq k'$, where $\Lambda_k\equiv\{\lambda\in\Lambda|\mu_k(\lambda)>0\}$ and $\Lambda_{k'}\equiv\{\lambda\in\Lambda|\mu_{k'}(\lambda)>0\}$.
Given the fact that $\Tr\Pi_k\rho_{k'}=\delta_{k,k'}$, we have in an ontological model
\begin{equation}
 \sum_{\lambda}\xi(\Pi_k|\lambda)\mu_{k'}(\lambda)=\delta_{k,k'},
\end{equation}
so that 
\begin{equation}
 \xi(\Pi_k|\lambda)=\begin{cases}
                     & 1 \text{ for all }\lambda\in\Lambda_k,\\
                     & 0 \text{ for all }\lambda\in\bigcup_{k'\neq k}\Lambda_{k'},
                    \end{cases}
\end{equation}
which is equivalent to 
\begin{equation}
 \xi(\Pi_k|\lambda)\xi(\Pi_{k'}|\lambda)=\delta_{k,k'} \quad \forall \lambda\in\cup_i\Lambda_i,
\end{equation}
proving outcome determinism for rank-1 projectors $\{\Pi_i\}$ over the ontic support of corresponding quantum states $\{\rho_i\}$.
In order to prove ODSM in ontological models of quantum theory, we now show that $\cup_i\Lambda_i=\Lambda$, which establishes outcome determinism over the full set of ontic states $\Lambda$. Here $\Lambda$ 
is the set of ontic states $\lambda$ which quantum states can sample from, i.e. $\lambda$ such that for every $\lambda\in\Lambda$, 
there exists a quantum state $\rho$ represented by $\mu_{\rho}(\lambda)>0$ in the ontological model, i.e.
\begin{equation}
 \Lambda=\{\lambda|\mu_{\rho}(\lambda)>0 \text{ for some }\rho\}.
\end{equation}

To show $\cup_i\Lambda_i=\Lambda$, note that 
\begin{equation}
 \frac{1}{d}\sum_i\rho_i=\frac{I}{d},
\end{equation}
an operational equivalence between a uniform mixture of states in $\{\rho_i\}$ and the maximally mixed state $\frac{I}{d}$ on a $d$-dimensional Hilbert space. Using the assumption of 
preparation noncontextuality applied to this operational equivalence, we have
\begin{equation}
 \frac{1}{d}\sum_i\mu_i(\lambda)=\mu_{\frac{I}{d}}(\lambda)\quad\forall \lambda\in\Lambda.
\end{equation}
This means that every ontic state in the support $\cup_i\Lambda_i$ also appears in the support $\Lambda_{I/d}$ and conversely, i.e. $\cup_i\Lambda_i=\Lambda_{I/d}$, where 
$\Lambda_{I/d}\equiv\{\lambda\in\Lambda|\mu_{I/d}(\lambda)>0\}$.

Now, since every quantum state $\rho$ appears in some convex decomposition of $\frac{I}{d}$, it follows that $\Lambda_{\rho}\subseteq\Lambda_{I/d}$ for all $\rho$, where 
$\Lambda_{\rho}\equiv\{\lambda\in\Lambda|\mu_{\rho}(\lambda)>0\}$. An example of such a convex decomposition of $I/d$ is the
following:
\begin{equation}
 \frac{1}{d}\rho+\left(1-\frac{1}{d}\right)\rho'=\frac{I}{d},
\end{equation}
where $\rho'=\frac{I-\rho}{d-1}$. Preparation noncontextuality applied to this operational equivalence requires that 
\begin{equation}
 \frac{1}{d}\mu_{\rho}(\lambda)+\left(1-\frac{1}{d}\right)\mu_{\rho'}(\lambda)=\mu_{\frac{I}{d}}(\lambda)\quad\forall \lambda\in\Lambda,
\end{equation}
so that $\Lambda_{\rho}\subseteq\Lambda_{I/d}$.

Since $\Lambda_{\rho}\subseteq\Lambda_{I/d}$ is true for any quantum state $\rho$, it follows that $\cup_{\rho}\Lambda_{\rho}=\Lambda_{I/d}$. But note that $\cup_{\rho}\Lambda_{\rho}$
is just the set of all ontic states which quantum states can sample from, i.e. $\cup_{\rho}\Lambda_{\rho}=\Lambda$. We therefore have $\Lambda_{I/d}=\Lambda$ and 
\begin{equation}
 \cup_i\Lambda_i=\Lambda,
\end{equation}
implying outcome determinism for rank-1 projectors in quantum theory,
\begin{equation}
 \xi(\Pi_k|\lambda)\xi(\Pi_{k'}|\lambda)=\delta_{k,k'} \quad \forall \lambda\in\Lambda,
\end{equation}
or $\xi(\Pi_k|\lambda)\in\{0,1\}$ for all $\lambda\in\Lambda$.

Any sharp measurement or PVM in quantum theory can be obtained by coarse graining rank-1 projectors and since the response functions of rank-1 projectors are deterministic,
any coarse-grainings corresponding to elements of a PVM will also be deterministic. This proves that PNC$\Rightarrow$ODSM in ontological models of quantum theory.
\end{proof}

We now prove a strengthening of the previous claim that, for ontological models of quantum theory, PNC$\Rightarrow$ODSM. It is, in fact, possible to show that PNC not only implies ODSM, it also
implies measurement noncontextuality (MNC), thus recovering the Kochen-Specker notion for noncontextuality for ontological models of quantum theory. 

\begin{theorem}
Preparation noncontextuality (PNC) implies KS-noncontextuality in ontological models of quantum theory.
\end{theorem}

\begin{proof}
Consider two PVMs, $\tilde{M}$ and $\tilde{M'}$ on a $d$-dimensional Hilbert space. Let $M$ be a coarse-graining of $\tilde{M}$ and $M'$ a coarse-graining of $\tilde{M'}$, such that 
both $M$ and $M'$ have binary outcomes labelled by $\{0,1\}$ and they are operationally equivalent, i.e. $M\simeq M'$ or
\begin{equation}
 p(k|M,\rho)=p(k|M',\rho)\quad\forall \rho,
\end{equation}
where $\rho$ is a $d$-dimensional quantum state and $k\in\{0,1\}$. MNC would then require that $\xi(k|M,\lambda)=\xi(k|M',\lambda)$ for all $\lambda\in\Lambda$ but we want to derive MNC starting
from PNC. Since the operational theory is quantum theory and since $M$ and $M'$ are PVMs (because they are coarse-grainings of PVMs), we have that for $M$ there exist two orthogonal preparations
$\rho_0$ and $\rho_1$ such that 
\begin{equation}
 p(k|M,\rho_k)=1,\quad k\in\{0,1\}.
\end{equation}
Since $M\simeq M'$, it's also the case that  $p(k|M',\rho_k)=1, k\in\{0,1\}.$ This is easy to see since, in general, $M$ can be represented by PVM elements of the form 
$[0|M]=\sum_{i=1}^m|\alpha_i\rangle\langle\alpha_i|$
and $[1|M]=\sum_{i=m+1}^d|\alpha_i\rangle\langle\alpha_i|$, where $1\leq m\leq d$ and $\{|\alpha_i\rangle\}_{i=1}^d$ is an orthonormal basis in the $d$-dimensional Hilbert space so that 
$\sum_{i=1}^d|\alpha_i\rangle\langle\alpha_i|=I$. Similarly, $M'$ can be represented through some other orthonormal basis $\{|\alpha'_i\rangle\}_{i=1}^d$ such that the presumed 
operational equivalences hold. Note that since $\tilde{M}$ and $\tilde{M'}$ are PVMs, they can either be the maximally fine-grained PVMs $\{|\alpha_i\rangle\langle\alpha_i|\}_{i=1}^d$ and 
$\{|\alpha'_i\rangle\langle\alpha'_i|\}_{i=1}^d$ respectively, or be obtainable from a coarse-graining of these maximally fine-grained PVMs. 

The preparation $\rho_0$ then has to lie in the subspace spanned by $\{|\alpha_i\rangle\langle\alpha_i|\}_{i=1}^m$ and the preparation $\rho_1$ in the subspace 
spanned by $\{|\alpha_i\rangle\langle\alpha_i|\}_{i=m+1}^d$ for $p(k|M,\rho_k)=1$, where $k\in\{0,1\}$, to hold. We will assume $\rho_0=\frac{1}{m}\sum_{i=1}^m|\alpha_i\rangle\langle\alpha_i|$
and $\rho_1=\frac{1}{d-m}\sum_{i=m+1}^d|\alpha_i\rangle\langle\alpha_i|$. Given the perfect predictability $p(k|M,\rho_k)=1$ and $p(k|M',\rho_k)=1$ for $k\in\{0,1\}$, it follows that 
\begin{eqnarray}\label{eq:perfpred}
&&\sum_{\lambda\in\Lambda}\xi(0|M,\lambda)\mu(\lambda|\rho_0)=1,\nonumber\\
&&\sum_{\lambda\in\Lambda}\xi(1|M,\lambda)\mu(\lambda|\rho_1)=1,\nonumber\\
&&\sum_{\lambda\in\Lambda}\xi(0|M',\lambda)\mu(\lambda|\rho_0)=1,\nonumber\\
&&\sum_{\lambda\in\Lambda}\xi(1|M',\lambda)\mu(\lambda|\rho_1)=1.
\end{eqnarray}
Since we have already established PNC$\Rightarrow$ODSM, we have $\xi(k|M,\lambda)$ and $\xi(k|M',\lambda)\in\{0,1\}$ for all $\lambda\in\Lambda$. We can now define 
\begin{eqnarray}
 \Lambda_{M,k}\equiv\{\lambda\in\Lambda|\xi(k|M,\lambda)=1\},\nonumber\\
 \Lambda_{M',k}\equiv\{\lambda\in\Lambda|\xi(k|M',\lambda)=1\}.
\end{eqnarray}
Since $\xi(1|M,\lambda)=1-\xi(0|M,\lambda)$ and similarly for $M'$, we have 
\begin{eqnarray}\label{eq:supportunion}
&\Lambda_{M,0}\cup\Lambda_{M,1}=\Lambda,\nonumber\\
&\Lambda_{M',0}\cup\Lambda_{M',1}=\Lambda.
\end{eqnarray}

Given ODSM, it's also the case that
\begin{eqnarray}\label{eq:supportintersection}
&\Lambda_{M,0}\cap\Lambda_{M,1}=\varnothing,\nonumber\\
&\Lambda_{M',0}\cap\Lambda_{M',1}=\varnothing.
\end{eqnarray}
It follows that $\xi(k|M,\lambda)=\xi(k|M',\lambda)\forall\lambda\in\Lambda$, where $k\in\{0,1\}$, if and only if $\Lambda_{M,0}=\Lambda_{M',0}$. We will now show that $\Lambda_{M,0}=\Lambda_{M',0}$.
From Eq.~\ref{eq:perfpred} we have 
\begin{eqnarray}\label{eq:overlaps}
&\sum_{\lambda\in\Lambda_{M,0}}\mu(\lambda|\rho_0)=1\nonumber\\
&\sum_{\lambda\in\Lambda_{M,1}}\mu(\lambda|\rho_1)=1\nonumber\\
&\sum_{\lambda\in\Lambda_{M',0}}\mu(\lambda|\rho_0)=1\nonumber\\
&\sum_{\lambda\in\Lambda_{M',1}}\mu(\lambda|\rho_1)=1.
\end{eqnarray}
From Eqs.~\ref{eq:supportunion}, \ref{eq:supportintersection}, and \ref{eq:overlaps}, we have
\begin{eqnarray}
 \sum_{\lambda\in\Lambda_{M,0}}\mu(\lambda|\rho_1)=0,\nonumber\\
 \sum_{\lambda\in\Lambda_{M',1}}\mu(\lambda|\rho_0)=0,
\end{eqnarray}
which implies that
\begin{equation}\label{eq:zerointegral}
 \sum_{\lambda\in\Lambda_{M,0}\cap\Lambda_{M',1}}\left(\frac{m}{d}\mu(\lambda|\rho_0)+\left(1-\frac{m}{d}\right)\mu(\lambda|\rho_1)\right)=0.
\end{equation}
Now note that 
\begin{equation}\
\frac{m}{d}\rho_0+\left(1-\frac{m}{d}\right)\rho_1=\frac{I}{d}.
\end{equation}
Since every quantum state $\rho$ appears in some convex decomposition of the maximally mixed state $\frac{I}{d}$, e.g. $\frac{I}{d}=p\rho+(1-p)\rho'$, where $\rho'=\frac{I-pd\rho}{d(1-p)}$,
the ontic support of $\rho$ is contained in the ontic support of $\frac{I}{d}$. Since $\Lambda$ is just the union of the ontic supports of all quantum states, it is the case that
it is equal to the ontic support of $\frac{I}{d}$. It then follows from preparation noncontextuality that 
\begin{equation}
 \frac{m}{d}\mu(\lambda|\rho_0)+\left(1-\frac{m}{d}\right)\mu(\lambda|\rho_1)=\mu(\lambda|I/d)\quad\forall\lambda\in\Lambda,
\end{equation}
so that the support of $\frac{m}{d}\mu(\lambda|\rho_0)+\left(1-\frac{m}{d}\right)\mu(\lambda|\rho_1)$ is $\Lambda$. This in turn means, following Eq.~\ref{eq:zerointegral}, that 
\begin{equation}
 \Lambda_{M,0}\cap\Lambda_{M',1}=\varnothing,
\end{equation}
which implies that
\begin{equation}
 \Lambda_{M,0}=\Lambda_{M',0}.
\end{equation}
We therefore have the conclusion we sought: given the operational equivalence $M\simeq M'$ and using preparation noncontextuality, we have shown that $\xi(k|M,\lambda)=\xi(k|M',\lambda)$ for 
all $\lambda\in\Lambda$, thus proving measurement noncontextuality (MNC). All in all, we have the conclusion
\begin{equation}
 \text{Preparation noncontextuality}\Rightarrow\text{KS-noncontextuality}.\footnote{The proof I have presented here is a generalization of an unpublished note 
due to R.~W.~Spekkens, which I thank him for sharing with me.
For an alternative proof that PNC$\Rightarrow$KS-noncontextuality for ontological models of quantum theory, see Leifer and Maroney\cite{leifermaroney}. We will not sketch this proof here since it makes use of
the notion of ``maximally psi-epistemic'' ontological models and we do not wish to introduce this aspect concerning the ``reality of the wavefunction'' 
here because the results presented in this thesis will not rely on these considerations. The curious reader may check Leifer and Maroney\cite{leifermaroney} and the references therein.}
\end{equation}\end{proof}

This result, then, should convince a skeptic of the relevance of the notion of preparation noncontextuality, which is not as \emph{ad hoc} as it may appear compared to KS-noncontextuality.\footnote{Particularly to a skeptic unconvinced 
by the argument that the motivation underlying preparation noncontextuality (the Leibnizian \emph{identity of indiscernibles}) is the same as that underlying measurement noncontextuality; if, on 
methodological grounds, one upholds measurement noncontextuality, then one must also uphold preparation noncontextuality.}

Every proof of the Kochen-Specker theorem is thus also a proof of preparation contextuality but not conversely.
Because of the strict implication PNC$\Rightarrow$KS-noncontextuality, preparation noncontextuality is a stronger notion of noncontextuality for ontological models of quantum theory than 
the traditional notion of KS-noncontextuality. It is indeed possible to rule out preparation noncontextuality without making any appeal to KS-contextuality, as shown in Ref.~\cite{genNC}.
\section{Bell's theorem}
Bell's theorem provides criteria (Bell inequalities) for ruling out a particular class of ontological models, traditionally called local hidden variable (LHV) models, for composite systems.
In particular, the predictions of quantum theory for composite systems rule out LHV ontological models of operational quantum theory. Further, violating Bell inequalities in an experiment rules out an LHV model of Nature itself,
rather than merely of the particular operational theory that we currently use to describe Nature. This means that any operational theory offering a 
putative replacement of quantum theory -- our current description of Nature -- would necessarily have to fail to admit an LHV model in order to account for experimental violations of 
Bell inequalities.

The theorem considers a bipartite system prepared by a source according to some distribution $\mu(\lambda|P)$ over the ontic states $\lambda\in\Lambda$, the preparation denoted by $P$. 
Each part of this bipartite system is sent to one of two parties\footnote{Bell's theorem is applicable beyond two parties as well, but it was first motivated by Einstein, Podolsky, and Rosen's
consideration of a two-party scenario\cite{EPR}. We illustrate it for this simplest two-party scenario, often also called
the Bell-CHSH scenario\cite{Bell64, chsh}},
Alice and Bob, who are spacelike separated from each other during each run of the experiment. In each run, the experiment involves 
each party performing a measurement on his/her part of the bipartite system, where the measurement performed is chosen uniformly randomly from two possibilities $\{0,1\}$, followed by recording the two-valued outcome $\{0,1\}$ of this measurement.
Alice and Bob implement several runs of the experiment to build up the statistics, $p(a,b|x,y;P)$. Here $x\in\{0,1\}$ denotes the two choices of measurements in Alice's lab,
$y\in\{0,1\}$ denotes the two choices available in Bob's lab, and $a,b\in\{0,1\}$ denote their respective outcomes for these measurement choices. $p(a,b|x,y;P)$ denotes the joint probability of Alice obtaining outcome 
$a$ when she performs measurement $x$ on her subsystem and Bob obtaining outcome $b$ when he performs measurement $y$ on his subsystem, where the bipartite system is prepared according to 
preparation $P$.

Note that since the parties are spacelike separated during the course of their measurements, there should be no signalling between them, i.e. it should not be possible for Alice (Bob) to infer 
the measurement setting of Bob (Alice) by just looking at her (his) local data $p(a|x;P)$ ($p(b|y;P)$). This no-signalling condition, implied by special relativity, 
can be expressed as 
\begin{eqnarray}
&p(a|x;P)=\sum_bp(a,b|x,y;P),\quad\forall y\in\{0,1\}, \forall a,x\in\{0,1\}\nonumber\\
&p(b|y;P)=\sum_ap(a,b|x,y;P),\quad\forall x\in\{0,1\}, \forall b,y\in\{0,1\}.
\end{eqnarray}

In the ontological models framework, we have
\begin{eqnarray}
p(a,b|x,y;P)&=&\sum_{\lambda\in\Lambda}\xi(a,b|x,y;\lambda)\mu(\lambda|P),\nonumber\\
&=&\sum_{\lambda\in\Lambda}\xi(a|x,y,b;\lambda)\xi(b|x,y;\lambda)\mu(\lambda|P),\nonumber\\
&=&\sum_{\lambda\in\Lambda}\xi(b|x,y,a;\lambda)\xi(a|x,y;\lambda)\mu(\lambda|P),\nonumber\\
\end{eqnarray}

Bell's assumption of \emph{local causality} which defines local hidden variable models then requires the following conditional independences given the ontic state $\lambda$ of the system
that is sampled from the source in a particular run of the experiment:

{\bf Parameter independence}, namely, the independence of one party's measurement outcome from the other party's measurement setting,
\begin{eqnarray}\label{eq:paramind}
\xi(b|x,y;\lambda)&=&\xi(b|y;\lambda),\nonumber\\
\xi(a|x,y;\lambda)&=&\xi(a|x;\lambda),\text{ and}
\end{eqnarray}
 
{\bf Outcome independence}, namely, the independence of one party's measurement outcome from the other party's measurement outcome,
\begin{eqnarray}\label{eq:outcomeind}
\xi(a|x,y,b;\lambda)&=&\xi(a|x,b;\lambda) \text{ (from parameter independence)}=\xi(a|x;\lambda),\nonumber\\
\xi(b|x,y,a;\lambda)&=&\xi(b|y,a;\lambda) \text{ (from parameter independence)}=\xi(b|y;\lambda).
\end{eqnarray}

These assumptions are motivated by the fact that spacelike separation should, in an ontological model, result in independence of the statistics in one lab from the 
statistics in the other lab if the full description of the system, its ontic state $\lambda$ sampled at the source, is specified. The correlations between the 
two labs arise purely due to ignorance of the exact $\lambda$ that is sampled from one run of the experiment to the other: we presume there exists some distribution
$\mu(\lambda|P)$ - characterizing our ignorance of $\lambda$ even when we know the preparation $P$ - according to which $\lambda\in\Lambda$ is sampled.\footnote{Viewed from
the perspective of causal explanations of correlations \cite{woodspekkens}, Bell's assumption of local causality envisages a causal structure where $\lambda$ 
provides a common cause explanation of the correlations between Alice and Bob's statistics since the other possibility of a causal explanation 
-- namely, a direct causal relation between variables in Alice and Bob's labs -- is not available on account of spacelike separation. Violation of a 
Bell inequality then rules out such a causal explanation without fine-tuning. See Wood and Spekkens \cite{woodspekkens} for this approach to Bell's theorem, which
takes its inspiration from Reichenbach's principle.}

The conjunction of parameter independence and outcome independence (called local causality) results in
\begin{equation}
 p(a,b|x,y;P)=\sum_{\lambda\in\Lambda}\xi(a|x;\lambda)\xi(b|y;\lambda)\mu(\lambda|P),
\end{equation}
a mathematical condition often called \emph{factorizability}, since it requires a factorization of the joint probability 
of measurement outcomes given measurement settings and the ontic state $\lambda$ of the system, i.e. $\xi(a,b|x,y;\lambda)=\xi(a|x;\lambda)\xi(b|y;\lambda)$.

{\bf Local causality $\Rightarrow$ no-signalling (but not conversely):} It is easy to see that local causality at the ontological level provides a natural account of no-signalling at 
the operational level, since 
\begin{eqnarray}
p(a|x;P)&\equiv&\sum_bp(a,b|x,y;P)\nonumber\\
&=&\sum_{\lambda\in\Lambda}\xi(a|x;\lambda)\sum_b\xi(b|y;\lambda)\mu(\lambda|P)\nonumber\\
&=&\sum_{\lambda\in\Lambda}\xi(a|x;\lambda)\mu(\lambda|P)
\end{eqnarray}
is obviously independent of the choice of $y$ on account of factorizability. Similarly, 
\begin{eqnarray}
p(b|y;P)&\equiv&\sum_ap(a,b|x,y;P)\nonumber\\
&=&\sum_{\lambda\in\Lambda}\sum_a\xi(a|x;\lambda)\xi(b|y;\lambda)\mu(\lambda|P)\nonumber\\
&=&\sum_{\lambda\in\Lambda}\xi(b|y;\lambda)\mu(\lambda|P)
\end{eqnarray}
is independent of $x$. However, notwithstanding the apparent plausibility of the \emph{local causality} hypothesis as the ontological counterpart of the special relativistic prohibition of 
faster-than-light propagation of physical influences, it is a matter of experiment whether the hypothesis holds up to scrutiny
when tested. These experiments, often called {\em Bell tests}, look for violations of constraints (Bell inequalities) on $p(a,b|x,y;P)$ arising from the assumption of local causality.
Quantum theory predicts violations of Bell inequalities. When experimentally verified, such violations can be said to be a property of the experiment/Nature without requiring that a quantum 
model of the experiment/Nature hold.\footnote{See Refs.~\cite{loopholefree,loopholefree2,loopholefree3} for recent ``loophole-free'' Bell tests and Ref.~\cite{Bellreview} for a fairly comprehensive
review of Bell nonlocality.}

The CHSH inequality provides the simplest example of a Bell inequality, that is,
\begin{equation}
\sum_{abxy}\frac{1}{4}p(a,b|x,y;P)\delta_{a\oplus b, x.y}\leq\frac{3}{4},
\end{equation}
where the sum is over those entries $p(a,b|x,y;P)$ for which the input/output correlation $a\oplus b=x.y$ holds. If we denote measurement events by $a_0\equiv(a|x=0)$,
$a_1\equiv(a|x=1)$, $b_0\equiv(b|y=0)$, and $b_1\equiv(b|y=1)$, then asking that this correlation be satisfied amounts to the following equations:
\begin{eqnarray}
a_0\oplus b_0&=&0,\\
a_0\oplus b_1&=&0,\\
a_1\oplus b_0&=&0,\\
a_1\oplus b_1&=&1.
\end{eqnarray}
Note that adding up the first three equations gives $a_1\oplus b_1=0$, contrary to the fourth equation $a_1\oplus b_1=1$. Thus, at most only three of the four equations can be simultaneously 
satisfied. The $3/4$ upper bound in the Bell-CHSH inequality arises from this fact.

That this Bell inequality admits a quantum violation can be seen by choosing quantum states and measurements as follows:

\subsubsection{Optimal Quantum State}
Consider a maximally entangled state of two qubits (each of which is sent to one of the two parties):
\begin{equation}
 |\Psi\rangle=\frac{1}{\sqrt{2}}\left(|0\rangle |0\rangle + |1\rangle |1\rangle\right)
\end{equation}
where we choose the computational basis:

$$|0\rangle=\begin{pmatrix}1 \\ 0\end{pmatrix} \text{ and } |1\rangle=\begin{pmatrix}0 \\ 1\end{pmatrix}$$
Note that $|\Psi\rangle$ lives in the tensor product Hilbert space $\mathcal{H_{A}} \otimes \mathcal{H_{B}}$, where $\mathcal{H_{A}}$ is the Hilbert space of Alice's subsystem (qubit) and $\mathcal{H_{B}}$
is the Hilbert space of Bob's subsystem (qubit). We have omitted the tensor product symbol `$\otimes$' in $|\Psi\rangle$ but the tensor product is presumed. One can now write the
corresponding density matrix:
\begin{equation}
 \rho=|\Psi\rangle\langle\Psi|=\frac{1}{2}\begin{pmatrix} 1 & 0 & 0 & 1 \\0 & 0 & 0 & 0\\0 & 0 & 0 & 0\\1 & 0 & 0 & 1\end{pmatrix}
\end{equation}
\subsubsection{Optimal Quantum Measurements}
Let us denote the optimal measurements that Alice and Bob perform on their subsystems by $\{A_{0},A_{1}\}$ and $\{B_{0},B_{1}\}$ respectively. They make spin measurements on
their qubits:

\begin{eqnarray}
 A_{0}=\sigma_{1}=\begin{pmatrix}0 & 1\\1 & 0 \end{pmatrix},\quad A_{1}=\sigma_{3}=\begin{pmatrix}1 & 0\\0 & -1 \end{pmatrix}\\\nonumber\\
 B_{0}=\frac{\sigma_{1}+\sigma_{3}}{\sqrt{2}}=\frac{1}{\sqrt{2}}\begin{pmatrix}1 & 1\\1 & -1\end{pmatrix},\quad B_{1}=\frac{\sigma_{1}-\sigma_{3}}{\sqrt{2}}=\frac{1}{\sqrt{2}}\begin{pmatrix}-1 & 1\\1 & 1\end{pmatrix}
\end{eqnarray}
Note that the outcomes for any spin measurement are $\{+1,-1\}$, where we label $+1$ by $`0$' and $-1$ by $`1$'. The winning probability for the CHSH game given this quantum strategy is: 
\begin{equation}
 p_{win}^{Q}=\frac{1}{4}\sum_{a,b,x,y \in \{0,1\}}\delta_{a\oplus b,x.y}p(a,b|x,y;\rho),
\end{equation}
where
$$p(a,b|x,y;\rho)=\Tr(A_{x}^{a} \otimes B_{y}^{b}\rho)=\langle\Psi| A_{x}^{a} \otimes B_{y}^{b} |\Psi\rangle \equiv \langle A_{x}^{a} \otimes B_{y}^{b} \rangle,$$
and
\begin{eqnarray}
A_{x}^{a}=\frac{I+(-1)^a A_{x}}{2}\\
B_{y}^{b}=\frac{I+(-1)^b B_{y}}{2}
\end{eqnarray}
Clearly,
\begin{equation}
 \langle A_{x}^{a} \otimes B_{y}^{b} \rangle=\frac{1}{4}\langle I \otimes I+(-1)^b I\otimes B_{y}+(-1)^a A_{x} \otimes I+(-1)^{a\oplus b}A_{x}\otimes B_{y}\rangle.
\end{equation}
We have:
\begin{eqnarray}
 p_{win}^{Q}&=&\frac{1}{4}\sum_{a,b,x,y \in \{0,1\}}\delta_{a\oplus b,x.y}\langle A_{x}^{a} \otimes B_{y}^{b} \rangle\nonumber\\
&=&\frac{1}{4}\sum_{a,b,x,y \in \{0,1\}}\delta_{a\oplus b,x.y}\frac{1}{4}\langle I \otimes I +(-1)^b I \otimes B_{y}+(-1)^a A_{x} \otimes I+(-1)^{a\oplus b}A_{x}\otimes B_{y}\rangle\nonumber\\
&=&\frac{1}{2}\left(1+\frac{\langle A_{0} \otimes B_{0}+A_{0} \otimes B_{1}+A_{1} \otimes B_{0}-A_{1} \otimes B_{1}\rangle}{4}\right)\nonumber\\
&=&\frac{1}{2}\left(1+\frac{\langle CHSH\rangle}{4}\right),
\end{eqnarray}
where
$$\langle CHSH\rangle\equiv\langle A_{0} \otimes B_{0}+A_{0} \otimes B_{1}+A_{1} \otimes B_{0}-A_{1} \otimes B_{1}\rangle.$$
Consider spin measurement on Alice's qubit along $\hat{\alpha}=(\alpha_{1},\alpha_{2},\alpha_{3})$ axis (i.e., measurement of $A^{\alpha}=\vec{\sigma}.\hat{\alpha}$), and measurement 
on Bob's qubit along $\hat{\beta}=(\beta_{1},\beta_{2},\beta_{3})$ axis (i.e., measurement of $B^{\beta}=\vec{\sigma}.\hat{\beta}$). Of course, $\vec{\sigma}=(\sigma_{1},\sigma_{2},\sigma_{3})$,
where $\sigma_1$, $\sigma_2$, and $\sigma_3$ are the three Pauli matrices:
\begin{eqnarray}
\sigma_1&=&\begin{pmatrix}0 & 1\\1 & 0 \end{pmatrix}\\
\sigma_2&=&\begin{pmatrix}0 & -i\\i & 0 \end{pmatrix}\\
\sigma_3&=&\begin{pmatrix}1 & 0\\0 & -1 \end{pmatrix}
\end{eqnarray}
corresponding to spin measurements along the $X$, $Y$, and $Z$ axis respectively.
Now,
\begin{eqnarray}
\langle A^{\alpha} \otimes B^{\beta}\rangle=\Tr(A^{\alpha} \otimes B^{\beta} \rho)=\alpha_{1}\beta_{1}-\alpha_{2}\beta_{2}+\alpha_{3}\beta_{3}.
\end{eqnarray}
Denoting any spin measurement $\vec{\sigma}.\hat{\gamma}$ by the corresponding unit vector $(\gamma_1,\gamma_2,\gamma_3)$ along which the measurement is made, we have
\begin{eqnarray}
 A_{0}=\sigma_{1}\rightarrow(1,0,0),\quad A_{1}=\sigma_{3} \rightarrow(0,0,1)\\\nonumber\\
 B_{0}=\frac{\sigma_{1}+\sigma_{3}}{\sqrt{2}}\rightarrow \left(\frac{1}{\sqrt{2}},0,\frac{1}{\sqrt{2}}\right),\quad B_{1}=\frac{\sigma_{1}-\sigma_{3}}{\sqrt{2}}=\left(\frac{1}{\sqrt{2}},0,-\frac{1}{\sqrt{2}}\right).
\end{eqnarray}
Then
\begin{eqnarray}\nonumber
\langle CHSH\rangle&=&\langle A_{0} \otimes B_{0}\rangle+\langle A_{0} \otimes B_{1}\rangle+\langle A_{1} \otimes B_{0}\rangle-\langle A_{1} \otimes B_{1}\rangle\\\nonumber
&=&(1,0,0).(\frac{1}{\sqrt{2}},0,\frac{1}{\sqrt{2}})+(1,0,0).(\frac{1}{\sqrt{2}},0,-\frac{1}{\sqrt{2}})\\\nonumber
&&+(0,0,1).(\frac{1}{\sqrt{2}},0,\frac{1}{\sqrt{2}})-(0,0,1).(\frac{1}{\sqrt{2}},0,-\frac{1}{\sqrt{2}})\\
&=&2\sqrt{2},
\end{eqnarray}
and
\begin{equation}
\therefore p_{win}^{Q}=\frac{1}{2}+\frac{1}{2\sqrt{2}}\approx 0.85,
\end{equation}
violating the Bell inequality bound of $0.75$.

\section{Bridging the gap between Bell's theorem and the Kochen-Specker theorem}
Along with its foundational implications, Bell's theorem\cite{Bell64,Bell66,Bell76} has also been the subject of a lot of research activity in quantum information theory.
It would not be an exaggeration to say that the seeds for the quantum information age were sown with John Bell's demonstration of the nontrivial implications of
entanglement for the locally causal worldview that quantum theory forces us to abandon. Bell was motivated by the questions Einstein, Podolsky and Rosen raised in 
their seminal paper\cite{EPR}
questioning the completeness of quantum theory as a description of reality. Today Bell's theorem underlies many quantum information protocols, ranging from cryptography
to randomness generation \cite{Bellreview}.

On the other hand, the Kochen-Specker theorem has not seen similar explosion of research activity when it comes to its applications to quantum information theory. 
This is {\em despite} the fact, as pointed out often in the recent literature\cite{CF, CSW, AFLS}, that mathematically speaking, both Bell's theorem
and the Kochen-Specker theorem lead to a marginal problem, i.e. determining whether a given set of variables, carved up into various subsets, can admit a joint probability distribution,
given joint probability distributions for the subsets.

The reasons for the relative lack of impact of the Kochen-Specker theorem (as opposed to Bell's theorem) on the development of quantum information  
have to do with the idealizations that make the KS theorem less suited to experimental testability and hence applications to quantum information processing. 
The reasons are foundational 
since there is a conceptual gap between Bell's theorem and the Kochen-Specker theorem 
that is not well-recognized when thinking of both \emph{merely} mathematically in terms of marginal problems. This gap refers to the following points of 
contrast between the two theorems, notwithstanding their mathematical 
similarities:
\begin{enumerate}

 \item Bell's theorem is not applicable to a single party. It {\em necessarily} requires at least two parties or laboratories for its assumption of local causality to be applied. On the other hand,
 Kochen-Specker theorem does not need a multipartite scenario for its assumptions to be applicable. One might expect this to count in favour of the Kochen-Specker theorem's 
 broader applicability but there are significant conceptual hurdles to that.

 \item Bell's theorem is theory-independent: an experimental violation of a Bell inequality rules out local causality 
 irrespective of the particular operational theory that may seek to model the experiment. In particular, one does not need 
 to presume a Hilbert space description of the system (as in quantum theory). On the other hand, the Kochen-Specker 
 theorem is very much a result specific to quantum theory: it presumes that measurement outcomes are represented by projectors on
 a Hilbert space. Bridging this particular gap between the two theorems requires an operationalization of 
 the Kochen-Specker theorem, a task we take up in Chapter 6.
 
 \item Nonsignalling between the parties during a Bell test is a prerequisite for any observed violations of Bell inequalities to be
 taken as evidence against local causality. Otherwise the violations can be easily attributed to signalling. Likewise, an operational 
 equivalence between measurement procedures is a prerequisite for any proof of the KS theorem to count as evidence against KS-noncontextuality.
 However, while nonsignalling is guaranteed by a physical principle independent of quantum theory -- namely, the special relativistic 
 prohibition of faster-than-light communication -- there is no obvious candidate principle, short of appealing to quantum 
 theory itself, that guarantees operational equivalence
 between two measurement procedures in a scenario where these measurement procedures may well be carried out on the same system.

 \item Bell's theorem does not require the assumption of outcome determinism while the Kochen-Specker theorem does.
 We take up this issue in Chapter 5.
 
 \item A fundamental conceptual hurdle to experimentally testing the KS theorem 
 is the existence of Meyer-Kent-Clifton (MKC) type KS-noncontextual models \cite{MKC} which can simulate the predictions of any set of measurements leading to the 
 Kochen-Specker theorem as long as the measurements are not infinitely precise. In other words, the finite precision of real-world measurements is a loophole that lets MKC-type models 
 reproduce the quantum statistics without giving up on KS-noncontextuality. Such a conceptual hurdle does not apply to 
 Bell's theorem because it makes no commitment, unlike the KS theorem, regarding the representation of measurements. In the 
 Spekkens' approach to noncontextuality, the MKC criticism is circumvented by avoiding any use of Hilbert spaces at all 
 -- properties of which are crucial to the MKC criticism -- akin to discussions of local causality. What remains to be confronted
 in the Spekkens' approach is the fact that {\em exact} operational equivalences between experimental procedures are an idealization
 never realized in practice. We take up this issue in the final section of Chapter 6. 
 
\end{enumerate}

\section{Overview of the thesis}
Following this introductory chapter, we will take up the issues highlighted in the last section in the rest of the thesis.
In Chapters 2 and 3, we will take the first steps beyond the Kochen-Specker theorem by considering nonprojective 
quantum measurements and their contextuality. Chapter 4 will show the 
realizability of arbitrary joint measurability structures when nonprojective measurements are allowed.
In Chapter 5, we will show how the assumption of outcome determinism 
is not as innocuous in the KS theorem as it is in Bell's theorem (where it's unnecessary). 
Chapters 6 and 7 will pursue a full-fledged operationalization of the Kochen-Specker theorem 
that does not presume quantum theory. Chapter 8 will conclude with some final remarks on 
the research directions initiated in this thesis and some speculations on what's next.

\chapter{A first look at Specker's scenario}
\setlength{\epigraphwidth}{\textwidth}
\epigraph{\it It is no exaggeration to say that this theorem -- the
Kochen-Specker theorem -- is one of the deepest facts about the foundations of quantum theory. The story of how Specker first started down the
road which led to this result is quite wonderful. It shows that even in an era where ``shut up and calculate'' is the mantra of many researchers,
deep philosophical questions can still lead to great advances in our understanding of the world. It is a story that will warm the heart of anyone
who believes that physics should be pursued in a romantic style...[Specker] was led to the critical question: could God know what outcome
would have occurred had a different quantum measurement been done to
the one that was actually done, without getting into contradiction? The
answer, he found, was that He could not.
}{R.W. Spekkens, Ernst Paul Specker (1920-2011), Mind \& Matter Vol. 9(2), pp. 121-128.}

In 1960, Ernst Specker wrote an article on the logic of quantum theory \cite{Spe60}. In it he uses a parable to illustrate the non-Boolean nature of quantum logic. Liang, Spekkens, and Wiseman \cite{LSW}
studied a modern rendition of Specker's parable with three unsharp (nonprojective) quantum measurements which are pairwise jointly measurable but not triplewise so. In this 
chapter, we consider the question of whether such measurements exhibit contextuality, something which was conjectured not to be the case in Ref.~\cite{LSW}.

This chapter is based on work reported in Ref.~\cite{KG}.

\section{Introduction}
Quantum theory does not admit Bell-local or KS-noncontextual ontological models. This is manifest in 
the Bell-nonlocality \cite{EPR, Bell64} and KS-contextuality \cite{KS67} of the theory. Both these features arise---at a
mathematical level---from the lack of a global joint probability distribution over measurement outcomes that can reproduce the
measurement statistics predicted by quantum theory. Traditionally, KS-contextuality has been shown with respect to projective measurements 
for Hilbert spaces of dimension three or greater \cite{KS67, peres, Clifton, Mermin2theorems,
Cab1,Cab2, KRK, KCBS, Oh}.

For projective measurements, KS-noncontextuality is the assumption that in an ontological model of quantum theory the outcome of a measurement $A$ is independent of 
whether it is performed together with a measurement $B$, where $[A,B]=0$, or with measurement $C$, where $[A,C]=0$ and $B$ 
and $C$ are not compatible, i.e. $[B,C]\neq0$. $B$ and $C$ provide contexts for measurement of $A$. A qubit cannot yield a
proof of KS-contextuality because, having only a maximum of two mutually orthogonal rank 1 projectors, it does not allow projective measurements $A,B,C$ such that $[A,B]=0, [A,C]=0$, and  $[B,C] \neq 0$.
The existence of a triple of such projective measurements is necessary for any proof of KS-contextuality. While a state-independent 
proof of KS-contextuality holds for any state preparation, a state-dependent proof requires a special choice of the prepared 
state. 
The minimal state-independent proof of KS-contextuality requires a qutrit and $13$ projectors \cite{Oh, cabmin}.
The minimal state-dependent proof \cite{KCBS, KRK}, first given by Klyachko et al., requires a qutrit and five projectors 
(Fig.~\ref{kcbs}).\footnote{Note, however, that the state-independent proof of Ref.\cite{Oh} is \emph{not} a traditional KS-uncolourability proof like that of Refs.~\cite{KS67,Cab1}. Indeed,
the KS-uncolourability proofs \cite{KS67,Cab1} are the strongest demonstrations of KS-contextuality because they need not rely on statistical inequalities: they concern scenarios where KS-noncontextual
value assignments are simply {\em impossible}. On the other hand, the proof in Ref.\cite{Oh} considers a scenario where KS-noncontextual value assignments {\em are possible}, but one can still identify 
a statistical inequality that {\em all quantum states} would violate for the given choice of projectors, hence the proof is state-independent in this sense. This state-independence, however, is weaker than for KS-uncolourability 
proofs since it cannot be guaranteed for operational theories other than quantum theory, while state-independence can be guaranteed for arbitrary operational theories when it comes to KS-uncolourability proofs.
This is easy to see: since there are no KS-noncontextual value assignments, any measurement noncontextual probability assignment at all will lead to KS-contextuality in KS-uncolourability proofs,
regardless of whether such an assignment arises from quantum states and measurements.
The state-dependent proofs of KS-contextuality such as the one in Ref.\cite{KCBS} rely on scenarios where KS-noncontextual value assignments {\em are possible} but they are constrained by 
statistical inequalities that can be violated by {\em a careful choice of quantum state}. Such state-dependent proofs are a weaker demonstration of KS-contextuality that both the KS-uncolourability
proofs and the state-independent proof in Ref.\cite{Oh}.}
Thus a qutrit is the simplest quantum system that allows a proof of KS-contextuality, both state-independent
and state-dependent. 

However, generalizations of KS-noncontextuality for a qubit have been considered earlier 
\cite{Cab3, busch, caves} in a manner that is conceptually distinct from the approach we adopt here. These generalizations typically apply the assumption of 
outcome determinism to unsharp (nonprojective) measurements. Such an application of this assumption cannot be justified from noncontextuality alone, as amply demonstrated in
Ref.~\cite{odum}. Our approach builds 
upon the work of Spekkens \cite{genNC} and Liang et. al. \cite{LSW}, and we consider generalized-noncontextuality 
proposed by Spekkens \cite{genNC} as the appropriate notion of noncontextuality for unsharp qubit measurements. 
This notion of noncontextuality allows one to consider outcome-indeterministic response functions for unsharp measurements in 
the ontological model, indeed it requires them\cite{odum}, while any naive application of KS-noncontextuality to POVMs would insist on outcome-deterministic response functions. 
Since we have already introduced the definitions of noncontextuality in Chapter 1, we will not repeat them here. 

We define a compatibility scenario as a collection of subsets, called `compatibility contexts'\footnote{Or, simply, `contexts' when it is clear that the type of context under consideration is a compatibility or joint measurability context.},
of the set of all measurements. A compatibility context refers to measurements that can be jointly implemented. 
\begin{figure}
\centering
\includegraphics[scale=0.3]{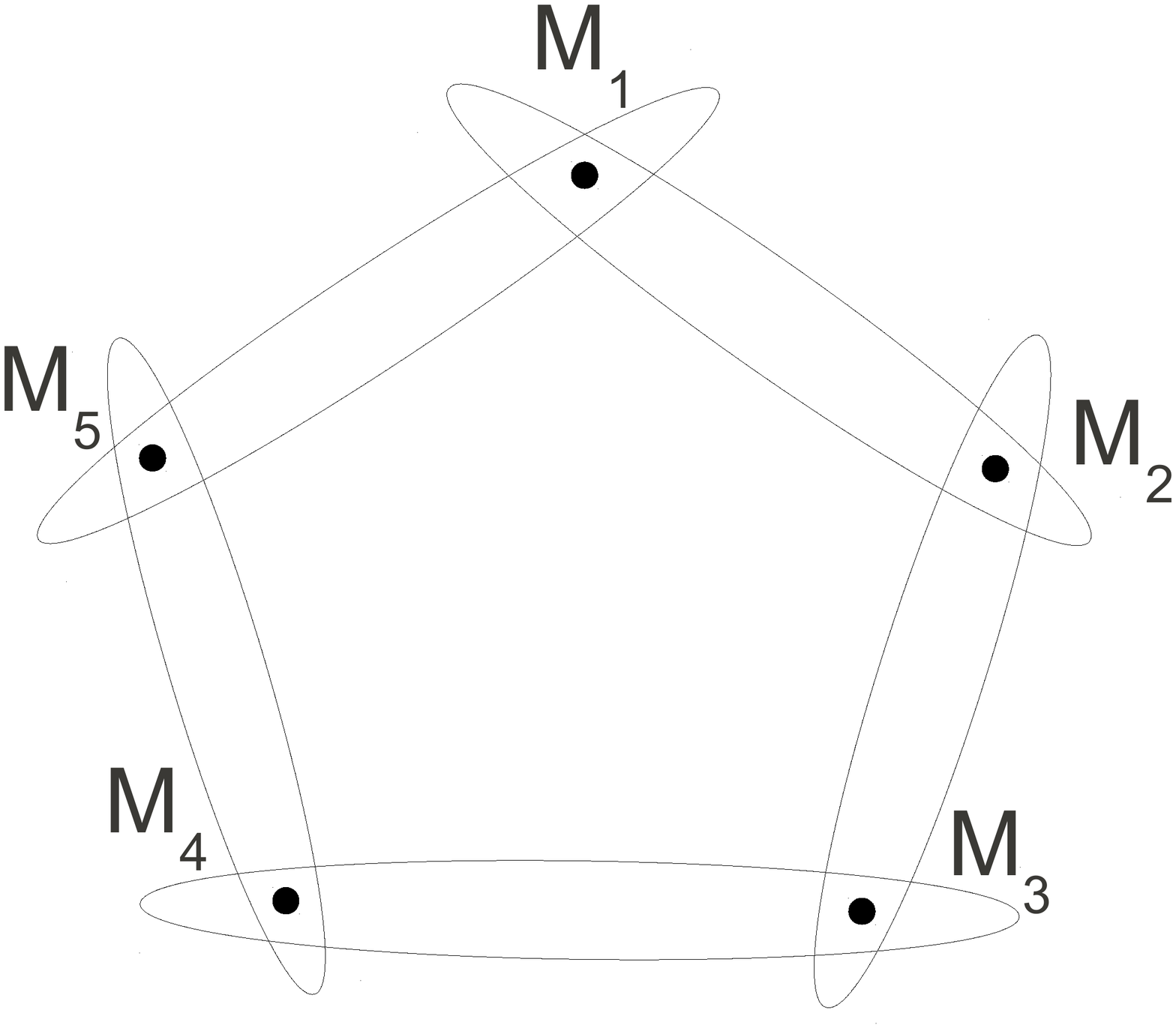}
\caption{The KCBS \cite{KCBS} compatibility scenario. The vertices represent the measurements and edges represent compatibility contexts (of jointly measurable observables).}
\label{kcbs}
\end{figure}
Conceptually, the simplest possible compatibility scenario capable of exhibiting KS-contextuality, first considered by Specker \cite{Spe60} (Fig.~\ref{specker}), requires three two-valued measurements, $\{M_1,M_2,M_3\}$, to allow for three non-trivial 
compatibility contexts: $\{\{M_1,M_2\},\{M_1,M_3\},\{M_2,M_3\}\}$. Any other choice of contexts will be trivially KS-noncontextual, e.g., $\{\{M_1,M_2\},\{M_1,M_3\}\}$ is KS-noncontextual because the joint probability distribution
$$p(M_1,M_2,M_3)\equiv p(M_1,M_2)p(M_1,M_3)/p(M_1)$$ reproduces the marginal statistics for $\{M_1,M_2\}$ and $\{M_1,M_3\}$.
In Specker's scenario, measurement statistics that always shows perfect anticorrelation between any two measurements sharing a context is necessarily KS-contextual. On assigning outcomes $\{+1,-1\}$ 
KS-noncontextually to the three measurements $\{M_1,M_2,M_3\}$,
it becomes obvious that the maximum number of anticorrelated contexts possible in a single assignment is two, e.g., for the assignment $\{M_1\rightarrow+1,M_2\rightarrow-1,M_3\rightarrow+1\}$,
$\{M_1,M_2\}$ and $\{M_2,M_3\}$ are anticorrelated but $\{M_1,M_3\}$ is not. This puts an upper bound of $\frac{2}{3}$ on the probability of anticorrelation when a context is chosen uniformly at random.
Specker's scenario precludes projective measurements because a set of three pairwise commuting projective measurements is trivially jointly measurable and cannot show any contextuality.
One may surmise that it represents a kind of contextuality that is not seen in quantum theory. However, as Liang et al. showed \cite{LSW}, Specker's scenario can be realized using noisy spin-1/2 observables. 
They showed that if one does not assume outcome-determinism for unsharp measurements and models them stochastically but noncontextually, then this noncontextual model 
for noisy spin-1/2 observables will obey a bound of $1-\frac{\eta}{3}$, where $\eta \in [0,1]$ is the sharpness associated with each observable.
Formally,
\begin{equation}
 R_3\equiv\frac{1}{3}\sum_{(ij)\in\{(12),(23),(13)\}}\textrm{Pr}(X_i\neq X_j|G_{ij})\leq 1-\frac{\eta}{3},
\end{equation}
where $\textrm{Pr}(X_i\neq X_j|G_{ij})$ is the probability of anticorrelation between the outcomes $X_i, X_j \in \{+1,-1\}$ of measurements $M_i$ and $M_j$,
respectively. $G_{ij}$ denotes the joint implementation of $M_i$ and $M_j$. We will refer to this noncontextuality inequality as the \emph{LSW (Liang-Spekkens-Wiseman) inequality}. Note that the LSW inequality
is \emph{not} a KS-noncontextuality inequality, for which the bound would be $\frac{2}{3}$. A violation of the LSW inequality will rule out Spekkens' 
generalized notion of noncontextuality and, 
by implication, KS-noncontextuality. 
After giving examples of orthogonal and trine spin-axes that did not seem to show a violation of this inequality, Liang et al.~left open the question of whether such a violation exists \cite{LSW}.
They conjectured that all such triples of POVMs will admit a noncontextual model \cite{genNC}, i.e. the LSW inequality will not be violated. We settle this conjecture by showing that, contrary to 
their expectation\cite{LSW}, the LSW inequality does admit a quantum violation.
\begin{figure}
\centering
\includegraphics[scale=0.3]{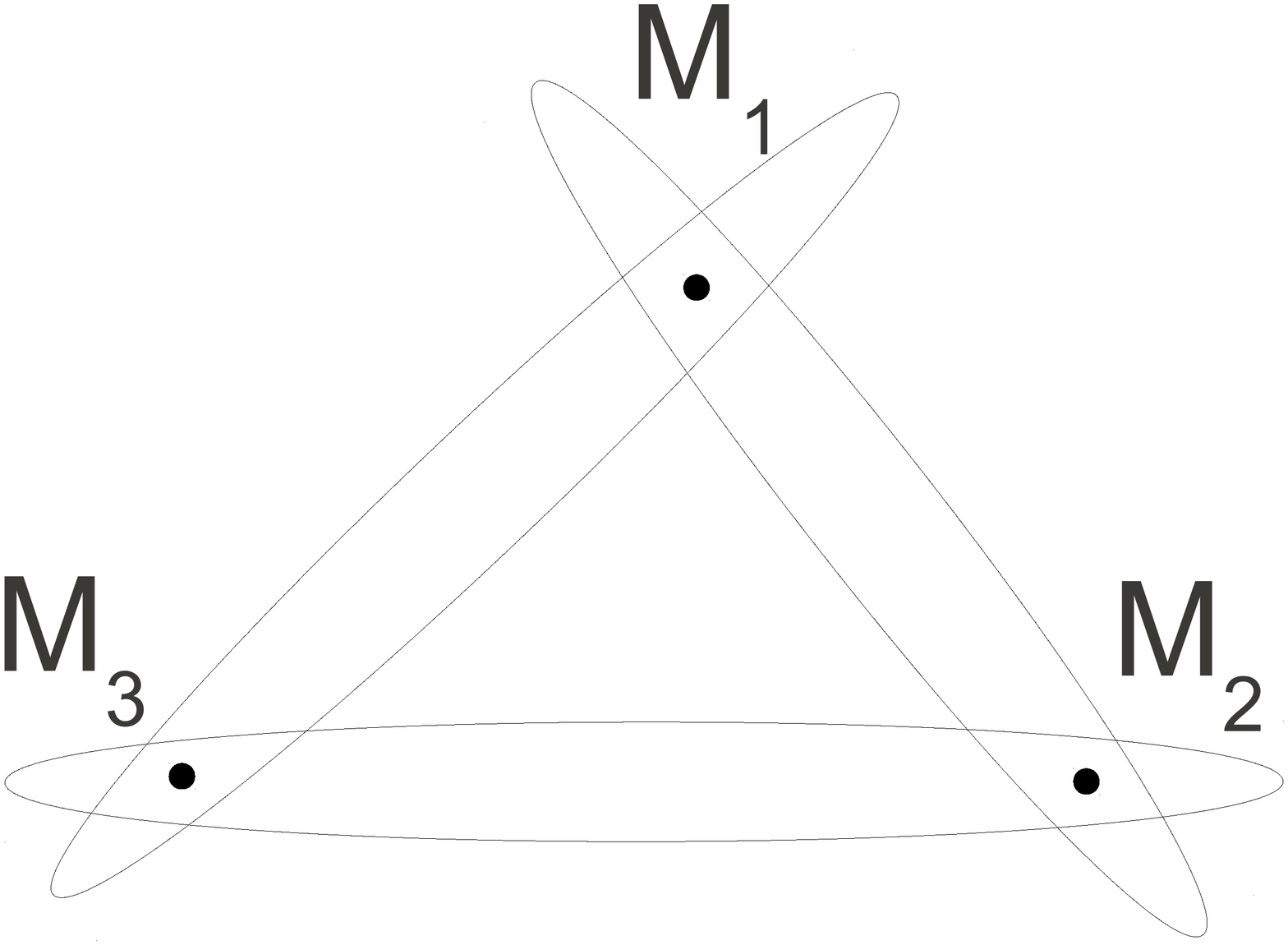}
\caption{Specker's scenario.}
\label{specker}
\end{figure}
\section{The LSW inequality}
The three POVMs considered, $M_k=\{E^k_+, E^k_-\}$, $k\in \{1,2,3\}$, are noisy spin-$\frac{1}{2}$ observables of the form 
\begin{equation}\label{qubitPOVMs}
 E^k_{\pm}\equiv \frac{1}{2}I\pm \frac{\eta}{2}\vec{\sigma}.\hat{n}_k, \quad 0\leq \eta \leq 1.
\end{equation}
That is,
\begin{equation}
 E^k_{\pm}=\frac{1-\eta}{2}I+\eta\Pi^k_{\pm},
\end{equation}
where $\Pi^k_{\pm}=\frac{1}{2}(I\pm\vec{\sigma}.\hat{n}_k)$ are the corresponding projectors. So $E^k_{\pm}$ are noisy versions of the projectors $\Pi^k_{\pm}$, and the observable $\{E^k_+,E^k_-\}$ is therefore a noisy 
(or unsharp) version of the projective measurement $P_k=\{\Pi^k_+,\Pi^k_-\}$ (for $k\in \{1,2,3\}$): $$\{E^k_+,E^k_-\}=\eta\{\Pi^k_+,\Pi^k_-\}+(1-\eta)\{I/2,I/2\}.$$
The LSW inequality concerns the following quantity:
\begin{equation}
R_3\equiv\frac{1}{3}\sum_{(ij)\in\{(12),(23),(13)\}}p(X_i\neq X_j|G_{ij}) 
\end{equation}
where $X_i, X_j \in \{+1,-1\}$ label measurement outcomes for $M_i$ and $M_j$, respectively. The joint measurement POVM for the context $\{M_i, M_j\}$ is denoted by $G_{ij}\equiv\{G_{++}^{ij},G_{+-}^{ij},G_{-+}^{ij},G_{--}^{ij}\}$.
$G^{ij}_{X_i,X_j}$ is the joint measurement effect corresponding to the effects $E^i_{X_i}$ and $E^j_{X_j}$, i.e.
$$\sum_{X_j} G^{ij}_{X_i,X_j}=E^i_{X_i}\text{ and }\sum_{X_i} G^{ij}_{X_i,X_j}=E^j_{X_j}.$$
$R_3$ is the average probability of anticorrelation when one of the three contexts is chosen uniformly at random.

Under the assumption of noncontextuality for these noisy spin-1/2 observables, the following bound on $R_3$ holds (cf. \cite{LSW}, Section 7.3)\footnote{We provide a separate derivation of this and other inequalities for Specker's
scenario in Chapter 5.}:
\begin{equation}\label{ncineq}
R_3\leq 1-\frac{\eta}{3} 
\end{equation}
The question is: Does there exist a triple of noisy spin-1/2 observables that will violate this inequality, perhaps for some specific state-preparation? 
\section{Joint measurability constraints}
Testing the LSW inequality for a quantum mechanical violation requires a special kind of joint measurability (or compatibility), denoted by compatibility contexts $\{\{M_1,M_2\}$, $\{M_2,M_3\}$, $\{M_1,M_3\}\}$,
i.e. pairwise joint measurability but no triplewise joint measurability. One can achieve this joint measurability criterion by adding noise to projective measurements along three different axes.
For a given choice of measurement directions $\{\hat{n}_1,\hat{n}_2,\hat{n}_3\}$ in Eq.~(\ref{qubitPOVMs}), denoting $\hat{n}_i.\hat{n}_j\equiv \cos \theta_{ij}$, a sufficient condition for this kind of joint measurability is
\begin{equation}\label{jointcriterion}
 \eta_l<\eta\leq \eta_u
\end{equation}
where
\begin{equation}\label{lower}
\eta_l=\frac{1}{3} \max_{X_1,X_2,X_3 \in \{\pm1\}} \left\{\sqrt{3+2\sum_{k,l \in \{1,2,3\}, k<l}X_kX_l\cos \theta_{kl}} \right\}
\end{equation}
and 
\begin{equation}\label{upper}
\eta_u=\min_{(ij)\in\{(12),(23),(13)\}}\left\{\frac{1}{\sqrt{1+|\sin \theta_{ij}|}}\right\}
\end{equation}
We note that this condition is necessary and sufficient for the special case of trine ($\hat{n}_i.\hat{n}_j=-1/2$) and orthogonal ($\hat{n}_i.\hat{n}_j=0$) spin axes.
These conditions are obtained as special cases of the more general joint measurability conditions we now prove, based on Refs.~\cite{LSW} and \cite{YuOh}.

{\bf Bounds on $\eta$ from joint measurability: }
In Appendix F of Ref.~\cite{LSW}, Theorem 13, the authors obtain necessary and sufficient conditions for joint measurability of noisy spin-1/2 observables. We note that the claimed necessary condition
in the aforementioned theorem is incorrect, while the sufficient condition holds. Hence we prove a necessary condition for joint measurability that 
we use for triplewise joint measurability, thereby correcting the argument for necessity made by Liang et al.:

\begin{theorem}
 Given a set of qubit POVMs, $\{\{E^k_{X_k}: X_k \in \{+1,-1\}\}|k \in \{1\dots N\}\}$, of the form 
\begin{equation}
 E^k_{X_k}=\frac{1}{2}I+\frac{1}{2}\vec{\sigma}.X_k\eta\hat{n}_k,
\end{equation}
and defining $2^N$ 3-vectors
\begin{equation}
 \vec{m}_{X_1\dots X_N}\equiv \sum_{k=1}^{N}X_k\hat{n}_k,
\end{equation}
a necessary condition for all the POVMs to be jointly measurable is that
\begin{equation}
 \eta \leq \frac{1}{N} \max_{X_1\dots X_N} \{|\vec{m}_{X_1\dots X_N}|\},
\end{equation}
and a sufficient condition is that 
\begin{equation}
 \eta \leq \frac{2^N}{\sum_{X_1\dots X_N}|\vec{m}_{X_1\dots X_N}|}.
\end{equation}
\end{theorem}

\textbf{Proof.} We will only prove the necessary condition and refer the reader to Ref.\cite{LSW}, Appendix F, for proof of the sufficient condition. Note that 
$\eta=\Tr\left[(\vec{\sigma}.X_k\hat{n}_k)E^k_{X_k}\right]$. Since this holds $\forall X_k,k$, we have
\begin{equation}
 \eta=\frac{1}{2N}\sum_{k=1}^{N}\sum_{X_k}\Tr\left[(\vec{\sigma}.X_k\hat{n}_k)E^k_{X_k}\right]
\end{equation}
If all the POVMs are jointly measurable, then we must necessarily have a joint POVM $\{E_{X_1\dots X_N}\}$ such that 
\begin{equation}
 E^k_{X_k}=\sum_{X_1\dots X_{k-1},X_{k+1}\dots X_N}E_{X_1\dots X_N}.
\end{equation}
Then,
\begin{equation}
 \eta=\frac{1}{2N}\sum_{X_1\dots X_N}\Tr\left[\left(\vec{\sigma}.\sum_{k=1}^{N}X_k\hat{n}_k\right)E_{X_1\dots X_N}\right],
\end{equation}
and using $\hat{m}_{X_1\dots X_N}\equiv\vec{m}_{X_1\dots X_N}/|\vec{m}_{X_1\dots X_N}|$, we have
\begin{equation}
 \eta=\frac{1}{2N}\sum_{X_1\dots X_N}|\vec{m}_{X_1\dots X_N}|\Tr\left[\left(\vec{\sigma}.\hat{m}_{X_1\dots X_N}\right)E_{X_1\dots X_N}\right].
\end{equation}
Further,
\begin{equation}
 \Tr\left[\left(\vec{\sigma}.\hat{m}_{X_1\dots X_N}\right)E_{X_1\dots X_N}\right]\leq \Tr\left[E_{X_1\dots X_N}\right],
\end{equation}
which yields the inequality
\begin{equation}
 \eta \leq \frac{1}{2N}\sum_{X_1\dots X_N}|\vec{m}_{X_1\dots X_N}|\Tr\left[E_{X_1\dots X_N}\right].
\end{equation}
Now, $\sum_{X_1\dots X_N} E_{X_1\dots X_N}=I$, and therefore, 
\begin{equation}
 \sum_{X_1\dots X_N} \frac{1}{2} \Tr\left[E_{X_1\dots X_N}\right]=1.
\end{equation}
Also, $0\leq\frac{1}{2} \Tr\left[E_{X_1\dots X_N}\right]\leq1$, so we have, by convexity of the set $\left\{\frac{1}{2} \Tr\left[E_{X_1\dots X_N}\right]\right\}_{X_1\dots X_N}$,
\begin{equation}
 \eta \leq \frac{1}{N}\max_{X_1\dots X_N}\left\{|\vec{m}_{X_1\dots X_N}|\right\},
\end{equation}
which is a necessary condition for joint measurability. For $N=3$ we obtain the necessary condition for triplewise joint measurability which we use for computing $\eta_l$.
The necessary and sufficient condition for pairwise joint measurability is given by
\begin{equation}
 1+\eta^4(\hat{n}_i.\hat{n}_j)^2-2\eta^2\geq 0, \quad \forall (ij)\in\{(12),(13),(23)\}.
\end{equation}
This is obtained as a special case, for the present problem, of the more general necessary and sufficient condition for joint measurability of pairs of unsharp qubit observables obtained in Ref. \cite{YuOh}.
Using $\hat{n}_i.\hat{n}_j=\cos \theta_{ij}$, this inequality becomes
\begin{equation}
 \left(\eta^2-\frac{1}{1-|\sin \theta_{ij}|}\right)\left(\eta^2-\frac{1}{1+|\sin \theta_{ij}|}\right)\geq 0.
\end{equation}
Since $0\leq \eta \leq1$, the necessary and sufficient condition for pairwise joint measurability becomes
\begin{equation}
 \eta \leq \min_{(ij)\in\{(12),(23),(13)\}} \left\{ \frac{1}{\sqrt{1+|\sin \theta_{ij}|}} \right\},
\end{equation}
which is used to compute $\eta_u$.

{\it Orthogonal spin axes:} $\hat{n}_i.\hat{n}_j=0$ for all $(ij) \in \{(12),(13),(23)\}$. The necessary and sufficient joint measurability condition is 
\begin{equation}
 \frac{1}{\sqrt{3}}<\eta\leq\frac{1}{\sqrt{2}}.
\end{equation}
\\
{\it Trine spin axes:} $\hat{n}_i.\hat{n}_j=-1/2$ for all $(ij) \in \{(12),(13),(23)\}$. The necessary and sufficient joint measurability condition is
\begin{equation}
 \frac{2}{3}<\eta\leq\sqrt{3}-1.
\end{equation}

{\bf Joint measurement effects: }
We construct the joint measurement POVM, $$G_{ij}=\{G^{ij}_{++},G^{ij}_{+-},G^{ij}_{-+},G^{ij}_{--}\},$$ such that the given POVMs, $M_i=\{E^i_+,E^i_-\}$ and $M_j=\{E^j_+,E^j_-\}$, are recovered 
as marginals, i.e., $\sum_{X_j} G^{ij}_{X_i,X_j}=E^i_{X_i}$, $\sum_{X_i} G^{ij}_{X_i,X_j}=E^j_{X_j}$, $0\leq G^{ij}_{X_i,X_j}\leq I$, and $\sum_{X_i, X_j} G^{ij}_{X_i,X_j}=I$, where $X_i, X_j \in \{+1,-1\}$.
The joint measurement POVM has the following general form:

\begin{eqnarray}\label{jointbeg}
 G^{ij}_{++}&=&\frac{1}{2}\left[\frac{\alpha_{ij}}{2}I+\vec{\sigma}.\frac{1}{2}\left(\eta(\hat{n}_i+\hat{n}_j)-\vec{a}_{ij}\right)\right],\\
 G^{ij}_{+-}&=&\frac{1}{2}\left[\left(1-\frac{\alpha_{ij}}{2}\right)I+\vec{\sigma}.\frac{1}{2}\left(\eta(\hat{n}_i-\hat{n}_j)+\vec{a}_{ij}\right)\right],\\
 G^{ij}_{-+}&=&\frac{1}{2}\left[\left(1-\frac{\alpha_{ij}}{2}\right)I+\vec{\sigma}.\frac{1}{2}\left(\eta(-\hat{n}_i+\hat{n}_j)+\vec{a}_{ij}\right)\right],\\
 G^{ij}_{--}&=&\frac{1}{2}\left[\frac{\alpha_{ij}}{2}I+\vec{\sigma}.\frac{1}{2}\left(\eta(-\hat{n}_i-\hat{n}_j)-\vec{a}_{ij}\right)\right],\label{jointend}
\end{eqnarray}
where $I$ is the $2\times2$ identity matrix, $\vec{\sigma}=(\sigma_x,\sigma_y,\sigma_z)$ are the $2\times2$ Pauli matrices, $\alpha_{ij} \in \mathbb{R}$, and $\vec{a}_{ij} \in \mathbb{R}^3$. The necessary and sufficient conditions for these to be valid 
qubit effects, $0\leq G^{ij}_{X_i,X_j}\leq I$, $\forall X_i,X_j \in \{+1,-1\}$, are equivalent to the following inequalities \cite{heinosaari},
\begin{equation}\label{valid1}
\sqrt{2\eta^2(1+\hat{n}_i.\hat{n}_j)+|\vec{a}_{ij}|^2+2\eta|(\hat{n}_i+\hat{n}_j).\vec{a}_{ij}|}\leq\alpha_{ij}
\end{equation}
\begin{equation}\label{valid2}
\alpha_{ij}\leq2-\sqrt{2\eta^2(1-\hat{n}_i.\hat{n}_j)+|\vec{a}_{ij}|^2+2\eta|(\hat{n}_i-\hat{n}_j).\vec{a}_{ij}|},
\end{equation}
where $\eta_l<\eta\leq \eta_u$.
The joint measurement effects corresponding to anticorrelation sum to
\begin{equation}\label{anticorr}
 G^{ij}_{+-}+G^{ij}_{-+}=\left(1-\frac{\alpha_{ij}}{2}\right)I+\frac{1}{2}\vec{\sigma}.\vec{a}_{ij}.
\end{equation}
We now come to the construction of the joint measurement POVM and derivation of the necessary and sufficient condition for its validity, 
Eqs.~(\ref{valid1})-(\ref{valid2}).

{\bf Constructing the joint measurement: }
The joint measurement POVM $G_{ij}$ for $\{M_i,M_j\}$ should satisfy the marginal condition:
\begin{eqnarray}\label{marg1}
G^{ij}_{++}+G^{ij}_{+-}=E^i_+,\quad
G^{ij}_{-+}+G^{ij}_{--}=E^i_-,\\
G^{ij}_{++}+G^{ij}_{-+}=E^j_+,\quad
G^{ij}_{+-}+G^{ij}_{--}=E^j_-.\label{marg2}
\end{eqnarray}
Also, the joint measurement should consist of valid effects:
\begin{equation}\label{valid}
0\leq G^{ij}_{++},G^{ij}_{+-},G^{ij}_{-+},G^{ij}_{--} \leq I,
\end{equation}
where $I$ is the $2\times2$ identity matrix. The general form of the joint measurement effects is:
\begin{eqnarray}
G^{ij}_{++}&=&\frac{1}{2}\left[\frac{\alpha_{ij}}{2}I+\vec{\sigma}.\vec{a}^{ij}_{++}\right],\\
G^{ij}_{+-}&=&\frac{1}{2}\left[\left(1-\frac{\alpha_{ij}}{2}\right)I+\vec{\sigma}.\vec{a}^{ij}_{+-}\right],\\
G^{ij}_{-+}&=&\frac{1}{2}\left[\left(1-\frac{\alpha_{ij}}{2}\right)I+\vec{\sigma}.\vec{a}^{ij}_{-+}\right],\\
G^{ij}_{--}&=&\frac{1}{2}\left[\frac{\alpha_{ij}}{2}I+\vec{\sigma}.\vec{a}^{ij}_{--}\right],
\end{eqnarray}
where each effect is parameterized by four real numbers. From the marginal condition, Eqs.~(\ref{marg1}-\ref{marg2}), it follows that:
\begin{eqnarray}\label{bega}
\vec{a}^{ij}_{++}+\vec{a}^{ij}_{+-}=\eta \hat{n}_i,\quad
\vec{a}^{ij}_{-+}+\vec{a}^{ij}_{--}=-\eta \hat{n}_i,\\
\vec{a}^{ij}_{-+}+\vec{a}^{ij}_{++}=\eta \hat{n}_j,\quad
\vec{a}^{ij}_{--}+\vec{a}^{ij}_{+-}=-\eta \hat{n}_j.\label{begb}
\end{eqnarray}
These can be rewritten as:
\begin{eqnarray}\label{beg1}
2\vec{a}^{ij}_{++}+\vec{a}^{ij}_{+-}+\vec{a}^{ij}_{-+}&=&\eta (\hat{n}_i+\hat{n}_j),\\
2\vec{a}^{ij}_{+-}+\vec{a}^{ij}_{++}+\vec{a}^{ij}_{--}&=&\eta (\hat{n}_i-\hat{n}_j),\\
2\vec{a}^{ij}_{-+}+\vec{a}^{ij}_{++}+\vec{a}^{ij}_{--}&=&\eta (-\hat{n}_i+\hat{n}_j),\\
2\vec{a}^{ij}_{--}+\vec{a}^{ij}_{+-}+\vec{a}^{ij}_{-+}&=&\eta (-\hat{n}_i-\hat{n}_j).\label{end1}
\end{eqnarray}
From Eqs.~(\ref{bega}-\ref{begb}) it follows that:
$$(\vec{a}^{ij}_{++}+\vec{a}^{ij}_{--})+(\vec{a}^{ij}_{-+}+\vec{a}^{ij}_{+-})=0.$$
So one can define:
$$\vec{a}_{ij}\equiv\vec{a}^{ij}_{+-}+\vec{a}^{ij}_{-+} \Rightarrow \vec{a}^{ij}_{++}+\vec{a}^{ij}_{--}=-\vec{a}_{ij}.$$
Now, from Eqs.~(\ref{beg1})-(\ref{end1}) the following are obvious:
\begin{eqnarray}
\vec{a}^{ij}_{++}&=&\frac{1}{2}\left[\eta\left(\hat{n}_i+\hat{n}_j\right)-\vec{a}_{ij}\right],\\
\vec{a}^{ij}_{+-}&=&\frac{1}{2}\left[\eta\left(\hat{n}_i-\hat{n}_j\right)+\vec{a}_{ij}\right],\\
\vec{a}^{ij}_{-+}&=&\frac{1}{2}\left[\eta\left(-\hat{n}_i+\hat{n}_j\right)+\vec{a}_{ij}\right],\\
\vec{a}^{ij}_{--}&=&\frac{1}{2}\left[\eta\left(-\hat{n}_i-\hat{n}_j\right)-\vec{a}_{ij}\right].
\end{eqnarray}
This gives the general form of the joint measurement POVMs. For qubit effects, $G^{ij}_{X_iX_j}$, where $X_i, X_j \in \{+1,-1\}$, the valid effect condition, Eq.~(\ref{valid}),
is equivalent to the following \cite{heinosaari}:
\begin{eqnarray}
 |\vec{a}^{ij}_{++}| &\leq& \frac{\alpha_{ij}}{2} \leq 2-|\vec{a}^{ij}_{++}|,\\
 |\vec{a}^{ij}_{+-}| &\leq& 1-\frac{\alpha_{ij}}{2} \leq 2-|\vec{a}^{ij}_{+-}|,\\
 |\vec{a}^{ij}_{-+}| &\leq& 1-\frac{\alpha_{ij}}{2} \leq 2-|\vec{a}^{ij}_{-+}|,\\
 |\vec{a}^{ij}_{--}| &\leq& \frac{\alpha_{ij}}{2} \leq 2-|\vec{a}^{ij}_{--}|.
\end{eqnarray}
These inequalities can be combined and rewritten as:
\begin{equation}
 2 \max \left\{|\vec{a}^{ij}_{++}|,|\vec{a}^{ij}_{--}|\right\}\leq\alpha_{ij}\leq 2-2 \max \left\{|\vec{a}^{ij}_{+-}|,|\vec{a}^{ij}_{-+}|\right\},
\end{equation}
where 
\begin{eqnarray}\nonumber
\max\left\{|\vec{a}^{ij}_{++}|,|\vec{a}^{ij}_{--}|\right\} \\ =\sqrt{\frac{\eta^2}{2}(1+\hat{n}_i.\hat{n}_j)+\frac{|\vec{a}_{ij}|^2}{4}+\frac{\eta}{2}|(\hat{n}_i+\hat{n}_j).\vec{a}_{ij}|}\nonumber
\end{eqnarray}
and
\begin{eqnarray}\nonumber
\max \left\{|\vec{a}^{ij}_{+-}|,|\vec{a}^{ij}_{-+}|\right\} \\ =\sqrt{\frac{\eta^2}{2}(1-\hat{n}_i.\hat{n}_j)+\frac{|\vec{a}_{ij}|^2}{4}+\frac{\eta}{2}|(\hat{n}_i-\hat{n}_j).\vec{a}_{ij}|}.\nonumber
\end{eqnarray}
This is the condition for a valid joint measurement used in inequalities of Eqs.~(\ref{valid1}-\ref{valid2}).

\section{No state-independent violation of LSW inequality}
We will now show that no state-independent violation of the LSW inequality with qubit POVMs is possible.
\begin{theorem}
There exists no state-independent violation of the LSW inequality $R_3 \leq 1-\frac{\eta}{3}$ using a triple of qubit POVMs, $M_k\equiv \{E^k_{\pm}\}, k\in\{1,2,3\}$,
that are pairwise jointly measurable (but not necessarily triplewise jointly measurable). 
\end{theorem}
\begin{proof}
In quantum theory, the probability $R^Q_3$ for anticorrelation of measurement outcomes for pairwise joint measurements of $M_k\equiv\{E^k_+,E^k_-\}$ (where $k\in\{1,2,3\}$) has the following form
for a qubit state $\rho$:
\begin{equation}
R^Q_3\equiv\frac{1}{3}\sum_{(ij)\in\{(12),(23),(13)\}}\textrm{Tr}\left(\rho\left(G^{ij}_{+-}+G^{ij}_{-+}\right)\right),
\end{equation}
The condition for violation of noncontextual inequality, Eq.~(\ref{ncineq}), is $R_3^Q > 1-\frac{\eta}{3}$. Using Eq.~(\ref{anticorr}), this reduces to 
\begin{equation}
 \textrm{Tr}\left(\rho\sum_{(ij)}(\alpha_{ij}I-\vec{\sigma}.\vec{a}_{ij})\right)<2\eta
\end{equation}
Using the standard $2 \times 2$ Pauli matrices and $\rho$ parameterized by $0\leq q\leq 1$ and $\hat{n}=(\sin \theta \cos \phi, \sin \theta \sin \phi, \cos \theta)$:
\begin{eqnarray}\label{state1}
 \rho&=&q|\psi\rangle\langle\psi|+(1-q)(I-|\psi\rangle\langle\psi|),\\\label{state2}
 |\psi\rangle&=&\left( \begin{array}{c}
\cos \frac{\theta}{2}\\
e^{i\phi}\sin \frac{\theta}{2}
\end{array}\right)=\cos \frac{\theta}{2}|0\rangle+e^{i\phi}\sin \frac{\theta}{2}|1\rangle,
\end{eqnarray}
the condition for violation becomes
\begin{equation}\label{viol}
 \sum_{(ij)}\alpha_{ij}+\lambda_{\rho}<2\eta,
\end{equation}
where 
\begin{equation}\label{statedep}
\lambda_{\rho}\equiv(1-2q)\vec{a}.\hat{n} \in \left[-|\vec{a}|,|\vec{a}|\right]
\end{equation}
denotes the state-dependent term in the condition and $\vec{a}=(a_x,a_y,a_z)$ is given by
\begin{equation}
 a_x=\sum_{(ij)}(\vec{a}_{ij})_x, \quad a_y=\sum_{(ij)}(\vec{a}_{ij})_y, \quad a_z=\sum_{(ij)}(\vec{a}_{ij})_z.
\end{equation}
For a state-independent violation, either the state-dependent term in Eq.~(\ref{viol}),
$\lambda_{\rho}$, must vanish for all qubit states $\rho$, or  $\sum_{(ij)}\alpha_{ij}+\max_{\rho}\lambda_{\rho}<2\eta$ should hold. The first case, $\lambda_{\rho}=0 \quad \forall \rho$, 
requires $\vec{a}=0$, since $\vec{a}$ is the only term in $\lambda_{\rho}$ that depends on the joint measurement POVM. 
This means $a_x=a_y=a_z=0$, so that $\lambda_{\rho}=0$ for all $\rho$. The second case requires $\sum_{(ij)}\alpha_{ij}+|\vec{a}|<2\eta$. In both cases,
we have the following lower bound on $\alpha_{ij}$, from inequality, Eq.~(\ref{valid1}):
\begin{equation}
\alpha_{ij}>\sqrt{2}\eta\sqrt{1+\hat{n}_i.\hat{n}_j}  
\end{equation}
Taking the sum of $\alpha_{ij}$, $(ij) \in \{(12),(23),(13)\}$, we have
\begin{equation}
\sum_{(ij)}\alpha_{ij}>\sqrt{2}\eta\sum_{(ij)}\sqrt{1+\hat{n}_i.\hat{n}_j}
\end{equation}
For the first case, the condition for state-independent violation is, $\sum_{(ij)}\alpha_{ij}<2\eta$, while for the second case the condition for such a violation is $\sum_{(ij)}\alpha_{ij}+|\vec{a}|<2\eta$.
Given the lower bound on $\sum_{(ij)}\alpha_{ij}$, it follows that a necessary condition for state-independent violation of the LSW inequality
is:
\begin{equation}
 \sum_{(ij)}\sqrt{1+\hat{n}_i.\hat{n}_j}<\sqrt{2}.
\end{equation}
We will show that there exists no choice of measurement directions that will satisfy this necessary condition, thereby ruling out a state-independent violation of the LSW inequality.
The particular cases of orthogonal axes ($\hat{n}_i.\hat{n}_j=0$) or trine spin axes ($\hat{n}_i.\hat{n}_j=-1/2$), used in \cite{LSW}, are clearly ruled out by this necessary condition.
Denoting $\hat{n}_i.\hat{n}_j\equiv\cos\theta_{ij}$, the necessary condition for violation is
\begin{equation}\label{necc}
  \left|\cos \frac{\theta_{12}}{2}\right|+\left|\cos \frac{\theta_{13}}{2}\right|+\left|\cos \frac{\theta_{23}}{2}\right|<1
\end{equation}
Without loss of generality, the three directions can be parameterized as:
\begin{eqnarray}\label{mmts1}
 \hat{n}_1&\equiv&(0,0,1),\\
 \hat{n}_2&\equiv&(\sin \theta_{12},0,\cos \theta_{12}),\\
 \hat{n}_3&\equiv&(\sin \theta_{13} \cos \phi_3, \sin \theta_{13} \sin \phi_3, \cos \theta_{13}).\label{mmts2}
\end{eqnarray}
where
$$0<\frac{\theta_{ij}}{2}<\frac{\pi}{2} \quad \forall (ij) \in \{(12),(13),(23)\},\quad 0\leq \phi_3< 2\pi,$$
and $\cos \theta_{23}=\sin \theta_{12}\sin \theta_{13} \cos \phi_3+\cos \theta_{12}\cos \theta_{13}.$
This implies:
\begin{equation}\label{angles}
\cos(\theta_{12}+\theta_{13})\leq \cos(\theta_{23}) \leq \cos(\theta_{12}-\theta_{13}).
\end{equation}
Then 
\begin{eqnarray}\nonumber
&&\min_{\theta_{12}, \theta_{13},\theta_{23}}\left\{\left|\cos \frac{\theta_{12}}{2}\right|+\left|\cos \frac{\theta_{13}}{2}\right|+\left|\cos \frac{\theta_{23}}{2}\right|\right\}\geq\\\nonumber
&&\min_{\theta_{12},\theta_{13}}\left\{\left|\cos \frac{\theta_{12}}{2}\right|+\left|\cos \frac{\theta_{13}}{2}\right|+\sqrt{\frac{1+\cos(\theta_{12}+\theta_{13})}{2}}\right\}>1. 
\end{eqnarray}
This contradicts the necessary condition (\ref{necc}). Hence, there is no state-independent violation of the LSW inequality (\ref{ncineq}) allowed by noisy spin-1/2 observables.
\end{proof}

\section{Quantum violation of the LSW inequality}
{\it State-dependent violation of the LSW inequality. ---}
The LSW inequality can be violated in a state-dependent manner.
From the condition for violation, Eq.~(\ref{viol}), it follows that a necessary condition for state-dependent violation is $\sum_{(ij)}\alpha_{ij}-|\vec{a}|<2\eta$.
An optimal choice of $\rho$ that yields $\lambda_{\rho}=-|\vec{a}|$ corresponds to $q=1$ and $\vec{a}.\hat{n}=|\vec{a}|$, i.e.,
$$\cos \theta = \frac{a_z}{|\vec{a}|}, \quad \tan \phi=\frac{a_y}{a_x}.$$

With this choice of $\rho$ the question becomes: Does there exist a choice of $\{\hat{n}_1,\hat{n}_2,\hat{n}_3\}, \eta, \{\alpha_{ij}, \vec{a}_{ij}\}$ such that
$\sum_{(ij)}\alpha_{ij}-|\vec{a}|<2\eta$? We show that this is indeed the case. We define
\begin{equation}\label{C}
 C\equiv2\eta-\left(\sum_{(ij)}\alpha_{ij}-|\vec{a}|\right),
\end{equation}
so that $C>0$ indicates a state-dependent violation. Note that violation of the LSW inequality $R_3^Q\leq 1-\frac{\eta}{3}$ is characterized by
\begin{equation}
 S\equiv R_3^Q-\left(1-\frac{\eta}{3}\right)=\frac{C}{6}
\end{equation}
where $S>0$ for a state-dependent violation. Given a coplanar choice of $\{\hat{n}_1,\hat{n}_2,\hat{n}_3\}$, and $\eta$ satisfying $\eta_l<\eta\leq\eta_u$, the optimal value of $C$
---denoted as $C^{\{\hat{n}_i\},\eta}_{\max}$---is given by 
\begin{equation}\nonumber
C^{\{\hat{n}_i\},\eta}_{\max}=2\eta+\sum_{(ij)}\left(\sqrt{1+\eta^4(\hat{n}_i.\hat{n}_j)^2-2\eta^2}-(1+\eta^2 \hat{n}_i.\hat{n}_j)\right),
\end{equation}
as we will prove after stating the criteria for violating the LSW inequality.
We obtain a state-dependent violation of the LSW inequality for trine axes (Fig. \ref{plane}):
\begin{theorem}\label{thm2}
The optimal violation of the LSW inequality for measurements along trine spin axes, i.e., $\{\hat{n}_1,\hat{n}_2,\hat{n}_3\}$ such that $\hat{n}_1.\hat{n}_2=\hat{n}_2.\hat{n}_3=\hat{n}_1.\hat{n}_3=-1/2$, occurs for $|\psi\rangle=\frac{1}{\sqrt{2}}(|0\rangle+i|1\rangle)$ if the plane of
measurements is the ZX plane. The lower and upper bounds on $\eta$ are $\eta_l=\frac{2}{3}\approx0.667$ and $\eta_u=\sqrt{3}-1\approx 0.732$. The joint measurement POVM is given by $\alpha_{ij}=1+\eta^2\hat{n}_i.\hat{n}_j$ and $\vec{a}_{ij}=(0,\sqrt{1+\eta^4(\hat{n}_i.\hat{n}_j)^2-2\eta^2},0)$.
The optimal violation corresponds to $\eta \rightarrow \eta_l$, so that $\alpha_{12}=\alpha_{13}=\alpha_{23}\rightarrow 1-\frac{\eta_l^2}{2}=\frac{7}{9}$, $|\vec{a}_{ij}|\rightarrow \frac{\sqrt{13}}{9} \quad \forall (ij)$,
$C^{trine}_{\max}\rightarrow \frac{\sqrt{13}}{3}-1 \approx 0.20185$, and $S^{trine}_{\max}=\frac{C^{\textrm{trine}}_{\max}}{6} \rightarrow 0.03364 \textrm{ or } 3.36\%$.
\end{theorem}

\begin{figure}\centering
\includegraphics[scale=0.3]{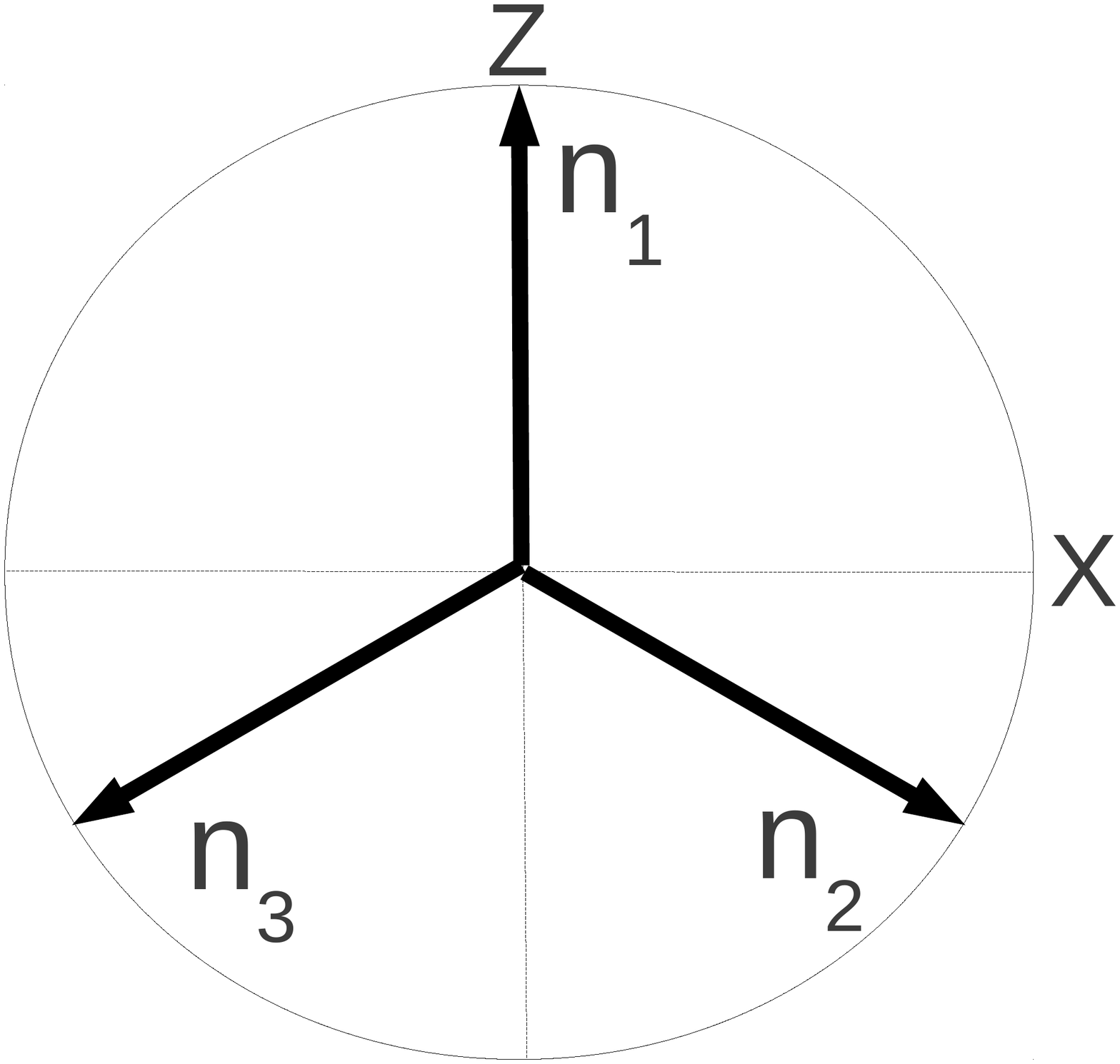}
\caption{Choice of measurement directions $\{\hat{n}_1,\hat{n}_2,\hat{n}_3\}$ along trine spin axes in the Z-X plane.}
\label{plane}
\end{figure}

Thus the quantum probability of anticorrelation can exceed the generalized-noncontextual bound by an amount arbitrarily close to $0.03364$ or about $3.36\%$ for trine spin axes. We conjecture that this is the optimal violation of the LSW inequality obtainable from qubit POVMs. 
The quantum degree of anti-correlation for this violation is $R_3^Q=S^{trine}_{\max}+\big(1-\frac{\eta}{3}\big)\rightarrow 0.8114$ and the generalized-noncontextual bound is $\big(1-\frac{\eta}{3}\big)\rightarrow\frac{7}{9}\approx 0.7778$.

The proof of Theorem \ref{thm2} follows:
\begin{proof}
{\bf Optimal state-dependent violation for measurements in a plane: }
We need to maximize $C\equiv 2\eta-\left(\sum_{(ij)}\alpha_{ij}-|\vec{a}|\right)$ to obtain the optimal violation of the LSW inequality. Subject to satisfaction of the joint measurability constraints 
in Eqs.~(\ref{valid1}-\ref{valid2}), we have
\begin{eqnarray}\nonumber
C_{max}&=&\max_{\{\hat{n}_1, \hat{n}_2, \hat{n}_3\},\{\vec{a}_{ij}\}, \eta} \left\{2\eta+|\vec{a}|-\sum_{(ij)}\alpha_{ij}\right\}\\\nonumber
&\leq& \max_{\{\hat{n}_1, \hat{n}_2, \hat{n}_3\},\{\vec{a}_{ij}\},\eta} \left\{2\eta+\sum_{(ij)}|\vec{a}_{ij}|-\sum_{(ij)}\sqrt{2\eta^2(1+\hat{n}_i.\hat{n}_j)+|\vec{a}_{ij}|^2}\right\}\\\nonumber
\end{eqnarray}
The inequality above follows from the fact that
\begin{equation}
 |\vec{a}|=\sqrt{\sum_{(ij)}|\vec{a}_{ij}|^2+2\left(\vec{a}_{12}.\vec{a}_{13}+\vec{a}_{12}.\vec{a}_{23}+\vec{a}_{13}.\vec{a}_{23}\right)},
\end{equation}
so that $|\vec{a}|\leq \sum_{(ij)}|\vec{a}_{ij}|$, and 
\begin{equation}
 \sum_{(ij)}\sqrt{2\eta^2(1+\hat{n}_i.\hat{n}_j)+|\vec{a}_{ij}|^2} \leq \sum_{(ij)}\sqrt{2\eta^2(1+\hat{n}_i.\hat{n}_j)+|\vec{a}_{ij}|^2+2\eta|(\hat{n}_i+\hat{n}_j).\vec{a}_{ij}|}
 \leq \sum_{(ij)}\alpha_{ij}.\nonumber
\end{equation}
Also, we have
\begin{equation}\nonumber
 \sum_{(ij)}\sqrt{2\eta^2(1+\hat{n}_i.\hat{n}_j)+|\vec{a}_{ij}|^2}\geq \sum_{(ij)}\sqrt{2\eta^2(1+\hat{n}_i.\hat{n}_j)+|\vec{a}_{ij}|^2} \big|_{\text{coplanar}, \phi_3=\pi}.
\end{equation}
That is, for a fixed $|\vec{a}_{ij}|$, $\sum_{(ij)}\sqrt{2\eta^2(1+\hat{n}_i.\hat{n}_j)+|\vec{a}_{ij}|^2}$ is smallest when the measurement directions
$\{\hat{n}_1, \hat{n}_2, \hat{n}_3\}$ are coplanar and $\phi_3=\pi$. From Eqs.~(\ref{mmts1}-\ref{mmts2}), we have $\hat{n}_2.\hat{n}_3=\cos \theta_{23}=\sin \theta_{12}\sin \theta_{13} \cos \phi_3+\cos \theta_{12}\cos \theta_{13}$.
When $\phi_3=0 \textrm{ or }\pi$, the three measurements are coplanar and there are only two free angles, $\hat{n}_1.\hat{n}_2=\cos \theta_{12}$
and $\hat{n}_1.\hat{n}_3=\cos \theta_{13}$, while the third angle is fixed by these two: $\hat{n}_2.\hat{n}_3=\cos \theta_{23}=\cos (\theta_{12}-\theta_{13}) \textrm{ or }\cos (\theta_{12}+\theta_{13})$. Since 
$\cos(\theta_{12}+\theta_{13})\leq \cos(\theta_{23}) \leq \cos(\theta_{12}-\theta_{13})$, for any given $\theta_{12}$ and $\theta_{13} \in (0,\pi)$, $\cos \theta_{23}$ is smallest when
$\phi_3=\pi$. Hence, we choose the three measurements to be coplanar such that $\phi_3=\pi$ and $\cos \theta_{23}=\cos (\theta_{12}+\theta_{13})$. Any other choice of $\{\hat{n}_1, \hat{n}_2, \hat{n}_3\}$
will give a larger value of $\cos \theta_{23}$, hence also $\sum_{(ij)}\sqrt{2\eta^2(1+\hat{n}_i.\hat{n}_j)+|\vec{a}_{ij}|^2}$. So,
\begin{equation}\nonumber
C_{max}\leq \max_{\hat{n}_1.\hat{n}_2, \hat{n}_1.\hat{n}_3, \{|\vec{a}_{ij}|\}, \eta} \left\{2\eta+\sum_{(ij)}|\vec{a}_{ij}|-\sum_{(ij)}\sqrt{2\eta^2(1+\hat{n}_i.\hat{n}_j)+|\vec{a}_{ij}|^2}\big|_{\text{coplanar}, \phi_3=\pi}\right\}.
\end{equation}
We will now argue that this inequality for $C_{max}$ can be replaced by an equality. Let us take coplanar measurement directions $\{\hat{n}_1, \hat{n}_2, \hat{n}_3\}$ such that $\phi_3=\pi$.
We also take all the $\vec{a}_{ij}$ parallel to each other, i.e., $\vec{a}_{12}.\vec{a}_{13}=|\vec{a}_{12}||\vec{a}_{13}|$, $\vec{a}_{12}.\vec{a}_{23}=|\vec{a}_{12}||\vec{a}_{23}|$,
and $\vec{a}_{13}.\vec{a}_{23}=|\vec{a}_{13}||\vec{a}_{23}|$, so that $|\vec{a}|=|\vec{a}_{12}|+|\vec{a}_{13}|+|\vec{a}_{23}|$. Besides, $|(\hat{n}_i+\hat{n}_j).\vec{a}_{ij}|=0$ $\forall (ij)\in \{(12),(13),(23)\}$.
From these conditions it follows that each $\vec{a}_{ij}$ is perpendicular to the plane and 
$\forall (ij) \in \{(12),(13),(23)\}$, $\vec{a}_{ij}.\hat{n}_i=\vec{a}_{ij}.\hat{n}_j=0$. This allows us to choose $\alpha_{ij}=\sqrt{2\eta^2(1+\hat{n}_i.\hat{n}_j)+|\vec{a}_{ij}|^2}$.
So, in our optimal configuration, the measurement directions are coplanar while the $\vec{a}_{ij}$'s are parallel to each other and perpendicular to the plane of measurements. Note that this also means $\vec{a}$ will be parallel to $\vec{a}_{ij}$ and therefore perpendicular to
the plane of measurements, and so will be the optimal state (which is parallel to $\vec{a}$). With these optimality conditions satisfied, the optimal violation can now be written as
\begin{equation}
 C_{max}=\max_{\hat{n}_1.\hat{n}_2, \hat{n}_1.\hat{n}_3, \{|\vec{a}_{ij}|\},\eta}\left\{2\eta+\sum_{(ij)}\left(|\vec{a}_{ij}|-\sqrt{2\eta^2(1+\hat{n}_i.\hat{n}_j)+|\vec{a}_{ij}|^2}\right)\right\}.
\end{equation}
The constraints from joint measurability, Eqs.~(\ref{valid1}-\ref{valid2}) become
\begin{equation}
|\vec{a}_{ij}|\leq \sqrt{1+\eta^4(\hat{n}_i.\hat{n}_j)^2-2\eta^2}.
\end{equation}
Now,
\begin{equation}\nonumber
C_{max}\leq\max_{\hat{n}_1.\hat{n}_2, \hat{n}_1.\hat{n}_3, \{|\vec{a}_{ij}|\},\eta}\left\{2\eta+\sum_{(ij)}\left(\sqrt{1+\eta^4(\hat{n}_i.\hat{n}_j)^2-2\eta^2}-(1+\eta^2 \hat{n}_i.\hat{n}_j)\right)\right\}.
\end{equation}
The upper bound follows from the fact that $f(x,y)=x-\sqrt{x^2+2\eta^2(1+y)}$, where $0\leq x \leq \sqrt{1+\eta^4y^2-2\eta^2}$ and $-1<y<1$, is an increasing function of 
$x$ for a fixed $y$, i.e., $(\frac{\partial{f}}{\partial{x}})_y>0$. Here $x\equiv|\vec{a}_{ij}|$ and $y\equiv \hat{n}_i.\hat{n}_j$. So, taking $|\vec{a}_{ij}|=\sqrt{1+\eta^4(\hat{n}_i.\hat{n}_j)^2-2\eta^2}$,
we have

\begin{equation}\nonumber
C_{max}^{\{\hat{n}_i\},\eta}\equiv 2\eta+\sum_{(ij)}\left(\sqrt{1+\eta^4(\hat{n}_i.\hat{n}_j)^2-2\eta^2}-(1+\eta^2 \hat{n}_i.\hat{n}_j)\right)
\end{equation}

Note that $\alpha_{ij}=1+\eta^2 \hat{n}_i.\hat{n}_j$ for $|\vec{a}_{ij}|=\sqrt{1+\eta^4(\hat{n}_i.\hat{n}_j)^2-2\eta^2}$. $C_{max}^{\{\hat{n}_i\},\eta}$ 
is the maximum value of $C$ for a given coplanar choice of measurement directions $\{\hat{n}_1,\hat{n}_2,\hat{n}_3\}$ and sharpness parameter $\eta$.
\end{proof}

\section{Chapter summary}
Having shown how to violate the LSW inequality, let us now examine why such a violation is interesting beyond the reasons already outlined. The LSW inequality 
takes into account, for example, the possibility that the measurement apparatus could introduce anticorrelations that have nothing
to do with ontic state(s) one could associate with the system's preparation.\footnote{As argued in the Appendix at the end of this chapter. See Ref.~\cite{LSW} for more discussion on this issue.}
This would allow violation of the KS-noncontextual
bound of $\frac{2}{3}$ when the measurement is unsharp ($\eta<1$) even though this violation could purely be a result of 
noise coming from elsewhere, such as the measurement apparatus, rather than a consequence of quantum theory. A violation of the LSW inequality bound, on
the other hand, rules out this possibility.

An interesting open question is whether such a violation is possible in higher dimensional systems and whether the amount of violation could be higher for these than for a qubit.
Our result also hints at the fact that perhaps all compatibility scenarios may be realizable and contextuality demonstrated for many of them if we consider the possibilities that general quantum measurements allow.
In particular, scenarios that involve pairwise compatibility between all measurements but no global compatibility may be realizable within quantum theory. Specker's scenario is the simplest such
example we have considered. Indeed, as we show in Chapter 4, quantum theory does admit all compatibility scenarios since it allows one to realize any conceivable set of (in)compatibility
relations between a set of observables.

Whether a state-independent violation of the LSW inequality is possible in higher dimensions also remains an open question.

In summary, the joint measurability allowed in a theory restricts the kind of compatibility scenarios that can arise in it. Quantum theory admits Specker's compatibility scenario with unsharp 
measurements \cite{genNC}.
Further, as we have shown, quantum theory allows violation of the LSW inequality in this scenario. Thus quantum theory is contextual even in the simplest compatibility scenario involving joint measurements
where contextuality is conceivable.
Whether, and to what extent, this is the case with more complicated compatibility scenarios realizable, for example, via the construction in Chapter 4 remains to be explored.

\section*{APPENDIX}

\subsection*{\bf LSW inequality: Noncontextuality vs. KS-noncontextuality}

The traditional assumption of KS-noncontextuality entails two things: measurement noncontextuality and outcome-determinism for sharp measurements \cite{genNC}.
Given a set of measurements $\{M_1,\dots,M_N\}$, measurement noncontextuality is the assumption that the response function for each measurement
is insensitive to contexts - jointly measurable subsets - that it may be a part of: $\forall M_i, p(X_i|M_i;\lambda) \in [0,1]$. Here $X_i$ is an 
outcome for measurement $M_i$ and $\lambda$ is the hidden variable associated with the system's preparation. Outcome-determinism
is the further assumption that $\forall M_i, p(X_i|M_i;\lambda) \in \{0,1\}$, i.e., response functions are outcome-deterministic.
A KS-noncontextual model is one that makes these two assumptions for sharp (projective) measurements. A KS-inequality is a constraint on 
measurement statistics obtained under these two assumptions. 
A noncontextual model \`a la Spekkens, on the other hand, derives outcome-determinism for sharp measurements as a consequence of preparation noncontextuality \cite{genNC},
as we also demonstrated in Chapter 1.
For unsharp measurements, however, outcome-determinism is not implied by noncontextuality and one needs to model these 
measurements by outcome-indeterministic response functions. In our case, the qubit effects we need to write the response functions for are of the form: $E^k_{\pm}=\eta \Pi^k_{\pm}+(1-\eta)\frac{I}{2}$.
We will relabel the outcomes according to $\{+1\rightarrow 0, -1\rightarrow 1\}$ so that $X_k \in \{0,1\}$ in what follows.
Liang, Spekkens and Wiseman (LSW) argued \cite{LSW} that the response function for these effects in a noncontextual model
should be $p(X_k|M_k;\lambda)=\eta[X_k(\lambda)]+(1-\eta)\left(\frac{1}{2}[0]+\frac{1}{2}[1]\right)$,
where $p(X)=[x]$ denotes the point distribution given by the Kronecker delta function $\delta_{X,x}$. For $\eta=1$ (sharp measurements) this would be the traditional KS-noncontextual model. When $\eta<1$ (unsharp measurements),
the second ``coin flip'' term in the response function, $(\frac{1}{2}[0]+\frac{1}{2}[1])$, begins to play a role. This term is not conditioned
by $\lambda$, the ontic state sampled by the system's preparation, but is instead the response function for tossing a fair
coin regardless of what measurement is being made. It characterizes the random noise introduced, for example, by the measuring apparatus. 
The important thing to note is that this noise is uncorrelated with the system's ontic state $\lambda$.

Given these single-measurement response functions, one needs to figure out pairwise response functions for pairwise
joint measurements of the three qubit POVMs. LSW \cite{LSW} argued that the pairwise response functions maximizing the average anti-correlation
$R_3$ and consistent with the single-measurement response functions are given by 
\begin{eqnarray}\nonumber
 p(X_i,X_j|M_{ij};\lambda)&=&\eta [X_i(\lambda)][X_j(\lambda)]\\
 &+&(1-\eta)\left(\frac{1}{2}[0][1]+\frac{1}{2}[1][0]\right),
\end{eqnarray}
for all pairs of measurements $(ij)\in \{(12),(13),(23)\}$. This noncontextual model for these measurements turns out to be 
KS-contextual in the sense that the three pairwise response functions do not admit a joint probability distribution over the three measurement
outcomes, $p(X_1,X_2,X_3|\lambda)$, that is consistent with all three of them. Indeed, this LSW-model maximizes the average anticorrelation
possible in Specker's scenario given the single-measurement response functions, thus allowing us to obtain the LSW inequality
$R_3\leq 1-\frac{\eta}{3}$. Let us note the two bounds separately:

\begin{eqnarray}
 R_3 &\leq& R_3^{KS}\equiv\frac{2}{3}\\
 R_3 &\leq& R_3^{LSW}\equiv1-\frac{\eta}{3}
\end{eqnarray}

We have shown that there exists a choice of measurement directions, $\{\hat{n}_1,\hat{n}_2,\hat{n}_3\}$, and
a choice of $\eta$ for some state $\rho$ such that the quantum probability of anticorrelation, $R_3^Q$, beats the generalized-noncontextual
bound $R_3^{LSW}$. This rules out the possibility of being able to generate these correlations by classical means, as in the LSW-model, for at least 
some values of sharpness parameter $\eta$. Of course, if $\eta=0$, then the noncontextual bound becomes trivial and the
question of violation does not arise - this situation corresponds to the case where for any of the three pairwise joint measurements,
the measuring apparatus outputs one of the two anticorrelated outcomes by flipping a fair coin and there is no functional dependence
of the response function on $\lambda$. In other words, one could generate perfect anti-correlation in a noncontextual
model if $\eta=0$. However, as long as one is performing a nontrivial measurement (where $\eta>0$) there is a constraint on
the degree of anticorrelation imposed by noncontextuality. What we establish is that noncontextuality
cannot account for the degree of anticorrelation observed in quantum theory. Clearly, quantum theory is nonclassical even given a more stringent 
benchmark than the one set by KS-noncontextuality. A violation of the KS-noncontextual bound, $R_3^{KS}$, is possible in a noncontextual model \`a la Spekkens,
so such a violation is not in itself a signature of nonclassicality. On the other hand, violation of the noncontextual bound, $R_3^{LSW}$, should 
be considered a signature of genuine nonclassicality in that it isn't attributable either to the system's ontic state or random noise (from the
measuring apparatus or elsewhere) in a noncontextual model.

\chapter{On the connection between joint measurability and contextuality}
In this chapter we take a critical look at a particular condition, namely triplewise incompatibility, that we required of the three measurements we considered in Chapter 2 for realizing 
Specker's scenario, following the prescription in the LSW paper\cite{LSW}. Following our demonstration of the violation of LSW inequality\cite{KG},
in Ref.~\cite{FYuOh} a peculiar feature of POVMs with respect to joint measurability was pointed out: that there exist three measurements which are pairwise 
jointly measurable and triplewise jointly measurable but for which there exist pairwise joint measurements which do not admit a triplewise joint measurement. 
We will focus on the logical relationship between joint measurability and the possibility of contextuality and in the process shed some light on 
the crucial differences between joint measurability of projective (sharp) and nonprojective (unsharp) measurements. We will see that the triplewise incompatibility 
of the three POVMs is, in fact, {\em not} necessary to see a violation of the LSW inequality. Throughout this chapter, `sharp measurement' will be synonymous with
projection-valued measures (PVMs) and `unsharp measurement' will be synonymous with POVMs that are not PVMs.

This chapter is based on work reported in Ref.~\cite{RK}.
\section{Uniqueness of joint measurement for sharp (or projective) measurements}
Since the peculiarity of positive-operator valued measures (POVMs) in cases of interest here arises from the nonuniqueness of joint measurements, I will first prove the uniqueness of joint measurements for projection-valued measures (PVMs).
This will help clarify how the distinction between sharp and unsharp measurements comes to play a role in Specker's scenario \cite{KG}. 

Consider a nonempty set $\Omega_i$ and a $\sigma$-algebra $\mathcal{F}_i$ of subsets of $\Omega_i$, for $i\in\{1,\dots,N\}$. The POVM $M_i$ is defined as the map $M_i: \mathcal{F}_i\rightarrow \mathcal{B}_+(\mathcal{H})$,
where $M_i(\bigcup X_i)=\sum_{X_i}M_i(X_i)=I$, $\bigcup X_i$ ($=\Omega_i$) being a union of disjoint subsets $X_i\in\mathcal{F}_i$, and $\mathcal{B}_+(\mathcal{H})$ denotes the set of positive semidefinite operators on a Hilbert space $\mathcal{H}$. $I$ is the identity operator
on $\mathcal{H}$. Therefore: $M_i\equiv\{M_i(X_i)|X_i\in\mathcal{F}_i\}$, where $X_i$ labels the elements of POVM $M_i$. $M_i$ becomes a projection-valued measure (PVM) under the 
additional constraint $M_i(X_i)^2=M_i(X_i)$ for all $X_i\in \mathcal{F}_i$.

\begin{theorem}\label{uniqueness}
Given a set of POVMs, $\{M_1,\dots,M_N\}$, all of which except at most one---say $M_N$---are PVMs, so that for $i\in\{1,\dots,N-1\}$
$$M_i\equiv\{M_i(X_i)|X_i\in\mathcal{F}_i, M_i(X_i)^2=M_i(X_i)\}$$ and $$M_N\equiv\{M_N(X_N)|X_N\in\mathcal{F}_N\},$$
the set of POVMs, $\{M_1,\dots,M_N\}$, is jointly measurable if and only if they commute pairwise, i.e.,
$$M_j(X_j)M_k(X_k)=M_k(X_k)M_j(X_j),$$
for all $j,k\in\{1,\dots,N\}$ and $X_j\in\mathcal{F}_j, X_k\in\mathcal{F}_k$. In this case, there exists a unique joint POVM $M$, 
defined as a map $$M:\mathcal{F}_1\times\mathcal{F}_2\times\dots\times\mathcal{F}_N \rightarrow \mathcal{B}_+(\mathcal{H}),$$ such that
$$M(X_1\times\dots\times X_N)=M_1(X_1)M_2(X_2)\dots M_N(X_N),$$
for all $X_1\times\dots\times X_N \in\mathcal{F}_1\times\dots\times \mathcal{F}_N.$
\end{theorem}

\begin{proof}
This proof is adapted from, and is a generalization of, the proof of Proposition 8 in the Appendix of Ref. \cite{heinosaari}. 

The first part of the proof is for the implication: joint measurability $\Rightarrow$ pairwise commutativity - A joint POVM for $\{M_1,\dots,M_N\}$ is defined as a map
$M:\mathcal{F}_1\times\mathcal{F}_2\times\dots\times\mathcal{F}_N \rightarrow \mathcal{B}_+(\mathcal{H})$, such that
\begin{equation}
 M_i(X_i)=\sum_{\{X_j\in\mathcal{F}_j|j\neq i\}}M(X_1\times\dots\times X_N)
\end{equation}
for all $X_i\in\mathcal{F}_i$, $i\in\{1\dots N\}$. Also, $M(X_1\times\dots\times X_N)\leq M_1(X_1)$, so the range of $M(X_1\times\dots\times X_N)$ is contained in the range of 
$M_1(X_1)$, and therefore:
\begin{equation}
M_1(X_1)M(X_1\times\dots\times X_N)=M(X_1\times\dots\times X_N).
\end{equation}
Using this relation for the complement $\Omega_1\backslash X_1 \in \mathcal{F}_1$:
\begin{eqnarray}
&&M_1(X_1)M(\Omega_1\backslash X_1\times\dots\times X_N)\nonumber\\
&&=(I-M_1(\Omega_1\backslash X_1))M(\Omega_1\backslash X_1\times\dots\times X_N)\nonumber\\
&&=0.
\end{eqnarray}
Taking the adjoints, it follows that
\begin{equation}
M(X_1\times\dots\times X_N)M_1(X_1)=M(X_1\times\dots\times X_N),
\end{equation}
and
\begin{equation}
M(\Omega_1\backslash X_1\times\dots\times X_N)M_1(X_1)=0.
\end{equation}
Denoting 
$$M^{(i)}(X_{i+1}\times\dots\times X_N)\equiv\sum_{\{X_j\in\mathcal{F}_j|j\leq i\}}M(X_1\times\dots\times X_N),$$
this implies:
\begin{eqnarray}
&&M_1(X_1)M^{(1)}(X_2\times\dots\times X_N)\nonumber\\
&=&M_1(X_1)M(X_1\times\dots\times X_N)\nonumber\\&&+M_1(X_1)M(\Omega_1\backslash X_1\times\dots\times X_N)\nonumber\\
&=&M_1(X_1)M(X_1\times\dots\times X_N)\nonumber\\
&=&M(X_1\times\dots\times X_N).
\end{eqnarray}
Taking the adjoint,
\begin{equation}
M^{(1)}(X_2\times\dots\times X_N)M_1(X_1)=M(X_1\times\dots\times X_N).
\end{equation}
Therefore:
\begin{eqnarray}
&&M_1(X_1)M^{(1)}(X_2\times\dots\times X_N)\nonumber\\
&=&M^{(1)}(X_2\times\dots\times X_N)M_1(X_1)\nonumber\\
&=&M(X_1\times\dots\times X_N).
\end{eqnarray}
Noting that $M^{(i-1)}(X_i\times\dots\times X_N)\leq M_i(X_i)$, one can repeat the above procedure for $M_i$, $i\in\{2,\dots,N-1\},$ to obtain:
\begin{eqnarray}
&&M^{(i-1)}(X_i\times\dots\times X_N)\nonumber\\
&=&M_i(X_i)M^{(i)}(X_{i+1}\times\dots\times X_N)\nonumber\\
&=&M^{(i)}(X_{i+1}\times\dots\times X_N)M_i(X_i).
\end{eqnarray}
Doing this recursively until $i=N-1$ and noting that $M^{(N-1)}(X_N)=M_N(X_N)$, it follows:
\begin{eqnarray}
&&M(X_1\times\dots\times X_N)\nonumber\\
&=&M_1(X_1)M^{(1)}(X_2\times\dots\times X_N)\nonumber\\
&=&M^{(1)}(X_2\times\dots\times X_N)M_1(X_1)\nonumber\\
&&\vdots\nonumber\\
&=&M_1(X_1)M_2(X_2)\dots M_{N-1}(X_{N-1})M_N(X_N)\nonumber\\
&=&M_N(X_N)M_{N-1}(X_{N-1})\dots M_2(X_2)M_1(X_1).\nonumber\\
\end{eqnarray}
For the last equality to hold, the POVM elements must commute pairwise, so that
\begin{equation}
M(X_1\times\dots\times X_N)=\prod_{i=1}^N M_i(X_i).
\end{equation}
This concludes the proof of the implication, joint measurability $\Rightarrow$ pairwise commutativity. The converse is easy to see since the joint POVM defined by taking the 
product of commuting POVM elements, $$\{M(X_1\times\dots\times X_N)=\prod_{i=1}^N M_i(X_i)|X_i\in\mathcal{F}_i\},$$ is indeed a valid POVM which coarse-grains to the given POVMs,
$\{M_1,\dots,M_N\}$.
\end{proof}

Indeed, pairwise commutativity $\Rightarrow$ joint measurability for any arbitrary set of POVMs, $\{M_1,\dots,M_N\}$, and it is only when all but (at most) one of these POVMs are PVMs that 
the converse---and the uniqueness of the joint POVM---holds.

\section{Specker's scenario}
Specker's scenario requires a set of three POVMs, $\{M_1,M_2,M_3\}$, that are pairwise jointly measurable, i.e., $\exists$ POVMs $M_{12}$, $M_{23}$, and $M_{31}$ which measure the respective pairs jointly.
An immediate consequence of the requirement of pairwise joint measurability of $\{M_1,M_2,M_3\}$ is that in quantum theory these three measurements cannot be realized as projective measurements 
(PVMs) and still be expected to show any contextuality. This is because for projective measurements or projection-valued measures (PVMs), a set of three measurements that are pairwise jointly measurable---and therefore admit \emph{unique} pairwise joint measurements---are 
also triplewise jointly measurable in the sense that there exists a \emph{unique} triplewise joint measurement which coarse-grains to each pairwise implementation of the three measurements and therefore also to the single measurements.

From Theorem \ref{uniqueness}, it follows that if $M_i$, $i\in\{1,2,3\}$, are PVMs then they admit unique pairwise and triplewise joint PVMs: 
\begin{eqnarray}
M_{ij}(X_i\times X_j)&=&M_i(X_i)M_j(X_j),\\
M(X_1\times X_2\times X_3)&=&M_1(X_1)M_2(X_2)M_3(X_3),
\end{eqnarray}

corresponding to the maps $M_{ij}:\mathcal{F}_i\times\mathcal{F}_j\rightarrow \mathcal{B}_+(\mathcal{H})$ and $M:\mathcal{F}_1\times\mathcal{F}_2\times\mathcal{F}_3\rightarrow \mathcal{B}_+(\mathcal{H})$,
respectively. Intuitively, this is  easy to see since joint measurability is equivalent to pairwise commutativity for a set of projective measurements and the joint measurement for each pair is unique \cite{heinosaari}.
The existence of a unique joint measurement implies that there exists a joint probability distribution realizable via this joint measurement for any given quantum state, thus explaining the pairwise statistics of the triple of measurements noncontextually in the traditional
Kochen-Specker sense.\footnote{KS-noncontextuality just means that there exists a joint probability distribution over the three measurement outcomes which marginalizes to the pairwise measurement statistics.
Violation of a KS inequality -- obtained under the assumption that a global joint distribution exists -- rules out KS-noncontextuality.}

Clearly, then, the three measurements $\{M_1, M_2, M_3\}$ must necessarily be unsharp for Specker's scenario to exhibit KS-contextuality. The uniqueness of joint measurements 
(pairwise or triplewise) need not hold in this case. I will refer to pairwise joint measurements as ``2-joints'' and triplewise joint measurements as ``3-joints''. Also,
I will use the phrases `joint measurability' and `compatibility' interchangeably since they will refer to the same notion. Consider 
the four propositions regarding the three measurements:

\begin{itemize}
 \item $\exists$ 2-joint: $\{M_1,M_2,M_3\}$ admit 2-joints, i.e. they are pairwise jointly measurable,
 \item $\nexists$ 2-joint: $\{M_1,M_2,M_3\}$ do not admit 2-joints, i.e. at least one pair is not jointly measurable,
 \item $\exists$ 3-joint: $\{M_1,M_2,M_3\}$ admit a 3-joint, i.e. they are triplewise jointly measurable.
 \item $\nexists$ 3-joint: $\{M_1,M_2,M_3\}$ do not admit a 3-joint, i.e. they are not triplewise jointly measurable.
\end{itemize}

The possible pairwise-triplewise propositions for the three measurements are: 
\begin{itemize}
 \item $(\exists \text{ 2-joint}, \exists \text{ 3-joint})$,
 \item $(\exists \text{ 2-joint}, \nexists \text{ 3-joint})$,
 \item $(\nexists \text{ 2-joint}, \nexists \text{ 3-joint})$.
\end{itemize}

Note that the proposition $(\nexists \text{ 2-joint}, \exists \text{ 3-joint})$ is trivially excluded because triplewise compatibility implies pairwise compatibility. Of the three remaining 
propositions, the ones of interest for contextuality are $(\exists \text{ 2-joint}, \exists \text{ 3-joint})$ and $(\exists \text{ 2-joint}, \nexists \text{ 3-joint})$,
since the remaining one is simply about observables that do not admit any joint measurement at all and hence no nontrivial compatibility contexts exist for this proposition.\footnote{
It is worth noting that, if $\{M_1,M_2,M_3\}$ were PVMs, then there are only two possibilities: $(\exists \text{ 2-joint}, \exists \text{ 3-joint})$ and  $(\nexists \text{ 2-joint}, \nexists \text{ 3-joint})$,
since for three PVMs, $\exists \text{ 2-joint} \Leftrightarrow \exists \text{ 3-joint}$, because pairwise commutativity is equivalent to joint measurability and 
the joint measurements are unique on account of Theorem \ref{uniqueness}. This is why KS-contextuality is impossible with PVMs in this scenario.}

It may seem that for purposes of contextuality even the proposition $(\exists \text{ 2-joint}, \exists \text{ 3-joint})$ is of no interest, but there is a subtlety involved here: one is only
considering whether 2-joints or a 3-joint exist for the set $\{M_1, M_2, M_3\}$. Since the statistics that is of relevance for Specker's scenario is the pairwise statistics \cite{LSW, KG}, 
one also needs to consider whether a given choice of 2-joints, $\{M_{12}, M_{23}, M_{31}\}$, admits a 3-joint, i.e., the proposition $(\exists \text{ 3-joint }\big|\text{ a choice of 2-joints})$ or its negation $(\nexists \text{ 3-joint }\big|\text{ a choice of 2-joints})$.
The four possible conjunctions are: 

\begin{itemize}
 \item $(\exists \text{ 2-joint}, \exists \text{ 3-joint})\bigwedge(\exists \text{ 3-joint }\big|\text{ a choice of 2-joints}),$
 \item $(\exists \text{ 2-joint}, \exists \text{ 3-joint})\bigwedge(\nexists \text{ 3-joint }\big|\text{ a choice of 2-joints}),$
 \item $(\exists \text{ 2-joint}, \nexists \text{ 3-joint})\bigwedge(\exists \text{ 3-joint }\big|\text{ a choice of 2-joints}),$
 \item $(\exists \text{ 2-joint}, \nexists \text{ 3-joint})\bigwedge(\nexists \text{ 3-joint }\big|\text{ a choice of 2-joints}).$
\end{itemize}

Of these, the first conjunction rules out the possibility of KS-contextuality, so it is not of interest for the present purpose. The third conjunction is false since the existence of a 3-joint for 
a given choice of 2-joints would also imply the existence of a 3-joint for the three single measurements, hence contradicting the fact that these admit no 3-joints. Thus the two remaining
conjunctions of interest are: 

\begin{itemize}
 \item \emph{Proposition 1}:\\$(\exists \text{ 2-joint}, \exists \text{ 3-joint})\bigwedge(\nexists \text{ 3-joint }\big|\text{ a choice of 2-joints})$,
 \item \emph{Proposition 2}:\\$(\exists \text{ 2-joint}, \nexists \text{ 3-joint})\bigwedge(\nexists \text{ 3-joint }\big|\text{ a choice of 2-joints})$\\
$\Leftrightarrow (\exists \text{ 2-joint}, \nexists \text{ 3-joint})$.
\end{itemize}

Note that for PVMs, both these propositions are false - the latter one especially so - and there is no KS-contextuality to be witnessed. These two possibilities lead to the following propositions:

\begin{itemize}
 \item \emph{Weak}: $(\exists \text{ 2-joint})\bigwedge(\nexists \text{ 3-joint}|\text{ a choice of 2-joints})$,
 \item \emph{Strong}:\\$(\exists \text{ 2-joint})\bigwedge(\nexists \text{ 3-joint}|\text{ for all choices of 2-joints})$\\
$\Leftrightarrow (\exists \text{ 2-joint}, \nexists \text{ 3-joint})$.
\end{itemize}
 
where \emph{Weak} $\Leftrightarrow$ \emph{Proposition 1} $\bigvee$ \emph{Proposition 2}, and \emph{Strong} $\Leftrightarrow$ \emph{Proposition 2}. The proposition \emph{Weak} relaxes 
the requirement of proposition \emph{Strong} that the three single measurements should themselves be incompatible (that is, not admit a 3-joint) to only the requirement that they admit a choice of 2-joints that in turn do not 
admit a 3-joint.
Obviously, under \emph{Strong}, there exists no 3-joint for all possible choices of 2-joints: \emph{Strong} $\Rightarrow$ \emph{Weak}.\footnote{
Note that for the case of PVMs, only the conjunction $(\exists \text{ 2-joint}, \exists \text{ 3-joint})\bigwedge(\exists \text{ 3-joint}|\text{ a choice of 2-joints})$
makes sense and that it is, in fact, equivalent to the proposition $(\exists \text{ 2-joint}, \exists \text{ 3-joint})$ since there is no ``choice of 2-joints'' available: the 2-joints,
if they exist, are unique and admit a unique 3-joint (cf. Theorem \ref{uniqueness}). Consequently, the propositions \emph{Weak} and \emph{Strong} are not admissible for PVMs.}

\subsection{A comment on Yu-Oh \cite{FYuOh} vis-\`a-vis Kunjwal-Ghosh \cite{KG}}

In Ref. \cite{KG}, contextuality---in the generalized sense of Spekkens \cite{genNC} and by implication in the Kochen-Specker sense---was shown keeping in mind the proposition \emph{Strong}, i.e., requiring that
the three measurements $\{M_1,M_2,M_3\}$ are pairwise jointly measurable but not triplewise jointly measurable. This was in keeping with the approach adopted in Ref. \cite{LSW}, where the construction 
used did not violate the LSW inequality \cite{LSW, KG}. Indeed, as shown in Theorem 1 of Ref. \cite{KG}, the construction used in Ref. \cite{LSW} could not have produced a violation because
it sought a state-independent violation. 

In Ref. \cite{FYuOh}, the authors -- under \emph{Proposition 1} -- use the construction first obtained in \cite{KG} (and demonstrated in Chapter 2) to show a higher violation of the LSW inequality than reported 
in Ref.~\cite{KG}. It is easy to check that the construction in Ref.~\cite{KG} recovers the violation reported in Ref.~\cite{FYuOh} when
the proposition \emph{Strong} is relaxed to the proposition \emph{Weak}: the expression for the quantum probability of anticorrelation in Ref.~\cite{KG} is given by

\begin{equation}\label{anti}
R_3^Q=\frac{C}{6}+\left(1-\frac{\eta}{3}\right)
\end{equation}
where $C>0$ for a state-dependent violation of the LSW inequality \cite{LSW,KG}. Given a coplanar choice of measurement directions $\{\hat{n}_1,\hat{n}_2,\hat{n}_3\}$, and $\eta$ satisfying $\eta_l<\eta\leq\eta_u$, the optimal value of $C$
---denoted as $C^{\{\hat{n}_i\},\eta}_{\max}$---is given by 
\begin{eqnarray}\label{Cmax}\nonumber
 &&C^{\{\hat{n}_i\},\eta}_{\max}=2\eta\\&+&\sum_{(ij)}\left(\sqrt{1+\eta^4(\hat{n}_i.\hat{n}_j)^2-2\eta^2}-(1+\eta^2 \hat{n}_i.\hat{n}_j)\right).
\end{eqnarray}

For trine measurements, $\hat{n}_i.\hat{n}_j=-\frac{1}{2}$ for each pair of measurement directions, $\{\hat{n}_i,\hat{n}_j\}$. Also, $\eta_l=\frac{2}{3}$
and $\eta_u=\sqrt{3}-1$. $\eta>\eta_l$ ensures that the three measurements corresponding to $\{\hat{n}_1,\hat{n}_2,\hat{n}_3\}$ do not admit a 3-joint while
$\eta\leq\eta_u$ is necessary and sufficient for 2-joints to exist: that is, $\eta_l<\eta\leq\eta_u$ corresponds to the proposition \emph{Strong}, $(\exists \text{ 2-joint}, \nexists \text{ 3-joint})$.
On relaxing the requirement $\eta_l<\eta$, we have $0\leq\eta\leq\eta_u$. This allows room for the proposition $(\exists \text{ 2-joint}, \exists \text{ 3-joint})$ when
$0\leq\eta\leq \eta_l$.

The quantity to be maximized is the quantum violation: $R_3^Q-(1-\frac{\eta}{3})=\frac{C}{6}$. Substituting the value $\hat{n}_i.\hat{n}_j=-\frac{1}{2}$ in Eq. (\ref{Cmax}), the quantum probability of anticorrelation from Eq. (\ref{anti}) for trine measurements is given by:
\begin{equation}
 R_3^Q=\frac{1}{2}+\frac{\eta^2}{4}+\frac{1}{2}\sqrt{1-2\eta^2+\frac{\eta^4}{4}},
\end{equation}
which is the same as the bound in Eq. (11) in Theorem 3 of Ref. \cite{FYuOh}. The quantum violation is given by:
\begin{equation}
 R_3^Q-(1-\frac{\eta}{3})=-\frac{1}{2}+\frac{\eta}{3}+\frac{\eta^2}{4}+\frac{1}{2}\sqrt{1-2\eta^2+\frac{\eta^4}{4}}.
\end{equation}

In Ref. \cite{KG}, this expression was maximized under the proposition \emph{Strong} ($\eta_l<\eta\leq \eta_u$)
and the quantum violation was seen to approach a maximum of $0.0336$ for $R_3^Q\rightarrow0.8114$ as $\eta\rightarrow \eta_l=\frac{2}{3}$. In Ref. \cite{FYuOh}, the same expression 
was maximized while relaxing proposition \emph{Strong} to proposition \emph{Weak} (allowing $\eta\leq\eta_l$) and the maximum quantum violation was seen to be $0.0896$
for $R_3^Q=0.9374$ and $\eta\approx 0.4566$.

Another comment in Ref. \cite{FYuOh} is the following: 

\emph{``Interestingly, there are three observables that are not triplewise jointly measurable but cannot violate LSW's inequality no matter
how each two observables are jointly measured.''} 

That is, \emph{Strong} $\nRightarrow$ Violation of LSW inequality. Equally, it is also the case that \emph{Weak} $\nRightarrow$ Violation of LSW inequality.
Neither of these is surprising given the discussion here. In particular, note the following implications ($0\leq\eta\leq 1$): 

\begin{enumerate}
 \item Violation of LSW inequality, i.e., $R_3^Q>1-\frac{\eta}{3}$ $\Rightarrow$ Violation of KS inequality, i.e., $R_3^Q>\frac{2}{3}$,
 \item Violation of KS inequality, i.e., $R_3^Q>\frac{2}{3}$ $\Rightarrow$ \emph{Weak}: $(\exists \text{ 2-joint})\bigwedge(\nexists \text{ 3-joint}|\text{ a choice of 2-joints})$,
 \item \emph{Strong} $\Rightarrow$ \emph{Weak}.
\end{enumerate}

Therefore, \emph{Weak} is a necessary condition for a violation of the LSW inequality. It can be satisfied either under \emph{Proposition 1} (as done in \cite{FYuOh}) or under 
\emph{Proposition 2} (or \emph{Strong}, as done in \cite{KG}).

\subsection{Joint measurability structures}

Let me end with a comment on the result proven in Ref.~\cite{KHF}, which is the basis of the next chapter in this thesis, where it was shown constructively that any conceivable joint measurability 
structure for a set of $N$ observables is realizable via
binary POVMs. With regard to contextuality, this result proves the admissibility in quantum theory of scenarios that are not realizable with PVMs alone. This should be easy to see, specifically, 
from the 
example of Specker's scenario, where PVMs do not suffice to demonstrate contextuality, primarily because they possess a very rigid joint measurability structure dictated by pairwise commutativity
and their joint measurements are unique (Theorem \ref{uniqueness}). If one can demonstrate contextuality given the scenarios obtained from more general joint measurability structures then a relaxation of
a sort similar to the case of Specker's scenario (from \emph{Strong} to \emph{Weak}) will also lead to contextuality. In this sense, an implication of the result of Ref. \cite{KHF} is that it allows 
one to consider the question of contextuality for joint measurability structures which admit no PVM realization in quantum theory on account of Theorem \ref{uniqueness}.

In particular, for PVMs, \emph{pairwise compatibility} $\Leftrightarrow$ \emph{global compatibility} because commutativity is a necessary and sufficient criterion for compatibility. On the other hand, POVMs allow for 
a failure of the implication \emph{pairwise compatibility} $\Rightarrow$ \emph{global compatibility} because pairwise compatibility is not equivalent to pairwise commutativity for POVMs:
\emph{pairwise commutativity} $\Rightarrow$ \emph{pairwise compatibility}, but not conversely.

\section{Chapter summary}
The main purpose of this chapter was to address any confusion that Ref.~\cite{FYuOh} might cause regarding the results presented in Chapter 2 and published in Ref.~\cite{KG}. 
With this out of the way, we can now turn to the question of admissible joint measurability structures in quantum theory, to be taken up in Chapter 4.

\chapter{All joint measurability structures are quantum realizable}
In many a traditional physics textbook, a quantum measurement is defined as a projective measurement represented by
a Hermitian operator. In quantum information theory, however, the concept of a measurement is dealt with in complete generality
and we are therefore forced to confront the more general notion of positive-operator valued measures (POVMs) which 
suffice to describe all measurements that can be implemented in quantum experiments. In this chapter, we study the (in)compatibility
of such POVMs and show that quantum theory realizes all possible (in)compatibility relations among sets 
of POVMs. This is in contrast to the restricted case of projective measurements for which commutativity is essentially equivalent to 
compatibility. We thus uncover a fundamental feature regarding the (in)compatibility of 
quantum observables that has no analog in the case of projective measurements.

This chapter is based on work reported in Ref.~\cite{KHF}.
\section{Introduction}
In the traditional textbook treatment of measurements in quantum theory one usually comes across projective measurements. 
For these measurements, commutativity of the associated Hermitian operators is necessary and sufficient for them to be compatible.
That is, commuting Hermitian operators represent quantum observables that can be jointly measured in a single experimental setup.
Furthermore, given a set of $N$ projective measurements, commutativity means pairwise commutativity and we have:
$$\text{pairwise compatibility}\Leftrightarrow\text{global compatibility}.$$

This equivalence is rather special since it reduces the problem of deciding whether a set of projective measurements is 
compatible to checking that every pair in the set commutes. Operationally, this also means that the measurement statistics 
obtained by performing these measurements sequentially on any preparation of a quantum system is independent of the sequence 
in which the measurements are performed, \textit{e.g.}, if $A$, $B$, $C$ are Hermitian operators that commute pairwise, then the
sequential measurements $ABC$, $ACB$, $BAC$, $BCA$, $CAB$ and $CBA$ are all physically equivalent.

However, once the projective property is relaxed and the resulting positive-operator valued measures (POVMs) are considered,
the implication ``pairwise compatibility $\Rightarrow$ global compatibility'' no longer holds. The converse implication is still 
true. Indeed, one can construct examples where a set of three POVMs is pairwise compatible but there is no global compatibility
between them~\cite{Kraus,LSW,KG,HRF}, a feature characteristic of the observables in Specker's scenario that we considered in Chapter 1. 
With this in mind, our purpose in this chapter is to explore whether there really is any 
constraint on the (in)compatibility relations that one could realize between quantum measurements (POVMs). If, for example, 
certain sets of (in)compatibility relations were not allowed in quantum theory then that would isolate conceivable 
joint measurability structures that are nevertheless forbidden in nature. 

\begin{figure}
\centering
\includegraphics[scale=1.8]{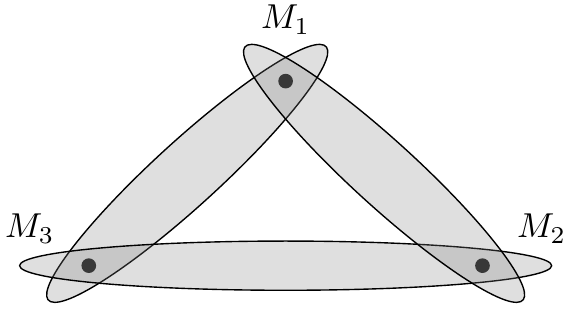}
\caption{Specker's scenario.}
\label{specker}
\end{figure}

It is worth noting that the impossibility of jointly implementing arbitrary sets of measurements is a key ingredient that 
enables a demonstration of the nonclassicality of quantum theory in proofs of Bell's theorem~\cite{Bell64} and the Kochen--Specker 
theorem~\cite{KS67}. A finite set of measurements is called \textit{jointly measurable} or \textit{compatible} if there exists a single
measurement whose various coarse-grainings recover the original measurements. The problem of characterizing the joint measurability 
of observables has been studied in the literature~\cite{heinosaari, SRH}, and at least the joint measurability of binary qubit observables has 
been completely characterized~\cite{BS, YuOh}. The connection between Bell inequality violations and the joint measurability of observables
has also been quantitatively studied~\cite{andersson, wolfetal}.

A natural question that arises when thinking about the (in)compatibility of observables is the following: given a set of 
(in)compatibility relations on a set of vertices representing observables, do they admit a quantum realization? That is, 
can one write down a positive-operator valued measure (POVM) for each vertex such that the (in)compatibility relations 
among the vertices are realized by the assigned POVMs? After formally defining these notions, we answer this question in the 
affirmative by providing an explicit construction of POVMs for any set of (in)compatibility relations.

We will use the terms `(not) jointly measurable' and `(in)compatible' interchangeably in this chapter. Part of our motivation in
studying this question comes from the simplest example of joint measurability relations realizable with POVMs but not with 
projective measurements. As we mentioned, this joint measurability scenario, referred to as Specker's scenario~\cite{Spe60, LSW, KG} in Chapter 1, involves
three binary measurements that can be jointly measured pairwise but not triplewise: that is, for the set of binary measurements
$\{M_1,M_2,M_3\}$, the (in)compatibility relations are given by the collection of compatible subsets $\{\{M_1,M_2\},\{M_2,M_3\},\{M_1,M_3\}\}$.
The remaining nontrivial subset (with at least two observables), namely $\{M_1,M_2,M_3\}$, is incompatible. This can be 
pictured as a hypergraph (Fig.~\ref{specker}).

In Chapter 1, we showed how Specker's scenario can be exploited to violate the LSW inequality using a set of three qubit POVMs 
realizing this scenario~\cite{genNC,LSW,KG}. This novel demonstration of contextuality in quantum theory
raises the question of whether there exist other scenarios---for example in an observable-based hypergraph approach as 
in~\cite{AB,CF}---that do not admit a proof of quantum contextuality using projective measurements, but do admit such a proof using POVMs.
A necessary first step towards answering this question is to figure out what compatibility scenarios are realizable in quantum
theory. One can then ask whether these scenarios allow nontrivial correlations that rule out noncontextuality~\cite{genNC}.
We take this first step by proving that, in principle, all joint measurability hypergraphs are realizable in quantum theory.
The realizability of all joint measurability graphs via projective measurements is known~\cite{HRF}. This prompted our 
question whether all joint measurability hypergraphs are realizable via POVMs. Our positive answer includes joint measurability
hypergraphs that do not admit a realization using projective measurements. For our construction, it suffices to consider binary observables on
finite-dimensional Hilbert spaces.

\section{Definitions}
\paragraph*{POVMs.} A positive-operator valued measure (POVM) on a Hilbert space $\mathcal{H}$ is a mapping $x\mapsto M(x)$ from an outcome set $X$ ($x\in X$) to the set of positive semidefinite operators
\[
M(x)\in\mathcal{B}(\mathcal{H}),\quad M(x)\geq 0,
\]
such that the POVM elements $M(x)$ sum 
to the identity operator,
\[
\sum_{x\in X}M(x)=I.
\]
If $M(x)^2=M(x)$ for all $x\in X$, then the POVM becomes a ``projection valued measure'', or simply, a projective measurement.

\paragraph*{Joint measurability of POVMs.} A finite set of POVMs 
\[
\{M_1,\dots,M_N\},
\]
where measurement $M_i$ has outcome set $X_i$,
is said to be \textit{jointly measurable} or \textit{compatible} if there exists a POVM $M$ with outcome set $X_1\times X_2 \times \dots \times X_N$ that marginalizes to each $M_i$ with outcome set $X_i$, meaning that
\[
M_i(x_i)=\sum_{x_1,\ldots,\cancel{x_i},\ldots,x_N}M(x_1,\dots,x_N) 
\]
for all outcomes $x_i\in X_i$.

\paragraph*{Joint measurability hypergraphs.} A \textit{hypergraph} consists of a set of vertices $V$, and 
a family $E \subseteq \{ e \mid e \subseteq V \}$ of subsets of $V$ called \textit{edges}.
We think of each vertex as representing a POVM, while an edge models joint measurability of the POVMs 
it links. Since every subset of a set of compatible measurements should also be compatible, a joint 
measurability hypergraph should have the property that any subset of an edge is also an edge,
\[
e\in E,\: e'\subseteq e \implies e'\in E.
\]
Additionally, we focus on the case where each edge $e$ is a finite subset of $V$.
This makes a joint measurability hypergraph into an abstract simplicial complex.

Every set of POVMs on $\mathcal{H}$ has such an associated joint measurability hypergraph. 
Hence characterizing joint measurability of quantum observables comes down to figuring out their joint measurability hypergraph.
Our main result solves the converse problem. Namely, every abstract simplicial complex arises from the joint measurability
relations of a set of quantum observables.

\section{Quantum realization of arbitrary joint measurability structures}

\begin{theorem}
 Every joint measurability hypergraph admits a quantum realization with POVMs.
\end{theorem}

\begin{proof}
We begin by proving a necessary and sufficient criterion for the joint measurability of $N$ binary POVMs $M_k:=\{E^k_{+},E^k_{-}\}$ of the form
\begin{equation}
\label{Ek}
 E^k_{\pm}:=\frac{1}{2}\left(I\pm\eta\Gamma_{k}\right),
\end{equation}
where the $\Gamma_{k}$ are generators of a Clifford algebra as in the Appendix. The variable $\eta\in[0,1]$ is a purity parameter.
Since $\Gamma_k^2=I$, the eigenvalues of $\Gamma_k$ are $\pm 1$, so that $E^k_\pm$ is indeed positive. The following derivation
of a joint measurability criterion is adapted from a proof first obtained in~\cite{LSW}, and subsequently revised in~\cite{KG},
for the joint measurability of a set of noisy qubit POVMs.
Because $\Gamma_k$ is traceless by~\eqref{traceless}, we can recover the purity parameter $\eta$ as
\[
 \Tr(\Gamma_k E^k_{\pm})=\pm\frac{\eta}{2}d,
\]
so that
\begin{equation}
\label{recovereta}
 \eta=\frac{1}{Nd}\sum_{k=1}^{N}\sum_{x_k\in X_k}\Tr(x_k\Gamma_k E^k_{x_k}),
\end{equation}
where we have introduced one separate outcome $x_k\in X_k:=\{+1,-1\}$ for each measurement $M_k$.

If all 
$M_k=\{E^k_+,E^k_-\}$ together are jointly measurable, then there exists
a joint POVM $M=\{E_{x_1\dots x_N}\}$ satisfying
\[
 E^k_{x_k}=\sum_{x_1,\ldots \cancel{x_k},\ldots x_N}E_{x_1\dots x_N}.
\]
Writing $\vec{x}:=(x_1,\ldots,x_N)$ and $\vec{\Gamma}:=(\Gamma_1,\ldots,\Gamma_N)$, this assumption together with~\eqref{recovereta} implies that
\begin{align*}
 \eta&=\frac{1}{Nd}\sum_{\vec{x}}\Tr\left[\left(\sum_{k=1}^{N}x_k\Gamma_k\right)E_{x_1\dots x_N}\right]\\[5pt]
 &\leq \frac{1}{Nd}\sum_{\vec{x}} \|\vec{x}\cdot\vec{\Gamma}\| \:\Tr\left[E_{\vec{x}}\right]\\[5pt]
 &= \frac{1}{N} \|\vec{x}\cdot\vec{\Gamma}\|,
\end{align*}
where the last step used the normalization $\sum_{\vec{x}} E_{\vec{x}}=I$. Since $(\vec{x}\cdot\vec{\Gamma})^2=\sum_k X_k^2=N\cdot I$ by~\eqref{cliffordprod}, we have $\|\vec{x}\cdot\vec{\Gamma}\|=\sqrt{N}$, and therefore
\[
\eta \leq  \frac{1}{\sqrt{N}} ,
\]
a necessary condition for joint measurability of $M_k$.
To show that this condition is also sufficient, we consider the joint POVM $M=\{E_{\vec{x}}\}$
given by
\begin{equation}
 E_{x_1\dots x_N}:=\frac{1}{2^N} \left(I+\eta\:\vec{x}\cdot\vec{\Gamma}\right).
\end{equation}
We start by showing that this indeed defines a POVM,
\[
E_{x_1\dots x_N}\geq 0, \quad \sum_{x_1,\dots, x_N} E_{x_1\dots x_N}=I.
\]
Positivity follows again from noting that the eigenvalues of $\vec{x}\cdot\vec{\Gamma}$ are $\pm\sqrt{N}$ by~\eqref{cliffordprod}, and normalization from $\sum_{\vec{x}}\vec{x}\cdot\vec{\Gamma}=0$.
Since
\[
 \sum_{x_1,\ldots,\cancel{x_k},\ldots,x_N} E_{x_1\dots x_N}=\frac{1}{2}\left(I+\eta x_k\Gamma_k\right)
\]
coincides with~\eqref{Ek}, we have indeed found a joint POVM marginalizing to the given $M_k$.

Thus $\eta\leq \frac{1}{\sqrt{N}}$ is a necessary and sufficient condition for the joint measurability of $M_1,\ldots,M_N$.

For arbitrary $N$, then, we can construct $N$ POVMs on a Hilbert space of appropriate dimension such that any 
$N-1$ of them are compatible, whereas all $N$ together are incompatible: 
simply take $M_1,\ldots, M_{N}$ from~\eqref{Ek} for any purity parameter $\eta$ satisyfing
\[
 \frac{1}{\sqrt{N}}<\eta\leq\frac{1}{\sqrt{N-1}}.
\]
For example, $\eta=1/\sqrt{N-1}$ will work. The above reasoning guarantees that all $N$ of them together are not compatible, and also that the $M_1,\ldots,M_{N-1}$ are compatible. By permuting the labels and 
observing that the above reasoning did not rely on any specific ordering of the $\Gamma_k$, we conclude that \textit{any} $N-1$ measurements among the $M_1,\ldots,M_N$ are compatible.

What we have established so far is that, 
if we are given any $N$-vertex joint measurability hypergraph 
where every subset of $N-1$ vertices is compatible (\textit{i.e.}\ belongs to a common edge), 
but the $N$-vertex set is incompatible, 
then the above construction provides us with a quantum realization of it.
These ``Specker-like'' hypergraphs are crucial to our construction. For example, for $N=3$, we
obtain a simple realization of Specker's scenario (Fig. \ref{specker}). For $N=2$, we simply obtain a pair of incompatible observables.
Given an arbitrary joint measurability hypergraph, the procedure to construct a quantum realization is now the following:

\begin{enumerate}
 \item Identify the minimal incompatible sets of vertices in the hypergraph. A minimal incompatible set is an incompatible
set of vertices such that any of its proper subsets \textit{is} compatible. In other words, it is a Specker-like hypergraph embedded
 in the given joint measurability hypergraph.
 \item For each minimal incompatible set, construct a quantum realization as above. Vertices that are outside this minimal incompatible set
 can be assigned a trivial POVM in which one outcome is deterministic, represented by the identity operator $I$. Let $\mathcal{H}_i$ denote
 the Hilbert space on which the minimal incompatible set is realized, where $i$ indexes the minimal incompatible sets.
 \item Having thus obtained a quantum representation of each minimal incompatible set, we simply ``stack'' these together in a direct
 sum over the Hilbert spaces on which each of the minimal incompatible sets are realized. On this larger direct sum Hilbert space
 $\mathcal{H}=\oplus_{i} \mathcal{H}_i$, we then have a quantum realization of the joint measurability hypergraph we started with.
\end{enumerate}
For any edge $e\in E$, the associated measurements are compatible on every $\mathcal{H}_i$, and therefore also on $\mathcal{H}$. On the other hand, every $e'\subseteq V$ that is not an edge is contained in some minimal incompatible set (or is itself already minimal), and therefore the associated POVMs are incompatible on some $\mathcal{H}_i$, and 
hence also on $\mathcal{H}$.
\end{proof}

\begin{figure}
\centering
\includegraphics[scale=1.8]{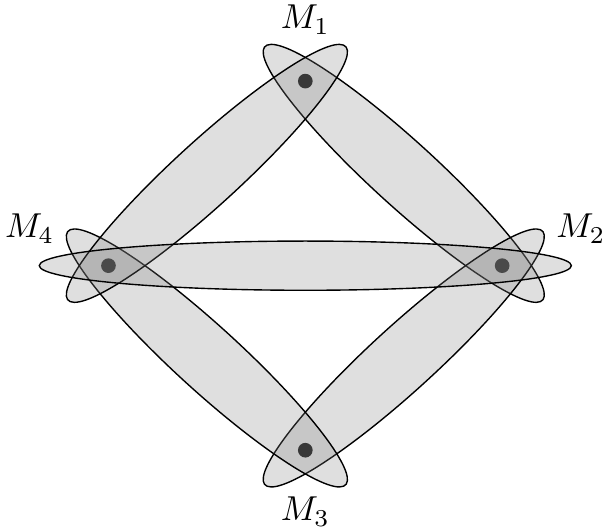}
\caption{A joint measurability hypergraph for $N=4$.}
\label{hyper}
\end{figure}

\begin{figure}
\centering
\includegraphics[scale=1.8]{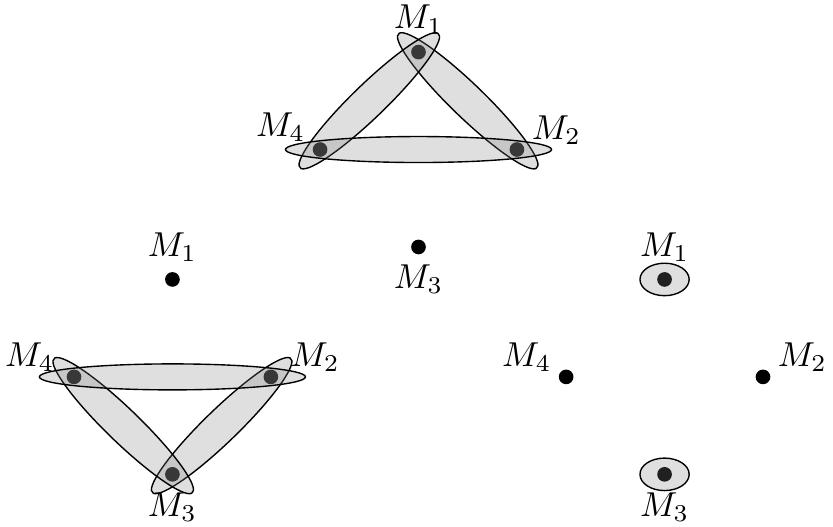}
\caption{Minimal incompatible sets for the joint measurability hypergraph in Fig.~\ref{hyper}.}
\label{incomp}
\end{figure}

\section{A simple example} 
To illustrate these ideas, we construct a POVM realization of a simple joint measurability hypergraph that does not admit a 
representation with projective measurements (Fig.~\ref{hyper}). This hypergraph can be decomposed into three minimal incompatible sets of vertices (Fig.~\ref{incomp}).
Two of these are Specker scenarios for $\{M_1,M_2,M_4\}$ and $\{M_2,M_3,M_4\}$, and the third one is a pair of incompatible vertices
$\{M_1,M_3\}$. For the minimal incompatible set $\{M_1,M_2,M_4\}$, we construct a set of three binary POVMs, $A_k\equiv\{A^k_+,A^k_-\}$ with $k\in \{1,2,4\}$ on a qubit Hilbert space $\mathcal{H}_1$ given by
\begin{equation}
A^k_{\pm}:=\frac{1}{2}\left(I\pm \frac{1}{\sqrt{2}}\Gamma_k\right),
\end{equation}
where the 
matrices $\{\Gamma_1,\Gamma_2,\Gamma_4\}$ can be taken to be the Pauli matrices,
\[
\Gamma_1 = \sigma_z,\quad \Gamma_2=\sigma_x,\quad \Gamma_4=\sigma_y,
\]
similar to~\eqref{pauli}.
The remaining vertex $M_3$ can be taken to be the trivial POVM $A_3=\{0,I\}$ on $\mathcal{H}_1$. A similar construction works for the second 
Specker scenario $\{M_2,M_3,M_4\}$ 
by setting
$B_k:=\{B^k_+,B^k_-\}$ with $k\in \{2,3,4\}$ to be
\begin{equation}
B^k_{\pm}:=\frac{1}{2}\left(I\pm \frac{1}{\sqrt{2}}\Gamma_k\right),
\end{equation}
where
\[
\Gamma_2=\sigma_z,\quad \Gamma_3=\sigma_x,\quad \Gamma_4=\sigma_y
\]
act on another qubit Hilbert space $\mathcal{H}_2$. The remaining vertex $M_1$
can be assigned the trivial POVM, $B_1=\{0,I\}$. The third minimal incompatible set $\{M_1,M_3\}$ can similarly 
be obtained on another qubit Hilbert space $\mathcal{H}_3$ as $C_k:=\{C^k_+,C^k_-\}$, with $k\in \{1,3\}$, given by
\begin{equation}
C^k_{\pm}:=\frac{1}{2}(I\pm \Gamma_k),
\end{equation}
where now \textit{e.g.}~$\Gamma_1=\sigma_z$ and $\Gamma_3=\sigma_x$. 
The remaining vertices $M_2$ and $M_4$ can both be 
assigned the trivial POVM $C_2=C_4:=\{0,I\}$ on $\mathcal{H}_3$.

In the direct sum Hilbert space
$\mathcal{H}:=\mathcal{H}_1\oplus \mathcal{H}_2\oplus \mathcal{H}_3$, we then have a POVM realization 
of the joint measurability hypergraph of Fig.~\ref{hyper}, given by
\[
 M^k_{\pm}:=A^k_{\pm}\oplus B^k_{\pm}\oplus C^k_{\pm}.
\]

\section{Chapter summary} We have shown, by construction, that any conceivable set of (in)compatibility relations for any number
of quantum measurements can be realized using a set of binary POVMs. Our result thus demonstrates that quantum theory is not 
constrained to admit only a restricted set of (in)compatibility relations, such as those where pairwise compatibility $\Leftrightarrow$
global compatibility, which is the case with projective measurements. Indeed, quantum theory admits all possible (in)compatibility
relations. With respect to (in)compatibility relations, therefore, quantum theory is as far away from classical theories 
(where there are no incompatibilities) as possible. By ``classical theories'' we mean those where all measurements can be jointly implemented.

Although our simple construction works for all joint measurability hypergraphs, it is probably not the most efficient one for
a given joint measurability hypergraph: for Fig.~\ref{hyper}, our representation lives on a six-dimensional Hilbert space. 
For a joint measurability hypergraph with a fixed number of vertices, the dimension of the Hilbert space $\mathcal{H}$ on which our construction
is realized depends on the number of minimal incompatible sets in the hypergraph: that is, $\dim \mathcal{H}=\sum_i \dim \mathcal{H}_i$,
where $\mathcal{H}_i$ is the Hilbert space on which the $i$th minimal incompatible set is realized. 
It remains open what the most efficient construction---requiring the smallest Hilbert space dimension---for a given joint 
measurability hypergraph is. 
Our result captures all conceivable (in)compatibility relations within the framework
of quantum theory, thus 
shedding light on the structure of quantum theory and what it really allows us to do.

\section*{Appendix: Clifford algebras}
A \textit{Clifford algebra} consists of a finite set of hermitian matrices $\Gamma_1,\ldots,\Gamma_N$ 
satisfying the relations\footnote{Strictly speaking, this is a \textit{representation} of a Clifford algebra, but the difference 
between algebras and their representations is not relevant here.}

\begin{equation}
\label{acrel}
\Gamma_j\Gamma_k + \Gamma_k\Gamma_j = 2\delta_{jk}I,
\end{equation}
Clifford algebras are the mathematical structure behind the definition of spinors and the Dirac equation~\cite{Lounesto}. 
They can be constructed recursively as follows~\cite[Sec.~16.3]{Lounesto}. Given $\Gamma_1,\dots,\Gamma_N$ living on a Hilbert
space $\mathcal{H}_{N}$, one obtains $\Gamma_1,\ldots,\Gamma_{N+2}$ on $\mathcal{H}_N\otimes \mathbb{C}^2$ by the following rules.
\begin{enumerate}
\item For each $i=1,\dots,N$, substitute
\[
\Gamma_i \rightarrow \Gamma_i\otimes\sigma_z.
\]
\item Further, define
\[
\Gamma_{N+1}:= I \otimes \sigma_x,\quad \Gamma_{N+2}:= I\otimes \sigma_y.
\]
\end{enumerate}
It is easy to show that if the original $\Gamma_i$ satisfy~\eqref{acrel}, then so do the new ones. One can simply start the 
recursion with $\Gamma_1=1$ on the one-dimensional Hilbert space $\mathcal{H}_1:=\mathbb{C}$, and then apply the construction as 
often as necessary to obtain any finite number of matrices satisfying~\eqref{acrel}. For example, 
a single iteration gives the Pauli matrices
\begin{equation}
\label{pauli}
\Gamma_1=\sigma_z,\quad \Gamma_2=\sigma_x,\quad \Gamma_3 = \sigma_y,
\end{equation}
while after two iterations one has 
\begin{eqnarray*}
&\Gamma_1 = \sigma_z\otimes\sigma_z, \quad \Gamma_2 = \sigma_x\otimes\sigma_z,\\[5pt]
&\Gamma_3 = \sigma_y\otimes\sigma_z, \quad \Gamma_4 = I\otimes\sigma_x,\quad \Gamma_5 = I\otimes\sigma_y.
\end{eqnarray*}
The Clifford algebra relations~\eqref{acrel} have many interesting consequences. For example for $N\geq 2$, one has for any $k$ and $j\neq k$,
\begin{eqnarray*}
\Tr(\Gamma_k) &= \Tr(\Gamma_k\Gamma_j\Gamma_j) = -\Tr(\Gamma_j\Gamma_k\Gamma_j) \\[5pt]
&= -\Tr(\Gamma_k\Gamma_j\Gamma_j) = -\Tr(\Gamma_k),
\end{eqnarray*}
so that
\begin{equation}
\label{traceless}
\Tr(\Gamma_k)=0. 
\end{equation}
Another consequence is that
\begin{equation}
\label{cliffordprod}
\left(\sum_k X_k \Gamma_k\right)^2 = \left(\sum_k X_k^2\right) \cdot I
\end{equation}
for arbitrary real coefficients $X_k$.

\chapter{Fine's theorem and its status in tests of noncontextuality}
In this chapter, we provide a characterization of noncontextual models which fall within the ambit of Fine's theorem\cite{Fine, Fine2}.
In particular, we make explicit the equivalence between the existence of three notions:
a joint probability distribution over the outcomes of all the measurements considered, a measurement-noncontextual and outcome-deterministic
(or KS-noncontextual) model for these measurements, and a measurement-noncontextual and factorizable 
model for them. A KS-inequality, therefore, is implied by each of these three notions. 
Following this characterization of noncontextual models that fall within the ambit of Fine's theorem,
non-factorizable noncontextual models which lie outside the domain of Fine's theorem are considered. While outcome determinism
for projective (sharp) measurements in quantum theory can be shown to follow from the assumption of preparation noncontextuality,
as we did in Chapter 1,
such a justification is not available for nonprojective (unsharp) measurements which ought to admit outcome-indeterministic response functions.
The Liang-Spekkens-Wiseman (LSW) inequality is an example of a noncontextuality inequality that should hold in any noncontextual model of quantum theory 
without assuming factorizability. Three other noncontextuality
inequalities, which turn out to be equivalent to the LSW inequality under relabellings of measurement outcomes,
are derived for Specker's scenario. We also characterize the polytope of correlations admissible in this scenario, given the 
operational equivalences between measurements (often called the ``no-disturbance'' condition in the literature on KS-contextuality).

This chapter is based on work reported in Ref.~\cite{finegen}.
\section{Introduction}

In attempts to provide a more complete description of reality than operational quantum theory in terms of a noncontextual ontological model, it is almost always assumed that 
whatever the ontic state $\lambda$ is, it must specify the outcomes of measurements exactly (an assumption called \emph{outcome determinism}) and any operational unpredictability in the measurement
outcomes is on account of coarse-graining over these ontic states $\lambda$. This chapter concerns itself with what can still be said about noncontextuality
if outcome determinism is not assumed: the ontic state is not always required to fix the outcomes of measurements but only their probabilities.
The physical motivation for this becomes clear once the following questions are asked:
\begin{enumerate}
 \item Do there exist noncontextual ontological models of quantum theory where the ontic state $\lambda$ fixes the outcomes of 
measurements?\\

The Kochen-Specker theorem \cite{KS67} rules out this possibility. Let us now remove the requirement of outcome determinism, namely, that $\lambda$ fix the outcomes of 
measurements, and ask the question:
\item Do there exist noncontextual ontological models of quantum theory where the ontic state $\lambda$ fixes the \emph{probabilities}
of outcomes of measurements?\\

The Kochen-Specker theorem \cite{KS67} is silent on this question since it presumes the ontic state $\lambda$ must fix the 
outcomes of (projective) measurements. This question is most naturally addressed in the framework of generalized noncontextuality
due to Spekkens \cite{genNC}.
\end{enumerate}

It is well-known that, in contrast to the Kochen-Specker theorem \cite{KS67}, Bell's theorem \cite{Bell76,Wiseman} does not require an assumption that
the ontic state $\lambda$ fixes the outcomes of the measurements. This becomes particularly clear in view of Fine's theorem \cite{Fine, Fine2} 
that, in a Bell scenario, a locally deterministic model \cite{Bell64} exists if and only if
a locally causal (or `Bell-local') model \cite{Bell76, Wiseman} exists, and how this is equivalent to requiring the existence of a joint probability distribution over
outcomes of all the measurements considered in a Bell scenario. Hence, even if the outcomes are only determined probabilistically by $\lambda$
in the local hidden variable model, Bell's theorem holds. The key issue in Bell scenarios is factorizability: the conditional independence of the outcomes of spacelike separated measurements given the ontic state
$\lambda$ of the system, 
\begin{eqnarray}
&&\xi(X_1,\dots,X_N|M_1,\dots,M_N,\lambda)\nonumber\\
&=&\xi(X_1|M_1,\lambda)\xi(X_2|M_2,\lambda)\dots\xi(X_N|M_N,\lambda),
\end{eqnarray}
where $X_i$ labels the outcome of measurement $M_i$ performed by the $i$th party, $i\in\{1,\dots,N\}$. All these response functions may be 
outcome-indeterministic, i.e., $\xi\in[0,1]$. Indeed, factorizability is a necessary consequence of any set of assumptions that may be used to derive Bell's theorem \cite{Wiseman}.

On the other hand, things are not as straightforward for contextuality \cite{Spe60,Bell66,KS67}. Mathematically, both Bell-local models and 
KS-noncontextual models rely on the existence of a joint probability distribution over all measurement outcomes
in a given scenario such that this distribution reproduces the observed statistics as marginals. Proofs of the KS theorem
that rely on uncolorability (such as the original proof in Ref.~\cite{KS67}) are such that there exists \emph{no} joint distribution
at all, given \emph{any} set of observed marginals. These proofs are often termed \emph{state-independent} since
they do not rely on preparing particular quantum states on which the measurement statistics are considered: any quantum state works.
Weaker proofs of the KS theorem (such as the one in Ref.~\cite{KCBS}) are such that
there exist joint distributions for \emph{some}, but \emph{not all}, sets of observed marginals: there exist sets of observable marginals
that admit no joint distribution and therefore violate some KS inequality arising from requiring the existence of a joint distribution.
These proofs are often termed \emph{state-dependent} since they rely on preparing particular quantum states that lead to 
marginal statistics violating a KS inequality. The state-dependent proofs of KS-contextuality are 
therefore analogous to proofs of Bell's theorem.

Given this correspondence between Bell's theorem
and the KS theorem, one may ask whether the assumption of outcome determinism is really required in the KS theorem and whether the KS theorem excludes
also all outcome-indeterministic noncontextual models
on account of Fine's theorem.\footnote{We will show that this is \emph{not} the case, i.e. the KS theorem does not rule out all outcome indeterministic noncontextual models.}

The outcome-indeterministic noncontextual models excluded by the KS theorem are precisely the ones where factorizability holds. However, in the absence of spacelike separation between 
measurements one does not have a compelling physical justification to assume that measurement outcomes are conditionally independent of each other 
(and the remote measurement settings) given the ontic state $\lambda$. The physical meaning 
of factorizability is this: that the measurement outcomes do not have any correlations that are not due to the ontic state $\lambda$ of the system. One could, on the other hand,
imagine
an adversarial situation where two measurement outcomes are correlated---which is physically possible if they are not spacelike separated---and this correlation is 
not mediated only by the ontic state $\lambda$ of the system but is perhaps encoded in the degrees of freedom of the measurement apparatus
by an adversary who wants to convince the experimenter that something nonclassical is going on (in the sense of KS-contextuality) but, really, 
it is correlated noise that's doing all the work of violating a KS inequality. The LSW inequality \cite{LSW,KG} is an example of a noncontextuality
inequality that takes this possibility into account and raises the bar for what correlations count as nonclassical. This is why
we need to consider noncontextual models which are not factorizable. Since all KS-noncontextual models are factorizable on account 
of Fine's theorem, noncontextual models which are not factorizable are exclusively taken into account \emph{only} in the generalized definition of noncontextuality \cite{genNC}.
This realization is a key conceptual insight of this chapter, pointing to the necessity of revising the traditional analyses of KS-noncontextuality to accommodate
the generalized notion of noncontextuality \cite{genNC}. To be clear, by \emph{outcome determinism} and \emph{factorizability}, we mean the following:

\paragraph*{Outcome determinism} is the assumption that every response function in the ontological model is deterministic, i.e., $\xi(k|M,\lambda)\in\{0,1\}$
for all measurement events $[k|M]$ and for all $\lambda\in\Lambda$.

Ontological models where outcome determinism doesn't hold are called outcome indeterministic. Of the class of outcome-indeterministic
ontological models, the ones that are related to outcome-deterministic models via Fine's theorem are those satisfying \emph{factorizability}:

\paragraph*{Factorizability} is the assumption that for every jointly measurable set of measurements $\{M_s^{(S)}|s\in S\}$, the response function for every 
outcome of a joint measurement $M_S$ is the product of the response functions of measurements in the jointly measurable set: 
$\xi(k_S|M_S,\lambda)=\prod_{s\in S}\xi(k_s|M_s^{(S)},\lambda)$. Note that $k_S\in\mathcal{K}_{M_S}$ and $k_s\in\mathcal{K}_{M_s^{(S)}}$, where $\mathcal{K}_{M_S}$ is 
the Cartesian product of the outcome sets $\mathcal{K}_{M_s^{(S)}}, s\in S$.

It should be clear that, for a given set of measurements $\{M_1,\dots,M_N\}$ with jointly measurable subsets $S\subset\{1,\dots,N\}$,
outcome determinism implies factorizability, but not conversely. Outcome determinism requires that $\xi(k_S|M_S,\lambda)\in\{0,1\}$
for all $S$, and $\xi(k_s|M_s^{(S)},\lambda)\in\{0,1\}$ for all $s\in S$. $\xi(k_S|M_S,\lambda)=1$ means that 
$\xi(k_s|M_s^{(S)},\lambda)=\sum_{k_{s'}: s'\neq s}\xi(k_S|M_S,\lambda)=1$ for all $s\in S$ and $\xi(k_S|M_S,\lambda)=0$
means that $\xi(k_s|M_s^{(S)},\lambda)=\sum_{k_{s'}: s'\neq s}\xi(k_S|M_S,\lambda)=0$ for at least one $s\in S$.
All in all, we have $\xi(k_S|M_S,\lambda)=\prod_{s\in S}\xi(k_s|M_s^{(S)},\lambda)$, which is what factorizability 
requires. That the converse is not true is easily seen by noting that one can have factorizability without requiring 
$\xi(k_S|M_S,\lambda)\in\{0,1\}$ or $\xi(k_s|M_s^{(S)},\lambda)\in\{0,1\}$.

Just as local causality does not require the assumption of outcome determinism, a good definition of noncontextuality should also not appeal to outcome determinism (or even 
factorizability). Experimental violations of Bell inequalities certify a kind of nonclassicality independent of the truth of quantum theory, a feature that makes Bell inequality violations an
invaluable resource in device-independent protocols \cite{nsqkd}.
In contrast, KS-noncontextuality has to refer to projective (sharp) measurements in quantum
theory and assume outcome-determinism for them in order to obtain a KS-inequality: neither of these is needed in a Bell-local model. The generalized notion of
noncontextuality offers the possibility of talking about noncontextuality without making the assumption that the operational theory is quantum theory.
The present chapter, however, restricts itself to generalized noncontextuality for operational quantum theory.

In the next section, we will see how factorizability - although it's a physically motivated assumption in locally causal models - is not well-motivated physically in the
general case of noncontextual models. Fine's theorem thus serves to delineate a mathematical boundary between KS-noncontextual models and 
measurement noncontextual models which are not factorizable.

\section{Fine's theorem for noncontextual models}

\begin{theorem}\label{genFine}
 Given a set of measurements $\{M_1,\dots,M_N\}$ with jointly measurable subsets $S \subset \{1,\dots,N\}$, where each measurement $M_s, s\in S$, takes 
 values labelled by $k_s\in\mathcal{K}_{M_s}$, the following propositions are equivalent:
\begin{enumerate}
\item For a given preparation $P\in\mathcal{P}$ of the system there exists a joint probability distribution $p(k_1,\dots,k_N|P)$ that recovers the marginal 
statistics for jointly measurable subsets predicted by the operational theory (such as quantum theory) under consideration, i.e., 
$\forall S \subset \{1,\dots,N\}$, $p(k_S|M_S;P)=\sum_{k_i: i \notin S} p(k_1,\dots,k_N|P)$,
where $k_S\in\mathcal{K}_{M_S}$. 

\item There exists a measurement-noncontextual and outcome-deterministic, i.e. KS-noncontextual, model for these measurements.

\item There exists a measurement-noncontextual and factorizable model for these measurements.
 \end{enumerate}
\end{theorem}

\emph{Proof.}
The proof of equivalence of the three propositions proceeds as follows: Proposition 3 $\Rightarrow$ Proposition 1, Proposition 1 $\Rightarrow$ Proposition 2, Proposition 2 $\Rightarrow$ Proposition 3.\\\\
\textbf{Proposition 3 $\Rightarrow$ Proposition 1:}\\

By Proposition 3, the assumption of measurement noncontextuality requires that the single-measurement response functions in the model be of the form $\xi(k_i|M_i;\lambda) \in [0,1]$, so
that each response function is independent of the contexts---jointly measurable subsets $S$---that the corresponding measurement may be a part of. Of course,
the assumption of measurement noncontextuality only applies once it is verified that for any $P\in\mathcal{P}$ 
the operational statistics $p(k_i|M_i;P)$ of measurement $M_i$ is the same across all the jointly measurable subsets, $S$, 
in which it appears. The response function is therefore conditioned only by $M_i$ and the ontic state $\lambda$ associated with the system
(and not on the jointly measurable subset $S$ that $M_i$ may be a part of).
Proposition 3 requires, in addition, factorizability, i.e., for all jointly measurable subsets $S \subset \{1,\dots,N\}$, $$\xi(k_S|M_S;\lambda)=\prod_{s\in S}\xi(k_s|M_s;\lambda).$$
Factorizability amounts to the assumption that the correlations between measurement outcomes are established only via the 
ontic state of the system---the measurement outcomes do not ``talk'' to each other except via $\lambda$.
Now define 
\begin{equation}
 \xi(k_1,\dots,k_N|\lambda) \equiv \prod_{i=1}^{N}\xi(k_i|M_i;\lambda),
\end{equation}
so that marginalizing this distribution over $k_i$, $i\notin S$, yields $\xi(k_S|M_S;\lambda)$ for every jointly measurable subset $S\subset \{1,\dots,N\}$.

Assuming the ontological model reproduces the operational statistics, there must exist a probability distribution $\mu(\lambda|P)$ 
for any $P\in\mathcal{P}$, such that
\begin{equation}
 \sum_{\lambda\in\Lambda}\xi(k_S|M_S; \lambda)\mu(\lambda|P)=p(k_S|M_S;P).
\end{equation}
Then define
\begin{equation}
 p(k_1\dots k_N|P)\equiv \sum_{\lambda\in\Lambda}\xi(k_1\dots k_N|\lambda)\mu(\lambda|P),
\end{equation}
which marginalizes on $k_S$ to
\begin{eqnarray}
 p(k_S|P)&=&\sum_{k_i: i \notin S}p(k_1\dots k_N|P)\\
 &=& \sum_{\lambda\in\Lambda}\sum_{k_i:i \notin S}\xi(k_1\dots k_N|\lambda)\mu(\lambda|P)\\
 &=& \sum_{\lambda\in\Lambda}\xi(k_S|M_S;\lambda)\mu(\lambda|P)\\
 &=& p(k_S|M_S;P).
\end{eqnarray}
Thus, Proposition 3 $\Rightarrow$ Proposition 1.\\\\
\textbf{Proposition 1 $\Rightarrow$ Proposition 2:}\\

By Proposition 1, for a given $P\in\mathcal{P}$ there exists a $p(k_1\dots k_N|P)$ such that $p(k_S|M_S;P)=\sum_{k_i:i \notin S}p(k_1\dots k_N|P)$,
for all jointly measurable subsets $S \subset \{1,\dots,N\}$. Now, there exists a probability distribution
over the ontic state space, $\mu(\lambda|P)$, such that 
\begin{equation}
 p(k_1\dots k_N|P)=\sum_{\lambda\in\Lambda}\xi(k_1\dots k_N|\lambda)\mu(\lambda|P)
\end{equation}
where $\xi(k_1\dots k_N|\lambda) \in \{0,1\}$. This is possible because any probability distribution can be decomposed as
a convex sum over deterministic distributions. Also, $p(k_j|M_j;P)=\sum_{k_i:i\neq j}p(k_1,\dots,k_N|P)$, so

\begin{equation}
 p(k_j|M_j;P)=\sum_{\lambda\in\Lambda}\mu(\lambda|P) \sum_{k_i:i\neq j}\xi(k_1\dots k_N|\lambda),
\end{equation}
which allows the definition
\begin{equation}
 \xi(k_j|M_j;\lambda)\equiv \sum_{k_i:i\neq j}\xi(k_1\dots k_N|\lambda) \in \{0,1\}, \forall j \in \{1\dots N\}.
\end{equation}
Since these are deterministic distributions, 
\begin{equation}
 \xi(k_1\dots k_N|\lambda)=\prod_{j=1}^{N}\xi(k_j|M_j;\lambda).
\end{equation}
Finally,
\begin{equation}
 p(k_S|M_S;P)=\sum_{\lambda\in\Lambda}\mu(\lambda|P)\prod_{s \in S}\xi(k_s|M_s;\lambda), 
\end{equation}
so there exists a measurement-noncontextual and outcome-deterministic model, i.e., Proposition 1 $\Rightarrow$ Proposition 2.\\
\\
\textbf{Proposition 2 $\Rightarrow$ Proposition 3:}\\\\
By Proposition 2, $\xi(k_i|M_i;\lambda) \in \{0,1\}, \forall i \in \{1\dots N\}$, such that
\begin{equation}
 p(k_S|M_S;P)=\sum_{\lambda\in\Lambda}\mu(\lambda|P)\prod_{s \in S}\xi(k_s|M_s;\lambda),
\end{equation}
$\forall$ jointly measurable subsets $S \subset \{1\dots N\}$. Clearly, this model is also a measurement-noncontextual 
and factorizable model because the assumption of outcome-determinism implies factorizability:
\begin{equation}
\xi(k_S|M_S;\lambda)=\prod_{s \in S}\xi(k_s|M_s;\lambda).
\end{equation}
\proofend

Note that this theorem itself is not new, but this particular reading of it in the framework of generalized noncontextuality \emph{is} new. 
In particular, the purpose of this restatement is to highlight why outcome-determinism is not an assumption
that can be taken for granted in noncontextual ontological models. Versions of this theorem have appeared in the literature
following Fine's original insight \cite{Fine,Fine2}. 
The fact that Proposition 2 implies Proposition 1 has been shown earlier in Ref.~\cite{LSW}. A similar result in the language 
of sheaf theory can be found Ref.~\cite{AB}, where the authors point out factorizability as the underlying assumption
in Bell-local and KS-noncontextual models: in effect they show the equivalence of Proposition 1 and Proposition 3. The sense in which 
Ref.~\cite{AB} refers to `non-contextuality'
is the notion of KS-noncontextuality, and while it is possible to provide a unified account of Bell-locality and KS-noncontextuality 
at a mathematical level, the generalized notion of noncontextuality \cite{genNC} does not admit such an account. In particular,
their definition of `non-contextuality' is stronger than the Spekkens' definition of measurement noncontextuality.
Indeed, as we have amply demonstrated in Chapter 1, generalized noncontextuality subsumes KS-noncontextuality but is not equivalent to it.

\paragraph*{Fine's theorem for Bell scenarios.}
Translating the preceding notions from noncontextual models to Bell-local models amounts to replacing `measurement-noncontextual and outcome-deterministic' by `locally deterministic'
and `measurement-noncontextual and factorizable' by `locally causal'. 
Consider the case of two-party Bell scenarios for simplicity, although the same considerations extend to general multiparty Bell 
scenarios in a straightforward manner. A two-party Bell scenario consists of measurements $\{M_1,\dots,M_N\}$, where $\{M_1,\dots,M_n\}$, $n<N-1$, are the measurement settings 
available to one party, say Alice, and $\{M_{n+1},\dots,M_N\}$ are the measurement settings available to the other party, say Bob.
The outcomes are denoted by $k_i\in\mathcal{K}_{M_i}$ for the respective measurement settings $M_i$. The jointly measurable subsets are given 
by $S\in \{\{i,j\}|i\in\{1,\dots,n\},j\in\{n+1,\dots,N\}\}$. Bell's assumption of local causality captures the notion of
a measurement noncontextual and factorizable model:

\begin{eqnarray}\nonumber
&&p(k_S|M_S;P)\\ 
&=&p(k_i,k_j|M_i,M_j;P)\\
&=&\sum_{\lambda\in\Lambda}\xi(k_i|M_i,\lambda)\xi(k_j|M_j,\lambda)\mu(\lambda|P).
\end{eqnarray}

Once factorizability is justified from Bell's assumption of local causality in this manner, Fine's theorem ensures that -- so far
as the existence of hidden variable models is concerned -- it is irrelevant whether the response functions for the measurement 
outcomes are deterministic or indeterministic. One does not need to worry about whether outcome-determinism for measurements is 
justified in Bell scenarios precisely because factorizability along with Fine's theorem absolves one of the need to provide such
a justification.
The crucial point, then, is the validity of factorizability in the more general case of noncontextual models. 
In general, factorizability is not justified in noncontextual models and, following Spekkens, one must
distinguish between the issue of noncontextuality and that of outcome-determinism when considering ontological models of an 
operational theory \cite{genNC}. If the goal is -- as it should be -- to obtain an experimental test of noncontextual models independent of the truth of 
quantum theory, then one needs to derive noncontextuality inequalities that do not rely on outcome-determinism at all. This is because
Fine's theorem for noncontextual models is of limited applicability -- namely, outcome-indeterministic response functions which satisfy factorizability
are shown by it to achieve no more generality than is already captured by outome-deterministic response functions in a KS-noncontextual model.
Outcome-indeterministic response functions that do not satisfy factorizability are not taken into account in a KS-noncontextual model.

For ontological models of operational \emph{quantum} theory, outcome-determinism for sharp (projective) measurements can be shown to follow from 
the assumption of preparation noncontextuality \cite{genNC}. Such a justification is not available for unsharp (nonprojective) measurements, which should therefore
be represented by outcome-indeterministic response functions. This issue has been discussed at length by Spekkens and the reader is referred
to Ref.~\cite{odum} for why and how this must be so. Therefore, to consider noncontextuality for unsharp measurements in full generality the noncontextuality
inequalities of interest are those which do not assume factorizability. An example is the LSW inequality for Specker's scenario \cite{LSW} that does not rely
on factorizability, although it does use the assumption of outcome determinism for sharp (projective) measurements.
The LSW inequality has been shown to be violated by quantum predictions \cite{KG}, thus ruling out noncontextual models of quantum theory
without invoking factorizability.
Note that the distinction between sharp and unsharp measurements is not part of the definition of a Bell-local model and one never has to worry
about this distinction to derive Bell's theorem. This distinction, however, becomes relevant for noncontextual models of quantum theory,
where the words `sharp' and `unsharp' have a clear meaning, the former referring to projective measurements and the latter to nonprojective
measurements.

In the next section, the polytope of correlations admissible in Specker's scenario is characterized.
\section{Correlations in Specker's scenario}

In this section three noncontextuality inequalities relevant to the correlations in Specker's
scenario are derived. They are shown to be equivalent to the known LSW inequality under relabelling of measurement
outcomes. This scenario involves three binary measurements, $\{M_1,M_2,M_3\}$, which are pairwise jointly measurable with outcomes 
labelled by $X_i\in\{0,1\}$ for $i\in\{1,2,3\}$. The statistics involved in Specker's scenario for a given preparation $P\in\mathcal{P}$ can be understood as a set of $12$ probabilities, $4$ for each pairwise joint measurement $M_{ij}$, 

\begin{equation}
\mathcal{S}\equiv\{p(X_iX_j|M_{ij};P)|X_i,X_j\in\{0,1\}, i,j\in\{1,2,3\},i<j\},
\end{equation}
subject to the obvious constraints of
positivity,
\begin{equation}
p(X_iX_j|M_{ij};P)\geq0 \quad \forall X_i,X_j,M_{ij},
\end{equation}
and normalization,
\begin{equation}
\sum_{X_i,X_j}p(X_iX_j|M_{ij};P)=1 \quad \forall M_{ij}.
\end{equation}

In addition to positivity and normalization, the statistics is assumed to obey the following condition:

\begin{eqnarray}\nonumber
&&\sum_{X_j}p(X_iX_j|M_{ij};P)\\
&=&\sum_{X_k}p(X_iX_k|M_{ik};P)\\
&\equiv&p(X_i|M_i;P),
\end{eqnarray}
for all $i<j,k$ where $i,j,k\in\{1,2,3\}$. Denoting $$\sum_{X_j}p(X_iX_j|M_{ij};P)\equiv p(X_i|M_i^j;P),$$ and $$\sum_{X_k}p(X_iX_k|M_{ik};P)\equiv p(X_i|M_i^k;P),$$ the condition
becomes
\begin{equation}
 p(X_i|M_i^j;P)=p(X_i|M_i^k;P)\equiv p(X_i|M_i;P).
\end{equation}

That is, the statistics of $M_i^j$, which is obtained by marginalizing the statistics of joint measurement $M_{ij}$, is identical to the statistics of $M_i^k$, which is obtained by 
marginalizing the statistics of joint measurement $M_{ik}$. If what has been measured is indeed a unique observable $M_i$ then its statistics relative to any preparation $P\in\mathcal{P}$ should
remain the same across joint measurements with different observables $M_j$ and $M_k$. Failure to meet this condition implies a failure of joint measurability: then
one can distinguish between $M_i^j$ and $M_i^k$ from their statistics relative to some preparation and they would therefore correspond to two different marginal observables $M_i^j$ and $M_i^k$ rather
than a unique observable $M_i$. This condition, $M_i^j\simeq M_i^k\simeq M_i$, is often called the \emph{no-disturbance} condition in the literature on contextuality. 
Operational quantum theory obeys the no-disturbance condition for joint measurements of generalized observables (which need not be projective or sequential).\footnote{
We will rarely use the terminology of ``no-disturbance'', preferring instead the operational equivalence between measurement procedures,
of which no-disturbance is a special case.}

\subsection{Kochen-Specker (KS) inequalities for Specker's scenario}
The four necessary and sufficient inequalities characterizing correlations which admit a KS-noncontextual model in Specker's scenario are given by:
\begin{equation}\label{ineq1}
R_3\equiv p(X_1\neq X_2|M_{12},P)+p(X_2\neq X_3|M_{23},P)+p(X_1\neq X_3|M_{13},P)\leq 2,
\end{equation}and
\begin{equation}\label{ineq2}
R_0\equiv p(X_1\neq X_2|M_{12},P)-p(X_2\neq X_3|M_{23},P)-p(X_1\neq X_3|M_{13},P)\leq 0,
\end{equation}
\begin{equation}\label{ineq3}
R_1\equiv p(X_2\neq X_3|M_{23},P)-p(X_1\neq X_3|M_{13},P)-p(X_1\neq X_2|M_{12},P)\leq 0,
\end{equation}
\begin{equation}\label{ineq4}
R_2\equiv p(X_1\neq X_3|M_{13},P)-p(X_1\neq X_2|M_{12},P)-p(X_2\neq X_3|M_{23},P)\leq 0.
\end{equation}
These inequalities have earlier appeared in Ref.~\cite{cabelloncycle}. A derivation is provided in the Appendix at the end of this chapter. 
Further, these inequalities
exhibit a curious property that no two of them can be violated by the same set of experimental statistics:
\begin{lemma}\label{kslemma}
There exists no set of distributions $\{p(X_i,X_j|M_{ij},P)|(ij)\in\{(12),(23),(13)\}\}$
that can violate any two of the four KS inequalities simultaneously.
\end{lemma}
\emph{Proof.}
Denoting $w_{12}\equiv p(X_1\neq X_2|M_{12},P)$, $w_{23}\equiv p(X_2\neq X_3|M_{23},P)$, and $w_{13}\equiv p(X_1\neq X_3|M_{13},P)$, the four
KS inequalities can be rewritten as:
\begin{eqnarray}
&&R_3\equiv w_{12}+w_{23}+w_{13}\leq 2,\\
&&R_0\equiv w_{12}-w_{23}-w_{13}\leq 0,\\
&&R_1\equiv w_{23}-w_{13}-w_{12}\leq 0,\\
&&R_2\equiv w_{13}-w_{23}-w_{12}\leq 0.
\end{eqnarray}

Now, violation of each of these is equivalent to the following, since $0\leq w_{12},w_{23},w_{13}\leq 1$:
\begin{eqnarray}\nonumber
&&R_3>2\Leftrightarrow w_{12}+w_{23}+w_{13}>2,\\\nonumber
&&R_0>0\Leftrightarrow w_{12}>w_{23}+w_{13} \Rightarrow w_{12}+w_{23}+w_{13}<2,\\\nonumber
&&R_1>0\Leftrightarrow w_{23}>w_{12}+w_{13} \Rightarrow w_{12}+w_{23}+w_{13}<2\\\nonumber
&&\text{ and }w_{12}<w_{23}-w_{13},\\\nonumber
&&R_2>0\Leftrightarrow w_{13}>w_{12}+w_{23} \Rightarrow w_{12}+w_{23}+w_{13}<2\\\nonumber
&&\text{ and }w_{12}<w_{13}-w_{23}.
\end{eqnarray}
It follows that violation of each inequality above is in conflict with a violation of each of the other three inequalities. Hence, there exist 
no conceivable measurement statistics that violate any two of the four KS inequalities simultaneously.\proofend

\subsection{Noncontextuality (NC) inequalities for Specker's scenario}

Consider the \emph{predictability} of each measurement $M_k$ defined as: 
\begin{equation}
\eta_{M_k}\equiv\max_P\{2\max_{X_k}p(X_k|M_k,P)-1\}, 
\end{equation}
where $P\in\mathcal{P}$ is any preparation of the system. Assuming the three measurements in Specker's scenario have the same predictability $\eta_0\equiv \eta_{M_1}=\eta_{M_2}=\eta_{M_3}$,
the following noncontextuality inequalities hold:\\\\
{\bf LSW inequality}
\begin{equation}\label{ncineq1}
R_3=p(X_1\neq X_2|M_{12},P)+p(X_2\neq X_3|M_{23},P)+p(X_1\neq X_3|M_{13},P)\leq 3-\eta_0,
\end{equation}
{\bf Three more inequalities}
\begin{equation}\label{ncineq2}
R_0=p(X_1\neq X_2|M_{12},P)-p(X_2\neq X_3|M_{23},P)-p(X_1\neq X_3|M_{13},P)\leq 1-\eta_0,
\end{equation}
\begin{equation}\label{ncineq3}
R_1=p(X_2\neq X_3|M_{23},P)-p(X_1\neq X_3|M_{13},P)-p(X_1\neq X_2|M_{12},P)\leq 1-\eta_0,
\end{equation}
\begin{equation}\label{ncineq4}
R_2=p(X_1\neq X_3|M_{13},P)-p(X_1\neq X_2|M_{12},P)-p(X_2\neq X_3|M_{23},P)\leq 1-\eta_0.
\end{equation}
These inequalities are derived in the Appendix. Note that violation of each of these inequalities implies the violation of 
the corresponding KS inequalities (recovered for $\eta_0=1$), but not conversely. 
\begin{lemma}
There exists no set of distributions $\{p(X_i,X_j|M_{ij},P)|(ij)\in\{(12),(23),(13)\}\}$
that can violate any two of the four NC inequalities simultaneously.
\end{lemma}
\emph{Proof.}
The proof trivially follows from Lemma \ref{kslemma}, since violation of any NC inequality implies violation of 
the corresponding KS inequality.\proofend

The predictability, $\eta_0$, quantifies how predictable a measurement
can be made in a variation over preparations: KS inequalities make sense only when $\eta_0=1$, i.e., it is possible to find a preparation which makes a 
given measurement perfectly predictable, a condition which is naturally satisfied by sharp (projective) measurements in quantum theory.
For the case of unsharp measurements, $\eta_0<1$, and the noncontextuality inequalities take this into account. When $\eta_0=0$, that is,
when the measurement outcomes are uniformly random (or completely unpredictable), the upper bounds in the noncontextuality inequalities become trivial
and a noncontextual model is always possible: simply ignore the system and toss a fair coin to decide whether to output $(X_i=0,X_j=1)$ or $(X_i=1,X_j=0)$ when
a pair of measurements $\{M_i,M_j\}$ is jointly implemented,
\begin{equation}
 p(X_i,X_j|M_{ij},P)=\frac{1}{2}\left(\delta_{X_i,0}\delta_{X_j,1}+\delta_{X_i,1}\delta_{X_j,0}\right).
\end{equation}
Clearly, $R_3=3$ for this, and $\eta_0=0$ since the marginal for each measurement $M_i$ is uniformly random independent of the preparation,
so the LSW inequality cannot be violated. This admits a noncontextual model since the response function for each measurement $M_i$ is a fair
coin flip independent of the system's ontic state and also of which other measurement it is jointly implemented with. The key feature that the
LSW inequality captures is this: that it is not possible to have a high degree of anticorrelation $R_3$ and a high degree of predictability
$\eta_0$ in a noncontextual model, and that there is a tradeoff between the two, given here by $R_3+\eta_0\leq 3$. Contextuality in this sense signifies
the ability to generate (anti)correlations which violate this tradeoff for values of $\eta_0<1$: the case $\eta_0=1$, as mentioned, is already covered
by the usual KS inequalities, and for $\eta_0=0$ there is no nontrivial tradeoff imposed by noncontextual models.
\subsection{Equivalence under relabelling of measurement outcomes}
The four NC inequalities (also the KS inequalities) are equivalent under relabelling measurement outcomes: To go from $R_3\leq 3-\eta_0$ to 
$R_0\leq 1-\eta_0$, simply relabel the measurement outcomes of $M_3$ as $X_3\rightarrow X'_3=1-X_3$, so that after the relabelling (denoted by primed quantities):
$w'_{12}=w_{12},w'_{23}=1-w_{23},w'_{13}=1-w_{13}$, and
$R'_3\equiv w'_{12}+w'_{23}+w'_{13}\leq 3-\eta_0$ becomes $w_{12}+(1-w_{23})+(1-w_{13})\leq 3-\eta_0$
which can be rewritten as $R'_3=R_0=w_{12}-w_{23}-w_{13}\leq 1-\eta_0$. Similarly, relabelling measurement outcomes of $M_2$ takes $R_3\leq 3-\eta_0$
to $R_2\leq 1-\eta_0$ and relabelling measurement outcomes of $M_1$ takes $R_3\leq 3-\eta_0$ to $R_1\leq 1-\eta_0$.
\subsection{Quantum violation of noncontextuality inequalities for Specker's scenario}
Quantum realization of Specker's scenario involves three unsharp qubit POVMs $M_k=\{E^k_0,E^k_1\}, k\in\{1,2,3\}$, where the effects are given by:
\begin{equation}
 E^k_{X_k}\equiv\frac{1}{2}I+(-1)^{X_k}\frac{\eta}{2}\vec{\sigma}.\hat{n}_k,\quad X_k\in\{0,1\}, 0\leq\eta\leq1.
\end{equation}
These can be rewritten as:
\begin{equation}
E^k_{X_k}=\eta\Pi^k_{X_k}+(1-\eta)\frac{I}{2}, 
\end{equation}
where $\Pi^k_{X_k}=\frac{1}{2}(I+(-1)^{X_k}\vec{\sigma}.\hat{n}_k)$ are the corresponding projectors. That is, $M_k$ is a noisy version of the projective measurement
of spin along the $\hat{n}_k$ direction, where the sharpness of the POVM is given by $\eta$. In this case, $p(X_k|M_k,P)=\Tr(\rho_P E^k_{X_k})$, where $\rho_P$ is the density matrix for preparation $P$ of the system and the predictability can be easily shown to be $\eta$: the preparation 
maximizing $\eta_{M_k}$ is a pure state along the $\hat{n}_k$ axis, i.e., $\rho_P=\Pi^k_{X_k}$.

Quantum violation of the LSW inequality has already been shown in Chapter 2 (based on Ref.~\cite{KG}). 
On account of the equivalence of the four NC inequalities
under relabelling of measurement outcomes, the violation of the other three NC inequalities besides LSW follows from appropriate 
relabellings of measurement outcomes in the quantum violation demonstrated in Chapter 2.\subsection{Specker polytope}
The statistics allowed in Specker's scenario, given that the no-disturbance condition holds, can be understood as a convex polytope in $\mathbb{R}^6$
with $12$ extreme points or vertices,
$8$ of which are deterministic and $4$ indeterministic. The measurement statistics are given by the vector of $12$ probabilities 
$\vec{v}(P)=(v^{ij}_{X_iX_j}(P)|X_i,X_j\in\{0,1\}, i,j\in\{1,2,3\},i<j)$, where $v^{ij}_{X_iX_j}(P)\equiv p(X_iX_j|M_{ij};P)$, constrained by the positivity, normalization 
and no-disturbance conditions which reduce the number of independent probabilities in $\vec{v}(P)$ from $12$ to $6$.

The deterministic vertices, which admit KS-noncontextual models, correspond to the $8$ possible tripartite joint distributions of the form,
$p(X_1,X_2,X_3|P)\equiv \delta_{X_1,X_1(P)}\delta_{X_2,X_2(P)},\delta_{X_3,X_3(P)}$, where $X_1(P),X_2(P),X_3(P)\in \{0,1\}$. The deterministic vertex $\vec{v}(P)$
can be obtained from this joint distribution as $v^{ij}_{X_iX_j}(P)=\sum_{X_k, k\neq i,j}p(X_1,X_2,X_3|P)=\delta_{X_i,X_i(P)}\delta_{X_j,X_j(P)}$.
These vertices are labelled lexicographically, $(X_1(P),X_2(P),X_3(P))$ as the decimal equivalent of binary number $X_1(P)X_2(P)X_3(P)$:

\begin{align}
\vec{v}_0(P): v^{12}_{00}(P)=v^{23}_{00}(P)=v^{13}_{00}(P)=1,\\
\vec{v}_1(P): v^{12}_{00}(P)=v^{23}_{01}(P)=v^{13}_{01}(P)=1,\\
\vec{v}_2(P): v^{12}_{01}(P)=v^{23}_{10}(P)=v^{13}_{00}(P)=1,\\
\vec{v}_3(P): v^{12}_{01}(P)=v^{23}_{11}(P)=v^{13}_{01}(P)=1,\\
\vec{v}_4(P): v^{12}_{10}(P)=v^{23}_{00}(P)=v^{13}_{10}(P)=1,\\
\vec{v}_5(P): v^{12}_{10}(P)=v^{23}_{01}(P)=v^{13}_{11}(P)=1,\\
\vec{v}_6(P): v^{12}_{11}(P)=v^{23}_{10}(P)=v^{13}_{10}(P)=1,\\
\vec{v}_7(P): v^{12}_{11}(P)=v^{23}_{11}(P)=v^{13}_{11}(P)=1.
\end{align}
Note that these deterministic vertices satisfy all the four KS inequalities, Eqs. (\ref{ineq1})-(\ref{ineq4}), and therefore also the four noncontextuality 
inequalities, Eqs. (\ref{ncineq1})-(\ref{ncineq4}). That is, they admit a KS-noncontextual model. Indeed, the convex set that these $8$ extreme points 
define is a KS-noncontextuality polytope, analogous to a Bell polytope in a Bell scenario. This polytope is a 
subset of the larger Specker polytope which in addition to these $8$ vertices includes the $4$ indeterministic vertices in Specker's scenario.

The indeterministic vertices, which do \emph{not} admit KS-noncontextual models, correspond to the $4$ sets of pairwise joint distributions given by:
\begin{eqnarray}
\vec{v}_8(P)&:&v^{12}_{01}(P)=v^{12}_{10}(P)=v^{23}_{00}(P)=v^{23}_{11}(P)=v^{13}_{00}(P)=v^{13}_{11}(P)=\frac{1}{2},\\
\vec{v}_9(P)&:&v^{12}_{00}(P)=v^{12}_{11}(P)=v^{23}_{01}(P)=v^{23}_{10}(P)=v^{13}_{00}(P)=v^{13}_{11}(P)=\frac{1}{2},\\
\vec{v}_{10}(P)&:&v^{12}_{00}(P)=v^{12}_{11}(P)=v^{23}_{00}(P)=v^{23}_{11}(P)=v^{13}_{01}(P)=v^{13}_{10}(P)=\frac{1}{2},\\
\vec{v}_{11}(P)&:&v^{12}_{01}(P)=v^{12}_{10}(P)=v^{23}_{01}(P)=v^{23}_{10}(P)=v^{13}_{01}(P)=v^{13}_{10}(P)=\frac{1}{2}.
\end{eqnarray}
The vertex $\vec{v}_8(P)$ violates inequalities (\ref{ineq2}) and (\ref{ncineq2}) ($\eta_0>0$), $\vec{v}_9(P)$ violates inequalities (\ref{ineq3})
and (\ref{ncineq3}) ($\eta_0>0$), $\vec{v}_{10}(P)$ violates inequalities (\ref{ineq4}) and (\ref{ncineq4}) ($\eta_0>0$),
and $\vec{v}_{11}(P)$ violates inequalities (\ref{ineq1}) and (\ref{ncineq1}) ($\eta_0>0$). 
Note that these vertices are equivalent under relabellings, that is, $\vec{v}_8(P)$ turns to $\vec{v}_{11}(P)$ on relabelling outcomes
of $M_3$, $\vec{v}_9(P)$ to $\vec{v}_{11}(P)$ on relabelling outcomes of $M_1$, and 
$\vec{v}_{10}(P)$ to $\vec{v}_{11}(P)$ on relabelling outcomes of $M_2$. Note that the vertex $\vec{v}_{11}(P)$ corresponds to the `overprotective seer' (OS)
correlations of Ref. \cite{LSW} which maximally violate the LSW inequality when $\eta_0<1$.

\subsection{Limitations of the joint probability distribution criterion for deciding contextuality}
All the Bell-Kochen-Specker type analyses of contextuality ultimately hinge on ruling out the existence of a joint probability
distribution that reproduces the operational statistics of various jointly measurable observables as marginals. Deciding whether
such a joint distribution exists is called a marginal problem \cite{CF}. That this is a limited criterion to 
decide the question of contextuality without also making the assumption of outcome determinism or factorizability is borne out by correlations 
in Specker's scenario that lie outside the polytope of correlations admissible in KS-noncontextual models but are realizable in noncontextual models.
Violation of the LSW inequality by unsharp measurements in quantum theory rules out such noncontextual models \cite{KG}. 

Once outcome determinism for unsharp measurements (ODUM, cf.\cite{odum}) is abandoned,
the existence of a joint distribution is no longer necessary to characterize noncontextual models.
Further, in the case of an arbitrary operational theory which isn't quantum theory it isn't obvious whether outcome-determinism for 
measurements can at all be justified from the assumption of preparation and measurement noncontextuality. An experimentally interesting
and robust noncontextuality inequality should not assume that the operational theory describing the experiment is quantum theory and 
instead derive from the assumption of noncontextuality alone, given some operational equivalences between preparation procedures or 
measurement procedures. Violation of the LSW inequality only indicates that \emph{quantum theory} does not admit a noncontextual ontological model.
The ideal to aspire for is something akin to Bell inequalities which are theory-independent. That such an ideal is achievable will be shown 
in the following chapters.

\section{Chapter summary}
To summarize, the chief takeaways from this chapter are the following:
\begin{enumerate}
 \item Fine's theorem for noncontextual models only applies in cases where the correlations between measurement outcomes are mediated exclusively by 
the ontic state $\lambda$ of the system. When this is not the case and factorizability fails, it's possible that the measurement outcomes share correlations
that are not on account of the measured system but an artifact of the measurement apparatus. Considering noncontextual models which are not factorizable allows
one to handle this situation.
 \item The no-disturbance polytope of Specker's scenario admits 4 indeterministic extremal points, related to each other by relabellings of measurement outcomes, that
are related to the `OS box' of Ref. \cite{LSW}. Corresponding to these 4 extremal points are 4 Kochen-Specker inequalities assuming outcome determinism, and
4 noncontextuality inequalities that do not assume outcome determinism.
\end{enumerate}
Hence, Fine's theorem, unlike its implications for Bell's theorem, does not absolve one of the need to justify outcome determinism 
in noncontextual ontological models. We therefore need to further investigate how a failure of outcome determinism or factorizability
in the case of more well-known KS inequalities should be handled.
Another open question is how to derive noncontextuality inequalities for arbitrary operational theories, rather than just quantum theory,
without any assumption of outcome determinism or factorizability. These questions will be taken up in forthcoming chapters.

\section*{Appendix}
\subsection*{Constraints on the operational statistics from positivity, normalization and no-disturbance}\label{opconstraints}
The notation here is simplified as follows: the measurements are denoted by $e\equiv M_1, f\equiv M_2, g\equiv M_3$, and their outcomes 
by $e_k\equiv (X_1=k)$, $f_k\equiv (X_2=k)$, $g_k\equiv (X_3=k)$, where $k\in\{0,1\}$. Thus there are three binary observables, $e, f, g$,
each taking values in $\{0,1\}$ and measured on a system prepared according to some preparation $P$. $e_0$ denotes the outcome $e=0$ and $e_1$ denotes $e=1$. 
Analogous notation applies for outcomes of $f$ and $g$ as well. The probability distributions on these observables associated with 
the preparation $P$ are denoted by $w_P(e)\equiv\{w_P(e_0),w_P(e_1)\},w_P(f)\equiv\{w_P(f_0),w_P(f_1)\},w_P(g)\equiv\{w_P(g_0),w_P(g_1)\}$.
The experimental statistics correspond to the joint measurement of every pair of observables: 
\begin{eqnarray}\nonumber
w_P(e,f)&\equiv&\{w_P(e_0,f_0),w_P(e_0,f_1),w_P(e_1,f_0),w_P(e_1,f_1)\},\\\nonumber
w_P(f,g)&\equiv&\{w_P(f_0,g_0),w_P(f_0,g_1),w_P(f_1,g_0),w_P(f_1,g_1)\},\\\nonumber
w_P(e,g)&\equiv&\{w_P(e_0,g_0),w_P(e_0,g_1),w_P(e_1,g_0),w_P(e_1,g_1)\}.
\end{eqnarray}

In addition to the usual positivity and normalization constraints for probability distributions,
the no-disturbance condition on the pairwise joint distributions yields:
\begin{eqnarray}\nonumber
w_P(e_0,f_0)+w_P(e_0,f_1)&=&w_P(e_0,g_0)+w_P(e_0,g_1)\\\nonumber
&\equiv& w_P(e_0),\\\nonumber
\Rightarrow w_P(e_1,f_0)+w_P(e_1,f_1)&=&w_P(e_1,g_0)+w_P(e_1,g_1)\\\nonumber
&\equiv& w_P(e_1),\\\nonumber
w_P(f_0,g_0)+w_P(f_0,g_1)&=&w_P(e_0,f_0)+w_P(e_1,f_0)\\\nonumber
&\equiv& w_P(f_0),\\\nonumber
\Rightarrow w_P(f_1,g_0)+w_P(f_1,g_1)&=&w_P(e_0,f_1)+w_P(e_1,f_1)\\\nonumber
&\equiv& w_P(f_1),\\\nonumber
w_P(e_0,g_0)+w_P(e_1,g_0)&=&w_P(f_0,g_0)+w_P(f_1,g_0)\\\nonumber
&\equiv& w_P(g_0),\\\nonumber
\Rightarrow w_P(e_0,g_1)+w_P(e_1,g_1)&=&w_P(f_0,g_1)+w_P(f_1,g_1)\\\nonumber
&\equiv& w_P(g_1).
\end{eqnarray}
Normalization gets rid of three parameters out of the twelve in the experimental statistics while no-disturbance eliminates
three more parameters. There are, therefore, six independent parameters describing the experimental statistics:
\begin{eqnarray}
 w_{12}=w_P(e_0,f_1)+w_P(e_1,f_0),\\
 w_{23}=w_P(f_0,g_1)+w_P(f_1,g_0),\\
 w_{13}=w_P(e_0,g_1)+w_P(e_1,g_0),\\
 p_1\equiv w_P(e_0),\\
 p_2\equiv w_P(f_0),\\
 p_3\equiv w_P(g_0),
\end{eqnarray}
subject to $0\leq w_{12},w_{23},w_{13},p_1,p_2,p_3\leq 1$.
Using the no-disturbance and normalization conditions:
\begin{eqnarray}\nonumber
w_P(e_0,f_1)=\frac{w_{12}+p_1-p_2}{2}&,& w_P(e_1,f_0)=\frac{w_{12}-p_1+p_2}{2},\\\nonumber
w_P(e_0,f_0)=\frac{p_1+p_2-w_{12}}{2}&,& w_P(e_1,f_1)=1-\frac{w_{12}+p_1+p_2}{2},\\\nonumber
w_P(f_0,g_1)=\frac{w_{23}+p_2-p_3}{2}&,& w_P(f_1,g_0)=\frac{w_{23}-p_2+p_3}{2},\\\nonumber
w_P(f_0,g_0)=\frac{p_2+p_3-w_{23}}{2}&,& w_P(f_1,g_1)=1-\frac{w_{23}+p_2+p_3}{2},\\\nonumber
w_P(e_0,g_1)=\frac{w_{13}+p_1-p_3}{2}&,& w_P(e_1,g_0)=\frac{w_{13}-p_1+p_3}{2},\\\nonumber
w_P(e_0,g_0)=\frac{p_1+p_3-w_{13}}{2}&,& w_P(e_1,g_1)=1-\frac{w_{13}+p_1+p_3}{2}.
\end{eqnarray}
The positivity requirements on these translate to the following inequalities:
\begin{eqnarray}
 |p_1-p_2|&\leq& w_{12}\leq p_1+p_2\leq 2-w_{12},\\
 |p_2-p_3|&\leq& w_{23}\leq p_2+p_3\leq 2-w_{23},\\
 |p_1-p_3|&\leq& w_{13}\leq p_1+p_3\leq 2-w_{13}.
\end{eqnarray}

\subsection*{Deriving the KS and NC inequalities}\label{deriv}

\subsubsection*{KS inequalities}
The KS inequalities derive from the existence of a joint probability distribution $p(X_1X_2X_3)$ such that
$p(X_iX_j|M_{ij},P)=\sum_{X_k}p(X_1X_2X_3)$, where $i,j,k$ are distinct indices in $\{1,2,3\}$. Therefore
the following must hold:
\begin{eqnarray}
p(001)&=&p(00|M_{12},P)-p(000),\nonumber\\
p(010)&=&p(00|M_{13},P)-p(000),\nonumber\\
p(100)&=&p(00|M_{23},P)-p(000),\nonumber\\
p(011)&=&p(01|M_{12},P)-p(010)\nonumber\\
&=&p(01|M_{12},P)-p(00|M_{13},P)+p(000),\nonumber\\
p(101)&=&p(10|M_{12},P)-p(100)\nonumber\\
&=&p(10|M_{12},P)-p(00|M_{23},P)+p(000),\nonumber\\
p(110)&=&p(10|M_{13},P)-p(100)\nonumber\\
&=&p(10|M_{13},P)-p(00|M_{23},P)+p(000),\nonumber\\
p(111)&=&1-p(00|M_{12},P)-p(01|M_{12},P)-p(10|M_{12},P)\nonumber\\
&-&p(10|M_{13},P)+p(00|M_{23},P)-p(000).\nonumber
\end{eqnarray}
Expressing the probabilities in terms of the six free parameters identified earlier, namely, 
the anticorrelation probabilities, $w_{12},w_{23},w_{13}$, and the marginals $p_1,p_2,p_3$, the positivity constraints, $0\leq p(X_1X_2X_3)\leq1$, require:
\begin{eqnarray}
&&0\leq p(000)\leq 1,\nonumber\\
&&0\leq p(001)\leq 1\nonumber\\
&&\Leftrightarrow p_1+p_2-2\leq w_{12}\leq p_1+p_2-2p(000)\nonumber\\
&&0\leq p(010)\leq 1\nonumber\\
&&\Leftrightarrow p_1+p_3-2\leq w_{13}\leq p_1+p_3-2p(000)\nonumber\\
&&0\leq p(100)\leq 1\nonumber\\
&&\Leftrightarrow p_2+p_3-2\leq w_{23}\leq p_2+p_3-2p(000)\nonumber\\
&&0\leq p(011)\leq 1\nonumber\\
&&\Leftrightarrow p_2+p_3-2p(000)\leq w_{12}+w_{13}\leq 2-2p(000)+p_2+p_3\nonumber\\
&&0\leq p(101)\leq 1\nonumber\\
&&\Leftrightarrow p_1+p_3-2p(000)\leq w_{12}+w_{23}\leq 2-2p(000)+p_1+p_3\nonumber\\
&&0\leq p(110)\leq 1\nonumber\\
&&\Leftrightarrow p_1+p_2-2p(000)\leq w_{13}+w_{23}\leq 2-2p(000)+p_1+p_2\nonumber\\
&&0\leq p(111)\leq 1\nonumber\\
&&\Leftrightarrow -2p(000)\leq w_{12}+w_{23}+w_{13}\leq 2-2p(000).\nonumber
\end{eqnarray}
Combining the inequalities to eliminate $p(000)$, and using the fact that $0\leq p(000)\leq1$:
\begin{eqnarray}
&&0\leq p(000)\leq1, 0\leq p(111)\leq 1\nonumber\\
&\Rightarrow&-2p(000)\leq 0\leq w_{12}+w_{23}+w_{13}\leq 2-2p(000)\leq 2,\nonumber\\
&&0\leq p(010)\leq1, 0\leq p(101)\leq1\nonumber\\
&\Rightarrow&0\leq w_{12}+w_{23}-w_{13}\leq 2 \leq 4-2p(000),\nonumber\\
&&0\leq p(110)\leq1, 0\leq p(001)\leq1\nonumber\\
&\Rightarrow&0\leq w_{23}+w_{13}-w_{12}\leq 2 \leq 4-2p(000),\nonumber\\
&&0\leq p(011)\leq1, 0\leq p(100)\leq1\nonumber\\
&\Rightarrow&0\leq w_{12}+w_{13}-w_{23}\leq 2 \leq 4-2p(000).\nonumber
\end{eqnarray}
Of these, the KS inequalities, which are not trivially true by normalization and positivity,
are the following:
\begin{eqnarray}
&&R_3\equiv w_{12}+w_{23}+w_{13}\leq 2,\\
&&R_0\equiv w_{12}-w_{23}-w_{13}\leq 0,\\
&&R_1\equiv w_{23}-w_{12}-w_{13}\leq 0,\\
&&R_2\equiv w_{13}-w_{12}-w_{23}\leq 0.
\end{eqnarray}
Note that $p(000)\leq p_1,p_2,p_3$, since $p_1=p(000)+p(001)+p(010)+p(011)$, etc. As long as the KS inequalities 
are satisfied, one can define a joint probability distribution by choosing a suitable 
$p(000)\leq\min\{p_1,p_2,p_3\}$.

To summarize, there are following constraints on the six parameters, $\{w_{12},w_{23},w_{13},p_1,p_2,p_3\}$, characterizing the
polytope of KS-noncontextual correlations:
\begin{eqnarray}
&&0\leq p_1,p_2,p_3,w_{12},w_{23},w_{13}\leq 1,\\
&&|p_1-p_2|\leq w_{12}\leq \min\{p_1+p_2, 2-p_1-p_2\},\\
&&|p_2-p_3|\leq w_{23}\leq \min\{p_2+p_3, 2-p_2-p_3\},\\
&&|p_1-p_3|\leq w_{13}\leq \min\{p_1+p_3, 2-p_1-p_3\},\\
&&w_{12}+w_{23}+w_{13}\leq 2,\\
&&w_{12}-w_{23}-w_{13}\leq 0,\\
&&w_{23}-w_{12}-w_{13}\leq 0,\\
&&w_{13}-w_{12}-w_{23}\leq 0.
\end{eqnarray}

\subsubsection*{NC inequalities}
In deriving the NC inequalities, I closely follow the derivation of the LSW inequality in Ref.~\cite{LSW}. For a more detailed explication of the 
principles underlying this derivation, the reader may consult Ref.~\cite{odum}. The assumptions used are: measurement noncontextuality
and preparation noncontextuality on account of the fact that in operational quantum theory, preparation noncontextuality implies outcome determinism
for sharp measurements \cite{genNC}.

Define
\begin{eqnarray}
&&R_3(\lambda)\equiv w_{12}(\lambda)+w_{23}(\lambda)+w_{13}(\lambda),\\
&&R_0(\lambda)\equiv w_{12}(\lambda)-w_{23}(\lambda)-w_{13}(\lambda),\\
&&R_1(\lambda)\equiv w_{23}(\lambda)-w_{12}(\lambda)-w_{13}(\lambda),\\
&&R_2(\lambda)\equiv w_{13}(\lambda)-w_{12}(\lambda)-w_{23}(\lambda),
\end{eqnarray}
where
$w_{ij}(\lambda)\equiv \xi(X_i\neq X_j|M_{ij};\lambda)$, for all $(ij)\in\{(12),(23),(13)\}$. Note that to any given preparation, the ontological model
associates a distribution $\mu(\lambda|P)\geq0$, where $\sum_{\lambda\in\Lambda}\mu(\lambda|P)=1$, and 
$p(X|M,P)=\sum_{\lambda\in\Lambda}\mu(\lambda|P)\xi(X|M;\lambda)$,
where $\xi(X|M;\lambda)\in[0,1]$ is the response function of outcome $X$ when measurement $M$ is performed and the system's ontic state is $\lambda$.
Therefore:
\begin{eqnarray}
R_3=\sum_{\lambda\in\Lambda} \mu(\lambda|P)R_3(\lambda)\leq\max_{\lambda}R_3(\lambda),\\
R_0=\sum_{\lambda\in\Lambda} \mu(\lambda|P)R_0(\lambda)\leq\max_{\lambda}R_0(\lambda),\\
R_1=\sum_{\lambda\in\Lambda} \mu(\lambda|P)R_1(\lambda)\leq\max_{\lambda}R_1(\lambda),\\
R_2=\sum_{\lambda\in\Lambda} \mu(\lambda|P)R_2(\lambda)\leq\max_{\lambda}R_2(\lambda).
\end{eqnarray}

To maximize $R_3(\lambda)$ in this noncontextual model, one needs to maximize each anticorrelation term $w_{12}(\lambda),w_{23}(\lambda),w_{13}(\lambda)$.
To maximize $R_0(\lambda)$, maximize $w_{12}(\lambda)$ and minimize $w_{23}(\lambda), w_{13}(\lambda)$. Similarly,
to maximize $R_1(\lambda)$, maximize $w_{23}(\lambda)$ and minimize $w_{12}(\lambda), w_{13}(\lambda)$, and 
to maximize $R_2(\lambda)$, maximize $w_{13}(\lambda)$ and minimize $w_{12}(\lambda), w_{23}(\lambda)$.

The single measurement response functions are given by
\begin{equation}
 \xi(X_i|M_i;\lambda)=\eta\delta_{X_i,X_i(\lambda)}+(1-\eta)\left(\frac{1}{2}\delta_{X_i,0}+\frac{1}{2}\delta_{X_i,1}\right),
\end{equation}
$i\in\{1,2,3\}$, in keeping with the assumption of outcome determinism for projectors but not so for nonprojective positive operators \cite{odum}. 
The general form the pairwise response function for measurements $\{M_i,M_j\}$ is given by:
\begin{eqnarray}
\xi(X_i,X_j|M_{ij};\lambda)&=&\alpha \delta_{X_i,X_i(\lambda)}\delta_{X_j,X_j(\lambda)}\\
&+&\beta \delta_{X_i,X_i(\lambda)}\left(\frac{1}{2}\delta_{X_j,0}+\frac{1}{2}\delta_{X_j,1}\right)\nonumber\\
&+&\gamma \left(\frac{1}{2}\delta_{X_i,0}+\frac{1}{2}\delta_{X_i,1}\right)\delta_{X_j,X_j(\lambda)}\nonumber\\
&+&\delta \left(\frac{1}{2}\delta_{X_i,0}\delta_{X_j,0}+\frac{1}{2}\delta_{X_i,1}\delta_{X_j,1}\right)\nonumber\\
&+&\epsilon \left(\frac{1}{2}\delta_{X_i,0}\delta_{X_j,1}+\frac{1}{2}\delta_{X_i,1}\delta_{X_j,0}\right).\nonumber
\end{eqnarray}
The marginals are
\begin{eqnarray}
\xi(X_i|M_{ij};\lambda)&=&(\alpha+\beta)\delta_{X_i,X_i(\lambda)}\\
&+&(\gamma+\delta+\epsilon)\left(\frac{1}{2}\delta_{X_i,0}+\frac{1}{2}\delta_{X_i,1}\right),\nonumber
\end{eqnarray}
and
\begin{eqnarray}
\xi(X_j|M_{ij};\lambda)&=&(\alpha+\gamma)\delta_{X_j,X_j(\lambda)}\nonumber\\
&+&(\beta+\delta+\epsilon)\left(\frac{1}{2}\delta_{X_j,0}+\frac{1}{2}\delta_{X_j,1}\right),
\end{eqnarray}
so that the following must hold on account of $\xi(X_i|M_{ij};\lambda)=\xi(X_i|M_i;\lambda)$ and $\xi(X_j|M_{ij};\lambda)=\xi(X_j|M_j;\lambda)$:
\begin{equation}
\alpha+\beta=\alpha+\gamma=\eta,
\end{equation}
\begin{equation}
\gamma+\delta+\epsilon=\beta+\delta+\epsilon=1-\eta.
\end{equation}

To maximize anticorrelation $w_{ij}(\lambda)$: the $\beta$ and $\gamma$ terms yield correlation as often as anticorrelation, 
so $\beta=\gamma=0$. The $\delta$ term always yields correlation, so $\delta=0$. Only $\alpha$ and $\epsilon$ terms 
allow for more anticorrelation than correlation. This means $\alpha=\eta$ and $\epsilon=1-\eta$. The pairwise response function
maximizing anticorrelation $w_{ij}(\lambda)$ is given by
\begin{eqnarray}
&&\xi(X_iX_j|M_{ij};\lambda)\nonumber\\
&=&\eta\delta_{X_i,X_i(\lambda)}\delta_{X_j,X_j(\lambda)}\nonumber\\
&+&(1-\eta)\left(\frac{1}{2}\delta_{X_i,0}\delta_{X_j,1}+\frac{1}{2}\delta_{X_i,1}\delta_{X_j,0}\right).\nonumber
\end{eqnarray}
This maximizing response function constrains the anticorrelation probability as:
\begin{equation}
1-\eta\leq w_{ij}(\lambda)\leq 1.
\end{equation}

To minimize anticorrelation $w_{ij}(\lambda)$: the $\beta$ and $\gamma$ terms yield correlation as often as anticorrelation, 
so $\beta=\gamma=0$. The $\epsilon$ term always yields anticorrelation, so $\epsilon=0$. Only $\alpha$ and $\delta$ terms 
allow for more correlation than anticorrelation. This means $\alpha=\eta$ and $\delta=1-\eta$. The pairwise response function
minimizing anticorrelation $w_{ij}(\lambda)$ is given by
\begin{eqnarray}
&&\xi(X_iX_j|M_{ij};\lambda)\nonumber\\
&=&\eta\delta_{X_i,X_i(\lambda)}\delta_{X_j,X_j(\lambda)}\nonumber\\
&+&(1-\eta)\left(\frac{1}{2}\delta_{X_i,0}\delta_{X_j,0}+\frac{1}{2}\delta_{X_i,1}\delta_{X_j,1}\right).\nonumber
\end{eqnarray}
This minimizing response function constrains the anticorrelation probability as:
\begin{equation}
0\leq w_{ij}(\lambda)\leq \eta.
\end{equation}

$R_3(\lambda)$ is maximized by considering the response function maximizing anticorrelation for each of $w_{ij}(\lambda)$,
and noting that of the eight possible assignments $\lambda\rightarrow (X_1(\lambda),X_2(\lambda),X_3(\lambda)) \in \{0,1\}^3$,
the assignments maximizing $R_3(\lambda)$ are $\{(001),(010),(011),(100),(101),(110)\}$, each of which has two anticorrelated pairs
and a third correlated pair such that the anticorrelation probability becomes
$$\max_{\lambda}R_3(\lambda)=2\eta+3(1-\eta)=3-\eta,$$ and therefore
\begin{equation}
R_3\leq 3-\eta,
\end{equation}
the LSW inequality.

$R_0(\lambda)$ is maximized by considering the response function maximizing anticorrelation for $w_{12}(\lambda)$ and
response functions minimizing anticorrelation for $w_{23}(\lambda)$ and $w_{13}(\lambda)$. Noting that of the eight possible
assignments $\lambda\rightarrow (X_1(\lambda),X_2(\lambda),X_3(\lambda)) \in \{0,1\}^3$,
the assignments maximizing $R_0(\lambda)$ are $\{(010),(011),(100),(101)\}$:
for $\{(010),(101)\}$, $w_{12}(\lambda)=1$, $w_{23}(\lambda)=\eta$, and $w_{13}(\lambda)=0$, and 
for $\{(011),(100)\}$, $w_{12}(\lambda)=1$, $w_{23}(\lambda)=0$, and $w_{13}(\lambda)=\eta$, so that
$$\max_{\lambda}R_0(\lambda)=1-\eta,$$ and therefore
\begin{equation}
R_0\leq 1-\eta.
\end{equation}
Similarly, the NC inequalities for $R_1$ and $R_2$ follow: $R_1\leq 1-\eta$ and $R_2\leq 1-\eta$.

\chapter{Beyond the Kochen-Specker theorem: its operationalization and experimental testability}
The Kochen-Specker theorem demonstrates that it is not possible to reproduce the predictions of quantum theory in terms of a 
hidden variable model where the hidden variables assign a value to every projector deterministically and noncontextually. A 
noncontextual value-assignment to a projector is one that does not depend on which other projectors -- the context -- are measured
together with it. In this chapter, using the generalized notion of noncontextuality that applies to both measurements and preparations \cite{genNC},
we propose a scheme -- inspired by the Kochen-Specker theorem -- for deriving inequalities that test whether a given set of experimental statistics is consistent with a noncontextual
model. Unlike previous inequalities inspired by the Kochen-Specker theorem, we do not assume that the value-assignments are
deterministic and therefore in the face of a violation of our inequality, the possibility of salvaging noncontextuality
by abandoning determinism is no longer an option. Our approach is operational in the sense that it does not presume quantum
theory: a violation of our inequality implies the impossibility of a noncontextual model for {\em any} operational theory that
can account for the experimental observations, including any successor to quantum theory. Our revision of 
the Kochen-Specker theorem thereby renders noncontextuality an experimentally testable hypothesis applicable to arbitrary operational theories.

In the last section of this chapter we discuss what a realistic experimental test of noncontextuality would entail,
based on the experiment reported in Ref.~\cite{exptl}. The experiment of Ref.~\cite{exptl} does not test a noncontextuality 
inequality directly inspired by the Kochen-Specker theorem but instead a simpler noncontextuality inequality that we will derive.
However, what we want to emphasize is not this experiment {\em per se} but 
the principles underlying it, so that the methodology of Ref.~\cite{exptl} can be generalized
for testing any noncontextuality inequality, including one derived from a Kochen-Specker construction. 
This chapter is based on 
work published in Refs.~\cite{KunjSpek} and \cite{exptl}.

\section{Introduction}
As we have seen in Chapter 1, an {\em ontological model} of operational quantum theory~\cite{harriganspekkens} associates to each quantum system a set of physical or {\em ontic} states, and imagines that each preparation
procedure is associated with a probability distribution over such ontic states and each measurement procedure is associated with a conditional probability distribution over its outcomes for 
each ontic state. 
This ontological models framework \cite{harriganspekkens} does not prejudice
the question of whether any of the variables remain unknown (i.e.
hidden) to one who knows the preparation procedure, hence we use the term ``ontic'' rather than ``hidden'' for 
the presumed physical states.
Contrary to na\"{i}ve impressions, such models have no difficulty reproducing quantum predictions \emph{unless} additional assumptions are made.
The Kochen-Specker theorem~\cite{KS67}  derives a contradiction
from an assumption  we term {\em KS-noncontextuality}.

Consider a set of 
measurements, each represented by an orthonormal basis, 
where some rays are common to more than one basis.  
Every ontic state
assigns a
definite value, 0 or 1, to each ray, {\em regardless of the basis (i.e. context) 
in which the ray appears}. If a ray is assigned 
value 1 (0) by an ontic
state $\lambda$, 
the measurement outcome associated with that ray 
occurs with
probability 1 (0) when any measurement 
including the ray is implemented on the system in ontic state $\lambda$. 
Hence, for every basis, precisely one
ray must be assigned the value 1 and the others 0.

Unlike the Kochen-Specker theorem, determinism\footnote{Note that, for simplicity, by ``determinism'' we mean the notion we have elsewhere called ``outcome determinism''. Unless otherwise 
specified, ``determinism'' will
mean ``outcome determinism'' in this thesis.} is not an assumption 
of Bell's theorem~\cite{Bell76, Wiseman}.
Even in Bell's 
1964 article~\cite{Bell64}, where deterministic assignments are used,
determinism is not assumed but
rather {\em derived} from local causality and the fact that quantum theory
predicts perfect correlations if the same observable is measured on two
parts of a maximally entangled state (an argument
from EPR~\cite{EPR} that Bell recycled\footnote{Indeed, in Ref.~\cite{Bellbook} (p. 157), Bell writes

\begin{quote}
My own first paper on [the subject of Bell's Theorem] ... starts with a summary
of the EPR argument from locality to deterministic hidden variables. But the
commentators have almost universally reported that it begins with deterministic
hidden variables.
\end{quote}
Although Wiseman has disputed Bell's account of the role of determinism in his
first paper~\cite{Wiseman}, see Norsen's response~\cite{Norsen}.}).
Similarly, rather than assuming determinism in
noncontextual models, one can derive it \cite{genNC} from a generalized notion of noncontextuality 
and from two facts about quantum theory: (i) the outcome of a 
measurement of some observable is perfectly predictable whenever the preceding
preparation is of an eigenstate of that observable, and (ii) the
indistinguishability, relative to all quantum
measurements, of different convex decompositions of the completely mixed state
into pure states.

Hence, in the Kochen-Specker theorem one can replace the assumption of determinism with
the generalized notion of noncontextuality and the quantum prediction of
perfect predictability. If perfect predictability is 
observed, then
in the face of the resulting contradiction, one must 
abandon
noncontextuality:
one 
cannot salvage 
it
by abandoning
determinism.

Of course, no real experiment ever yields {\em perfect} predictability, so this
manner of ruling out noncontextuality is not robust to experimental error. 
We will show how to contend with the lack of perfect predictability of
measurements and how to derive an experimentally-robust noncontextuality inequality for any KS-uncolourability proof of the Kochen-Specker theorem.

The original  proof of the KS theorem  required 117 rays in a 3-dimensional Hilbert space~\cite{KS67}.
We use the simpler proof in Ref.~\cite{Cab1}, requiring
a 4-dimensional Hilbert space and 18 rays that appear in 9 orthonormal bases,
each ray appearing in two bases
(Fig.~\ref{CEGAhypergraph}(a)). There is no
0-1 assignment to these rays that respects KS-noncontextuality: the
hypergraph is
{\em KS-uncolourable} (as shown in Fig.~\ref{CEGAhypergraph}(b)).
Of course, if the value assigned to a ray were allowed to be 0 in one basis and
1 in the other (KS-contextuality)
then there is no contradiction.
\begin{figure}
 \includegraphics[scale=0.55]{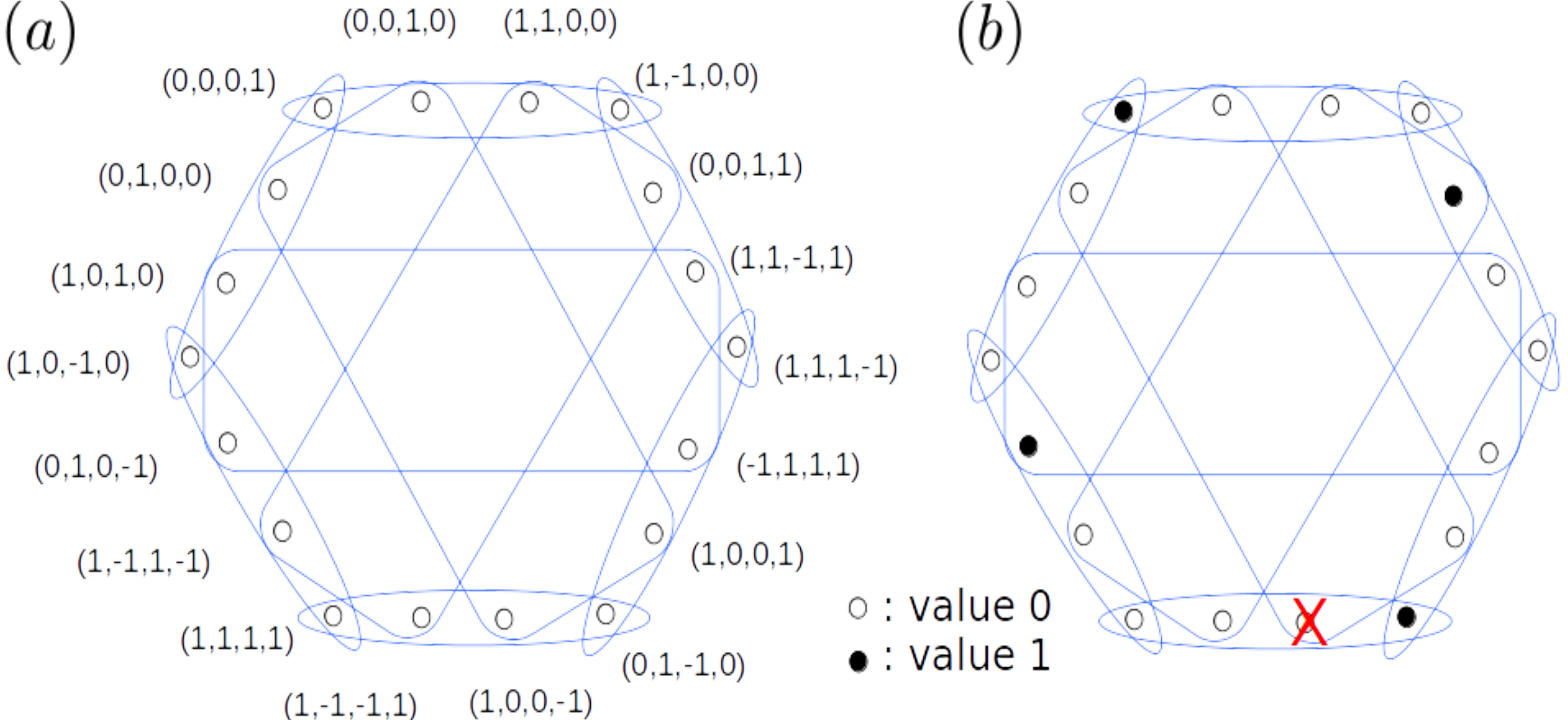}
 \caption{
Each of the $18$ rays is depicted by a node, and the $9$ orthonormal bases 
are depicted by $9$ edges, each 
enclosing $4$ nodes. There is no 
KS-noncontextual assignment 
to these nodes.
For instance, 
a noncontextual assignment of 0s and 1s to 17 of
the nodes cannot be completed to an assignment to all 18 
because neither 0 nor 1 can be assigned to the remaining node (marked X): 
one edge enclosing it requires it to be 0, the other 1.
}
\label{CEGAhypergraph}
\end{figure}

Can the possibility of a
KS-noncontextual ontological model be tested experimentally? One view is that it cannot, that
the KS theorem merely constrains 
the possibilities for {\em interpreting} the quantum formalism
\cite{MKC,merminquote}.
This answer, however, is inadequate.
One {\em can} and {\em should} ask: what is the minimal set of operational
predictions of quantum theory that need to be experimentally verified in order
to show that it does not admit of a noncontextual model?

We show that this minimal set is a far cry from 
the whole of quantum theory and
is therefore consistent with 
other possible operational theories such as the framework of generalized probabilistic theories \cite{hardy5axioms,barrettgpt}. As such,
our no-go result 
shows that none of these theories violating our noncontextuality inequality admits of a
noncontextual model.
Therefore, if corroborated by experiment, it implies that any future theory of physics that might replace quantum theory also
fails to admit of a noncontextual model.\footnote{The minimal necessary requirements such a theory should satisfy
in order to be able to decide whether it admits a
noncontextual ontological model or not is the existence of operationally equivalent experimental procedures 
(preparations and/or measurements) and a finite set of tomographically complete preparations and/or measurements to 
verify the operational equivalences in an experiment.}

Let us briefly recall some definitions: An {\em operational theory} is a triple
$(\mathcal{P},\mathcal{M},p)$ where $\mathcal{P}$ is a set of preparations,
$\mathcal{M}$ is a set of measurements, and $p$ specifies, for every pair of
preparation and measurement, the probability distribution over outcomes for that
measurement if it is implemented on that preparation.  
Denoting the set of outcomes of measurement $M$ by $\mathcal{K}_M$,
we have $\forall P \in \mathcal{P},\; \forall M\in \mathcal{M}$,
$p(\cdot|P,M):  \mathcal{K}_{M} \to [0,1]$.

An {\em ontological model} of an operational theory
$(\mathcal{P},\mathcal{M},p)$ is a triple $(\Lambda, \mu,\xi)$, where $\Lambda$
denotes a space of possible ontic states for the physical system,
$\mu$ specifies a probability distribution over the ontic
states for every preparation procedure, that is, $\forall P\in \mathcal{P}, \;
\mu(\cdot|P): \Lambda \rightarrow [0,1]$, such that
$\sum_{\lambda\in\Lambda}\mu(\lambda|P)=1$, and 
$\xi$ specifies, for every
measurement, the conditional probability of obtaining a given outcome if the
system is in a particular ontic state,  that is, $\forall M\in \mathcal{M},\;
\xi(k|M,\cdot):\Lambda \rightarrow [0,1]$, such that
$\sum_{k\in\mathcal{K}_M}\xi(k|M,\lambda)=1$.
The ontological model should 
reproduce the statistical predictions of
the operational theory:
\beq
p(k|P,M)=\sum_{\lambda\in\Lambda} \xi(k|M,\lambda)\mu(\lambda|P)
\label{empirical}
\eeq
for all $P \in \mathcal{P}$, and $M\in \mathcal{M}$.

We denote the event of obtaining outcome $k$ of measurement $M$ by $[k|M]$.
If $[k|M]$
is assigned a deterministic outcome by every ontic state
in the ontological model, i.e., if $\xi(k|M,\cdot):\Lambda \rightarrow
\{0,1\}$,
then it is said to be {\em outcome-deterministic} in that model.

\section{Upgrading the Kochen-Specker theorem}
We explain how to derive an experimental test of noncontextuality 
using a sequence of four refinements on the standard account of the KS theorem. In the first step, we identify the operational 
grounds that warrant applying an assumption of KS-noncontextuality to a measurement. In the second step, we explain why the 
assumption of outcome determinism for projective measurements -- which is part of KS-noncontextuality -- is unjustified, and
in the third step, we show how it can be justified for a perfectly predictable measurement given an assumption of noncontextuality
for preparations. In the fourth step, which is the central contribution of this chapter, we describe how to contend with the lack
of perfect predictability that is characteristic of any real experiment.

\subsection{Operationalizing the notion of KS-noncontextuality}
In a KS-noncontextual model of operational quantum theory,
the value (0 or 1) 
assigned to the event $[k|M]$ by 
$\lambda$ is the same as the value 
assigned to 
$[k'|M']$
whenever these two events 
correspond to the same ray of Hilbert space
(here, $M$ and $M'$ are 
assumed to be maximal projective measurements). 
We get to the crux of 
KS-noncontextuality, therefore, by
describing the {\em operational grounds} for associating the same ray to $[k|M]$
as is associated to $[k'|M']$.
Letting $\Pi_{k|M}$ and $\Pi_{k'|M'}$ represent the corresponding rank-1
projectors, the grounds for concluding that $\Pi_{k|M}=\Pi_{k'|M'}$
are that ${\rm tr}(\rho \Pi_{k|M})= {\rm tr}(\rho \Pi_{k'|M'})$ for an
appropriate set of density operators $\rho$.  It is clearly sufficient for the
equality to hold for the set of {\em all} density operators, but it is also
sufficient to have equality for certain smaller sets of density operators, namely, those
{\em complete for measurement tomography}.

What then should the operational grounds be for assigning the same value to
$[k|M]$ and $[k'|M']$ in a general operational theory, where preparations are
not represented by density operators? The answer, clearly, is that the event
$[k|M]$ {\em occurs with the same probability} as the event $[k'|M']$ for {\em
all} preparation procedures of the system,
\beq
p(k|M,P)=p(k'|M',P) {\rm\;\; for\; all\;\;  } P\in\mathcal{P},
\eeq
or equivalently, if this holds
for a subset of $\mathcal{P}$ that is tomographically complete.
In this case, we shall say that $[k|M]$ and $[k'|M']$ are {\em operationally
equivalent}, and denote this as $[k|M]$ $\simeq$ $[k'|M']$.  We can therefore
define a notion of KS-noncontextuality for any operational theory as follows: 
an ontological model $(\Lambda,\mu,\xi)$ of an operational theory
$(\mathcal{P},\mathcal{M},p)$ is KS-noncontextual if (i) operational equivalence
of events implies equivalent representations in the model, i.e., $[k|M]\simeq
[k'|M']\Rightarrow \xi(k|M,\lambda)=\xi(k'|M',\lambda)$  for all
$\lambda\in\Lambda$, and (ii) the model is outcome-deterministic,
$\xi(k|M,\cdot): \Lambda \to \{0,1\}.$

The operational equivalences 
among measurements 
required for the 
KS construction 
in Fig.~\ref{CEGAhypergraph}(a) are made
explicit in Fig.~\ref{mmtequivs}(a),
where 
every measurement event $[k|M]$ 
is represented 
by a distinct node, and a
novel type of edge between nodes 
specifies 
operational equivalence between events.
This affords a nice 
depiction of contextual value assignments,
as in Fig.~\ref{mmtequivs}(b).
It follows that 
{\em any} operational theory 
with nine four-outcome measurements 
satisfying the operational
equivalences
depicted in Fig.~\ref{mmtequivs}(a) fails to admit of a
KS-noncontextual model.
\begin{figure}
 \includegraphics[scale=0.58]{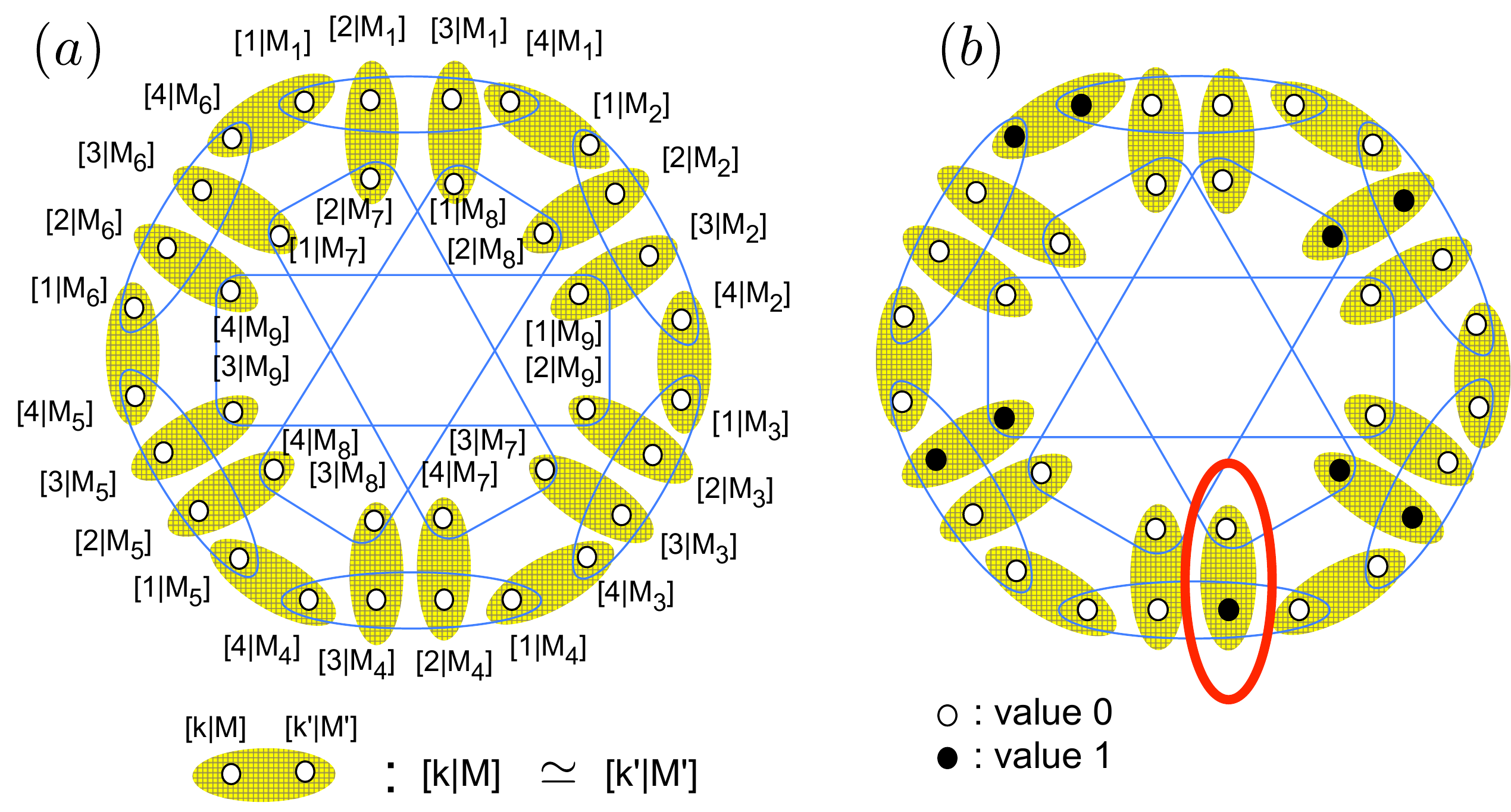}
 \caption{
(a) Nine four-outcome measurements,
each depicted by a set of four nodes encirled by a blue loop.
A yellow hashed region enclosing a set of nodes 
implies that the corresponding events are operationally equivalent. (b) 
An illustration of the fact that there is no outcome-deterministic noncontextual
assignment to the measurement events. 
The depicted outcome-deterministic assignment breaks the assumption of noncontextuality for the 
highlighted pair.
}
\label{mmtequivs}
\end{figure}

\subsection{Defining noncontextuality without outcome determinism}
The essence of 
noncontextuality is that context-independence at
the operational level should imply context-independence at the ontological
level. 
The operationalized version of KS-noncontextuality, however, makes an additional assumption about {\em what sort of thing} should be context-independent at the ontological level, namely, a 
deterministic assignment of an outcome.
However, one can equally well assume
that the ontic state merely assigns a probability distribution over outcomes,
and take {\em this distribution} to be the thing that is context-independent.  
In Ref.~\cite{genNC}, this 
revised notion of noncontextuality was
termed {\em measurement noncontextuality}:
\begin{quote}
\emph{Measurement noncontextuality} is satisfied by an ontological model
$(\Lambda,\mu,\xi)$ of an operational theory $(\mathcal{P},\mathcal{M},p)$ if
$[k|M]\simeq [k'|M']$ implies $\xi(k|M,\lambda)=\xi(k'|M',\lambda)$ for all
$\lambda\in \Lambda$.
\end{quote}
Here, 
$\xi(k|M,\cdot)\in[0,1]$,
so that outcome determinism is not assumed.

\subsection{Justifying outcome determinism for perfectly predictable measurements}
Outcome determinism can, however, be justified sometimes
if one
assumes a notion of noncontextuality for {\em preparations}~\cite{genNC}.
First, a definition: $P$ and $P'$ are said to be operationally equivalent,
denoted $P\simeq P'$, if for every measurement event $[k|M]$, $P$ assigns the
same probability to this event as $P'$ does, that is,
\beq
p(k|M,P)=p(k|M,P') {\rm\;\; for\; all\;\;  } k\in \mathcal{K}_M, {\rm\;\; for\;
all\;\;  } M\in\mathcal{M}.
\eeq
A preparation-noncontextual ontological model is then defined as
follows:
\begin{quote}
\emph{Preparation noncontextuality} is satisfied by an ontological model
$(\Lambda,\mu,\xi)$ of an operational theory $(\mathcal{P},\mathcal{M},p)$ if
$P\simeq P'$ implies $\mu(\lambda|P)=\mu(\lambda|P')$  for all
$\lambda\in\Lambda$.
\end{quote}
Insofar as both measurement and preparation noncontextuality 
are inferences from operational equivalence 
to ontological equivalence, it is most natural to
assume {\em both}, that is, to assume {\em universal noncontextuality}.

As we showed in Chapter 1 (based on Ref.~\cite{genNC}), in a preparation-noncontextual model of quantum theory,
all projective measurements must be represented outcome-deterministically. 
Here, we provide a version of this argument for the 18 ray construction.

Suppose that one has experimentally identified thirty-six preparation procedures
organized into nine ensembles of four each, $\{ P_{i,k}: i \in \{ 1,\dots, 9\},
k\in \{1,\dots,4\}\}$, such that for all $i$, measurement $M_i$ on preparation
$P_{i,k}$ yields the $k$th outcome with certainty, 
\beq
\forall i , \forall k : p(k|M_i, P_{i,k})= 1.
\label{eq:perfectpredictability}
\eeq
We call this property {\em perfect predictability}.
In quantum theory, it suffices to let $P_{i,k}$ be the preparation associated
with the pure state corresponding to the $k$th element of the $i$th measurement
basis. 

\begin{figure}
 \includegraphics[scale=0.7]{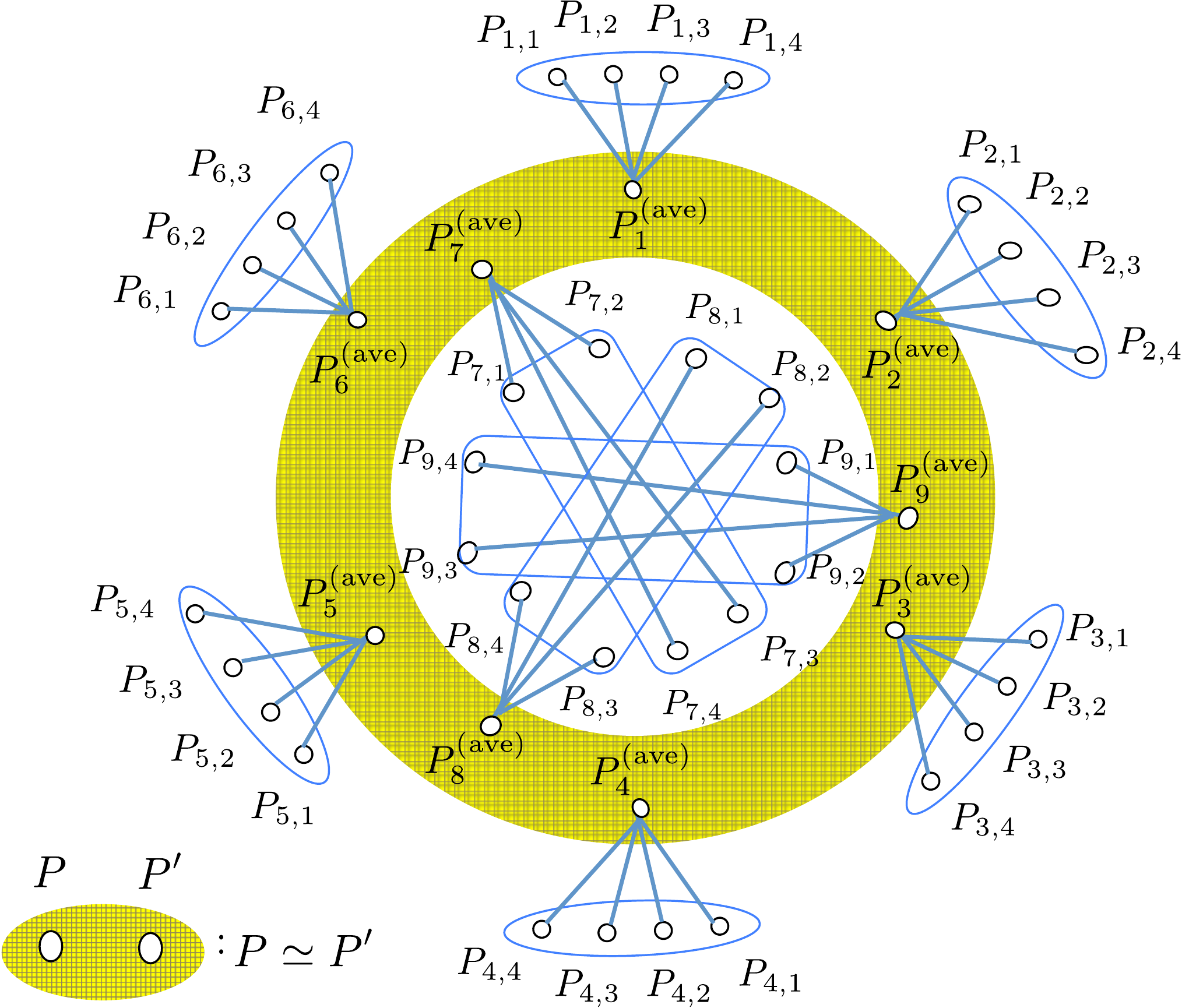}
 \caption{36 preparation procedures,
nine ensembles of four each.
A node 
connected to the elements of an
ensemble represents 
the effective preparation procedure 
achieved by sampling uniformly from
the ensemble. 
}
\label{prepequivs}
\end{figure}

Define the effective preparation $P_i^{\rm (ave)}$ as the procedure obtained by
sampling $k$ uniformly at random and then implementing $P_{i,k}$.
Suppose 
one has experimentally verified the operational 
equivalences (Fig.~\ref{prepequivs})
\beq
P_i^{\rm (ave)} \simeq P_{i'}^{\rm (ave)} {\rm \;\; for\; all \;}
i,i'\in\{1,\dots,9\}.
\label{eq:prepequivs}
\eeq
They hold 
in our quantum example because each
$P_i^{\rm (ave)}$
corresponds to the completely mixed state.

Given Eq.~\eqref{eq:prepequivs} and the assumption of preparation
noncontextuality, 
there is a single distribution over $\Lambda$,
denoted $\nu(\lambda)$, such that 
\beq
\mu(\lambda| P_i^{\rm (ave)})= \nu(\lambda)  {\rm \;\; for\; all \;} i
\in\{1,\dots,9\}.
\label{eq:nu}
\eeq
Given the definition of $P_i^{\rm (ave)}$, it follows that 
\beq
\frac{1}{4} \sum_k \mu(\lambda| P_{i,k})= \nu(\lambda)  {\rm \;\; for\; all \;}
i \in\{1,\dots,9\}.
\label{eq:nu2}
\eeq
Furthermore, recalling Eq.~\eqref{empirical}, for the ontological model to
reproduce Eq.~\eqref{eq:perfectpredictability}, we must have 
\beq
\forall i, \forall k: \sum_\lambda \xi(k|M_i,\lambda) \mu(\lambda| P_{i,k})= 1.
\label{eq:pp2}
\eeq
Because every $\lambda$ in the support of $\nu(\lambda)$ appears in the support
of $\mu(\lambda| P_{i,k})$ for some $k$, it follows that if $\xi(k|M_i,\lambda)$
had an indeterministic response
on any such $\lambda$, we would have a contradiction with Eq.~\eqref{eq:pp2}. 
Consequently, for all $i$ and $k$, 
the measurement event $[k|M_i]$ 
must be outcome-deterministic for all $\lambda$ in the support of
$\nu(\lambda)$.  

To summarize then, if one has experimentally verified the operational
equivalences depicted in Figs.~\ref{mmtequivs}(a) and \ref{prepequivs} and the
measurement statistics described in Eq.~\eqref{eq:perfectpredictability}, then
universal noncontextuality implies that the value assignments to measurement
events should be deterministic and noncontextual, hence KS-noncontextual, and we
obtain a contradiction\footnote{Here and elsewhere in this thesis, ``$\bigwedge$'' denotes a logical conjunction or ``AND''.}:
\beqa
&\textrm{universal  noncontextuality} \bigwedge \textrm{operational
equivalences}\nonumber\\
&\bigwedge 
\textrm{perfect predictability}
\Rightarrow \textrm{contradiction}.
\label{inference}
\eeqa

\subsection{Contending with imperfect predictability in real experiments}
In real experiments, the ideal of 
perfect predictability
described by
Eq.~\eqref{eq:perfectpredictability} is never achieved, so we cannot derive a
contradiction from it. However, Eq.~\eqref{inference} is logically equivalent
to the following inference:
\beqa
&\textrm{universal  noncontextuality} \bigwedge \textrm{operational
equivalences}\nonumber\\
&\Rightarrow
\textrm{failure of perfect predictability}.
\label{inference2}
\eeqa
That is, the degree of predictability,
averaged over all $i$ and $k$, will
necessarily be bounded away from 1.
It is this bound that defines
the operational noncontextuality inequality.
For the 18 ray example, we prove that
\begin{align}
A \equiv  \frac{1}{36} \sum_{i=1}^{9} \sum_{k=1}^{4} p(k|M_{i},P_{i,k}) \le
\frac{5}{6}.
\label{main}
\end{align}
We now outline how the bound in Eq.~\eqref{main} is obtained. First, we use
Eq.~\eqref{empirical} to express $A$ in terms of $\xi(k|M_{i},\lambda)$ and
$\mu(\lambda|P_{i,k})$.  
Defining the {\em max-predictability} of a measurement $M$ given an ontic state
$\lambda$ by
 \beq
\zeta(M,\lambda)\equiv \max_{k'\in\mathcal{K}_M} \xi(k'|M,\lambda),
\label{eq:eta}
\eeq
we deduce that
\begin{eqnarray}
A &\le&\sum_{\lambda} \left( \frac{1}{9} \sum_{i}\zeta(M_{i},\lambda)
\left[ \frac{1}{4} \sum_{k}  \mu(\lambda|P_{i,k}) \right] \right)\nonumber \\
&=&  \sum_{\lambda} \left( \frac{1}{9} \sum_{i}\zeta(M_{i},\lambda) \right)
\nu(\lambda)\nonumber \\
&\le& \max_{\lambda} \left( \frac{1}{9} \sum_{i} \zeta(M_{i},\lambda)  \right),
\end{eqnarray}
where we have used Eq.~\eqref{eq:nu2}.

The measurements can
have indeterministic responses,
$\xi(k|M,\cdot):\Lambda \rightarrow [0,1]$, but measurement noncontextuality
implies that
$\xi(k|M_i,\lambda)=\xi(k'|M_{i'},\lambda)$ for the
operationally equivalent pairs $\big\{[k|M_i], [k'|M_{i'}]\big\}$.
There are many such assignments.  Every unit-trace positive operator, for
instance, specifies an indeterministic noncontextual assignment via the Born
rule, and there are other, nonquantum\footnote{In the sense that it is impossible to associate projectors to the nodes of the hypergraph
such that a quantum state leads to the depicted probability assignments via Born rule. This can be seen by noting that 
the assignments in Fig.~\ref{highpredictability} would require three pairwise orthogonal projectors such that each of them is 
assigned a value $1/2$ by a quantum state. Since three pairwise orthogonal projectors add up to no more than the identity, there is no way
a quantum state can assign probabilities to them that add up to more than $1$, such as the $3/2$ required by the given assignment.} assignments as well, such as the one
depicted in Fig.~\ref{highpredictability}.

\begin{figure}
\includegraphics[scale=0.5]{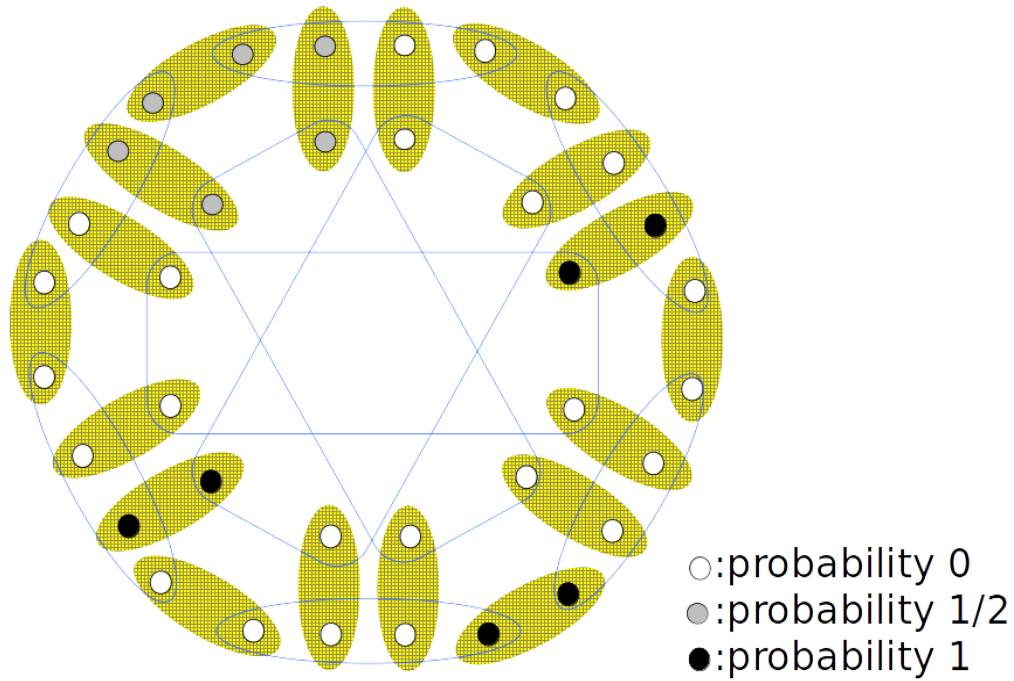}
\caption{A noncontextual outcome-indeterministic assignment}
\label{highpredictability}
\end{figure}

Consider the average max-predictability achieved by the assignment of
Fig~\ref{highpredictability}.  Here, six measurements have max-predictability
of 1, three of $\frac{1}{2}$, implying that
$\frac{1}{9} \sum_{i} \zeta(M_{i},\lambda) = \frac{1}{9} ( 6 \cdot 1 + 3 \cdot
\frac{1}{2}) = \frac{5}{6}$.
As we demonstrate in the formal proof of our noncontextuality inequality in the next section, no ontic state can do better,
so that
$\max_{\lambda} \left( \frac{1}{9} \sum_{i} \zeta(M_{i},\lambda)  \right) \le
\frac{5}{6}$, thereby establishing the noncontextual bound on $A$.
The logical limit for the value of $A$ is $1$, so the noncontextual bound of
$\frac{5}{6}$ is nontrivial. The quantum realization of the 18 ray
construction
achieves $A=1$.

Following our formal proof of the inequality in the next section, we discuss the noise tolerance of our inequality, 
and we criticise a previous proposal for a noncontextuality inequality \cite{Cabelloexpt} on two main grounds: 
(i) that logic alone rules out the possiblity of satisfying it, and (ii) that all operational theories supporting the measurement 
equivalences of Fig.~\ref{mmtequivs}(a) necessarily violate it, regardless of whether or not they admit of a noncontextual model.

Although we have used the  proof of Ref.~\cite{Cab1} to illustrate our approach to deriving robust 
noncontextuality inequalities, our scheme 
can turn any proof of the Kochen-Specker theorem based on a KS-uncolourable set into an experimentally testable
inequality, as we will show in the following section.


\section{Proof of the inequality}

Our noncontextuality inequality inspired by the proof of the Kochen-Specker theorem for the 18 ray KS-uncolourable set of 
Figure \ref{CEGAhypergraph} can be summmarized by the following theorem:
\begin{theorem}
Consider an operational theory $(\mathcal{P},\mathcal{M},p)$.  Let $\{ M_i \in \mathcal{M} :  i \in\{ 1,\dots,9\}\}$ be nine four-outcome measurements.
Let $[k|M_{i}]$ denote the $k$th outcome of the $i$th measurement, where $k \in \{1,\dots,4\}$. 
Let  $\{ P_{i,k} \in \mathcal{P}: i \in\{ 1,\dots,9\}$, $ k \in \{1,2,3,4\}\}$ be thirty-six preparation procedures, organized into nine sets of four. Let $P^{\rm (ave)}_i\in \mathcal{P}$ be 
the preparation procedure obtained by sampling $k \in \{1,2,3,4\}$ uniformly at random and implementing $P_{i,k}$.

Suppose that one has experimentally verified the operational preparation equivalences depicted in Fig.~\ref{prepequivs}, namely,
 \begin{align}\label{eq:optlequivP1to9}
 P^{\rm (ave)}_1 \simeq P^{\rm (ave)}_2 \simeq \dots \simeq P^{\rm (ave)}_9,
\end{align} 
and the operational equivalences depicted in Fig.~\ref{mmtequivs}(a), namely,
 \begin{align}\label{eq:optlequivM18ray}
[k| M_{i}] \simeq [k'|M_{i'}],
\end{align} 
for the eighteen pairs specifed therein. 

If one assumes that the operational theory admits of a universally noncontextual ontological model,
that is, one which is both measurement-noncontextual and preparation-noncontextual, then the following inequality on operational probabilities holds:
\begin{align}
A \equiv  \frac{1}{36} \sum_{i=1}^{9} \sum_{k=1}^{4} p(k|M_{i},P_{i,k}) \le \frac{5}{6}.
\end{align}
\end{theorem}

\begin{figure}
 \includegraphics[scale=0.8]{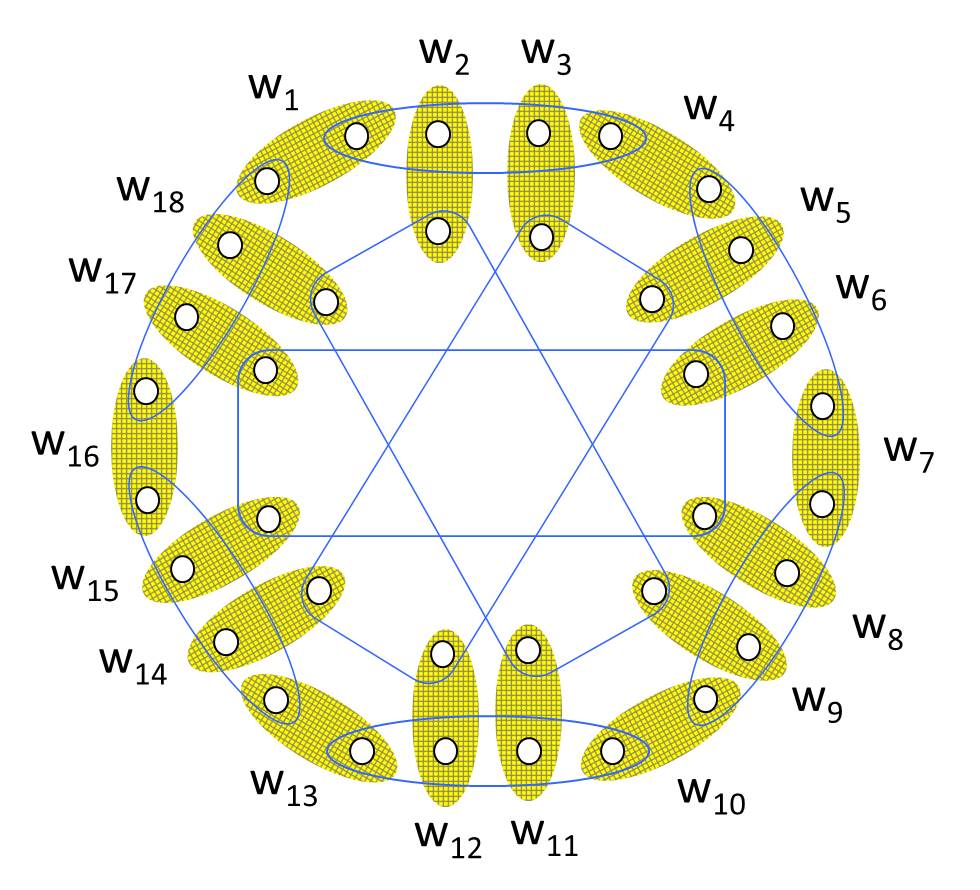}
 \caption{
 A choice of labelling of the eighteen equivalence classes of measurement events.  Here, $w_{\kappa}$ denotes the probability assigned to the equivalence class labelled by $\kappa$ in a 
 noncontextual outcome-indeterministic ontological model.}
\label{legend}
\end{figure}

We now provide the proof.  For clarity, we expand on some of the steps presented in previous section in our proof sketch. 

The quantity $A$ can be expressed in terms of the distributions and response functions of the ontological model, using Eq.~\eqref{empirical}, as
\begin{align}
A = \frac{1}{36}\sum_{i=1}^{9}  \sum_{k=1}^{4} \sum_{\lambda}\xi(k|M_{i},\lambda)\mu(\lambda|P_{i,k}).
\end{align}
Using the definition of the max-probability $\zeta(M_{i},\lambda)$, given in Eq.~\eqref{eq:eta}, we have
\begin{align}
A \le \frac{1}{9}\sum_{i=1}^{9}  \sum_{\lambda}\zeta(M_{i},\lambda) \left( \frac{1}{4}\sum_{k=1}^{4} \mu(\lambda|P_{i,k}) \right).
\end{align}

Assuming that one experimentally verifies the operational preparation equivalences of Eq.~\eqref{eq:optlequivP1to9}, the assumption of preparation noncontextuality implies that
\beq
\mu(\lambda| P_1^{\rm (ave)})= \mu(\lambda| P_{2}^{\rm (ave)}) = \cdots = \mu(\lambda| P_{9}^{\rm (ave)}).
\eeq
It follows that there exists a single distribution, which we denote $\nu(\lambda)$, such that
\beq
\mu(\lambda| P_i^{\rm (ave)})= \nu(\lambda)  {\rm \;\; for\; all \;} i \in\{1,\dots,9\}.
\label{muPave}
\eeq
Recall that $P_i^{\rm (ave)}$ is the preparation procedure that samples $k$ uniformly from $\{1,2,3,4\}$ and implements $P_{i,k}$.  Given that the probability of the system being in a given 
ontic state $\lambda$ given the preparation $P_{i,k}$ is $\mu(\lambda|P_{i,k})$, and given that the probability of $P_{i,k}$ being implemented is $\frac{1}{4}$ for each value of $k$, it follows
that the probability of the system being in a given ontic state $\lambda$ given the preparation $P_{i}^{\rm (ave)}$ is $\mu(\lambda|P_i^{\rm (ave)}) = \frac{1}{4} \sum_{\lambda} \mu(\lambda|P_{i,k})$.
Combining this with Eq.~\eqref{muPave}, we conclude that
\beq
\frac{1}{4} \sum_{\lambda} \mu(\lambda|P_{i,k}) = \nu(\lambda)  {\rm \;\; for\; all \;} i \in\{1,\dots,9\},
\eeq
and therefore that
\begin{align}
A \le \frac{1}{9} \sum_{\lambda}\sum_{i=1}^{9} \zeta(M_{i},\lambda)\nu(\lambda).
\end{align}
This in turn implies
\begin{align}\label{eq:aa}
A \le \max_{\lambda}\frac{1}{9} \sum_{i=1}^{9} \zeta(M_{i},\lambda).
\end{align}

Assuming that one experimentally verifies the operational measurement equivalences of Eq.~\eqref{eq:optlequivM18ray}, the assumption of measurement noncontextuality implies that
\begin{align}
\xi(k|M_i,\lambda)=\xi(k'|M_{i'},\lambda),
\end{align}
for the eighteen pairs of operationally equivalent measurement events $([k|M_i], [k'|M_{i'}])$
specifed in Fig.~\ref{mmtequivs}(a). 

It is useful to simplify the notation at this stage.  We introduce the variable $\kappa \in \{1,\dots, 18\}$ to range over
the eighteen operational equivalence classes of measurement events.
We introduce the shorthand notation
\begin{equation}
w_{\kappa}\equiv \xi(k|M_i,\lambda)=\xi(k'|M_{i'},\lambda),
\end{equation}
for the probability assigned to the $\kappa$th equivalence class, where the dependence on $\lambda$ is left implicit. The variable $\kappa$ enumerates the equivalence classes 
in Fig.~\ref{mmtequivs}(a), starting from 
$[1|M_1]$ and proceeding clockwise around the hypergraph, as depicted in Fig.~\ref{legend}.

In this notation, the constraint that each response function is probability-valued, $\xi(k|M_i.\lambda) \in [0,1]$, is simply
\beq
0 \le w_{\kappa} \le 1, \;\;\forall \kappa \in\{1,\dots,18\},
\label{eq:positivity}
\eeq
while the constraint that the set of response functions for each measurement sum to 1, $\sum_{k=1}^4 \xi(k|M_i,\lambda)=1$, can be captured by the matrix equality  
\beq\label{eq:bb}
Z \vec{w} =  \vec{u}
\eeq
where $\vec{w}\equiv (w_1,\dots,w_{18})^T$, $\vec{u}\equiv (1,1,1,1,1,1,1,1,1)^T$, and
\beq
Z \equiv \left(
\begin{array}{cccccccccccccccccc}
1&1&1&1&0&0&0&0&0&0&0&0&0&0&0&0&0&0\\
0&0&0&1&1&1&1&0&0&0&0&0&0&0&0&0&0&0\\
0&0&0&0&0&0&1&1&1&1&0&0&0&0&0&0&0&0\\
0&0&0&0&0&0&0&0&0&1&1&1&1&0&0&0&0&0\\
0&0&0&0&0&0&0&0&0&0&0&0&1&1&1&1&0&0\\
1&0&0&0&0&0&0&0&0&0&0&0&0&0&0&1&1&1\\
0&1&0&0&0&0&0&0&1&0&1&0&0&0&0&0&0&1\\
0&0&1&0&1&0&0&0&0&0&0&1&0&1&0&0&0&0\\
0&0&0&0&0&1&0&1&0&0&0&0&0&0&1&0&1&0
\end{array}
\right).
\eeq
Finally, we can express the quantity to be maximized as
\beq
\frac{1}{9} \sum_{i=1}^{9} \zeta(M_{i},\lambda) = \frac{1}{9} \sum_{i=1}^{9}  \max_{\kappa : Z_{i\kappa}=1} w_{\kappa},
\label{eq:cc}
\eeq
or, more explicitly, as
\begin{align}
&\frac{1}{9}\sum_{i=1}^{9} \zeta(M_{i},\lambda)\nonumber\\
&= \frac{1}{9} [ \max\{w_1,w_2,w_3,w_4\}
+ \max\{w_4,w_5,w_6,w_7\}\nonumber\\
&+\max\{w_7,w_8,w_9,w_{10}\}
+ \max\{w_{10},w_{11},w_{12},w_{13}\}\nonumber\\
&+ \max\{w_{13},w_{14},w_{15},w_{16}\}
+\max\{w_{16},w_{17},w_{18},w_1\}\nonumber\\
&+ \max\{w_{18},w_2,w_9,w_{11}\}
+ \max\{w_3,w_5,w_{12},w_{14}\}\nonumber\\
&+\max\{w_6,w_8,w_{15},w_{17}\} ].
\end{align}
The matrix equality of Eq.~\eqref{eq:bb} implies that there are only nine independent variables in the set $\{ w_1, w_2,\dots, w_{18} \}$ and that these satisfy linear inequalities.
The space of possibilities for the vector $\vec{w}$ therefore forms a nine-dimensional polytope in the hypercube described by Eq.~\eqref{eq:aa}.

The value of $\frac{1}{9} \sum_{i=1}^{9} \zeta(M_{i},\lambda)$ 
on any of the interior points of this polytope will be an average of its values at the vertices because it is a convex function of $\vec{w}$.  Therefore, to implement the maximization over 
$\lambda$, it suffices to maximize over the vertices of this polytope.

Using the numerical software Sage, in particular the Polyhedron class in SageMathCloud\cite{sagemath}, we were able to infer all 146 vertices of our 9-dimensional polytope from its 
characterization in terms of the linear inequalities obtained from Eqs.~\eqref{eq:positivity} and \eqref{eq:bb}. 
From this brute-force enumeration of all the vertices of the polytope, the maximum possible value of $\frac{1}{9} \sum_{i=1}^{9} \zeta(M_{i},\lambda)$ was found to be $\frac{5}{6}$. 
An example of a vertex achieving this value is $\vec{w}=(1,0,0,0,1,0,0,0,\tfrac{1}{2},\tfrac{1}{2},\tfrac{1}{2},0,0,0,1,0,0,0)^{\rm T}$, which is depicted in Fig.~\ref{highpredictability}. This concludes the proof.\proofend

Our proof technique can be adapted to derive a similar noncontextuality inequality correponding to any 
proof of the KS theorem based on the KS-uncolourability of a set of rays of Hilbert space. One begins by completing every set of orthogonal rays into a basis of the Hilbert space, and then 
forming the hypergraph depicting the orthogonality relations among these rays (the analogue of Fig.~\ref{CEGAhypergraph}).
One then forms the hypergraph depicting all of the measurements events, with one type of edge denoting which events correspond to the outcomes of a single measurement, and the other type of 
edge denoting when a set of measurement events are operationally equivalent (the analogue of Fig.~\ref{mmtequivs}(a)).
One then associates a set of preparations with every measurement in the hypergraph, one preparation for every outcome.  For each such set of preparations, we define the effective preparation 
that is the uniform mixture of the set's elements, and we presume that all of the effective preparations so defined are operationally equivalent (as is the case in quantum theory, where the 
effective 
preparation for every set corresponds to the completely mixed state). We consider the correlation between the measurement outcome and 
the choice of preparation in the set associated with that measurement, averaged over all measurements.  This average correlation is the quantity $A$ that appears on the left-hand side of the 
operational noncontextuality inequality.  

The KS-uncolourability of the hypergraph means that there are no noncontextual deterministic assignments
to the measurement events, hence the polytope of probabilistic assignments to the measurement events has no deterministic vertices either. Each vertex of this polytope, that is, each 
convexly-extremal probabilistic assignment, will necessarily yield an indeterministic assignment to some of the measurement events.  
Using the operational equivalences and the assumption of universal noncontextuality, one can infer from this 
that the average correlation $A$ is always bounded away from 1. 
For any KS-uncolourable hypergraph, a quantum realization would achieve the logical limit $A=1$ by construction, so the noncontextuality inequality we derive is necessarily violated by quantum 
theory in each case. The exact upper bound on $A$ will depend on the vertices of the polytope of measurement noncontextual
probability assignments possible on the KS-uncolourable hypergraph and will therefore vary from case to case (but will always be less than 1).

One can understand this violation as being due to the fact that assignments of density operators that are independent of the  preparation context can achieve higher predictability for the 
respective measurements than assignments of probability distributions over ontic states that are independent of the preparation context. This is the feature of quantum theory that allows it 
to maximally violate the 
noncontextual bound of $A\leq 5/6$.

\section{Vertices of the polytope}

In this section, we describe the vertices of the polytope of possible probabilistic assignments to the $18$ equivalence
classes of measurement events in the hypergraph of Fig.~\ref{legend}. These vertices correspond to the 
convexly-extremal probabilistic assignments. Since we use these vertices in proving our noncontextuality inequality, we 
discuss some of their characteristics below.

\subsection{Odd n-cycles in the hypergraph}
We begin by noting a 
property of the hypergraph of Fig.~\ref{legend}, namely, the presence
of odd $n$-cycles. An odd $n$-cycle of equivalence classes of measurement events is an ordered sequence of $n$ such classes wherein
adjacent elements in the sequence contain measurement events that are distinct outcomes of a single measurement. For instance, the
sequence of equivalence classes $(1,2,18)$ in Fig.~\ref{legend} forms a $3$-cycle because for the first adjacent pair in the 
sequence, $(1,2)$, there is a node within the class 1 and a node within the class 2 that appear together in the same edge, and 
similarly for the two other adjacent pairs, $(2,18)$, and $(18,1)$.  
 
We note the presence of the following odd
$n$-cycles in Fig.~\ref{legend} which will be of interest further on:
\begin{enumerate}
 \item $3$-cycle: $(1,2,18)$
 \item $5$-cycle: $(8,9,11,13,15)$
 \item $7$-cycle: $(1,3,5,7,10,11,18)$
 \item $9$-cycle: $(1,4,6,15,14,12,10,9,18)$
\end{enumerate}

\subsection{Quantum probabilistic assignments: the projective and the nonprojective cases}
A given probabilistic assignment $\{w_{\kappa}\}_{\kappa=1}^{18}$ is quantum-realizable 
if one can associate an effect $E_{\kappa}$---that is, a positive operator less than identity, $0\le E_{\kappa}\le I$---to each node $\kappa$,
such that each of the edges of the hypergraph correspond to positive operator-valued measures, 
$\sum_{\kappa \in {\rm edge}} E_{\kappa} = I$,  and one can find a unit-trace positive operator $\rho$ such that
$w_{\kappa} = {\rm tr} (E_{\kappa} \rho)$.  In other words a given probabilistic assignment to the measurements is 
quantum-realizable if there is a set of quantum measurements and a quantum state that yield this probabilistic assignment
via the Born rule. Such quantum realizability of probabilistic assignments when the POVMs are restricted to projective measurements
is the question of interest in the traditional Kochen-Specker type approach to contextuality, most recently exemplified in 
the work of Refs.~\cite{CSW} and \cite{AFLS}. In our approach, this question is not the one of interest, particularly because
even if we restrict ourselves to quantum theory\footnote{Which, of course, we do not {\em a priori} do: our treatment of contextuality 
does not presume the operational theory of interest is strictly quantum theory. It could be any generalized probabilistic
theory (GPT) \cite{barrettgpt}.}, we want to be able to deal with all POVMs -- projective or nonprojective -- on
an equal footing without having to artificially constrain ourselves to just projective measurements.
However, we will still study this question 
so that the distinction between the traditional approach \cite{AFLS,CSW} and our approach becomes clearer.

As we will show, none of the vertices of the polytope of probabilistic assignments (Fig.~\ref{legend}) 
on the KS-uncolourable hypergraph of 
Fig.~\ref{CEGAhypergraph}(a)
is quantum realizable via projectors. That is, considering every assignment
of projectors to nodes of the hypergraph (such that projectors in an edge of the hypergraph sum to identity) and every quantum state,
the probabilistic assignment that results is never a vertex of the polytope. Consider, for example, the extremal probabilistic 
assignment to the equivalence classes in Fig.~\ref{highpredictability}: this requires the assignment $w_1=w_2=w_{18}=1/2$ to the 
$3$-cycle $(1,2,18)$. If the experiment is modelled by the quantum formalism, then because joint measurability is represented by 
orthogonality of projectors, any odd $n$-cycle must be represented by a sequence of projectors that are orthogonal for contiguous 
pairs. Any odd $n$-cycle with such assignments ($w_{\kappa}=1/2$ for all $\kappa$ in the odd $n$-cycle) does not admit a realization with 
projectors as the following lemma shows:

\begin{lemma}\label{noquantumlemma}
Given a set of $n$ projectors (where $n$ is odd), $\{\Pi_1,\Pi_2,\dots,\Pi_n\}$, satisfying orthogonality between adjacent pairs, 
$\Pi_i \Pi_{i\oplus 1}=0$ for all $i\in\{1,2,\dots,n\}$
(we call this set of projectors an odd $n$-cycle),
there exists no quantum state $\rho$, $\rho\geq 0$ and ${\rm Tr} \rho=1$, such that
\begin{equation}\label{lemmaeqn}
\forall i\in\{1,2,\dots,n\}: \Tr(\rho\Pi_i)=\frac{1}{2}.
\end{equation}
\end{lemma}
{\bf Proof.} This follows from noting that Eq.~(\ref{lemmaeqn}) implies $\forall i\in\{1,2,\dots,n\}$:
$\Tr \left(\rho(\Pi_i+\Pi_{i\oplus1})\right)=1$. 
Since $\Pi_i$ and $\Pi_{i\oplus1}$ are orthogonal, we have $\Pi_i+\Pi_{i\oplus1}\leq I$ for all $i$.
We also have $\Pi_i+\Pi_i^{\perp}=I$ for all $i$, where $\Pi_i^{\perp}$ is the projector onto the subspace orthogonal to
$\Pi_i$. Hence, we must have $\Pi_{i\oplus1}\leq \Pi_i^{\perp}$, but since we are given $\Tr(\rho\Pi_{i\oplus1})=\frac{1}{2}$
and can infer $\Tr(\rho\Pi_i^{\perp})=\frac{1}{2}$ (because $\Tr(\rho\Pi_i)=\frac{1}{2}$), we in fact have 
$\Pi_{i\oplus1}=\Pi_i^{\perp}$ for all $i$. When $n$ is odd, this means that the list of projectors reads as 
$\{\Pi_1=\Pi_n^{\perp},\Pi_2=\Pi_1^{\perp},\Pi_3=\Pi_2^{\perp}=\Pi_1,\Pi_4=\Pi_1^{\perp},\dots,\Pi_n=\Pi_1\}$,
leading to the contradiction $\Pi_1=\Pi_1^{\perp}$ (impossible because $\Pi_1$ is a projector). Hence such a valuation as in 
Eq.~\eqref{lemmaeqn} is impossible in quantum theory with projectors for any quantum state $\rho$.\footnote{Note that 
when $n$ is even, it is possible to have the list of projectors reading 
$\{\Pi_1=\Pi_n^{\perp},\Pi_2=\Pi_1^{\perp},\Pi_3=\Pi_2^{\perp}=\Pi_1,\Pi_4=\Pi_1^{\perp},\dots,\Pi_n=\Pi_1^{\perp}\}$
which does not lead to a contradiction. That even $n$ cannot lead to a contradiction is clear from a qubit example:
simply take $\Pi_1=|0\rangle\langle0|$ and $\Pi_1^{\perp}=|1\rangle\langle1|$ in the list of projectors with 
$\rho=I/2$.}\proofend

It follows from this lemma that $3$-cycle correlations of this form are not realizable via quantum projectors. These correspond
to the indeterministic extremal points of the polytope of correlations in Specker's scenario we discussed in Chapter 5 (see also Ref.~\cite{finegen}), 
in particular the ``overprotective seer'' correlations
of Ref.~\cite{LSW}. 

More generally, we can use Lemma \ref{noquantumlemma} to argue that none of the vertices admit a quantum realization with projectors.
This is because every vertex involves an odd $n$-cycle with probabilities $1/2$ assigned to nodes in that $n$-cycle.
On the other hand, any vertex with assignments $\{w_{\kappa}\}_{\kappa=1}^{18}$ can always be realized quantumly by assigning trivial effects of the form $w_{\kappa}I$ to each node labelled by $\kappa\in\{1,2,\dots,18\}$, where
$I$ denotes the identity operator. However, such trivial quantum realizations are indeed \emph{trivial}:
they do not admit contextuality in our approach because they are too noisy to ever achieve $A>5/6$, as we will show in the 
next section on the noise robustness of our inequality. 
In fact, something stronger can be said:
the trivial effect assignments, $\{w_{\kappa}I\}_{\kappa=1}^{18}$,
corresponding to \emph{any} point belonging to the polytope (including the vertices of the polytope)
achieve $A=1/4$, which is way below the noncontextual upper bound of $5/6$. This is easy to see: the statistics of $w_{\kappa}I$,
given the Born rule, is independent of what preparation (density operator) precedes such a ``measurement'' outcome, so we have
$p(k|M_i,P_{i,k})=\Tr (\rho_{i,k} w_{\kappa}I)=w_{\kappa}$, independent of the density operator $\rho_{i,k}$ that may be associated
with preparation $P_{i,k}$.
A trivial quantum realization, therefore, achieves $\sum_{k=1}^4p(k|M_i,P_{i,k})=1$ for each $i\in\{1,2,\dots,9\}$, and we have $A=\frac{1}{36}(9)=1/4$.
(Note that for a projective quantum realization we can have $\sum_{k=1}^4p(k|M_i,P_{i,k})=4$ for each $i$ and therefore $A=1$.)

Hence, \emph{all} such trivial quantum realizations (and not only those of the vertices) are on the same footing so far as our noncontextuality inequality is concerned.
The operational significance of this should be clear: in order to even approach the noncontextual upper bound of our inequality, 
one needs to have some projective component to the quantum measurements; they can't all be so noisy that they are rendered trivial.

{\em Thus our noncontextuality inequality clarifies the precise sense in which such assignments are trivial, something
which is never clear in the traditional Kochen-Specker approach to noncontextuality where the vertices of the polytope could be
naively deemed ``maximally (KS)contextual'', and the fact that they admit trivial quantum realizations of this sort 
is ignored by restricting attention to sharp (projective) measurements (see, for example, Refs.~\cite{CSW, AFLS}).}

\subsection{Classification of the vertices}
Classified in terms of their average predictability, $\frac{1}{9}\sum_{i=1}^9\zeta(M_i,\lambda)$, the $146$ vertices
of the $9$-dimensional polytope fall into four types, depicted in Table \ref{vertexdescription}.
Each row of the table corresponds to a particular type. The first column labels the type of vertex; the second column notes 
the average predictability of vertices of this type; the third column describes an example vertex of the given type; the fourth
column specifies, for the given example, the particular $n$-cycle that makes it impossible on account of Lemma \ref{noquantumlemma}
to build a (projective) quantum realization of that vertex, and notes the value of $n$ 
which makes a projective quantum realization impossible for all vertices of that type;
the last column notes the number of instances of the given type of vertex in the set of 146 vertices.
See Tables \ref{vertices1}, \ref{vertices2}, \ref{vertices3}, and \ref{vertices4} for the complete list of
all the four types of vertices.

\begin{table}[t]
\begin{center}
     \resizebox{\textwidth}{!}{\begin{tabular}{ | c || c | c | c !{\vrule width 1pt} c | c |}
     \hline
     Type of vertex&$\frac{1}{9}\sum_{i=1}^9\zeta(M_i,\lambda)$& Example vertex
     & $n$-cycle excluding projective realization (Lemma \ref{noquantumlemma}) & Number of vertices\\ \hline\hline
     $1$ & $\frac{1}{9}(6+3\frac{1}{2})=5/6$ & See figure \ref{type1} & $(1,2,18)$, $n=3$ & 24\\\hline
      $2$ & $\frac{1}{9}(4+5\frac{1}{2})=13/18$ & See figure \ref{type2}
      & $(8,9,11,13,15)$, $n=5$ & $36$ \\ \hline     

     $3$ & $\frac{1}{9}(2+7\frac{1}{2})=11/18$ & See figure \ref{type3}
     & $(1,3,5,7,10,11,18)$, $n=7$ & $36$ \\ \hline
     $4$ & $\frac{1}{9}(0+9\frac{1}{2})=1/2$ & See figure \ref{type4}
     & $(1,4,6,15,14,12,10,9,18)$, $n=9$ & $50$ \\ \hline
     \end{tabular}}
 \end{center}
 \caption{Classification of the $146$ vertices of the $9$-dimensional polytope of probabilistic assignments to nodes of the hypergraph of Fig.~\ref{legend}. 
}
 \label{vertexdescription}
 \end{table}
 
\begin{figure}[htb!]
\centering
\includegraphics[scale=0.4]{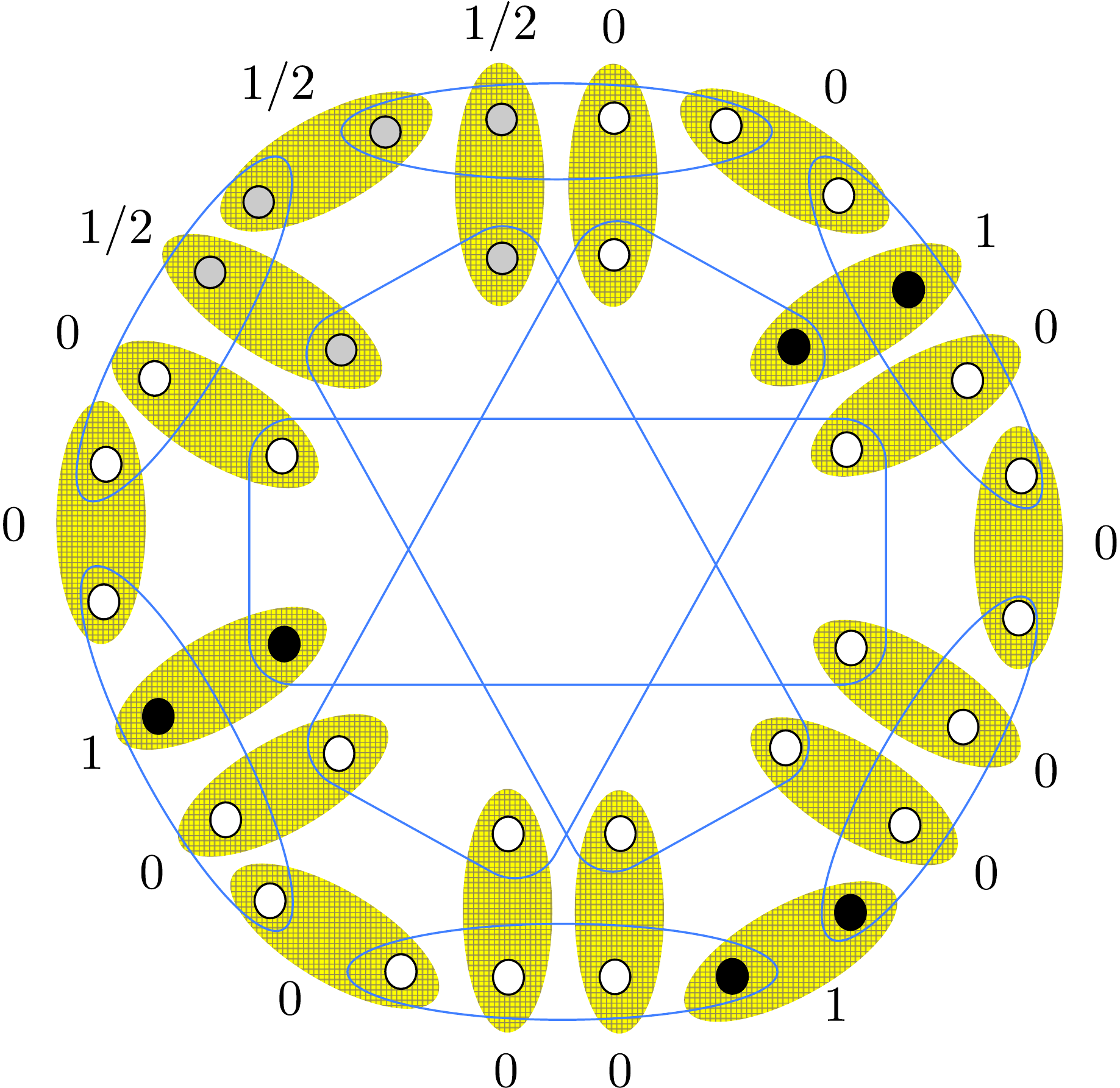}
 \caption{An example of a vertex of type $1$.}
\label{type1}
\end{figure}
\begin{figure}[htb!]
\centering
\includegraphics[scale=0.4]{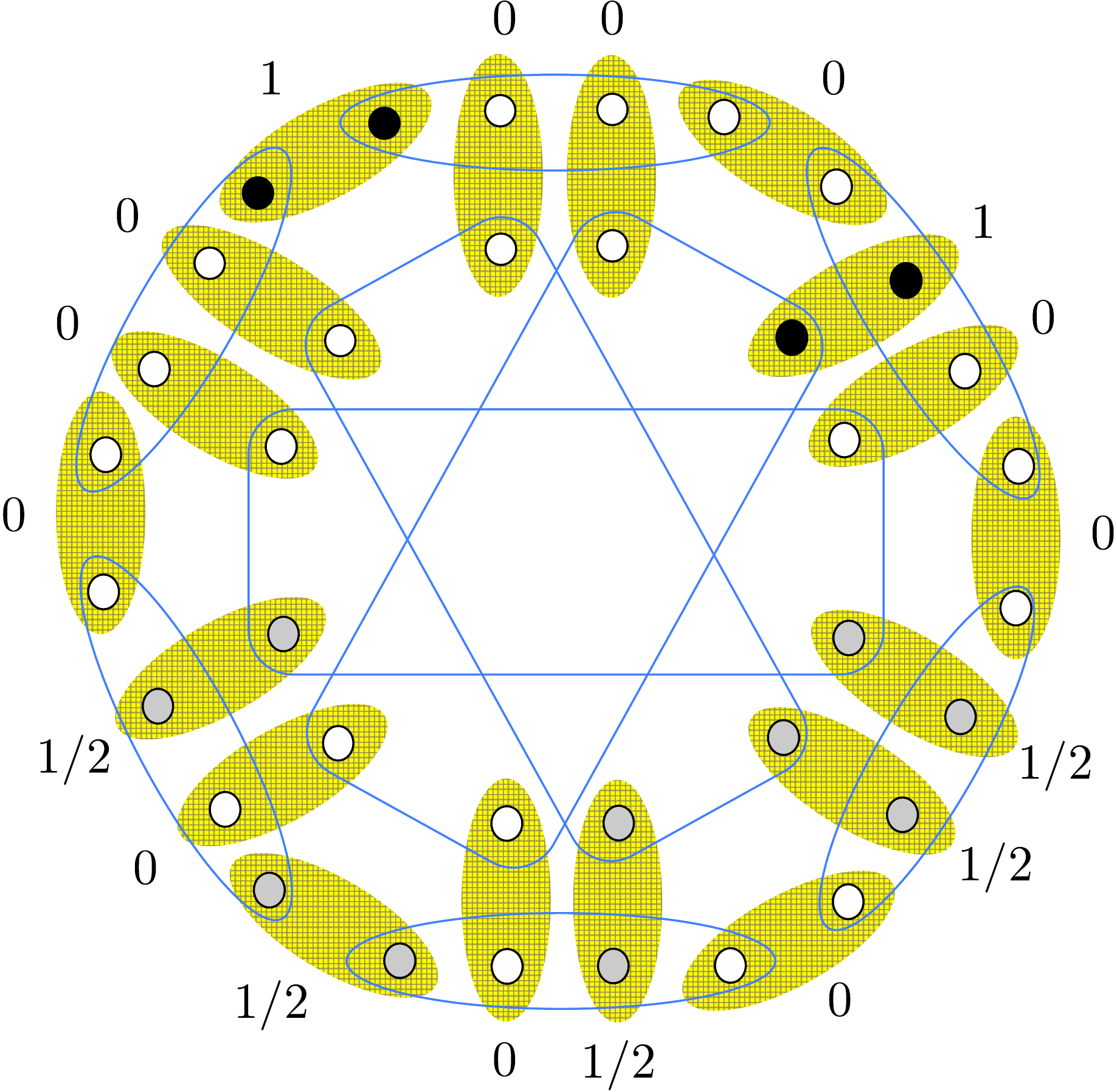}
 \caption{An example of a vertex of type $2$.}
\label{type2}
\end{figure}
\begin{figure}
 \centering
 \includegraphics[scale=0.4]{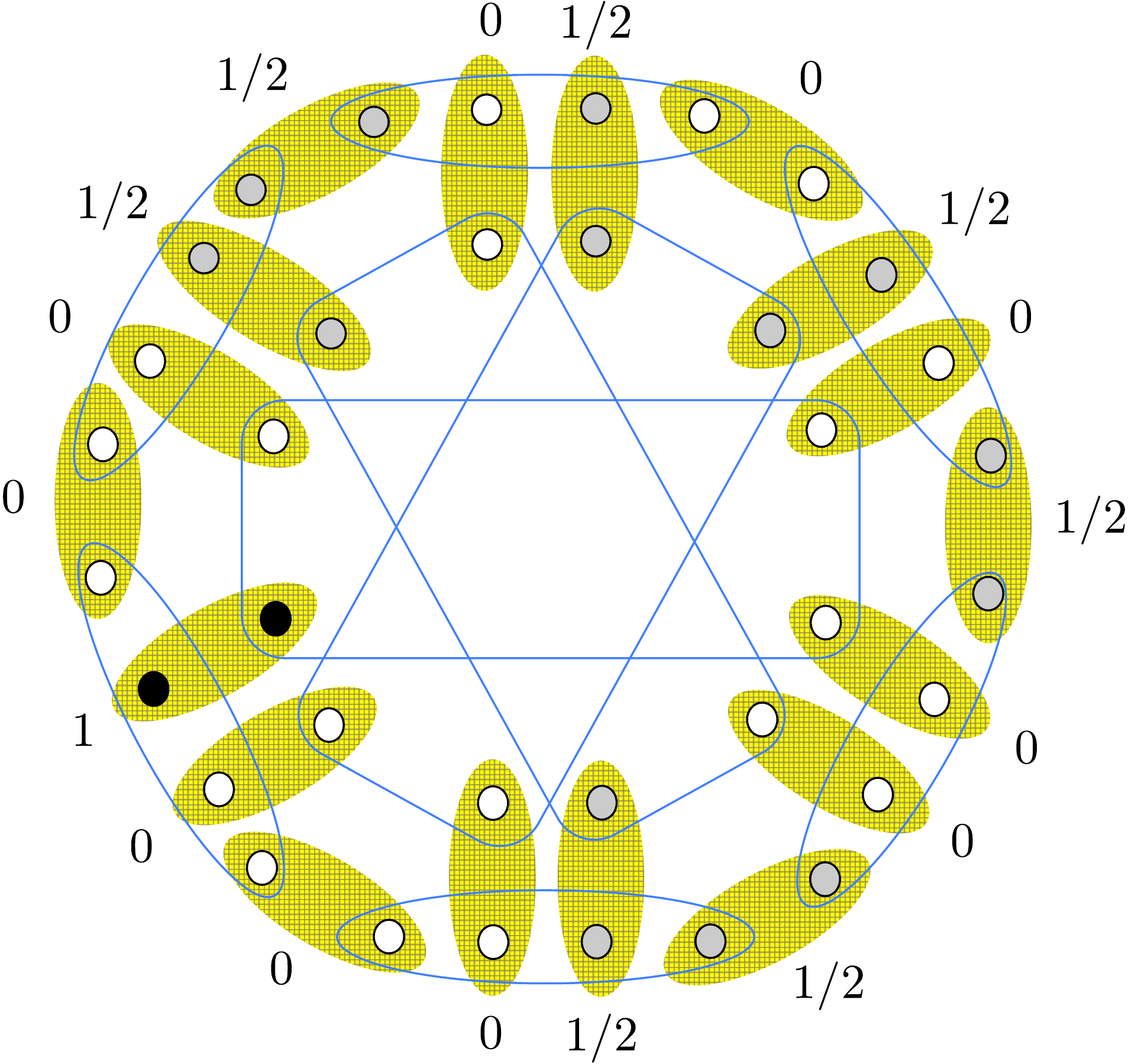}
 \caption{An example of a vertex of type $3$.}
\label{type3}
\end{figure}
\begin{figure}
 \centering
 \includegraphics[scale=0.4]{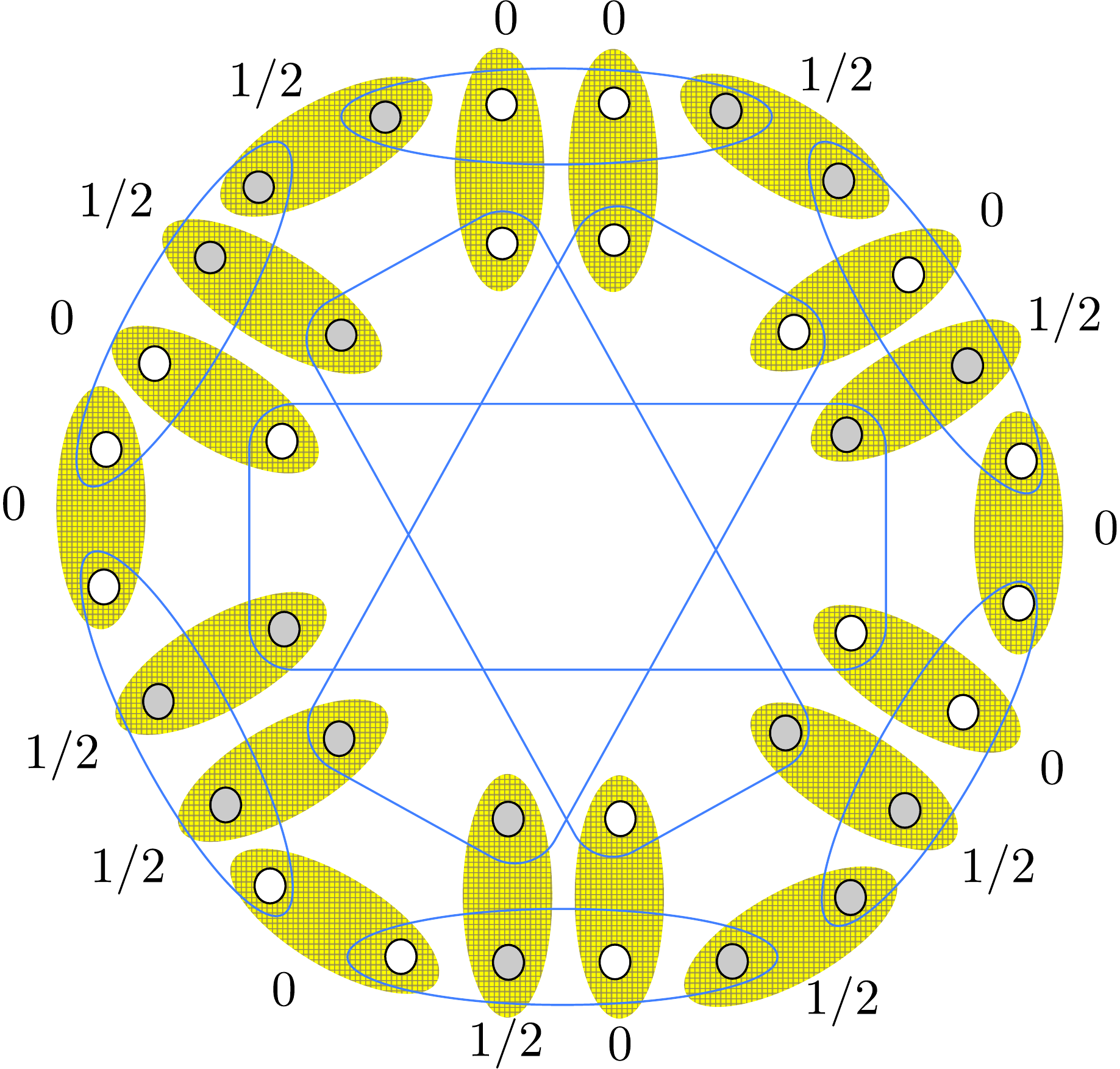}
 \caption{An example of a vertex of type $4$.}
\label{type4}
\end{figure}

\begin{table}
\begin{center}
\resizebox{\textwidth}{!}{\begin{tabular}{|c|c|c|c|c|c|c|c|c|c|c|c|c|c|c|c|c|c|}
\hline
$w_1$ & $w_2$ & $w_3$ & $w_4$ & $w_5$ & $w_6$ &$w_7$ & $w_8$ & $w_9$ &$w_{10}$ & $w_{11}$ & $w_{12}$ &$w_{13}$ & $w_{14}$ & $w_{15}$ &$w_{16}$ & $w_{17}$ & $w_{18}$\\ \hline
1/2 & 1/2 & 0 & 0 & 1 & 0 & 0 & 0 & 0 & 1 & 0 & 0 & 0 & 0 & 1 & 0 & 0 & 1/2\\ \hline
1/2 & 1/2 & 0 & 0 & 0 & 0 & 1 & 0 & 0 & 0 & 0 & 1 & 0 & 0 & 1 & 0 & 0 & 1/2\\ \hline
1/2 & 1/2 & 0 & 0 & 0 & 1 & 0 & 0 & 0 & 1 & 0 & 0 & 0 & 1 & 0 & 0 & 0 & 1/2\\ \hline
1/2 & 1/2 & 0 & 0 & 1 & 0 & 0 & 1 & 0 & 0 & 0 & 0 & 1 & 0 & 0 & 0 & 0 & 1/2\\ \hline
0 & 0 & 1/2 & 1/2 & 1/2 & 0 & 0 & 0 & 0 & 1 & 0 & 0 & 0 & 0 & 1 & 0 & 0 & 1\\ \hline
0 & 0 & 1/2 & 1/2 & 1/2 & 0 & 0 & 1 & 0 & 0 & 0 & 0 & 1 & 0 & 0 & 0 & 0 & 1\\ \hline
0 & 0 & 1/2 & 1/2 & 1/2 & 0 & 0 & 0 & 1 & 0 & 0 & 0 & 1 & 0 & 0 & 0 & 1 & 0\\ \hline
0 & 0 & 1/2 & 1/2 & 1/2 & 0 & 0 & 1 & 0 & 0 & 1 & 0 & 0 & 0 & 0 & 1 & 0 & 0\\ \hline
0 & 0 & 1 & 0 & 0 & 1/2 & 1/2 & 1/2 & 0 & 0 & 0 & 0 & 1 & 0 & 0 & 0 & 0 & 1\\ \hline
1 & 0 & 0 & 0 & 0 & 1/2 & 1/2 & 1/2 & 0 & 0 & 1 & 0 & 0 & 1 & 0 & 0 & 0 & 0\\ \hline
0 & 1 & 0 & 0 & 0 & 1/2 & 1/2 & 1/2 & 0 & 0 & 0 & 1 & 0 & 0 & 0 & 1 & 0 & 0\\ \hline
0 & 0 & 1 & 0 & 0 & 1/2 & 1/2 & 1/2 & 0 & 0 & 1 & 0 & 0 & 0 & 0 & 1 & 0 & 0\\ \hline
1 & 0 & 0 & 0 & 1 & 0 & 0 & 0 & 1/2 & 1/2 & 1/2 & 0 & 0 & 0 & 1 & 0 & 0 & 0\\ \hline
1 & 0 & 0 & 0 & 0 & 1 & 0 & 0 & 1/2 & 1/2 & 1/2 & 0 & 0 & 1 & 0 & 0 & 0 & 0\\ \hline
0 & 0 & 0 & 1 & 0 & 0 & 0 & 0 & 1/2 & 1/2 & 1/2 & 0 & 0 & 1 & 0 & 0 & 1 & 0\\ \hline
0 & 0 & 1 & 0 & 0 & 1 & 0 & 0 & 1/2 & 1/2 & 1/2 & 0 & 0 & 0 & 0 & 1 & 0 & 0\\ \hline
0 & 0 & 0 & 1 & 0 & 0 & 0 & 1 & 0 & 0 & 0 & 1/2 & 1/2 & 1/2 & 0 & 0 & 0 & 1\\ \hline
0 & 0 & 0 & 1 & 0 & 0 & 0 & 0 & 1 & 0 & 0 & 1/2 & 1/2 & 1/2 & 0 & 0 & 1 & 0\\ \hline
0 & 1 & 0 & 0 & 0 & 0 & 1 & 0 & 0 & 0 & 0 & 1/2 & 1/2 & 1/2 & 0 & 0 & 1 & 0\\ \hline
1 & 0 & 0 & 0 & 0 & 1 & 0 & 0 & 1 & 0 & 0 & 1/2 & 1/2 & 1/2 & 0 & 0 & 0 & 0\\ \hline
0 & 1 & 0 & 0 & 0 & 0 & 1 & 0 & 0 & 0 & 0 & 1 & 0 & 0 & 1/2 & 1/2 & 1/2 & 0\\ \hline
0 & 1 & 0 & 0 & 1 & 0 & 0 & 0 & 0 & 1 & 0 & 0 & 0 & 0 & 1/2 & 1/2 & 1/2 & 0\\ \hline
0 & 0 & 0 & 1 & 0 & 0 & 0 & 0 & 1 & 0 & 0 & 1 & 0 & 0 & 1/2 & 1/2 & 1/2 & 0\\ \hline
0 & 0 & 1 & 0 & 0 & 0 & 1 & 0 & 0 & 0 & 1 & 0 & 0 & 0 & 1/2 & 1/2 & 1/2 & 0\\ \hline
\end{tabular}}
\caption{The $24$ vertices of Type 1.}
\label{vertices1}
\end{center}
\end{table}

\begin{table}
\begin{center}
\resizebox{\textwidth}{!}{
\begin{tabular}{|c|c|c|c|c|c|c|c|c|c|c|c|c|c|c|c|c|c|}
 \hline
$w_1$ & $w_2$ & $w_3$ & $w_4$ & $w_5$ & $w_6$ &$w_7$ & $w_8$ & $w_9$ &$w_{10}$ & $w_{11}$ & $w_{12}$ &$w_{13}$ & $w_{14}$ & $w_{15}$ &$w_{16}$ & $w_{17}$ & $w_{18}$\\ \hline
1& 0& 0& 0& 1& 0& 0& 1/2& 1/2& 0& 1/2& 0& 1/2& 0& 1/2& 0& 0& 0\\ \hline
1& 0& 0& 0& 1/2& 0& 1/2& 0& 1/2& 0& 1/2& 1/2& 0& 0& 1& 0& 0& 0\\ \hline
1& 0& 0& 0& 1/2& 1/2& 0& 0& 1& 0& 0& 1/2& 1/2& 0& 1/2& 0& 0& 0\\ \hline
1& 0& 0& 0& 1/2& 0& 1/2& 1/2& 0& 0& 1& 0& 0& 1/2& 1/2& 0& 0& 0\\ \hline
1& 0& 0& 0& 1/2& 0& 1/2& 0& 1/2& 0& 1/2& 1/2& 0& 0& 1& 0& 0& 0\\ \hline

0& 1& 0& 0& 1/2& 0& 1/2& 0& 0& 1/2& 0& 0& 1/2& 1/2& 0& 0& 1& 0\\ \hline
0& 1& 0& 0& 1/2& 1/2& 0& 1/2& 0& 1/2& 0& 1/2& 0& 0& 0& 1& 0& 0\\ \hline
0& 1& 0& 0& 1/2& 1/2& 0& 0& 0& 1& 0& 0& 0& 1/2& 0& 1/2& 1/2& 0\\ \hline
0& 1& 0& 0& 1& 0& 0& 1/2& 0& 1/2& 0& 0& 1/2& 0& 0& 1/2& 1/2& 0\\ \hline

0& 0& 1& 0& 0& 1& 0& 0& 1/2& 1/2& 0& 0& 1/2& 0& 0& 1/2& 0& 1/2\\ \hline
0& 0& 1& 0& 0& 1/2& 1/2& 0& 1/2& 0& 0& 0& 1& 0& 0& 0& 1/2& 1/2\\ \hline
0& 0& 1& 0& 0& 1/2& 1/2& 0& 0& 1/2& 0& 0& 1/2& 0& 1/2& 0& 0& 1\\ \hline
0& 0& 1& 0& 0& 0& 1& 0& 0& 0& 1/2& 0& 1/2& 0& 1/2& 0& 1/2& 1/2\\ \hline

0& 0& 0& 1& 0& 0& 0& 1/2& 1/2& 0& 0& 1& 0& 0& 1/2& 1/2& 0& 1/2\\ \hline
0& 0& 0& 1& 0& 0& 0& 1/2& 0& 1/2& 1/2& 0& 0& 1& 0& 0& 1/2& 1/2\\ \hline
0& 0& 0& 1& 0& 0& 0& 1/2& 0& 1/2& 0& 1/2& 0& 1/2& 1/2& 0& 0& 1\\ \hline
0& 0& 0& 1& 0& 0& 0& 1& 0& 0& 1/2& 1/2& 0& 1/2& 0& 1/2& 0& 1/2\\ \hline

0& 0& 1/2& 1/2& 0& 0& 1/2& 0& 0& 1/2& 0& 1/2& 0& 0& 1& 0& 0& 1\\\hline
0& 0& 1/2& 1/2& 0& 1/2& 0& 0& 0& 1& 0& 0& 0& 1/2& 1/2& 0& 0& 1\\ \hline

0& 1/2& 0& 1/2& 0& 0& 1/2& 0& 0& 1/2& 1/2& 0& 0& 1& 0& 0& 1& 0\\ \hline
0& 1/2& 0& 1/2& 0& 1/2& 0& 0& 0& 1& 0& 0& 0& 1& 0& 0& 1/2& 1/2\\ \hline
0& 1/2& 0& 1/2& 0& 1/2& 0& 1/2& 1/2& 0& 0& 1& 0& 0& 0& 1& 0& 0\\ \hline
0& 1/2& 0& 1/2& 1/2& 0& 0& 1& 0& 0& 1/2& 1/2& 0& 0& 0& 1& 0& 0\\ \hline

0& 1/2& 1/2& 0& 0& 0& 1& 0& 0& 0& 1/2& 0& 1/2& 1/2& 0& 0& 1& 0\\ \hline
0& 1/2& 1/2& 0& 1/2& 0& 1/2& 0& 1/2& 0& 0& 0& 1& 0& 0& 0& 1& 0\\ \hline
0& 1/2& 1/2& 0& 0& 1& 0& 0& 0& 1& 0& 0& 0& 1/2& 0& 1/2& 0& 1/2\\ \hline
0& 1/2& 1/2& 0& 0& 1& 0& 0& 1/2& 1/2& 0& 1/2& 0& 0& 0& 1& 0& 0\\ \hline

1/2& 0& 0& 1/2& 0& 0& 1/2& 0& 1/2& 0& 0& 1& 0& 0& 1& 0& 0& 1/2\\ \hline
1/2& 0& 0& 1/2& 0& 0& 1/2& 1/2& 0& 0& 1& 0& 0& 1& 0& 0& 1/2& 0\\ \hline
1/2& 0& 0& 1/2& 0& 1/2& 0& 0& 1& 0& 0& 1& 0& 0& 1/2& 1/2& 0& 0\\ \hline
1/2& 0& 0& 1/2& 1/2& 0& 0& 1& 0& 0& 1& 0& 0& 1/2& 0& 1/2& 0& 0\\ \hline

1/2& 0& 1/2& 0& 1/2& 1/2& 0& 0& 1& 0& 0& 0& 1& 0& 0& 0& 1/2& 0\\ \hline
1/2& 0& 1/2& 0& 0& 1& 0& 0& 1& 0& 0& 1/2& 1/2& 0& 0& 1/2& 0& 0\\ \hline
1/2& 0& 1/2& 0& 0& 0& 1& 0& 0& 0& 1& 0& 0& 1/2& 1/2& 0& 1/2& 0\\ \hline
1/2& 0& 1/2& 0& 0& 0& 1& 0& 0& 0& 1/2& 1/2& 0& 0& 1& 0& 0& 1/2\\ \hline

1/2& 1/2& 0& 0& 1& 0& 0& 1/2& 1/2& 0& 0& 0& 1& 0& 0& 0& 1/2& 0\\ \hline
1/2& 1/2& 0& 0& 1& 0& 0& 1& 0& 0& 1/2& 0& 1/2& 0& 0& 1/2& 0& 0\\ \hline
\end{tabular}}
\caption{The $36$ vertices of Type 2.}
\label{vertices2}
\end{center}
\end{table}

\begin{table}
\begin{center}
\resizebox{\textwidth}{!}{\begin{tabular}{|c|c|c|c|c|c|c|c|c|c|c|c|c|c|c|c|c|c|}
\hline
$w_1$ & $w_2$ & $w_3$ & $w_4$ & $w_5$ & $w_6$ &$w_7$ & $w_8$ & $w_9$ &$w_{10}$ & $w_{11}$ & $w_{12}$ &$w_{13}$ & $w_{14}$ & $w_{15}$ &$w_{16}$ & $w_{17}$ & $w_{18}$\\ \hline
1/2& 0& 1/2& 0& 1/2& 0& 1/2& 0& 0& 1/2& 1/2& 0& 0& 0& 1& 0& 0& 1/2\\ \hline
1/2& 0& 1/2& 0& 1/2& 1/2& 0& 1/2& 1/2& 0& 0& 0& 1& 0& 0& 0& 0& 1/2\\ \hline
1/2& 0& 1/2& 0& 0& 1& 0& 0& 1/2& 1/2& 0& 0& 1/2& 1/2& 0& 0& 0& 1/2\\ \hline
1/2& 0& 1/2& 0& 1/2& 0& 1/2& 1/2& 0& 0& 1& 0& 0& 0& 1/2& 1/2& 0& 0\\ \hline

1/2& 0& 0& 1/2& 1/2& 0& 0& 0& 1/2& 1/2& 0& 1/2& 0& 0& 1& 0& 0& 1/2\\ \hline
1/2& 0& 0& 1/2& 0& 1/2& 0& 1/2& 0& 1/2& 1/2& 0& 0& 1& 0& 0& 0& 1/2\\ \hline
1/2& 0& 0& 1/2& 1/2& 0& 0& 1& 0& 0& 1/2& 0& 1/2& 1/2& 0& 0& 0& 1/2\\ \hline
1/2& 0& 0& 1/2& 1/2& 0& 0& 0& 1& 0& 0& 1/2& 1/2& 0& 1/2& 0& 1/2& 0\\ \hline

0& 0& 1/2& 1/2& 0& 1/2& 0& 1/2& 0& 1/2& 0& 0& 1/2& 1/2& 0& 0& 0& 1\\ \hline
0& 0& 1/2& 1/2& 0& 0& 1/2& 0& 1/2& 0& 1/2& 0& 1/2& 1/2& 0& 0& 1& 0\\ \hline
0& 0& 1/2& 1/2& 0& 0& 1/2& 1/2& 0& 0& 0& 1/2& 1/2& 0& 1/2& 0& 0& 1\\ \hline
0& 0& 1/2& 1/2& 0& 1/2& 0& 0& 1& 0& 0& 1/2& 1/2& 0& 0& 1/2& 1/2& 0\\ \hline
0& 0& 1/2& 1/2& 0& 0& 1/2& 1/2& 0& 0& 1& 0& 0& 1/2& 0& 1/2& 1/2& 0\\ \hline
0& 0& 1/2& 1/2& 0& 1/2& 0& 1/2& 1/2& 0& 1/2& 1/2& 0& 0& 0& 1& 0& 0\\ \hline

0& 1/2& 0& 1/2& 1/2& 0& 0& 0& 1/2& 1/2& 0& 0& 1/2& 1/2& 0& 0& 1& 0\\ \hline
0& 1/2& 0& 1/2& 1/2& 0& 0& 1& 0& 0& 0& 1/2& 1/2& 0& 0& 1/2& 0& 1/2\\ \hline
0& 1/2& 0& 1/2& 0& 0& 1/2& 1/2& 0& 0& 0& 1& 0& 0& 1/2& 1/2& 0& 1/2\\ \hline
0& 1/2& 0& 1/2& 1/2& 0& 0& 0& 0& 1& 0& 0& 0& 1/2& 1/2& 0& 1/2& 1/2\\ \hline

1/2& 1/2& 0& 0& 0& 1/2& 1/2& 0& 0& 1/2& 1/2& 0& 0& 1& 0& 0& 1/2& 0\\ \hline
1/2& 1/2& 0& 0& 0& 1& 0& 0& 1/2& 1/2& 0& 1/2& 0& 1/2& 0& 1/2& 0& 0\\ \hline
1/2& 1/2& 0& 0& 1& 0& 0& 1/2& 0& 1/2& 1/2& 0& 0& 0& 1/2& 1/2& 0& 0\\ \hline
1/2& 1/2& 0& 0& 1& 0& 0& 0& 1/2& 1/2& 0& 0& 1/2& 0& 1/2& 0& 1/2& 0\\ \hline
1/2& 1/2& 0& 0& 0& 0& 1& 0& 0& 0& 1/2& 1/2& 0& 1/2& 1/2& 0& 1/2& 0\\ \hline
1/2& 1/2& 0& 0& 0& 1/2& 1/2& 0& 1/2& 0& 0& 1& 0& 0& 1/2& 1/2& 0& 0\\ \hline

0& 1/2& 1/2& 0& 1/2& 0& 1/2& 1/2& 0& 0& 0& 0& 1& 0& 0& 0& 1/2& 1/2\\ \hline
0& 1/2& 1/2& 0& 1/2& 1/2& 0& 0& 0& 1& 0& 0& 0& 0& 1/2& 1/2& 0& 1/2\\ \hline
0& 1/2& 1/2& 0& 1/2& 1/2& 0& 1/2& 0& 1/2& 1/2& 0& 0& 0& 0& 1& 0& 0\\ \hline
0& 1/2& 1/2& 0& 0& 0& 1& 0& 0& 0& 0& 1/2& 1/2& 0& 1/2& 0& 1/2& 1/2\\ \hline

1& 0& 0& 0& 1/2& 1/2& 0& 1/2& 1/2& 0& 1/2& 0& 1/2& 1/2& 0& 0& 0& 0\\ \hline
1& 0& 0& 0& 0& 1/2& 1/2& 0& 1/2& 0& 1/2& 1/2& 0& 1/2& 1/2& 0& 0& 0\\ \hline

0& 1& 0& 0& 0& 1/2& 1/2& 0& 0& 1/2& 0& 1/2& 0& 1/2& 0& 1/2& 1/2& 0\\ \hline
0& 1& 0& 0& 1/2& 0& 1/2& 1/2& 0& 0& 0& 1/2& 1/2& 0& 0& 1/2& 1/2& 0\\ \hline

0& 0& 1& 0& 0& 1/2& 1/2& 0& 1/2& 0& 1/2& 0& 1/2& 0& 0& 1/2& 1/2& 0\\ \hline
0& 0& 1& 0& 0& 1/2& 1/2& 0& 0& 1/2& 1/2& 0& 0& 0& 1/2& 1/2& 0& 1/2\\ \hline

0& 0& 0& 1& 0& 0& 0& 0& 1/2& 1/2& 0& 1/2& 0& 1/2& 1/2& 0& 1/2& 1/2\\ \hline
0& 0& 0& 1& 0& 0& 0& 1/2& 1/2& 0& 1/2& 1/2& 0& 1/2& 0& 1/2& 1/2& 0\\ \hline
\end{tabular}}
\caption{The $36$ vertices of Type 3.}
\label{vertices3}
\end{center}
\end{table}

\begin{table}
\begin{center}
\resizebox{\textwidth}{!}{\begin{tabular}{|c|c|c|c|c|c|c|c|c|c|c|c|c|c|c|c|c|c|}
\hline
$w_1$ & $w_2$ & $w_3$ & $w_4$ & $w_5$ & $w_6$ &$w_7$ & $w_8$ & $w_9$ &$w_{10}$ & $w_{11}$ & $w_{12}$ &$w_{13}$ & $w_{14}$ & $w_{15}$ &$w_{16}$ & $w_{17}$ & $w_{18}$\\ \hline
1/2& 0& 0& 1/2& 0& 1/2& 0& 0& 1/2& 1/2& 0& 1/2& 0& 1/2& 1/2& 0& 0& 1/2\\ \hline
1/2& 0& 0& 1/2& 1/2& 0& 0& 1/2& 1/2& 0& 1/2& 1/2& 0& 0& 1/2& 1/2& 0& 0\\ \hline
1/2& 0& 0& 1/2& 0& 1/2& 0& 1/2& 1/2& 0& 1/2& 1/2& 0& 1/2& 0& 1/2& 0& 0\\ \hline
1/2& 0& 0& 1/2& 0& 0& 1/2& 0& 1/2& 0& 1/2& 1/2& 0& 1/2& 1/2& 0& 1/2& 0\\ \hline
1/2& 0& 0& 1/2& 0& 0& 1/2& 1/2& 0& 0& 1/2& 1/2& 0& 1/2& 1/2& 0& 0& 1/2\\ \hline
1/2& 0& 0& 1/2& 1/2& 0& 0& 1/2& 0& 1/2& 1/2& 0& 0& 1/2& 1/2& 0& 0& 1/2\\ \hline
1/2& 0& 0& 1/2& 1/2& 0& 0& 1/2& 1/2& 0& 1/2& 0& 1/2& 1/2& 0& 0& 1/2& 0\\ \hline
1/2& 0& 0& 1/2& 1/2& 0& 0& 1/2& 1/2& 0& 0& 1/2& 1/2& 0& 1/2& 0& 0& 1/2\\ \hline
0& 0& 1/2& 1/2& 0& 1/2& 0& 0& 1/2& 1/2& 0& 1/2& 0& 0& 1/2& 1/2& 0& 1/2\\ \hline
0& 0& 1/2& 1/2& 0& 1/2& 0& 1/2& 1/2& 0& 0& 1/2& 1/2& 0& 0& 1/2& 0& 1/2\\ \hline
0& 0& 1/2& 1/2& 1/2& 0& 0& 0& 1/2& 1/2& 1/2& 0& 0& 0& 1/2& 1/2& 1/2& 0\\ \hline
0& 0& 1/2& 1/2& 0& 1/2& 0& 1/2& 0& 1/2& 1/2& 0& 0& 1/2& 0& 1/2& 0& 1/2\\ \hline
0& 0& 1/2& 1/2& 0& 0& 1/2& 1/2& 0& 0& 1/2& 1/2& 0& 0& 1/2& 1/2& 0& 1/2\\ \hline
0& 0& 1/2& 1/2& 0& 1/2& 0& 0& 1/2& 1/2& 0& 0& 1/2& 1/2& 0& 0& 1/2& 1/2\\ \hline
0& 0& 1/2& 1/2& 0& 0& 1/2& 0& 0& 1/2& 1/2& 0& 0& 1/2& 1/2& 0& 1/2& 1/2\\ \hline
0& 0& 1/2& 1/2& 0& 0& 1/2& 0& 1/2& 0& 0& 1/2& 1/2& 0& 1/2& 0& 1/2& 1/2\\ \hline
0& 0& 1/2& 1/2& 0& 0& 1/2& 1/2& 0& 0& 1/2& 0& 1/2& 1/2& 0& 0& 1/2& 1/2\\ \hline
0& 1/2& 0& 1/2& 1/2& 0& 0& 1/2& 1/2& 0& 0& 1/2& 1/2& 0& 0& 1/2& 1/2& 0\\ \hline
0& 1/2& 0& 1/2& 0& 1/2& 0& 0& 1/2& 1/2& 0& 1/2& 0& 1/2& 0& 1/2& 1/2& 0\\ \hline
0& 1/2& 0& 1/2& 1/2& 0& 0& 1/2& 0& 1/2& 1/2& 0& 0& 1/2& 0& 1/2& 1/2& 0\\ \hline
0& 1/2& 0& 1/2& 1/2& 0& 0& 1/2& 0& 1/2& 0& 1/2& 0& 0& 1/2& 1/2& 0& 1/2\\ \hline
0& 1/2& 0& 1/2& 0& 0& 1/2& 1/2& 0& 0& 1/2& 1/2& 0& 1/2& 0& 1/2& 1/2& 0\\ \hline
0& 1/2& 0& 1/2& 0& 1/2& 0& 1/2& 0& 1/2& 0& 1/2& 0& 1/2& 0& 1/2& 0& 1/2\\ \hline
0& 1/2& 0& 1/2& 1/2& 0& 0& 1/2& 0& 1/2& 0& 0& 1/2& 1/2& 0& 0& 1/2& 1/2\\ \hline
0& 1/2& 0& 1/2& 0& 0& 1/2& 0& 0& 1/2& 0& 1/2& 0& 1/2& 1/2& 0& 1/2& 1/2\\ \hline

0& 1/2& 1/2& 0& 1/2& 1/2& 0& 0& 1/2& 1/2& 0& 0& 1/2& 0& 0& 1/2& 1/2& 0\\ \hline
0& 1/2& 1/2& 0& 0& 1/2& 1/2& 0& 0& 1/2& 0& 1/2& 0& 0& 1/2& 1/2& 0& 1/2\\ \hline
0& 1/2& 1/2& 0& 1/2& 0& 1/2& 1/2& 0& 0& 1/2& 0& 1/2& 0& 0& 1/2& 1/2& 0\\ \hline
0& 1/2& 1/2& 0& 0& 1/2& 1/2& 0& 0& 1/2& 1/2& 0& 0& 1/2& 0& 1/2& 1/2& 0\\ \hline
0& 1/2& 1/2& 0& 0& 1/2& 1/2& 0& 1/2& 0& 0& 1/2& 1/2& 0& 0& 1/2& 1/2& 0\\ \hline
0& 1/2& 1/2& 0& 1/2& 1/2& 0& 1/2& 0& 1/2& 0& 0& 1/2& 0& 0& 1/2& 0& 1/2\\ \hline
0& 1/2& 1/2& 0& 1/2& 0& 1/2& 0& 0& 1/2& 0& 0& 1/2& 0& 1/2& 0& 1/2& 1/2\\ \hline
0& 1/2& 1/2& 0& 0& 1/2& 1/2& 0& 0& 1/2& 0& 0& 1/2& 1/2& 0& 0& 1/2& 1/2\\ \hline

1/2& 1/2& 0& 0& 1/2& 0& 1/2& 1/2& 0& 0& 1/2& 1/2& 0& 0& 1/2& 1/2& 0& 0\\ \hline
1/2& 1/2& 0& 0& 1/2& 1/2& 0& 1/2& 0& 1/2& 1/2& 0& 0& 1/2& 0& 1/2& 0& 0\\ \hline
1/2& 1/2& 0& 0& 1/2& 1/2& 0& 1/2& 1/2& 0& 0& 1/2& 1/2& 0& 0& 1/2& 0& 0\\ \hline
1/2& 1/2& 0& 0& 1/2& 1/2& 0& 0& 1/2& 1/2& 0& 1/2& 0& 0& 1/2& 1/2& 0& 0\\ \hline
1/2& 1/2& 0& 0& 0& 1/2& 1/2& 1/2& 0& 0& 0& 1/2& 1/2& 1/2& 0& 0& 0& 1/2\\ \hline
1/2& 1/2& 0& 0& 1/2& 1/2& 0& 0& 1/2& 1/2& 0& 0& 1/2& 1/2& 0& 0& 1/2& 0\\ \hline
1/2& 1/2& 0& 0& 1/2& 0& 1/2& 0& 0& 1/2& 1/2& 0& 0& 1/2& 1/2& 0& 1/2& 0\\ \hline
1/2& 1/2& 0& 0& 1/2& 0& 1/2& 0& 1/2& 0& 0& 1/2& 1/2& 0& 1/2& 0& 1/2& 0\\ \hline
1/2& 1/2& 0& 0& 1/2& 0& 1/2& 1/2& 0& 0& 1/2& 0& 1/2& 1/2& 0& 0& 1/2& 0\\ \hline

1/2& 0& 1/2& 0& 1/2& 1/2& 0& 1/2& 1/2& 0& 1/2& 0& 1/2& 0& 0& 1/2& 0& 0\\ \hline
1/2& 0& 1/2& 0& 0& 1/2& 1/2& 0& 1/2& 0& 1/2& 1/2& 0& 0& 1/2& 1/2& 0& 0\\ \hline
1/2& 0& 1/2& 0& 1/2& 0& 1/2& 1/2& 0& 0& 1/2& 0& 1/2& 0& 1/2& 0& 0& 1/2\\ \hline
1/2& 0& 1/2& 0& 0& 1/2& 1/2& 0& 0& 1/2& 1/2& 0& 0& 1/2& 1/2& 0& 0& 1/2\\ \hline
1/2& 0& 1/2& 0& 0& 1/2& 1/2& 0& 1/2& 0& 1/2& 0& 1/2& 1/2& 0& 0& 1/2& 0\\ \hline
1/2& 0& 1/2& 0& 0& 1/2& 1/2& 0& 1/2& 0& 0& 1/2& 1/2& 0& 1/2& 0& 0& 1/2\\ \hline
1/2& 0& 1/2& 0& 1/2& 1/2& 0& 0& 1/2& 1/2& 0& 0& 1/2& 0& 1/2& 0& 0& 1/2\\ \hline
1/2& 0& 1/2& 0& 1/2& 0& 1/2& 0& 1/2& 0& 1/2& 0& 1/2& 0& 1/2& 0& 1/2& 0\\ \hline
\end{tabular}}
\caption{The $50$ vertices of Type 4.}
\label{vertices4}
\end{center}
\end{table}

\section{Noise robustness of the noncontextuality inequality}
How much noise can one add to the measurements and preparations while still violating our noncontextuality inequality? We answer this question here assuming that the experimental operations 
are well-modelled by quantum theory.
According to quantum theory,
\begin{equation}
 p(k|M_i,P_{i,k})=\Tr(E_{k|M_i}\rho_{i,k}),
\end{equation}
where $E_{k|M_i}$ denotes the positive operator representing the measurement event $[k|M_i]$ and $\rho_{i,k}$  denotes the  density operator representing the preparation $P_{i,k}$. To be precise,
for every $i$, the set $\{ E_{k|M_i}\}_k$ is a positive operator valued measure, so that 
$0\leq E_{k|M_i}\leq I$, and $\sum_k E_{k|M_i}=I$, and for every $i$ and $k$, $\rho_{i,k}$ is positive, $\rho_{i,k}\geq 0$, and has unit trace, $\Tr\rho_{i,k}=1$.

In quantum theory, a noiseless and maximally informative measurement is represented by a POVM whose elements are rank-1 projectors, that is, 
\begin{equation}
 E_{k|M_i}=\Pi_{i,k},
\end{equation}
where for each $k$, $\Pi_{i,k}$ is a projector, hence idempotent, $\Pi_{i,k}^2=\Pi_{i,k}$, and is rank $1$, so that $\Pi_{i,k}=|\psi_{i,k}\rangle \langle \psi_{i,k}|$, where for each $i$, the 
set $\{ |\psi_{i,k}\rangle \}_k$ is an orthonormal basis of the Hilbert space.
If we furthermore set 
\beq
\rho_{i,k}=\Pi_{i,k},
\eeq 
then we find  $p(k|M_i,P_{i,k})={\rm Tr}(E_{k|M_i} \rho_{i,k})=1$ for each $(i,k)$, and consequently $A=1$.  We see, therefore, that the maximum possible value of $A$ is attained when 
preparations and measurements satisfy the noiseless ideal. We can now consider the consequence of adding noise. 

We begin by considering a very simple noise model wherein the preparations and measurements both deviate from the noiseless ideal by the action of a depolarizing channel, that is, a channel of 
the form
\begin{equation}
\mathcal{D}_p(\cdot) = p I (\cdot) I + (1-p) \frac{1}{4}I \;{\rm Tr}(\cdot),
\end{equation}
which with probability $p$ implements the identity channel and with probability $1-p$ generates the completely mixed state. 
If the quantum states are the image of the ideal states under a depolarizing channel with parameter $p_1$, and the POVM is obtained by acting the depolarizing channel with parameter $p_2$ 
followed by the ideal projector-valued measure (such that the POVM elements are the images of the projectors under the {\em adjoint} of the channel), then
\begin{eqnarray}
\rho_{i,k}&=&\mathcal{D}_{p_1} (\Pi_{i,k}) = p_1 \Pi_{i,k}+(1-p_1)\frac{1}{4}I,\\
E_{k|M_i}&=& \mathcal{D}^{\dag}_{p_2} (\Pi_{i,k})=  p_2 \Pi_{i,k}+(1-p_2) \frac{1}{4} I,
\end{eqnarray}
Here, the POVM $\{ E_{k|M_i}\}_k$ is a mixture of $\{ \Pi_{i,k} \}_k$ and a POVM $\{ \frac{1}{4}I,\frac{1}{4}I,\frac{1}{4}I,\frac{1}{4}I \}$ which simply samples $k$ uniformly at random 
regardless of the input state.
It follows that for each $(i,k)$, if we consider $p(k|M_i,P_{i,k})= {\rm Tr}(E_{k|M_i} \rho_{i,k})$, we find perfect predictability for the term having weight $p_1 p_2$ while for the three 
other terms, we have a uniformly random outcome, so that in all
\begin{equation}
 p(k|M_i,P_{i,k})=  p_1p_2 + (1-p_1p_2) \frac{1}{4}.
 \end{equation}
It follows that
\begin{equation}
A\equiv \frac{1}{36}\sum_{i=1}^9\sum_{k=1}^{4} p(k|M_i,P_{i,k})=\frac{1}{4}+\frac{3}{4}p_1 p_2,
\end{equation}
Thus a violation of the noncontextuality inequality, i.e. $A>\frac{5}{6}$, occurs if and only if
\begin{equation}
p_1 p_2 >\frac{7}{9}.
\end{equation}

It turns out that one can derive similar bounds for more general noise models as well.   Suppose that instead of a depolarizing channel, we have one of the form 
\begin{equation}
\mathcal{N}_{p,\rho}(\cdot) = p I (\cdot) I + (1-p) \rho \;{\rm Tr}(\cdot).
\end{equation}
With probability $p$, this implements the identity channel and with probability $1-p$ it reprepares a state $\rho$  that need not be the completely mixed state, but which is independent of the 
input to the channel.  The analogous sort of noise acting on the measurement corresponds to acting on the POVM elements by the adjoint of this channel, that is,
\begin{equation}
\mathcal{N}^{\dag}_{p,\rho}(\cdot) = p I (\cdot) I + (1-p) I \;{\rm Tr}(\rho \; \cdot).
\end{equation}

Therefore, if this sort of noise is applied to the ideal states and measurements, with the parameters in each noise model allowed to depend on $i$, we obtain
\begin{eqnarray}
\rho_{i,k}&=&\mathcal{N}_{p^{(i)}_1,\rho^{(i)}} (\Pi_{i,k}) = p^{(i)}_1 \Pi_{i,k}+(1-p^{(i)}_1) \rho^{(i)},\label{bb1}\\
E_{k|M_i}&=& \mathcal{N}^{\dag}_{p^{(i)}_2,\rho^{(i)}} (\Pi_{i,k})=  p^{(i)}_2 \Pi_{i,k}+(1-p^{(i)}_2) s(k|i) I,\nonumber \\\label{bb2}
\end{eqnarray}
where $s(k|i)\equiv {\rm Tr}(\rho^{(i)} \Pi_{i,k})$ is a probability distribution over $k$ for each value of $i$.  
Here, the POVM $\{ E_{k|M_i}\}_k$ is a mixture of $\{ \Pi_{i,k} \}_k$ and a POVM $\{ s(k|i) I\}_k$ which simply samples $k$ at random from the distribution $s(k|i)$, regardless of the quantum 
state.  Compared to the simple model considered above, the innovation of this one is that for both preparations and measurements, the noise is allowed to be biased. 

For the case of $p_1^{(i)}=0$, which by Eq.~\eqref{bb1} implies that $\rho_{i,k}=\rho^{(i)}$, we find that, regardless of the measurement, $p(k|M_i,P_{i,k})$ is just a normalized probability 
distribution over $k$ (because there is no $k$ dependence in the state).  Hence, in this case, $\frac{1}{4}\sum_{k=1}^4p(k|M_i,P_{i,k})= \frac{1}{4}$.

Similarly, for the case of $p_2^{(i)}=0$, that is, when the POVM corresponds to a random number generator $E_{k|M_i}=s(k|i) I$, we find that, regardless of the preparation, $p(k|M_i,P_{i,k})$ 
is again just a normalized probability distribution over $k$.  Hence, in this case again, $\frac{1}{4}\sum_{k=1}^4p(k|M_i,P_{i,k})= \frac{1}{4}$.  

It follows that for generic values of $p_1^{(i)}$ and $p_2^{(i)}$, we have $\frac{1}{4}\sum_{k=1}^4p(k|M_i,P_{i,k})= p_1^{(i)} p_2^{(i)} + (1-p_1^{(i)} p_2^{(i)}) \frac{1}{4}$.  In all then, 
we have 
\begin{equation}
A\equiv \frac{1}{36}\sum_{i=1}^9\sum_{k=1}^{4} p(k|M_i,P_{i,k})=\frac{1}{4}+\frac{3}{4} \left( \frac{1}{9}\sum_{i=1}^9 p^{(i)}_1 p^{(i)}_2 \right).
\end{equation}
Consequently, a violation of the noncontextuality inequality, i.e. $A>\frac{5}{6}$, occurs if and only if the noise parameters satisfy
\begin{equation}
\frac{1}{9}\sum_{i=1}^9 p^{(i)}_1 p^{(i)}_2  >\frac{7}{9}.
\end{equation}
Because the parameters $p_1^{(i)}$ and $p_2^{(i)}$ decrease as one increases the amount of noise, this inequality specifies an upper bound on the amount of noise that can be tolerated if one 
seeks to violate the noncontextuality inequality. 

This analysis highlights how the approach to deriving noncontextuality inequalities described in this article has no trouble accommodating noisy POVMs.  This contrasts with previous proposals 
for experimental tests based on the traditional notion of noncontextuality, which can only be applied to projective measurements.  This is one way to see how previous proposals are not 
applicable to realistic experiments, where every measurement has some noise and consequently is necessarily {\em not} represented 
projectively.\footnote{For more general noise models, a similar analysis will go through because we know that the ideal quantum projectors 
lead to a maximal violation of the inequality $A\leq 5/6$, i.e., $A=1$. Any introduction of noise to these 
will bring $A$ below 1 and, for any given noise model, one can figure out the regime in which 
$A$ does not fall below $5/6$ so that it is possible to demonstrate contextuality even in the presence of some noise. This is the sense in
which our noncontextuality inequality is noise-tolerant.}

\section{Comparison to ``state-independent contextuality'' (SIC) inequalities}

We have proposed a technique for deriving noncontextuality inequalities from proofs of the Kochen-Specker theorem.  It is useful to compare our approach with one that has previously been 
proposed by Cabello~\cite{Cabelloexpt} to derive the so-called ``state-independent contextuality'' (SIC) inequalities
based on proofs of the Kochen-Specker theorem.
We do so by explicitly comparing the two proposals in the case of the 18 ray construction of Ref.~\cite{Cab1}.   Indeed, the fact that Ref.~\cite{Cabelloexpt} proposes an inequality for this 
construction 
is part of our motivation for choosing it as our illustrative example.

For each of the 18 operational equivalence classes of measurement events, labelled by $\kappa \in \{1,\dots,18\}$ as depicted in Fig.~\ref{legend}, we associate a $\{-1,+1\}$-valued 
variable, denoted $S_{\kappa} \in \{-1,+1\}$.
A given ontic state $\lambda$ is assumed to assign a value to each $S_{\kappa}$.
The fact that there is only a {\em single} variable associated to each equivalence class implies that any assignment of such values is necessarily noncontextual.

Ref.~\cite{Cabelloexpt} considers a particular linear combination of expectation values of products of these variables: 
\begin{align}
\alpha \equiv &-\langle S_1 S_2 S_3 S_4\rangle - \langle S_4 S_5 S_6 S_7\rangle- \langle S_7 S_8 S_9 S_{10}\rangle\nonumber\\
& - \langle S_{10} S_{11} S_{12} S_{13}\rangle -  \langle S_{13} S_{14} S_{15} S_{16}\rangle - \langle S_{16} S_{17} S_{18} S_{1}\rangle\nonumber\\
&- \langle S_{18} S_2 S_9 S_{11}\rangle- \langle S_3 S_5 S_{12} S_{14}\rangle\nonumber\\
&-\langle S_6 S_8 S_{15} S_{17}\rangle,
\end{align}
and derives the following inequality for it:
\begin{align}
\alpha \le 7
\label{Cabelloinequality}
\end{align}
(Note that Ref.~\cite{Cabelloexpt} used a labelling convention for the eighteen measurement events that is different from the one we use here; to translate between the two conventions, it 
suffices to compare Fig. 1 in that article with Fig.~\ref{legend} in ours.)  Each term in $\alpha$ refers to a quadruple of variables that can be measured together, that is, which can be 
computed from the outcome of a single measurement.   Different terms correspond to measurements that are incompatible.

In Ref.~\cite{Cabelloexpt}, the following justification is given for the inequality \eqref{Cabelloinequality}.  We are asked to consider 
the $2^{18}$ possible assignments to $(S_1,\dots, S_{18})$ that result from the two possible assignments to $S_{\kappa}$, namely $-1$ or $+1$, for each $\kappa \in \{1,\dots,18\}$.
It is then noted that among all such possibilities, the maximum value of $\alpha$ that can be achieved is 7.

Ref.~\cite{Cabelloexpt} states that a violation of this inequality should be considered evidence of a failure of noncontextuality.  We disagree with this conclusion, and the rest of this 
section seeks to explain why.

\subsection{The most natural interpretation}

It is useful to recast the inequality of Eq.~\eqref{Cabelloinequality} in terms of variables $v_{\kappa}$ with values in $\{0,1\}$ rather than $\{-1,+1\}$.  Specifically, we take
\beq
v_{\kappa} \equiv \frac{S_{\kappa} +1}{2}.
\eeq
Under this translation, products of the $S_{\kappa}$ correspond to sums (modulo 2) of the $v_{\kappa}$.  For instance,  an equation such as $S_{\kappa_1} S_{\kappa_2}=-1$ corresponds to the 
equation $v_{\kappa_1} \oplus v_{\kappa_2} =1$, where $\oplus$ denotes sum modulo 2, while $S_{\kappa_1} S_{\kappa_2}= +1$ corresponds to $v_{\kappa_1} \oplus v_{\kappa_2} =0$, so that 
$v_{\kappa_1} \oplus v_{\kappa_2} =  \frac{-S_{\kappa_1} S_{\kappa_2} +1}{2}$.  In particular, we also have
\beq
v_{\kappa_1} \oplus v_{\kappa_2} \oplus v_{\kappa_3}\oplus v_{\kappa_4} =  \frac{-S_{\kappa_1} S_{\kappa_2} S_{\kappa_3}S_{\kappa_4} +1}{2}
\eeq
or equivalently,
\beq
- S_{\kappa_1} S_{\kappa_2} S_{\kappa_3}S_{\kappa_4}= 2(v_{\kappa_1} \oplus v_{\kappa_2} \oplus v_{\kappa_3}\oplus v_{\kappa_4}) -1,
\label{jjj}
\eeq 
We can therefore consider a quantity $\alpha'$, defined as
\begin{align}
\alpha' \equiv &\langle v_1 \oplus  v_2 \oplus v_3 \oplus v_4\rangle + \langle v_4 \oplus v_5 \oplus v_6 \oplus v_7\rangle\nonumber\\
&+ \langle v_7 \oplus v_8 \oplus v_9 \oplus v_{10}\rangle
 + \langle v_{10} \oplus v_{11} \oplus v_{12} \oplus v_{13}\rangle\nonumber\\
 & +  \langle v_{13} \oplus v_{14} \oplus v_{15} \oplus v_{16}\rangle +\langle v_{16} \oplus v_{17} \oplus v_{18} \oplus v_{1}\rangle\nonumber\\
&+ \langle v_{18} \oplus v_2 \oplus v_9 \oplus v_{11}\rangle + \langle v_3 \oplus v_5 \oplus v_{12} \oplus v_{14}\rangle\nonumber\\
&+ \langle v_6 \oplus v_8 \oplus v_{15} \oplus v_{17}\rangle,
\label{alphaprime}
\end{align}
so that $\alpha = 2\alpha' - 9$, 
and we can re-express inequality \eqref{Cabelloinequality} as
\beq
\alpha' \le 8.
\label{newinequality}
\eeq
Of course, rather than using Eq.~\eqref{jjj} to translate  \eqref{Cabelloinequality} from $\{-1,+1\}$-valued variables into $\{ 0,1\}$-valued variables, 
one can also just derive the inequality \eqref{newinequality} directly: among the $2^{18}$ possible assignments of values in $\{0,1\}$ to each of the $v_{\kappa}$, the maximum value of 
$\alpha'$ is 8.  Two examples of such assignments are provided in Fig.~\ref{LogicDefyingAssignments}.

\begin{figure}[htb]
\centering
 \includegraphics[scale=0.5]{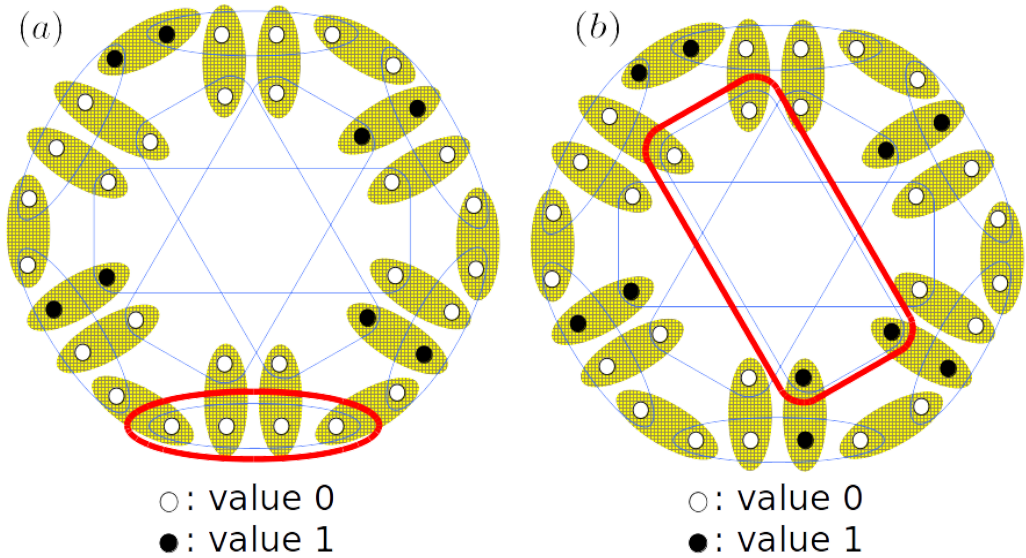}
 \caption{Examples of noncontextual assignments of $\{ 0,1\}$-values to the measurement events in the 18 ray construction
 where it is {\em not} required that every measurement has precisely one outcome that is assigned value 1 and three outcomes that are assigned the value 0. That is, normalization
 is not respected in these assignments. Example (a) depicts an assignment wherein
 there is a measurement all of whose outcomes receive probability 0, hence a \emph{subnormalized} assignment (the probabilities add up to less than 1 for the 
 highlighted measurement). 
 Example (b) depicts one wherein there is a measurement two of whose outcomes recieve probability 1, hence a \emph{supernormalized}
 assignment (the probabilities add up to more than 1 for the highlighted measurement). }
\label{LogicDefyingAssignments}
\end{figure}

It is useful to use a notation that specfies whether a given expectation value of some variable $X$ is relative to a preparation procedure $P$, in which case it is denoted $\langle X \rangle_P$,
or relative to an ontic state $\lambda$, in which case it is denoted $\langle X \rangle_{\lambda}$.  
We denote by $\alpha'(P)$ the quantity defined in \eqref{alphaprime} if the expectation values contained therein are relative to preparation $P$, and we denote by $\alpha'(\lambda)$ the case 
where the expectation values are relative to ontic state $\lambda$.
Under the assumption of an ontological model, each expectation value relative to a preparation $P$ can be expressed as a function of the expectation value relative to an ontic state $\lambda$,
via
\beq
\langle X \rangle_P = \sum_{\lambda} \langle X \rangle_{\lambda} \;\mu(\lambda|P),
\label{previous}
\eeq
where $\mu(\lambda|P)$ is the distribution over ontic states associated with preparation $P$. 
We can infer from Eq.~\eqref{previous} that
\beq
\alpha'(P) = \sum_{\lambda} \alpha'(\lambda) \mu(\lambda|P).
\label{nnn}
\eeq

With these notational conventions, we can summarize the argument of Ref.~\cite{Cabelloexpt} as follows.  In any noncontextual ontological model, every ontic state $\lambda$ satisfies
\beq
\alpha'(\lambda) \le 8.
\label{lll}
\eeq
But this in turn implies, through Eq.~\eqref{nnn}, that for all preparations $P$,
\beq
\alpha'(P) \le 8,
\label{ooo}
\eeq
which is an inequality constraining operational quantities. Cabello \cite{Cabelloexpt} then goes on to show that this inequality is violated for 
\emph{any} quantum state, hence violating this inequality proves ``state-independent contextuality''. We will argue that 
although the violation of this inequality is indeed state-independent, such a violation is not due to 
contextuality. Instead, this state-independent violation is on account of a fundamental  
inconsistency in the sort of ontological model that is implicitly assumed in Cabello's discussion.
We will show that the violation of this inequality is \emph{necessary} for any ontological model -- noncontextual or otherwise --
to even make sense and that this violation and its state-independence is therefore not a signature of contextuality.

Let us now describe the problem with the inequality \eqref{ooo}, or equivalently inequality \eqref{Cabelloinequality}, and thus with the claim of Ref.~\cite{Cabelloexpt}.
First, we highlight the physical interpretation of the variables $v_{\kappa}$.  If $v_{\kappa}$ is assigned value 1 by the ontic state $\lambda$, then this means that if the system is in the 
ontic state $\lambda$, 
and  a measurement that includes $\kappa$ as an outcome is implemented on it, then the outcome $\kappa$ is certain to occur, while if $v_{\kappa}$ is assigned value 0 by $\lambda$, then the 
outcome $\kappa$ is certain {\em not} to occur.
But each of the $2^{18}$ different assignments to $(v_1 , \dots, v_{18})$  is such that for at least one measurement either
{\em none} of the outcomes occur, as in the example of Fig.~\ref{LogicDefyingAssignments}(a), or {\em more than one} outcome occurs, as in the example of Fig.~\ref{LogicDefyingAssignments}(b).  (This is precisely what is implied by the fact that the 18 measurement events are Kochen-Specker {\em KS-uncolourable}.)  Such assignments involve a {\em logical contradiction} given that the four outcomes of each measurement are mutually excusive and jointly exhaustive possibilities.

It follows that the sort of model that  a violation of inequality \eqref{ooo} 
rules out 
can already be ruled out {\em by logic alone}; no experiment is required. 
To put it another way, discovering that quantum theory and nature violate inequality \eqref{ooo}
only allows one to conclude that neither quantum theory nor nature involve a logical contradiction, which one presumably already knew prior to noting the violation.  

We have already argued that the notion of KS-noncontextuality, insofar as it assumes outcome-determinism, is not suitable for devising experimentally robust inequalities given that 
every real measurement involves some noise.
The problem with inequality \eqref{ooo} can also be traced back to the use of the assumption of KS-noncontextuality.  Suppose we ask the following question: given the existence of nine 
four-outcome measurements satisfying the operational 
equivalences of Fig.~\ref{mmtequivs}(a), how are the operational probabilities that are assigned to these measurement events constrained if we presume that KS-noncontextual assignments 
underlie the operational statistics?
On the face of it, the question seems well-posed.  On further reflection, however, one sees that it is not.  There are simply {\em no} KS-noncontextual assignments to these measurement events,
so it is simply impossible to imagine that such assignments could underlie the operational statistics.  There is nothing to be tested experimentally, as the hypothesis under consideration is 
seen to be false as a matter of logic.  

Here is another way to see that the inequality \eqref{ooo} does not provide a test of noncontextuality.  Consider the expectation value $\langle v_{\kappa_1} \oplus v_{\kappa_2} \oplus v_{\kappa_3} \oplus v_{\kappa_4} \rangle_{P}$
for a preparation $P$, where $\kappa_1$, $\kappa_2$, $\kappa_3$ and $\kappa_4$ correspond to the four outcomes of some measurement.   Regardless of which of the four outcomes of the measurement
occurs in a given run where preparation $P$ is implemented---i.e. regardless of whether $(v_{\kappa_1} ,v_{\kappa_2} ,v_{\kappa_3}, v_{\kappa_4})$ comes out as (1,0,0,0) or (0,1,0,0) or (0,0,1,0) 
or (0,0,0,1) in that run---the variable $v_{\kappa_1} \oplus v_{\kappa_2} \oplus v_{\kappa_3} \oplus v_{\kappa_4} $ has the value 1.  We can think of it this way: the variable $v_{\kappa_1} \oplus v_{\kappa_2} \oplus v_{\kappa_3} \oplus v_{\kappa_4} $
is a trivial variable because it is a constant function of the measurement outcome. (This is analogous to how, in quantum theory, for a four-outcome 
measurement associated with four projectors, although each projector is a nontrivial observable, their sum is the identity operator, which has expectation value 1 for all quantum states, and 
therefore corresponds to a trivial observable.)  It follows that regardless of what distribution over the four outcomes is assigned by $P$, the expectation value $\langle v_{\kappa_1} \oplus v_{\kappa_2} \oplus v_{\kappa_3} \oplus v_{\kappa_4} \rangle_P$
will be 1.  Given that each of the nine terms in $\alpha'(P)$ is of this form, it follows that $\alpha'(P)=9$.  

So, for {\em any} operational theory that admits of nine four-outcome measurements with the operational equivalence relations depicted in Fig.~\ref{mmtequivs}(a), we will find that 
$\alpha'(P)=9$ for all $P$.
Therefore,  we can conclude that the inequality $\alpha'(P)\le 8$ is violated for all $P$. One can reach this conclusion without ever considering the question of whether the operational 
predictions can be explained by some underlying noncontextual model.
  
Another consequence  of the triviality of the variables of the form $v_{\kappa_1} \oplus v_{\kappa_2} \oplus v_{\kappa_3} \oplus v_{\kappa_4} $ is that the inequality \eqref{ooo} can be violated
regardless of how noisy the measurements are.  Suppose, for instance, that quantum theory describes our experiment, but that the nine four-outcome measurements are not the projective 
measurements
described in Fig.~\ref{CEGAhypergraph}, but rather noisy versions thereof.
For instance, one can imagine that each measurement 
is associated with a positive operator-valued measure that is the image under a depolarizing map of the projector valued measure associated with the ideal measurement. The amount of 
depolarization can be taken arbitrarily large and, as long as it is the {\em same} amount of depolarization for each of the measurements, the nine noisy 
measurements that result will still satisfy precisely the same operational equivalences as the original nine, namely, those depicted in Fig.~\ref{mmtequivs}(a).
For such noisy measurements, we can still identify variables $v_{\kappa}$ associated to the eighteen equivalence classes of measurement events, and we still find that regardless of which of 
the four outcomes of the measurement occurs, the variable $v_{\kappa_1} \oplus v_{\kappa_2} \oplus v_{\kappa_3} \oplus v_{\kappa_4} $ has the value 1, so that
regardless of what distribution over the four outcomes is assigned by $P$, the expectation value $\langle v_{\kappa_1} \oplus v_{\kappa_2} \oplus v_{\kappa_3} \oplus v_{\kappa_4} \rangle_P$  
will be 1 and therefore $\alpha'(P)=9$, which is a violation of the inequality \eqref{ooo}.

According to the generalized notion of noncontextuality proposed in Ref.~\cite{genNC}, if one adds enough noise to the preparations and measurements in an experiment, it always becomes possible
to represent the experimental statistics by a noncontextual model. One way to prove this is to note that: (i)  if all of the preparations and the measurements in an experiment admit of positive
Wigner representations, then, as demonstrated in Ref.~\cite{Spe08}, the Wigner representation defines a noncontextual model, and (ii) if one adds enough noise to the preparations and 
measurements, it is possible to ensure that they admit of positive Wigner representations.

This analysis of the effect of noise accords with intuition: noncontextuality is meant to represent a notion of classicality, so that a failure of noncontextuality is only expected to occur in 
a quantum experiment if one's experimental operations have a high degree of coherence.
It follows that  there should always exist a threshold of noise above which an experiment cannot be used to demonstrate the failure of noncontextuality.   
One can turn this observation into a minimal criterion that should be satisfied by any noncontextuality inequality, that there should exist a threshold of experimental noise above which it 
cannot be violated. 

As we have just noted, the inequality proposed in Ref.~\cite{Cabelloexpt} fails this minimal criterion. By contrast, the noncontextuality inequality proposed 
in this chapter identifies such a threshold for the 18 ray construction:
the noise must be kept low enough that the average of the measurement predictabilities is above $5/6$. We have even provided an
analysis of such noise thresholds for quantum theory in the section on noise robustness of our noncontextuality inequality.

\subsection{Alternative interpretation}

The inequality proposed in Ref.~\cite{Cabelloexpt} 
can be given a different interpretation to the one we have just provided.
This  interpretation is more charitable in some ways, but it still does not vindicate the proposed inequality as delimiting the boundary of noncontextual models. 

The idea is to imagine that for each of the nine measurements, there are in fact {\em five} rather than four outcomes that are mutually exclusive and jointly exhaustive. 
Thus, in this interpretation, it is assumed that the hypergraph describing compatibility relations and operational equivalences is {\em not} the one of Fig.~\ref{mmtequivs}(a), 
but rather a modification wherein there are nine additional nodes---one additional node appended to each of the nine measurements---as depicted in Fig.~\ref{Subnormalization}(a).

\begin{figure}[t]
 \centering
 \includegraphics[scale=0.5]{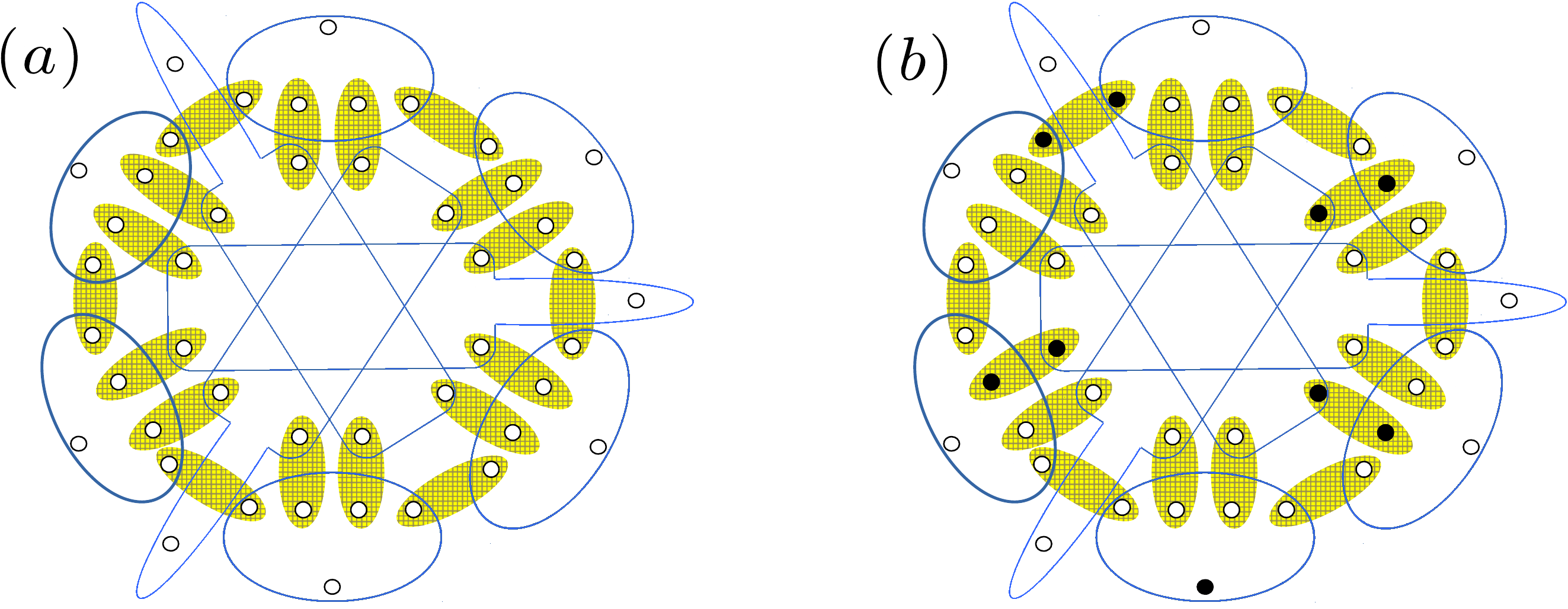}
 \caption{(a) The hypergraph wherein each measurement is assigned an additional fifth outcome.  (b) A normalized noncontextual deterministic assignment to the hypergraph of (a) that recovers 
 the subnormalized noncontextual deterministic assignment of Fig.~\ref{LogicDefyingAssignments}(a) on the appropriate subgraph.}
\label{Subnormalization}
\end{figure}

If $\{ \kappa_1, \kappa_2, \kappa_3, \kappa_4\}$  are the original four outcomes of a given measurement, then the variable $v_{\kappa_1} \oplus v_{\kappa_2} \oplus v_{\kappa_3} \oplus v_{\kappa_4} $ 
is no longer a constant function of the measurement outcome because its value varies depending on whether or not the fifth outcome occurs.  
If $\kappa_5$ denotes the fifth outcome of the measurement, 
then the trivial variable is $v_{\kappa_1} \oplus v_{\kappa_2} \oplus v_{\kappa_3} \oplus v_{\kappa_4} \oplus v_{\kappa_5}$,
taking the value 1 regardless of the outcome.

In this case, the assignments of the type depicted in Fig.~\ref{LogicDefyingAssignments}(a)---the noncontextual deterministic assignments that are {\em subnormalized}---can be embedded into 
noncontextual deterministic {\em normalized} assignments on the larger hypergraph, as depicted in Fig.~\ref{Subnormalization}(b).  (The possibility of such an embedding for the subnormalized 
noncontextual deterministic assignments considered in Cabello, Severini and Winter~\cite{CSW} was noted in Acin, Fritz, Leverrier, Sainz~\cite{AFLS}.)

Of course, such a move does not provide any way of understanding the deterministic noncontextual assignments of the type depicted in Fig.~\ref{LogicDefyingAssignments}(b), because the latter 
violate normalization by having the probabilities of the different outcomes of the measurement summing to greater than 1---they are {\em supernormalized}.  

So, while the supernormalized noncontextual deterministic assignments can be ruled out by logic alone, the subnormalized noncontextual deterministic assignments may be entertained without 
logical inconsistency if they are considered as reductions to a subgraph of a normalized noncontextual deterministic assignment on a larger hypergraph.

Because the justification given in Ref.~\cite{Cabelloexpt} for the inequality derived there asks one to consider {\em all} of the noncontextual deterministic assignments, including the 
supernormalized ones, the interpretation of this inequality as a constraint on subnormalized assignments is in tension with the manner in which the inequality is justified.
This interpretation is a better fit with Cabello's later work, such as Ref.~\cite{CSW}, wherein the restriction to subnormalized assignments is explicit. 
In any case,  if the inequality holds for {\em all} noncontextual deterministic assignments, regardless of normalization, then it holds for the special case of the subnormalized assignments, 
so the inequality can still be derived within this interpretation.

The fundamental problem with this subnormalization interpretation is the following:

While Cabello insists that the quantum model of the experiment involves complete sets of orthonormal rays, thereby implicitly assuming that the operational statistics has no nondetection 
events,\footnote{By ``nondetection'' events we mean measurement events where none of the outcomes of interest is detected.}
he then allows the response functions corresponding to each basis to be subnormalized in the ontological model, allowing nondetection events when none are seen operationally.
This inconsistency -- no nondetections in the operational theory, but nondetections in the ontological model -- is the reason Cabello's inequality and its violation don't signify a failure of 
KS-noncontextuality. Normalization (no nondetection) in the operational theory is equivalent to normalization (no nondetection) in
the ontological model: if the ontological model fails to take this into account, it is a priori ruled out, without even considering
the question of its KS-noncontextuality. To see this, note that $p(X|M,P)=\sum_{\lambda}\xi(X|M,\lambda)\mu(\lambda|P)$,
and $\sum_{X}p(X|M,P)=1$ if and only if $\sum_{X}\xi(X|M,\lambda)=1$ for all $\lambda$ such that $\mu(\lambda|P)>0$.

Cabello goes on to claim a state-independent violation of his inequality by all quantum states: this is trivially true because of the completeness of the bases chosen and because 
the scenario is restricted to a 4-dimensional Hilbert space. Any probability assignment induced by the Born rule will then obviously
violate Cabello's inequality: this says nothing about KS-noncontextuality.
As we have noted, in order for the subnormalized response functions to make sense, one has to allow the possibility that the 
operational theory admits nondetection events by, say, considering a 5-dimensional Hilbert space, where the fifth outcome is a 
nondetection event in each measurement (see Fig.~\ref{Subnormalization}). In this situation Cabello's inequality no longer admits a 
``state-independent'' violation: any quantum state that is orthogonal to the 4-dimensional subspace in which the 18 rays live will never violate the inequality, it will just give $\alpha'(P)=0$.
This means that in a 5-dimensional Hilbert space where subnormalization makes sense, 
violation of Cabello's inequality is not ``state-independent'', although the violation can be considered 
a signature of KS-contextuality for the extended hypergraph of Fig.~\ref{Subnormalization}(a).\footnote{The same cannot be said of the 
hypergraph of Fig.~\ref{mmtequivs}(a), which does not admit any KS-noncontextual assignments of values: Cabello's inequality is
inapplicable to this case. On the other hand, 
the extended hypergraph of Fig.~\ref{Subnormalization}(a) does admit KS-noncontextual assignments of values, which makes it possible
to think of Cabello's inequality as a test of KS-contextuality for {\em this} scenario.}

It is instructive to compare Cabello's inequality \cite{Cabelloexpt} with the KCBS inequality \cite{KCBS} 
to understand why the former is not a test of KS-noncontextuality for the 18 ray construction in 4 dimensions (Fig.~\ref{CEGAhypergraph}(a))
while the latter is 
a legitimate test of KS-noncontextuality in 3 dimensions. If one thinks of the five projectors in the Klyachko et al.~5-cycle \cite{KCBS}, 
it is the case that adjacent pairs do not add up to identity.\footnote{Contrast this with the fact that the four projectors in each measurement
add up to identity in the 18 ray construction.} Hence it is possible to have a nonzero 
nondetection probability operationally (corresponding to the projector orthogonal to the subspace defined by an adjacent pair of 
projectors), and it then makes sense to allow subnormalization in the ontological model
and write down the Kochen-Specker type inequality for this case in terms of the independence number \cite{CSW}.\footnote{The same approach does not make sense
in the 18 ray construction because nondetection events are {\em a priori} ruled out in a 4-dimensional Hilbert space.}

\section{Experimental testability}
Notwithstanding our operationalization of the Kochen-Specker theorem which lets us deal with the problem of noisy measurements, 
there is still a fundamental difficulty in experimentally testing operational noncontextuality inequalities: namely, the {\em problem of inexact operational equivalences}.
That is, in general, it will be the case that two experimental procedures (preparations or measurements) that are intended to be operationally equivalent turn out to be 
only approximately (or {\em inexactly}) so. Strictly speaking, we then have no justification for applying the assumption of noncontextuality in the absence of exact operational 
equivalence. Recall that noncontextuality is an inference from exact operational equivalence in the operational theory to exact ontological equivalence in the ontological model.
The equivalence in the ontological model is meant to be a natural explanation of why we see the operational equivalence we see in an experiment. But if such operational 
equivalences are not seen in the first place, then we have no reason to believe in any ontological equivalence.\footnote{
As regards the traditional account of the Kochen-Specker theorem,
the fact that one might in practice fail to measure the exact projector that a Kochen-Specker construction requires - by however small a margin - leads to a ``nullification'' of the theorem, 
meaning
that it's impossible to test the Kochen-Specker theorem in a real experiment with finite precision measurements.
This ``nullification'' is on account of the fact that failure to 
do infinitely precise measurements opens up the possibility of KS-noncontextual ontological models (called Meyer-Kent-Clifton (MKC) type models) of quantum statistics \cite{MKC}.
In our formulation of noncontextuality following Spekkens \cite{genNC}, we have circumvented the difficulty posed by MKC-type models by abandoning any reference to Hilbert spaces
(properties of which the MKC-type models exploit). 
What we are still left with, however, is the difficulty of achieving exact operational equivalences in real experiments.}

A second difficulty -- the {\em problem of universal quantifiers} -- with the experimental testability of operational 
noncontextuality 
inequalities is that the number of 
measurements (preparations) it might take to verify the operational equivalence of a pair of preparation (measurement) 
procedures
can be potentially infinite: this is because we judge operational equivalence of two preparation (measurement) procedures 
relative
to \emph{all} \footnote{Hence the reference to the \emph{universal quantifier} ``$\forall$'' in definitions of operational
equivalence.}
measurements (preparations) in the operational theory. Unless the operational theory admits of a \emph{finite} set of 
tomographically complete preparations and measurements, experimentally verifying the operational equivalence of two 
procedures would be practically impossible.

In this section we will show how the first difficulty -- the problem of inexact operational equivalences -- is solved by a 
convexity argument and the second 
difficulty -- the problem of universal quantifiers -- is tackled by assuming the tomographical completeness of a 
finite number of preparations and measurements. The experiment reported in Ref.~\cite{exptl} looks for violation of a simple
operational noncontextuality inequality, not directly related to the traditional Kochen-Specker theorem,
that we will derive here. The solution to the problem of inexact operational equivalences
that will be discussed here is largely due to Matt Pusey who first came up with the idea reported in Ref.~\cite{exptl}.\footnote{See also Ref.~\cite{MattP}.}
The methods that resolve these difficulties can be adapted to the test of any operational noncontextuality inequality.
In particular, these methods also apply to operational noncontextuality inequalities directly inspired by the Kochen-Specker theorem,
such as the one we have already derived in this chapter.

\subsection{Fair coin flip (FCF) inequality}
We will now derive a simple noncontextuality inequality that we call the fair coin flip (FCF) inequality.
\subsubsection{The setup}
Consider a measurement procedure $M_*$ (with outcomes $X\in\{0,1\}$) that is operationally indistinguishable from a fair coin 
flip, that is,

\begin{equation}\label{mnc1opeq}
 p(X=0|M_*,P)=p(X=1|M_*,P)=\frac{1}{2},\forall P\in\mathcal{P}.
\end{equation}
By the assumption of measurement noncontextuality -- namely, if it's operationally indistinguishable from a fair coin flip then it's
also ontologically indistinguishable from a fair coin flip -- we have for the response function
\begin{equation}\label{mnc1}
 \xi(X=0|M_*,\lambda)=\xi(X=1|M_*,\lambda)=\frac{1}{2},\forall \lambda\in\Lambda.
\end{equation}
Consider also three preparation procedures $P_1,P_2,P_3\in\mathcal{P}$ which are operationally indistinguishable, that is,
\begin{equation}\label{pnc1opeq}
 \forall M\in\mathcal{M}: p(X|M,P_1)=p(X|M,P_2)=p(X|M,P_3),\text{ where } X\in\{0,1\}.
\end{equation}
By the assumption of preparation noncontextuality -- namely, if the preparation procedures are operationally indistinguishable
then they are also ontologically indistinguishable -- we have for the associated distributions
\begin{equation}\label{pnc1}
 \forall \lambda\in\Lambda: \mu(\lambda|P_1)=\mu(\lambda|P_2)=\mu(\lambda|P_3).
\end{equation}
Suppose that $M_*$ can be realized as the uniform mixture of three binary-outcome measurements $M_1,M_2,M_3\in\mathcal{M}$:
that is uniformly randomly pick $t\in\{1,2,3\}$, implement $M_t$, and report
the outcome of $M_t$ as the outcome of $M_*$ (coarse graining over $t$).

Suppose also that each preparation procedure $P_t\in\mathcal{P}$ can be realized as the uniformly random 
mixture of two other preparation procedures $P_{t,0},P_{t,1}\in\mathcal{P}$.

Consider an experiment involving a measurement of $M_t$ on $P_{t,b}$ (with $b\in\{0,1\}$). We are interested in the average \emph{degree of correlation}
between the outcome $X$ of $M_t$ and preparation variable $b$ in $P_{t,b}$:
\begin{equation}\label{A'expr}
 A'\equiv\frac{1}{6}\sum_{t\in\{1,2,3\}}\sum_{b\in\{0,1\}}p(X=b|M_t,P_{t,b}).
\end{equation}
We will now show that the assumption of noncontextuality places a nontrivial bound on $A'$, namely, $A'\leq\frac{5}{6}$.

\subsubsection{The FCF inequality}
Let us first see why we can't have $A'=1$ in a noncontextual ontological model. In order to have $A'=1$, we require perfect correlation between $X$ and $b$ for each $M_t$ and $P_{t,b}$, i.e. 
each of the $6$ terms $p(X=b|M_t,P_{t,b})=1$. This means that any ontic state in the support of $P_{t,b}$ assigns deterministic 
outcome to $M_t$, for otherwise one cannot have $p(X=b|M_t,P_{t,b})=1$. We therefore have $p(X=0|M_t,P_{t,0})=p(X=1|M_t,P_{t,1})=1$. 
Since $P_t$ is an equal mixture of $P_{t,0}$ and $P_{t,1}$, it follows that all the ontic states in the support of $P_t$ yield deterministic 
outcomes for $M_t$. The operational equivalence of the three preparation procedures $P_t$ then implies that \emph{all} the ontic states relevant to the experiment 
(ontic support of any $P_t$) assign deterministic outcomes to the three measurements $M_t$. In turn, the operational equivalence between 
$M_*$ and a fair coin flip followed by the assumption of measurement noncontextuality applied to this equivalence requires that 
\begin{equation}\label{mnc11}
\forall \lambda\in\Lambda: \frac{1}{3}\xi(X=0|M_1,\lambda)+\frac{1}{3}\xi(X=0|M_2,\lambda)+\frac{1}{3}\xi(X=0|M_3,\lambda)=\frac{1}{2}.
\end{equation}
Since $M_t$ have deterministic response functions, this equation clearly cannot be satisfied, hence we arrive at a contradiction between $A'=1$ and the assumption of 
noncontextuality. To satisfy the condition from measurement noncontextuality at least one of the three response functions has to be indeterministic and that would reduce 
the degree of correlation to $A'<1$. A noncontextual ontological model of such an experiment therefore necessarily requires $A'<1$.

A formal proof of the noncontextuality inequality, $A'\leq\frac{5}{6}$, follows:
\begin{proof}
In the ontological models framework,
\begin{equation}
 A'=\sum_{\lambda\in\Lambda}\frac{1}{6}\sum_{t\in\{1,2,3\}}\sum_{b\in\{0,1\}}\xi(X=b|M_t,\lambda)\mu(\lambda|P_{t,b}).
\end{equation}

There is an upper bound on each response function that is independent of the value of $b$, namely,
\beq
\xi(X=b|M_t,\lambda)\le \eta(M_t,\lambda),
\eeq
where
\beq
\eta(M_t,\lambda)\equiv \max_{b'\in \{0,1\}} \xi(X=b'|M_t, \lambda).
\eeq
We therefore have
\beq
A'\le \frac{1}{3} \sum_{t\in \{ 1,2,3\}} \sum_{\lambda\in \Lambda}\eta(M_t,\lambda)\left( \frac{1}{2} \sum_{b\in\{0,1\}}  \mu(\lambda|P_{t,b})\right),
\eeq
Recalling that $P_t$ is an equal mixture of $P_{t,0}$ and $P_{t,1}$, so that
\beq
\mu(\lambda|P_{t})= \frac{1}{2} \mu(\lambda|P_{t,0}) + \frac{1}{2} \mu(\lambda|P_{t,1}),
\label{mixmu}
\eeq
we can rewrite the bound as simply
\beq
A'\le \frac{1}{3} \sum_{t\in \{ 1,2,3\}}  \sum_{\lambda\in \Lambda} \eta(M_t,\lambda) \mu(\lambda|P_{t}).
\label{Aexplicit2}
\eeq
But recalling Eq.~(\ref{pnc1}) (preparation noncontextuality),
\beq
\forall \lambda \in \Lambda: \mu(\lambda|P_1)=\mu(\lambda|P_2)=\mu(\lambda|P_3),
\label{NCimpl2supp}
\eeq
we see that the distribution $\mu(\lambda|P_{t})$ is independent of $t$, so we denote it by $\nu(\lambda)$ and rewrite the bound as
\beq
A'\le \sum_{\lambda\in \Lambda} \left( \frac{1}{3} \sum_{t\in \{ 1,2,3\}}   \eta(M_t,\lambda) \right)\nu(\lambda).
\eeq
This last step is the first use of noncontextuality in the proof because Eq.~\eqref{NCimpl2supp} is derived from preparation noncontextuality and the operational equivalence of Eq.~(\ref{pnc1opeq}). 
It then follows that
\beq
A'\le \max_{\lambda\in \Lambda} \left( \frac{1}{3} \sum_{t\in \{ 1,2,3\}}   \eta(M_t,\lambda) \right).
\label{Aexplicit3}
\eeq

Therefore, if we can provide a nontrivial upper bound on $\frac{1}{3} \sum_{t}  \eta(M_t,\lambda)$ for an arbitrary ontic state $\lambda$, we obtain a nontrivial upper bound on $A'$.
We infer constraints on the possibilities for the triple $(\eta(M_1,\lambda),\eta(M_2,\lambda), \eta(M_3,\lambda))$ from constraints on the possibilities for the triple $(\xi(X\text{=}0|M_1,\lambda),\xi(X\text{=}0|M_2,\lambda), \xi(X\text{=}0|M_3,\lambda))$.

The latter triple
is constrained by measurement noncontextuality as
\beq
\frac{1}{3}\sum_{t\in \{1,2,3\}} \xi(X\text{=}0|M_t,\lambda) = \frac{1}{2}.
\label{constraintsupp}
\eeq
This is the second use of noncontextuality in our proof, because Eq.~\eqref{constraintsupp} is derived from the operational 
equivalence of Eq.~(\ref{mnc1opeq}) and the assumption of measurement 
noncontextuality (Eqs.~(\ref{mnc1},\ref{mnc11})).

The fact that the range of each response function is $[0,1]$ implies that the vector
$$(\xi(X\text{=}0|M_1,\lambda),\allowbreak \xi(X\text{=}0|M_2,\lambda), \xi(X\text{=}0|M_3,\lambda))$$ is 
constrained to the unit cube. The linear constraint of Eq.~\eqref{constraintsupp} implies that these vectors are confined 
to a two-dimensional plane. The intersection of the plane and the cube 
defines the polygon depicted in Fig.~\ref{fcfpolytope}.
\begin{figure}
 \centering
 \includegraphics[scale=0.5]{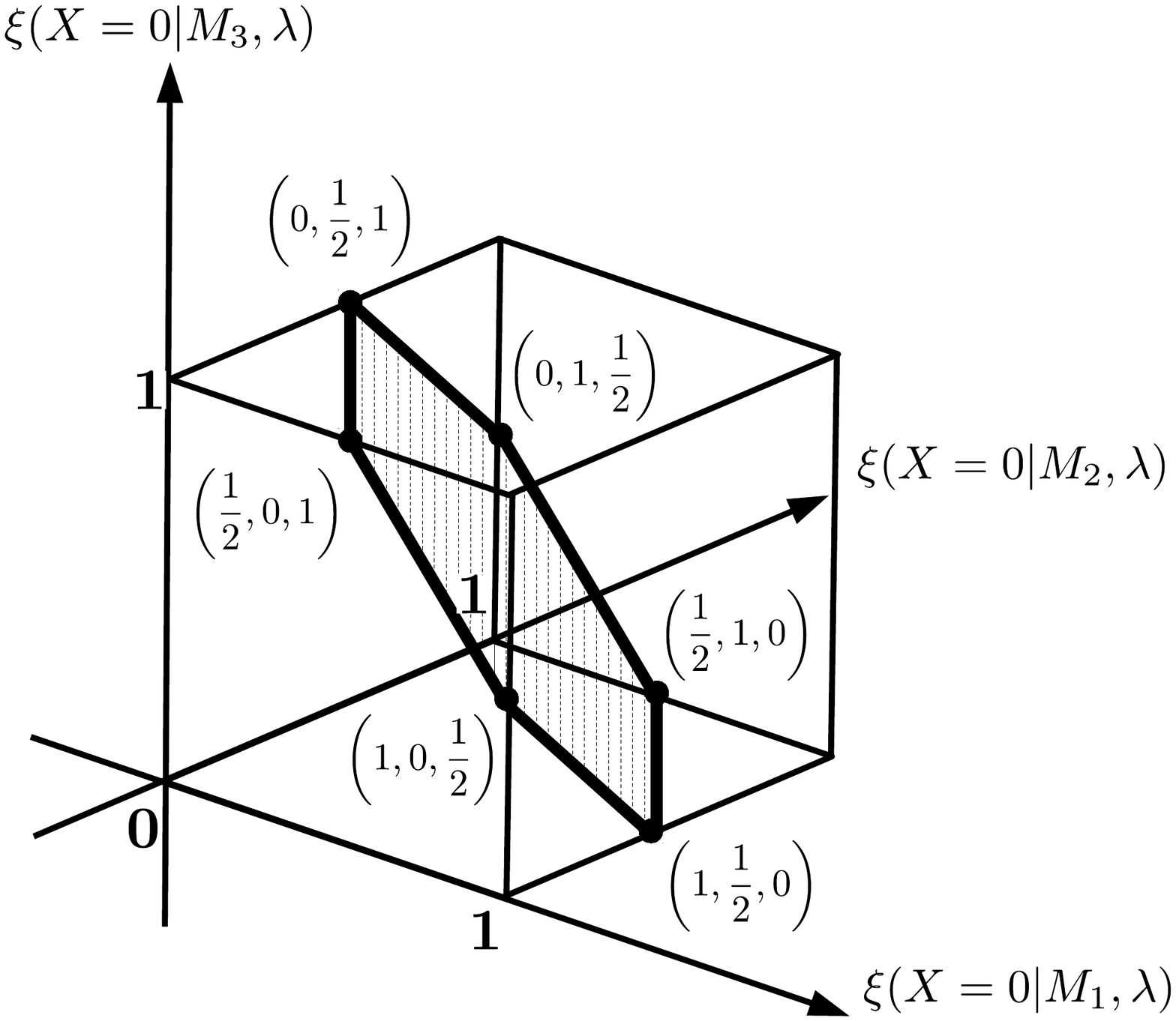}
 \caption{The possible values of $(\xi(X\text{=}0|M_1,\lambda),\xi(X\text{=}0|M_2,\lambda), \xi(X\text{=}0|M_3,\lambda))$.}
\label{fcfpolytope}
\end{figure}
The six vertices of this polygon have coordinates that are a permutation of $(1,\frac{1}{2},0)$.
For every $\lambda$, the vector $$(\xi(X\text{=}0|M_1,\lambda),\xi(X\text{=}0|M_2,\lambda), \xi(X\text{=}0|M_3,\lambda))$$ 
corresponds to a point in the convex hull of these extreme points and given that $\frac{1}{3} \sum_{t}\eta(M_t,\lambda)$
is a convex function of this vector, it suffices to find a bound on the value of this function at the extreme points.
If $\lambda$ corresponds to the extreme point 
$$(\xi(X\text{=}0|M_1,\lambda),\xi(X\text{=}0|M_2,\lambda), \xi(X\text{=}0|M_3,\lambda))=(1,\frac{1}{2},0),$$
then we have $(\eta(M_1,\lambda),\eta(M_2,\lambda), \eta(M_3,\lambda))= (1,\frac{1}{2},1)$, and the other extreme points are 
simply permutations thereof.  It follows that
\beq
\frac{1}{3} \sum_{t}  \eta(M_t,\lambda) \le \frac{5}{6}.
\eeq
Substituting this bound into Eq.~\eqref{Aexplicit3}, we have our result.
\end{proof}

\subsubsection{Tightness of bound: two ontological models}
We now provide an explicit example of a noncontextual ontological model
that saturates our noncontextuality inequality, thus
proving that the noncontextuality inequality is tight, i.e., the upper bound of the
inequality cannot be reduced any further for a noncontextual model.

We also
provide an example of an ontological model that is preparation noncontextual but fails to be measurement noncontextual (i.e. it is measurement {\em contextual}) and that exceeds the bound of 
our noncontextuality inequality.
This makes it clear that preparation noncontextuality alone does not suffice to justify the precise bound in our inequality, the assumption of measurement noncontextuality is a necessary 
ingredient as well. 

Note that there is no point inquiring about the bound for models that are measurement noncontextual but preparation contextual because, as shown in Ref.~\cite{genNC}, quantum theory 
admits of models of this type---the ontological model wherein the pure quantum states are the ontic states (the $\psi$-complete ontological model  in the terminology of Ref.~\cite{harriganspekkens}) is of this sort.

For the two ontological models we present, we begin by specifying the ontic state space $\Lambda$. These are depicted in Figs.~\ref{PNCMNC} and \ref{PNCMC} as pie charts with each slice 
corresponding to a different element of $\Lambda$. We specify the six preparations $P_{t,b}$ by the distributions over $\Lambda$ that they correspond to, denoted $\mu(\lambda|P_{t,b})$ 
(middle left of Figs.~\ref{PNCMNC} and \ref{PNCMC}). We specify the three measurements $M_t$ by the response functions for the $X=0$ outcome, denoted $\xi(0|M_t, \lambda)$ 
(top right of Figs.~\ref{PNCMNC} and \ref{PNCMC}). Finally, we compute the operational probabilities for the various preparation-measurement pairs, using $p(X|M,P)=\sum_{\lambda\in\Lambda}\xi(X|M,\lambda)\mu(\lambda|P)$,
and display the results in 
the $6\times 4$ upper-left-hand corner of Tables~\ref{tablepncmnc} and \ref{tablepncmc}.

\begin{figure}
 \centering
 \includegraphics[scale=0.5]{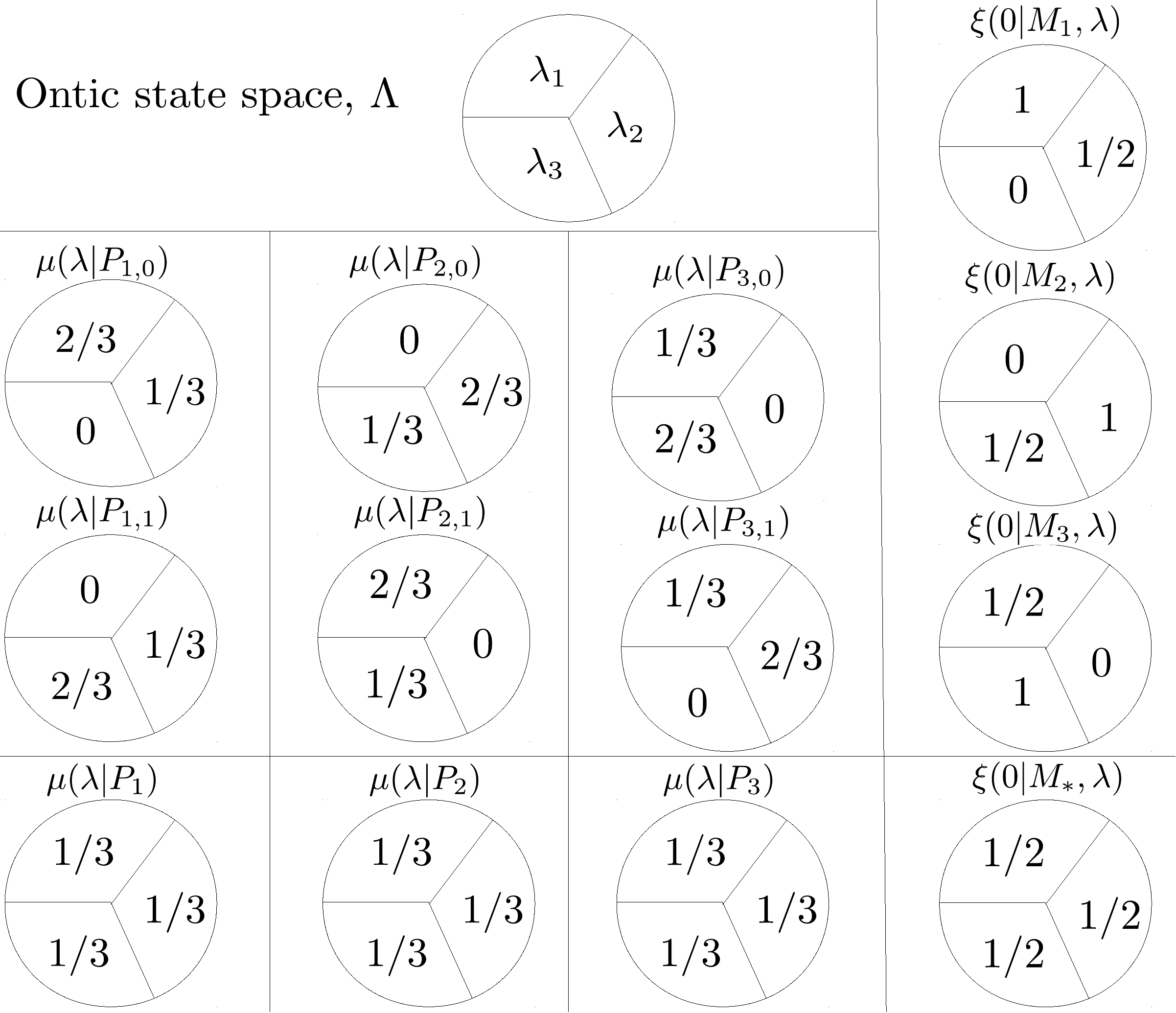}
 \caption{A noncontextual ontological model that saturates the noncontextal bound of our inequality, exhibiting that the bound is tight.}
\label{PNCMNC}
\end{figure}

\begin{table}
\begin{center}
    \begin{tabular}{ | c || c | c | c !{\vrule width 1pt} c | c |}
    \hline
    & $[0|M_1]$ & $[0|M_2]$ & $[0|M_3]$ & $[0|M_*]$ \\ \hline\hline%
    $P_{1,0}$ & \cellcolor{gray} $5/6$ & $1/3$ & $1/3$ & $1/2$ \\ \hline
    $P_{1,1}$ & \cellcolor{gray} $1/6$ & $2/3$ & $2/3$ & $1/2$ \\ \hline

    $P_{2,0}$ & $1/3$ & \cellcolor{gray} $5/6$ & $1/3$ & $1/2$ \\ \hline
    $P_{2,1}$ & $2/3$ & \cellcolor{gray} $1/6$ & $2/3$ & $1/2$ \\ \hline

    $P_{3,0}$ & $1/3$ & $1/3$ & \cellcolor{gray} $5/6$ & $1/2$ \\ \hline
    $P_{3,1}$ & $2/3$ & $2/3$ & \cellcolor{gray} $1/6$ & $1/2$ \\ \hlinewd{1pt}
    $P_1$     & $1/2$ & $1/2$ & $1/2$ & $1/2$ \\ \hline
    $P_2$     & $1/2$ & $1/2$ & $1/2$ & $1/2$ \\ \hline
    $P_3$     & $1/2$ & $1/2$ & $1/2$ & $1/2$ \\ \hline
    \end{tabular}
\end{center}
\caption{Operational statistics from the noncontextual ontological model of Fig.~\ref{PNCMNC}, achieving $A'=5/6$. The shaded
cells correspond to the ones relevant for calculating $A'$.}
\label{tablepncmnc}
\end{table}

\begin{figure}
 \centering
 \includegraphics[scale=0.5]{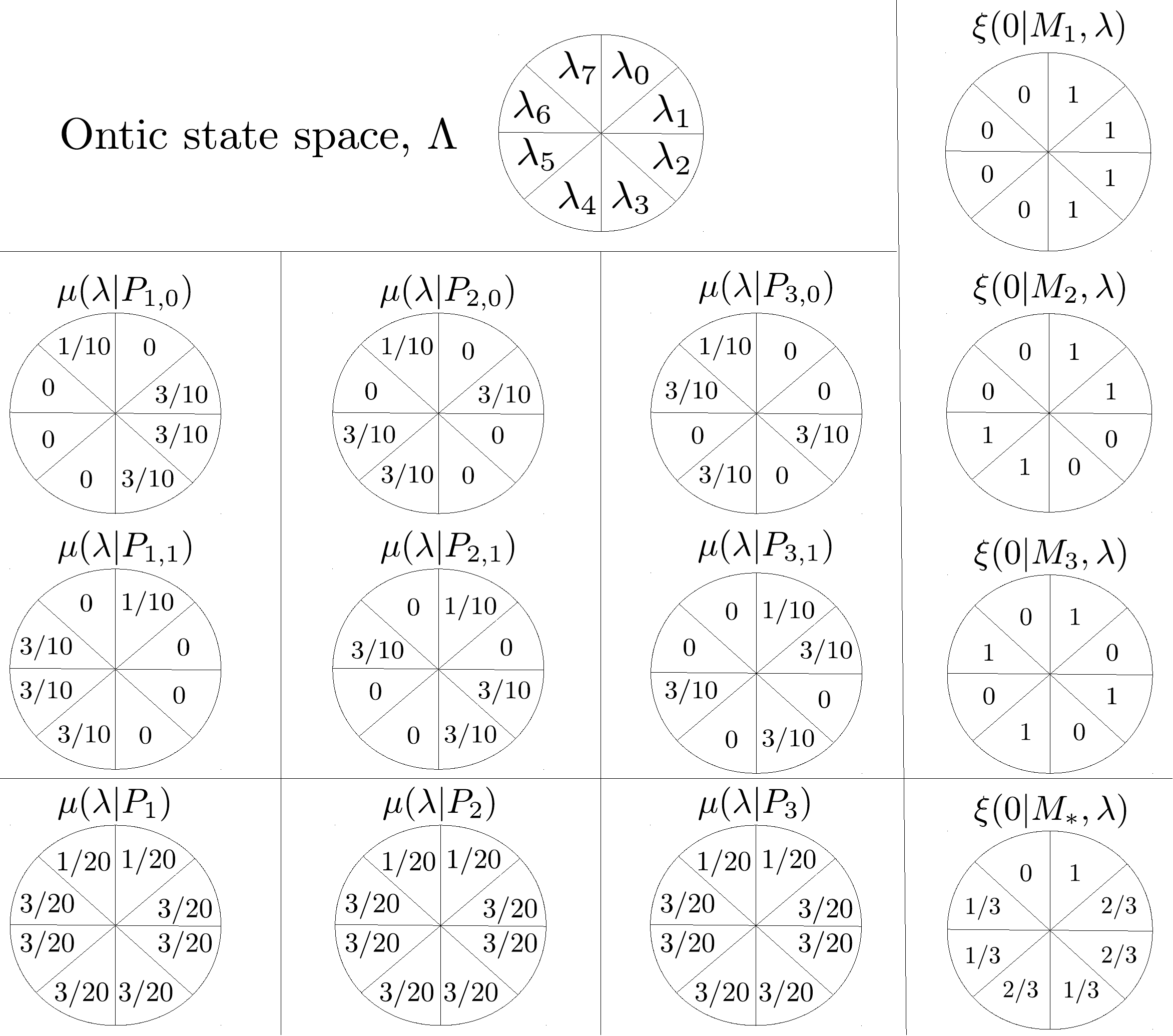}
 \caption{An ontological model that is preparation noncontextual but measurement contextual and that violates our inequality.}
\label{PNCMC}
\end{figure}

\begin{table}
\begin{center}
    \begin{tabular}{ | c || c | c | c !{\vrule width 1pt} c | c |}
    \hline
     & $[0|M_1]$ & $[0|M_2]$ & $[0|M_3]$ & $[0|M_*]$ \\ \hline\hline%
    $P_{1,0}$ & \cellcolor{gray}$9/10$ & $3/10$ & $3/10$ & $1/2$ \\ \hline
    $P_{1,1}$ & \cellcolor{gray}$1/10$ & $7/10$ & $7/10$ & $1/2$ \\ \hline
    $P_{2,0}$ & $3/10$ & \cellcolor{gray}$9/10$ & $3/10$ & $1/2$ \\ \hline
    $P_{2,1}$ & $7/10$ & \cellcolor{gray}$1/10$ & $7/10$ & $1/2$ \\ \hline
    $P_{3,0}$ & $3/10$ & $3/10$ & \cellcolor{gray}$9/10$ & $1/2$ \\ \hline
    $P_{3,1}$ & $7/10$ & $7/10$ & \cellcolor{gray}$1/10$ & $1/2$ \\ \hlinewd{1pt}
    $P_1$     & $1/2$ & $1/2$ & $1/2$ & $1/2$ \\ \hline
    $P_2$     & $1/2$ & $1/2$ & $1/2$ & $1/2$ \\ \hline
    $P_3$     & $1/2$ & $1/2$ & $1/2$ & $1/2$ \\ \hline
    \end{tabular}
\end{center}
\caption{Operational statistics from the preparation noncontextual and measurement contextual ontological model of Fig.~\ref{PNCMC}, achieving $A'=9/10$. The shaded cells
correspond to the ones relevant for calculating $A'$.}
\label{tablepncmc}
\end{table}

In the remainder of each table, we display the operational probabilities for the effective preparations, $P_t$, which are computed from the operational probabilities for the $P_{t,b}$ and the 
fact that $P_t$ is the uniform mixture of $P_{t,0}$ and $P_{t,1}$.  We also display the operational probabilities for the effective measurement $M_*$, which is computed from the operational
probabilities for the $M_t$ and the fact that $M_*$ is a uniform mixture of $M_1$, $M_2$ and $M_3$.

From the tables, we can verify that our two ontological models imply the operational equivalences that we use in the derivation of our noncontextuality inequality.  Specifically, the three 
preparations $P_1$, $P_2$ and $P_3$ yield exactly the same statistics for all of the measurements, and the measurement $M_*$ is indistinguishable from a fair coin flip for all the preparations.

Figs.~\ref{PNCMNC} and \ref{PNCMC} also depict $\mu(\lambda|P_t)$ for $t\in\{1,2,3\}$ for each model (bottom left).  These are determined from the $\mu(\lambda|P_{t,b})$ via Eq.~\eqref{mixmu}.
Similarly, the response function $\xi(0|M_*, \lambda)$, which is determined from $\xi(X=b|M_*,\lambda) = \frac{1}{3}\sum_{t\in \{1,2,3\}} \xi(X=b|M_t,\lambda)$, is displayed in each case 
(bottom right).

Given the operational equivalence of $P_1$, $P_2$ and $P_3$, an ontological model is preparation noncontextual if and only if $\mu(\lambda|P_1)=\mu(\lambda|P_2)=\mu(\lambda|P_3)$ for all 
$\lambda\in \Lambda$. We see, therefore, that both models are preparation noncontextual. Similarly given the operational equivalence of $M_*$ and a fair coin flip, an ontological model is 
measurement noncontextual if and only if $\xi(0|M_*, \lambda)=\frac{1}{2}$ for all $\lambda\in \Lambda$.  We see, therefore, that only the first model is measurement noncontextual.
Note that in the second model, $M_*$ manages to be operationally equivalent to a fair coin flip, despite the fact that when one conditions on a given ontic state $\lambda$, it does not have a
uniformly random response.  This is possible only because the set of distributions is restricted in scope, and the overlaps of these distributions with the response functions always generates 
the uniformly random outcome.  This highlights how an ontological model can do justice to the operational probabilities while failing to be noncontextual.

Finally, using the operational probabilities in the tables, one can compute the value of $A'$ for each model.  It is determined entirely by the operational probabilities in the shaded cells. 
One thereby confirms that $A'=\frac{5}{6}$ in the first model, while $A'=\frac{9}{10}$ in the second model.

\subsubsection{Quantum violation of the noncontextuality inequality}\label{quantumviolation}
Quantum theory predicts that there is a set of preparations and measurements on a qubit 
having the supposed properties and achieving $A'=1$, maximally violating the noncontextuality inequality. 
Take the $M_t$ to be 
represented by the observables $\vec{\sigma} \cdot \hat{n}_t$ where $\vec{\sigma}$ is the vector of Pauli operators and the unit vectors $\{ \hat{n}_1,\hat{n}_2, \hat{n}_3 \}$ are separated by 
120$^{\circ}$ in the $\hat{x}-\hat{z}$ plane of the Bloch sphere of qubit states. The $P_{t,b}$ are  the eigenstates of these observables, where we associate the positive 
eigenstate $ |\text{+}\hat{n}_t\rangle \langle +\hat{n}_t|$ with $b=0$.
To see that the statistical equivalence of Eq.~\eqref{mnc1opeq} is satisfied, it suffices to note that
\beq
 \frac{1}{3} |\text{+}\hat{n}_1\rangle \langle\text{+}\hat{n}_1 |+\frac{1}{3} |\text{+}\hat{n}_2\rangle \langle \text{+}\hat{n}_2 |+\frac{1}{3} |\text{+}\hat{n}_3\rangle \langle\text{+}\hat{n}_3 | =\frac{1}{2}\idn,
 \label{qopequiv1}
\eeq
and to recall that for any density operator $\rho$, ${\rm tr}(\rho \frac{1}{2}\idn)=\frac{1}{2}$.
To see that the statistical equivalence of Eq.~\eqref{pnc1opeq} is satisfied, it suffices to note that for all pairs $t,t'\in\{1,2,3\}$,
\beqa
&  \frac{1}{2} |\text{+}\hat{n}_t\rangle \langle\text{+}\hat{n}_t | +\frac{1}{2} |{-} \hat{n}_t\rangle \langle{-}\hat{n}_t |\nonumber\\
&=\frac{1}{2} |\text{+}\hat{n}_{t'}\rangle \langle\text{+}\hat{n}_{t'} | +\frac{1}{2} |{-}\hat{n}_{t'}\rangle \langle{-}\hat{n}_{t'} | \label{qopequiv2},
\eeqa
which asserts that the average density operator for each value of $t$ is the same, and therefore leads to precisely the same statistics for all measurements.
Finally, it is clear that the outcome of the measurement of $\vec{\sigma} \cdot \hat{n}_t$ is necessarily perfectly correlated with whether the state was $ |\text{+}\hat{n}_t\rangle \langle \text{+}\hat{n}_t |$ or  $|{-}\hat{n}_t\rangle \langle {-}\hat{n}_t |$, so that $A'=1$.

These quantum measurements and preparations are what we seek to implement experimentally, so we refer to them as {\em ideal}, and denote them by $M^{\rm i}_t$ and $P^{\rm i}_{t,b}$.
However, since we want our analysis of the data to be theory-independent, we will have to explicitly verify the relevant operational equivalences instead of taking them 
for granted as we do in quantum theory.

Note that our noncontextuality inequality can accommodate noise in both the measurements and the preparations, up to the point where the average of 
$p(X=b|M_t,P_{t,b})$ drops below $\frac{5}{6}$.  It is in this sense that our inequality does not presume the idealization of noiseless measurements.

\subsection{Tackling the problem of inexact operational equivalences: Secondary procedures}\label{secondarytrick}
Let us recall the problem of inexact operational equivalence: two experimental procedures are seen to be ``close'' to operationally equivalent but not quite exactly equivalent
and the assumption of noncontextuality then cannot, strictly speaking, be applied in the absence of exact operational equivalence. Our general 
solution to this problem with experimental tests of noncontextuality, first implemented explicitly for a particular test in the experiment of 
Ref.~\cite{exptl}, is described in the following steps:\footnote{The procedure outlined here is a direct generalization of the 
specific procedure used in Ref.~\cite{exptl}.}
\begin{enumerate}
 \item {\bf Raw data matrix:} Take the set of experimental procedures (preparations or measurements) that are implemented in the laboratory
and for which the raw data is collected. Tabulate the raw data in a matrix $\mathbb{D}^r$ with entries $\{f_{ij}\}$, where 
$f_{ij}$ is the fraction of $[0|M_i]$ outcomes that occur when binary-outcome measurement procedure $M_i$, with outcome set $\{0,1\}$,
is implemented following preparation procedure $P_j$:

\begin{equation}
 \mathbb{D}^r=\begin{pmatrix}
 f_{11} & f_{12} & \dots & f_{1n}\\
 f_{21} & f_{22} & \dots & f_{2n}\\
 \vdots\\
 f_{m1} & f_{m2} & \dots & f_{mn}
\end{pmatrix},
\end{equation}
where we consider an experiment in which $m$ binary-outcome measurements are implemented on each of $n$ preparations.
Each row denotes a measurement procedure while each column denotes a preparation procedure in this matrix of raw data.
We assume that the primary measurement and preparation procedures implemented in the laboratory to obtain this raw data
arise from a generalized probabilistic theory (GPT) in which $m_t$ ($\leq m$) of the $m$ measurements implemented are 
tomographically complete. This assumption is necessary because verifying the operational equivalence between
two preparations requires that we implement a tomographically complete set of measurements. 
We will obtain estimates of the probabilities associated with the primary procedures by fitting 
the raw data (which is obtained from a large but {\em finite} number of runs of the experiment) to this GPT.

In the experiment of Ref.~\cite{exptl}, we have $m=4$ and $n=8$, and we have assumed $m_t=3$.

\item {\bf Estimating probabilities (fitting raw data to GPT):} Assuming that a finite number, $m_t$, of independent binary-outcome measurements are tomographically complete, 
we fit the raw data in $\mathbb{D}^r$ to a matrix $\mathbb{D}^{\rm p}$ of estimated probabilities arising from a GPT with 
$m_t$ tomographically complete measurements. As with $\mathbb{D}^r$, the rows of $\mathbb{D}^{\rm p}$ denote the primary measurements
and the columns of $\mathbb{D}^{\rm p}$ denote the primary preparations. We will describe the fitting procedure later.
The entries of $\mathbb{D}^{\rm p}$ are given by $p_{ij}\equiv p(0|M_i^{\rm p},P_j^{\rm p})$, i.e. $p_{ij}$ are the estimated probabilities --
given the raw frequencies in $\mathbb{D}^r$ -- of outcome 
$[0|M_i^{\rm p}]$ of measurement procedure $M_i^{\rm p}$ following preparation procedure $P_j^{\rm p}$:
\begin{equation}
 \mathbb{D}^{\rm p}=\begin{pmatrix}
 p_{11} & p_{12} & \dots & p_{1n}\\
 p_{21} & p_{22} & \dots & p_{2n}\\
 \vdots\\
 p_{m1} & p_{m2} & \dots & p_{mn}
\end{pmatrix}.
\end{equation}

\item {\bf Failure of exact operational equivalences:} The primary measurement and preparation procedures will typically 
fail to satisfy -- by however small a margin -- the operational equivalences required in
a particular test of noncontextuality. Can one salvage anything interesting from the primary data despite this failure
of exact operational equivalences? Yes: by inferring secondary procedures.

\item {\bf Secondary procedures (restoring exact operational equivalences):} 
For the set of primary procedures, consider their convex hull, i.e. the set of all procedures that can be obtained from the primary set by 
probabilistically mixing the elements of the primary set. Since all these procedures constitute post-processing 
of the probabilities estimated from the raw data, we can define \emph{secondary procedures} corresponding to the \emph{primary procedures}
such that for each primary procedure relevant to the noncontextuality inequality we have a secondary procedure that is 
as close as possible to the primary procedure but which is chosen so that the operational equivalences required for testing noncontextuality are exactly satisfied for the 
chosen set of secondary procedures. Being able to find such a set of secondary procedures signals a good experiment that can meaningfully test noncontextuality: if such 
secondary procedures cannot be found, then the data will have to be discarded and the experiment repeated more carefully. This needs
to be done only for those procedures (preparations and \slash or measurements) which are required to satisfy some operational equivalences for the particular
test of noncontextuality.

In the case of preparations, the post-processing involves probabilistic mixtures of 
the primary preparation procedures (the columns of $\mathbb{D}^{\rm p}$, denoted by ${\bf P}^{\rm p}_k$ with $k=1,\dots,n$)
to define the secondary preparation procedures (the columns of
secondary data matrix $\mathbb{D}^{\rm s}$, denoted by ${\bf P}^{\rm s}_j$ with $j=1,\dots,n'$, $n'\leq n$):
\begin{equation}
{\bf P}^{\rm s}_j=\sum_{k=1}^n u_{jk}{\bf P}^{\rm p}_k, \text{ where }u_{jk}\geq 0\quad\forall j,k, \text{ and }\sum_ku_{jk}=1 \quad\forall j\in\{1,2,\dots,n'\}.
\end{equation}
To keep the secondary preparation procedures as close as possible to the primary ones, we maximize the function
\begin{equation}
 C_{\rm P}\equiv\frac{1}{n'}\sum_{j=1}^{n'} u_{jj},
\end{equation}
subject to the constraints imposed by the operational equivalences required between the preparation procedures
for the particular test of noncontextuality.

In the case of measurements, extremal\footnote{i.e. those which cannot be obtained from probabilistic
mixing of primary events $[0|M_l^{\rm p}]$.} post-processings
of primary events $[0|M_l^{\rm p}]$ also include: 1) relabelling the outcomes of a 
measurement $M_l^{\rm p}$ to define a flipped-outcome measurement $M_{\lnot l}^{\rm p}$ such that $[X|M_l^{\rm p}]\simeq[1-X|M_{\lnot l}^{\rm p}]$, $X\in\{0,1\}$,
2) the measurement event that is certain to occur (i.e. obtaining outcome `0' for all preparation procedures), denoted by a row vector {\bf1} 
with all probabilities $1$, 3) the measurement event that is certain {\em not} to occur (i.e. not obtaining outcome 
`0' for any of the preparation procedures or, equivalently, obtaining outcome `1' for all preparation procedures),
denoted by a row vector {\bf0} with all probabilities $0$. The secondary measurement procedures (the rows of secondary 
data matrix $\mathbb{D}^{\rm s}$, denoted by ${\bf M}^{\rm s}_i$ with $i=1,\dots,m'$, $m'\leq m$) 
are obtained by probabilistic mixing of all the primary events $[0|M_l^{\rm p}]$ (denoted by row vectors ${\bf M}^{\rm p}_l$ with $l=1,\dots,m$)
with their 
extremal post-processings $[0|M_{\lnot l}^{\rm p}]$ (denoted by row vectors ${\bf 1}-{\bf M}^{\rm p}_l$ with $l=1,\dots,m$), {\bf0}, and {\bf1} events:
\begin{equation}
 {\bf M}^{\rm s}_i=\sum_{l=1}^m v_{il}{\bf M}^{\rm p}_l+v_{i0}{\bf 0}+v_{i1}{\bf 1}+\sum_{l'=1}^m v_{i,\lnot l'}({\bf 1}-{\bf M}^{\rm p}_{l'}),
\end{equation}
where $\sum_{l=1}^m v_{il}+v_{i0}+v_{i1}+\sum_{l'=1}^m v_{i,\lnot l'}=1$ for all $i$ and all the weights $v_{il},v_{i,\lnot l'},v_{i0},v_{i1}$
are nonnegative. Note that the purpose of allowing mixtures of {\bf0}, {\bf1}, and  ${\bf 1}-{\bf M}^{\rm p}_l$ besides the 
original set of ${\bf M}^{\rm p}_l$ is to permit greater freedom in the choice of secondary measurement procedures 
that can satisfy all the required operational equivalences exactly and are also as close as possible to the primary 
set of measurement procedures. The latter requirement is motivated by the fact that an experimenter intends to 
implement primary procedures -- both preparations and measurements -- that maximize the violation of noncontextuality, hence
we want to infer secondary procedures as close to the intended ones as possible while at the same time satisfying the prerequisite
of exact operational equivalences for a test of noncontextuality.

To keep the secondary measurement procedures as close as possible to the primary ones, we maximize the function
\begin{equation}
 C_{\rm M}\equiv\frac{1}{m'}\sum_{i=1}^{m'} v_{ii},
\end{equation}
subject to the constraints imposed by the operational equivalences required between the measurement procedures
for the particular test of noncontextuality.

All in all, the entries $s_{ij}\equiv p(0|M^{\rm s}_i,P^{\rm s}_j)$ of the $m'\times n'$ secondary data matrix $\mathbb{D}^{\rm s}$ are given by:
\begin{equation}
s_{ij}=\sum_{k=1}^{n}u_{jk}\left[\sum_{l=1}^{m}v_{il}p_{lk}+v_{i0}0+v_{i1}1+\sum_{l'=1}^{m}v_{i,\lnot l'}\left(1-p_{l'k}\right)\right]
\end{equation}
These entries belonging to $\mathbb{D}^{\rm s}$ are the ones that are finally used in checking whether a particular noncontextuality inequality
is violated or not, since they satisfy the required operational equivalences exactly.

In the experiment of Ref.~\cite{exptl}, $m'=3$ and $n'=6$ and the relevant operational equivalences that need to be exactly satisfied
are those of Eqs.~(\ref{mnc1opeq}) and (\ref{pnc1opeq}).
\end{enumerate}
Note that in describing this algorithm for handling failure of operational equivalence, we have not restricted ourselves to the 
FCF inequality, although the only known experiment (so far) where this algorithm has been implemented tests the FCF inequality \cite{exptl}. 
This algorithm should work, in principle, for \emph{any} test of a noncontextuality inequality including the one 
inspired by the Kochen-Specker theorem derived earlier in this chapter. It is our hope that describing this algorithm in general
terms will help carry out robust experimental tests of noncontextuality in the spirit of the experiment of Ref.~\cite{exptl}.

\subsubsection{Estimating $\mathbb{D}^{\rm p}$: fitting raw data to GPT}
To fit the raw data matrix $\mathbb{D}^r$ to a matrix $\mathbb{D}^{\rm p}$ arising from a GPT, we first need to characterize the set of 
$\mathbb{D}^{\rm p}$ that can arise from a GPT with a specified number of tomographically complete effects.
The following theorem characterizes such $\mathbb{D}^{\rm p}$:

\begin{theorem}\label{Dptheorem}
 A matrix $\mathbb{D}^{\rm p}$ with rows corresponding to effects $$\left\{[0|M_1],[0|M_2],\dots,[0|M_m]\right\}$$
 can arise from a GPT in which $m_t=|\mathcal{M}_{\rm tomo}|<m$ of these effects, say $$\mathcal{M}_{\rm tomo}\equiv\left\{[0|M_1],[0|M_2],\dots,[0|M_{m_t}]\right\},$$
 are tomographically
 complete if and (with a measure zero set of exceptions) only if 
 \begin{equation}\label{Dp}
 \forall c\in\{m_t+1,m_t+2,\dots,m\}: \sum_{i=1}^{m_t}a_{ci}p_{ij}+a_cp_{cj}-1=0\quad\forall j\in\{1,2,\dots,n\},
 \end{equation}
 for some real constants $\{a_{c1},a_{c2},\dots,a_{cm_t},a_c\}_{c=m_t+1}^m$. This set of $m-m_t$ (one for each value of $c$) linear
 constraints characterize the set of $m\times n$ matrices $\mathbb{D}^{\rm p}$ arising from a GPT with $m_t$ tomographically complete effects.
\end{theorem}

\begin{proof}
 {\bf ``only if'':} We assume that $\mathbb{D}^{\rm p}$ belongs to a GPT with $m_t=|\mathcal{M}_{\rm tomo}|$ tomographically complete measurements.
 Let $\mathcal{M}_{\rm tomo}=\{M_{1}^{\rm p},M_{2}^{\rm p},\dots,M_{m_t}^{\rm p}\}$ be such a set of \emph{fiducial} measurements tomographically complete 
 for a system, so that the state of the system given a preparation procedure $P$ is specified by the column vector
 \begin{equation}
  {\bf p}_P=\begin{pmatrix}
           1\\
           p(0|M_{1}^{\rm p},P)\\
           p(0|M_{2}^{\rm p},P)\\
           \vdots\\
           p(0|M_{m_t}^{\rm p},P)
          \end{pmatrix},
 \end{equation}
where the first entry ensures normalization of the state. As shown in Refs.~\cite{hardy5axioms,barrettgpt}, convexity then requires that
the probability of outcome `0' for any measurement $M\in\mathcal{M}$ is given by ${\bf r}_M.{\bf p}_P$ for some vector ${\bf r}_M$.
Let $\{{\bf r}_1,{\bf r}_2,\dots,{\bf r}_m\}$ denote the vectors corresponding to outcome `0' of 
the measurements $\{M_1^{\rm p},M_2^{\rm p},\dots,M_m^{\rm p}\}$ respectively. The measurement event that \emph{always} occurs (e.g. the event of obtaining
outcome `0' or `1' in any binary-outcome measurement) must be represented by ${\bf r}_{\mathbb{I}}=(1,0,0,\dots,0)\in\mathbb{R}^{m_t+1}$ 
(so that ${\bf r}_{\mathbb{I}}.{\bf p}_P=1$ for 
all $P\in\mathcal{P}$).
 
Note that in order to have an experiment that allows for operational equivalences to be verified, the number of measurements carried
out in the experiment, $m$, must be at least as large as the cardinality of the tomographically complete set of measurements, $m_t$,
i.e. $m\geq m_t$. We will therefore presume that our experimental test of noncontextuality has $m>m_t$ in what follows.
Since $\{\{{\bf r}_i\}_{i=1}^m,{\bf r}_{\mathbb{I}}\}$ is a set of $m+1$ vectors in $\mathbb{R}^{m_t+1}$, any $m_t+2$ element 
subset of them, say $\{{\bf r}_1,\dots,{\bf r}_{m_t},{\bf r}_{\mathbb{I}},{\bf r}_c\}$, must necessarily be linearly dependent:
\begin{equation}
 \sum_{i=1}^{m_t}a'_{ci}{\bf r}_i+a'_c{\bf r}_c+e'_c{\bf r}_{\mathbb{I}}=0,
\end{equation}
with $(a'_{c1},a'_{c2},\dots,a'_{cm_t},a'_c,e'_c)\neq(0,0,\dots,0,0,0)$. The set of $\{\{{\bf r}_i\}_{i=1}^{m_t},{\bf r}_c, {\bf r}_{\mathbb{I}}\}$
for which $e'_c=0$ are those where ${\bf r}_{\mathbb{I}}$ is not in the span of $\{\{{\bf r}_i\}_{i=1}^{m_t},{\bf r}_c\}$, which is a 
set of measure zero. To see this, note that the set of vectors $\{{\bf r}_{\mathbb{I}},{\bf r}_{1},{\bf r}_{2},\dots,{\bf r}_{m_t}\}$ forms an
orthonormal basis for $\mathbb{R}^{m_t+1}$, where ${\bf r}_{i}=(0,\dots,1,\dots,0)$ ($i+1$th entry 1, rest 0)
represents the measurement $M_{i}^{\rm p}$ from the fiducial set. Any ${\bf r}_i$ should be expressible as a linear combination of 
$\{\{{\bf r}_{i}\}_{i=1}^{m_t},{\bf r}_{\mathbb{I}}\}$. Restricting the set $\{\{{\bf r}_i\}_{i=1}^{m_t},{\bf r}_c\}$ to the 
subspace orthogonal to ${\bf r}_{\mathbb{I}}$ --- so that $e'_c=0$ and ${\bf r}_{\mathbb{I}}$ is not in the span of 
$\{\{{\bf r}_i\}_{i=1}^{m_t},{\bf r}_c\}$ --- means that the
$\{\{{\bf r}_i\}_{i=1}^{m_t},{\bf r}_c\}$ lie in a one dimension lower subspace $\mathbb{R}^{m_t}$ of $\mathbb{R}^{m_t+1}$, 
hence they form a measure
zero set. Therefore, we can generically ensure $e'_c\neq0$ and divide throughout by $-e'_c$ to obtain
\begin{equation}
\sum_{i=1}^{m_t}a_{ci}{\bf r}_i+a_c{\bf r}_c-{\bf r}_{\mathbb{I}}=0,\text{ where } a_c=-a'_c/e'_c, a_{ci}=-a'_{ci}/e'_c \quad\forall i.
\end{equation}
Let ${\bf p}_j$ denote the vector representing preparation $P_j^{\rm p}$, then we have $p_{ij}={\bf r}_i.{\bf p}_j$, where 
$p_{ij}=p(0|M^{\rm p}_i,P^{\rm p}_j)$ is the $(i,j)$th entry of $\mathbb{D}^{\rm p}$. Taking dot product with ${\bf p}_j$, we then have
\begin{equation}
 \sum_{i=1}^{m_t}a_{ci}p_{ij}+a_cp_{cj}-1=0
\end{equation}
for some real constants $\{a_{c1},a_{c2},\dots,a_{cm_t},a_c\}$, for every $c\in\{m_t+1,\dots,m\}$. We now prove the converse.

{\bf ``if'': } Since we require some subset of $\{M_1^{\rm p},M_2^{\rm p},\dots,M_m^{\rm p}\}$ to be a fiducial set in order to be able to verify
operational equivalences for a test of noncontextuality (recall that $m>m_t$ is a necessary
condition for this), we take $\{M_1^{\rm p},M_2^{\rm p},\dots,M_{m_t}^{\rm p}\}$ to be a fiducial set of measurements as before,
so that preparation procedure $P_j^{\rm p}$ corresponds to the vector
\begin{equation}
 {\bf p}_j=\begin{pmatrix}
            1\\
            p_{1j}\\
            p_{2j}\\
            \vdots\\
            p_{m_tj}
           \end{pmatrix},
\end{equation}
where $p_{ij}=p(0|M^{\rm p}_i,P^{\rm p}_j)$.

Now we can recover $\mathbb{D}^{\rm p}$ if $M_i^{\rm p}$ ($i=1,\dots,m_t$) is represented by ${\bf r}_i=(0,\dots,0,1,0,\dots,0)$ ($(i+1)$th entry 1 and the rest 0),
so that $p_{ij}={\bf r}_i.{\bf p}_j$ for $i\in\{1,\dots,m_t\}$ and, since we are given $\sum_{i=1}^{m_t}a_{ci}p_{ij}+a_cp_{cj}-1=0$
for every $c\in\{m_t+1,\dots,m\}$, we must have
\begin{equation}
\sum_{i=1}^{m_t}a_{ci}p_{ij}+a_cp_{cj}-1=0\Leftrightarrow\left(\sum_{i=1}^{m_t}a_{ci}{\bf r}_i+a_c{\bf r}_c-{\bf r}_{\mathbb{I}}\right).{\bf p}_j=0.
\end{equation}
That is
\begin{equation}
a_c{\bf r}_c.{\bf p}_j=-\left(\sum_{i=1}^{m_t}a_{ci}{\bf r}_i-{\bf r}_{\mathbb{I}}\right).{\bf p}_j,
\end{equation}
so that we can take
\begin{equation}
 {\bf r}_c=(1/a_c,-a_{c1}/a_c,-a_{c2}/a_c,\dots,-a_{cm_t}/a_c)
\end{equation}
for every $c\in\{m_t+1,\dots,m\}$ and thus reconstruct the $m\times n$ matrix $\mathbb{D}^{\rm p}$.
This proves the converse. We therefore have a characterization of $\mathbb{D}^{\rm p}$.
\end{proof}
Geometrically, this theorem means that the $n$ columns of $\mathbb{D}^{\rm p}$ -- each represented by a vector 
${\bf p}_j\in\mathbb{R}^{m_t+1}$ -- lie in the common intersection of
$m_t$-dimensional hyperplanes ($m-m_t$ of them) defined
by the constants $$\{a_{c1},a_{c2},\dots,a_{cm_t},a_c\}_{c=m_t+1}^m,$$ where $c\in\{m_t+1,\dots,m\}$ can be thought of as 
a label for the hyperplanes.

{\bf Note on the number of tomographically complete preparations and measurements:}
Note that if we supplement $\mathbb{D}^{\rm p}$ with an additional row of all 1s,\footnote{This additional row corresponds to 
the trivial measurement that always yields $p(0|M^{\rm p}_{\mathbb{I}},P)=1$ for any preparation $P\in\mathcal{P}$ and is
represented in the GPT by ${\bf r}_{\mathbb{I}}=(1,0,\dots,0)\in\mathbb{R}^{m_t+1}$.} then we have the rank of the ``new'' 
matrix $$\mathbb{D}^{\rm p}_{\rm new}=\begin{pmatrix}
                                1& 1& \dots &1\\
                                p_{11}& p_{12}&\dots &p_{1n}\\
                                \vdots\\
                                p_{m1}& p_{m2}&\dots &p_{mn}                               
                               \end{pmatrix}$$ given by $m_t+1$.
This is easy to see from Eq.~\eqref{Dp}, which can be rewritten in terms of the rows of $\mathbb{D}^{\rm p}_{\rm new}$ as 
\begin{equation}
\forall c\in\{m_t+1,\dots,m\}: \sum_{i=1}^{m_t}a_{ci}{\bf}R_i+a_cR_c-R_{\mathbb{I}}=0,
\end{equation}
where the first row of $\mathbb{D}^{\rm p}_{\rm new}$ is denoted by 
$R_{\mathbb{I}}=(1,1,\dots,1)=({\bf r}_{\mathbb{I}}.{\bf p}_1,{\bf r}_{\mathbb{I}}.{\bf p}_2,\dots,{\bf r}_{\mathbb{I}}.{\bf p}_n)$,
the next $m_t$ rows are denoted by $R_i=\left(p_{ij}\right)_{j=1}^n=\left({\bf r}_i.{\bf p}_j\right)_{j=1}^n$ for $i\in\{1,2,\dots,m_t\}$,
and the remaining rows are denoted by $R_c=\left(p_{cj}\right)_{j=1}^n=\left({\bf r}_c.{\bf p}_j\right)_{j=1}^n$ for 
$c\in\{m_t+1,m_t+2,\dots,m\}$. That is, any row $R_c$ can be expressed as a linear combination of the rows $\{\{R_i\}_{i=1}^{m_t},R_{\mathbb{I}}\}$.
Hence the row rank of $\mathbb{D}^{\rm p}_{\rm new}$ is no more than $m_t+1$.
We also know from the proof characterizing $\mathbb{D}^{\rm p}$ that $\{\{{\bf r}_i\}_{i=1}^{m_t},{\bf r}_{\mathbb{I}}\}$ are linearly 
independent in $\mathbb{R}^{m_t+1}$, hence the corresponding rows $\{\{R_i\}_{i=1}^{m_t},R_{\mathbb{I}}\}$ in $\mathbb{D}^{\rm p}_{\rm new}$
are also linearly independent. Therefore the rank of $\mathbb{D}^{\rm p}_{\rm new}$ is indeed exactly $m_t+1$.

Since row rank = column rank for any matrix (in particular, $\mathbb{D}^{\rm p}_{\rm new}$),
we have $m_t+1$ preparations that are tomographically complete for the measurements in the GPT 
and the number of nontrivial effects that are tomographically complete is $m_t$ 
(excluding the trivial effect -- corresponding to doing nothing to the system -- represented by ${\bf r}_{\mathbb{I}}$ 
in the GPT 
and corresponding to the row $R_{\mathbb{I}}$). Thus, as we discussed in Chapter 1, for a qubit in quantum theory,
the trivial effect (corresponding to identity) is not counted in the list of tomographically complete effects because the 
quantum state is normalized to trace 1 and only requires three real parameters to be specified. On the other hand, 
a qubit effect
requires $4$ real parameters to be specified.

{\bf Limitations of the tomographic completeness assumption:} 
Ideally, we would like an experiment to involve as few measurements as possible in order to reduce the complexity involved
in carrying it out as well as in the data analysis: the minimum $m$, as we have noted, is $m_t+1$ if we want to be able to 
verify
operational equivalences, besides gathering evidence that $m_t$ effects are indeed tomographically complete.\footnote{Indicated
by a good fit of data from $m>m_t$ effects carried out in an experiment to a GPT with $m_t$ tomographically complete effects.
If we only had data from $m=m_t$ effects, then the fit can always be done and there is no opportunity for the experiment to 
hint at a failure of the assumption of tomographic completeness.}
In the experiment of Ref.~\cite{exptl}, the number of effects $m=4$ and the GPT has $m_t=3$ tomographically complete effects.

At the same time, the more the number of measurements carried out ($m>m_t$), the more are the opportunities 
($m-m_t$ linear constraints)
for a failure of the assumption of tomographic completeness when trying to fit the raw data to GPT. For example, 
for $m=m_t$, such a fit
can always be done for the experimental data, for $m=m_t+1$ it is a matter of experimental data being consistent 
with a GPT with $m_t$
tomographically complete effects (one linear constraint), and for $m>m_t$ the experimental data has to contend with
$m-m_t$ linear constraints.
Inability to find a good fit to GPT could be on account of 
either a bad experiment or a failure of the assumption of tomographic completeness required to perform the fit.
However, inability to find a good fit {\em despite careful and repeated experiments} could well be on account of
a failure of assuming fewer tomographically complete measurements than are necessary to explain the data from the experiment.

Justifying the assumption of tomographic completeness of some finite set of effects is therefore a crucial step towards 
robust experimental tests of noncontextuality. In principle, it is always possible that some measurement that 
we did not do in a noncontextuality experiment would have demonstrated a failure of tomographic completeness: 
there is no way to
{\em definitively} rule out the existence of such a hypothetical measurement on purely operational grounds: 
all we can do is gather 
evidence by fitting as many measurements as possible to the GPT. On the other hand, if we were assuming quantum 
theory,\footnote{In
particular, that the system under investigation is a quantum system of a known dimension $d$.} we would not have to 
worry about this assumption because then the structure of the theory specifies the cardinality of the tomographically 
complete set
and insofar as we believe the experiment is well-modelled by quantum theory, we can definitively rule out a failure of 
tomographic completeness (since that would imply a failure of quantum theory).

We have hinted at ways in which the assumption of tomographic completeness can be seen to fail in 
an experiment but more exhaustive tests of this assumption would be crucial in making tests of noncontextuality as free of 
this ``failure of tomographic completeness'' loophole as possible.

{\bf The fitting procedure: } We have to fit $m-m_t$ hyperplanes
to the $n$ points that make up the columns of 
$\mathbb{D}^r$. Each column of $\mathbb{D}^r$ is then mapped to its closest point on the hyperplanes and these points 
make up the columns of $\mathbb{D}^{\rm p}$.

Following Ref.~\cite{exptl}, we can perform a weighted total least-squares fit as follows:
\begin{enumerate}
 \item Let use denote by $\triangle f_{ij}$ the uncertainty (say, due to statistical errors in the experiment) associated with
 the corresponding element $f_{ij}$ of $\mathbb{D}^r$.
 The weighted distance, $\chi_j$, between the $j$th columns of $\mathbb{D}^r$ and $\mathbb{D}^{\rm p}$ is then defined as 
 \begin{equation}
  \chi_j\equiv\sqrt{\sum_{i=1}^m\frac{\left(f_{ij}-p_{ij}\right)^2}{\left(\Delta f_{ij}\right)^2}}
 \end{equation}
 \item Finding the hyperplanes-of-best-fit is then the following minimization problem
 \begin{equation}
\begin{aligned}
& \underset{\{\{p_{ij},a_{ci}\}_{i=1}^{m_t},p_{cj},a_c\}_{c=m_t+1}^m}{\text{minimize}}
& & \chi^2=\sum_{j=1}^n\chi_j^2 \\
& \text{ subject to }
& &  \forall c\in\{m_t+1,\dots,m\}: \sum_{i=1}^{m_t}a_{ci}p_{ij}+a_cp_{cj}-1=0 \forall j\in\{1,\dots,n\}.
\end{aligned}
\end{equation}
\end{enumerate}

\subsection{Applying the method of inferring secondary procedures to the case of FCF inequality}
Let us illustrate the general procedure outlined in the last subsection with the example of the FCF inequality.
The experimental test of the FCF inequality reported in Ref.~\cite{exptl} implemented this procedure. 
Fig.~\ref{fg:setup} shows a schematic of the experimental setup that was used in Ref.~\cite{exptl}.
We refer to Ref.~\cite{exptl} for details of the experiment -- which uses the polarization of single photons -- 
and focus here solely on the theoretical ideas at play.

\begin{figure}[htb!]
  \centering
  \includegraphics[width=0.9\textwidth]{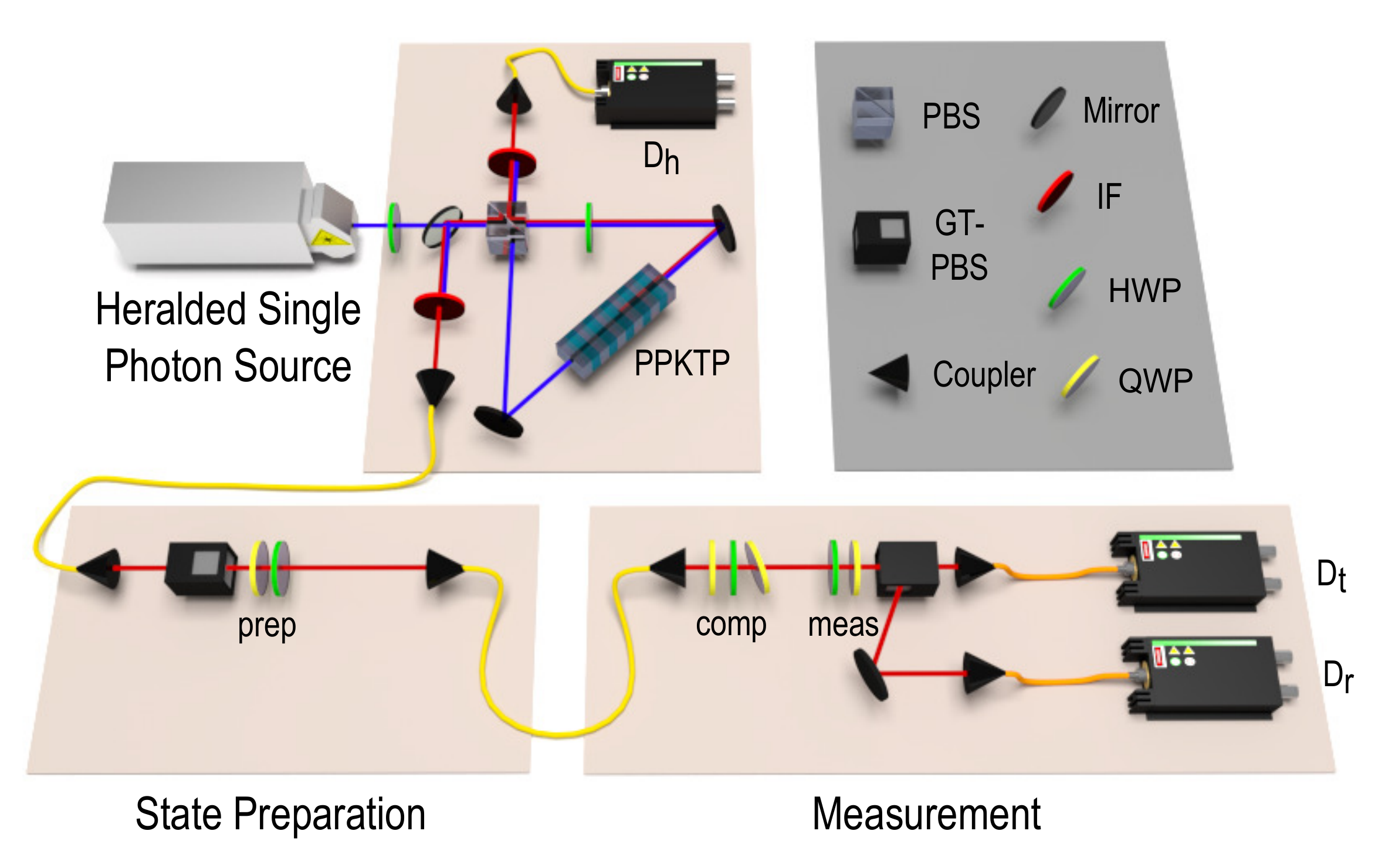}
  \caption{The experimental setup. Polarization-separable photon pairs are created via parametric downconversion, and 
  detection of a photon at $D_h$ heralds the presence of a single photon. The polarization state of this photon is prepared 
  with a polarizer and two waveplates (prep). A single-mode fibre is a spatial filter that decouples beam deflections caused 
  by the state-preparation \text{and measurement} waveplates from the coupling efficiency into the detectors. Three waveplates
  (comp) are set to undo the polarization rotation caused by the fibre. Two waveplates (meas), a polarizing beamsplitter, and 
  detectors $D_r$ and $D_t$ perform a two-outcome measurement on the state. PPKTP, periodically poled potassium titanyl 
  phosphate; PBS, polarizing beamsplitter; GT-PBS, Glan-Taylor polarizing beamsplitter; IF, interference filter; HWP,
  half-waveplate; QWP, quarter-waveplate. Figure credit: Kevin Resch \cite{exptl}.}
  \label{fg:setup}
\end{figure}

\subsubsection{Secondary preparations in quantum theory}
It is easiest to describe the details of our procedure for defining secondary preparations if we make the assumption that 
quantum theory correctly describes the 
experiment.\footnote{Recall that we have identified the {\em ideal} set of quantum states and measurements that lead to maximal
violation of the FCF inequality in Section \ref{quantumviolation}.}
Further on, we will describe the procedure for a generalised probabilistic theory (GPT) for this experiment.

\begin{figure}
  \centering
  \includegraphics[width=1.0\textwidth]{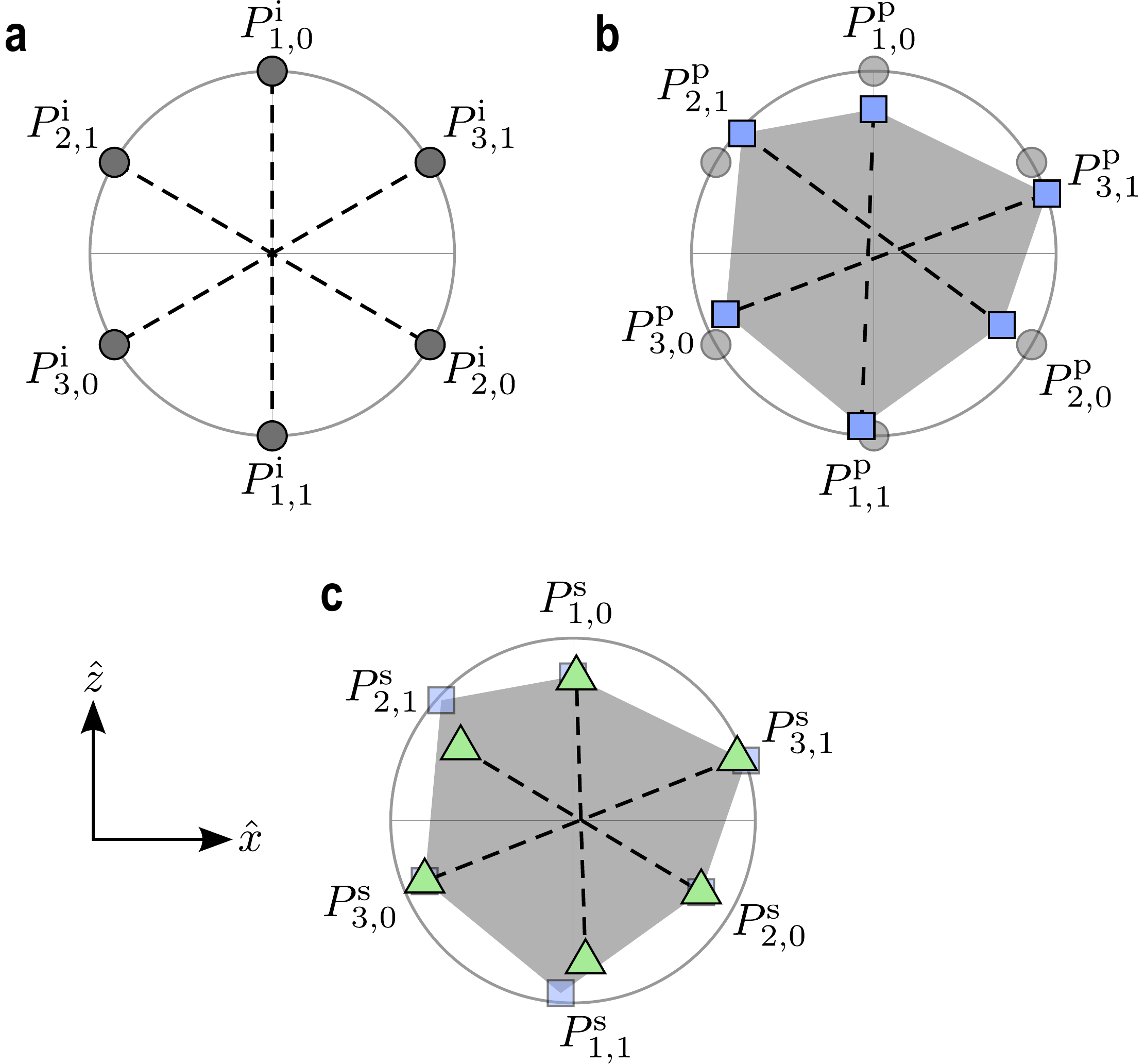}
  \caption{Illustration of our solution to the problem of the failure to achieve strict operational equivalences of preparations
  (under the simplifying assumption that these are confined to the $\hat{x}-\hat{z}$ plane of the Bloch sphere). For a given 
  pair, $P_{t,0}$ and $P_{t,1}$, the midpoint along the line connecting the corresponding points represents their equal mixture, 
  $P_t$. \textbf{a}. The target preparations $P^{\rm i}_{t,b}$, with the coincidence of the midpoints of the three lines 
  illustrating that they satisfy the operational equivalence \eqref{pnc1opeq} exactly. 
  \textbf{b}. Illustration of how errors in the experiment (exaggerated in magnitude) will imply that the realized preparations
  $P^{\rm p}_{t,b}$ (termed primary) will deviate from the ideal. The lines indicate that not only do these preparations fail 
  to satify the operational equivalence \eqref{pnc1opeq}, but since the lines do not meet, no mixtures of the $P^{\rm p}_{t,0}$
  and $P^{\rm p}_{t,1}$ can be found at a single point independent of $t$. The set of preparations corresponding to probabilistic
  mixtures of the $P^{\rm p}_{t,b}$ are depicted by the grey region.
  \textbf{c}. Secondary preparations $P^{\rm s}_{t,b}$ have been chosen from this grey region, with the coincidence of the 
  midpoints of the three lines indicating that the operational equivalence \eqref{pnc1opeq} has been restored. Note that we 
  require only that the mixtures of the three pairs of preparations be the same, not that they correspond to the completely 
  mixed state. Figure credit: Michael Mazurek and Matthew Pusey \cite{exptl}.}
  \label{fg:convex_combs}
\end{figure}

The actual preparations and measurements in the experiment, which we call the {\em primary} procedures and denote by 
$P^{\rm p}_{1,0}$, $P^{\rm p}_{1,1}$, $P^{\rm p}_{2,0}$, $P^{\rm p}_{2,1}$, $P^{\rm p}_{3,0}$, $P^{\rm p}_{3,1}$ and 
$M^{\rm p}_1$, $M^{\rm p}_2$, $M^{\rm p}_3$, necessarily deviate from the ideal versions and consequently their mixtures,
that is, $P^{\rm p}_1$, $P^{\rm p}_2$, $P^{\rm p}_3$ and $M^{\rm p}_*$, fail to achieve strict equality in 
Eqs.~\eqref{mnc1opeq} and \eqref{pnc1opeq}.\footnote{Note that these primary preparation and measurement procedures 
correspond to the quantum states and effects that best fit the raw data collected in the experiment. Of course,
in the actual data analysis reported in the experiment\cite{exptl}, what we fit to the raw data are states and 
effects from a generalized probabilistic theory rather than restricting ourselves to quantum theory.}

We solve this problem as follows. 

\begin{itemize}
 \item 
From the outcome probabilities on the six primary preparations, one can infer the outcome 
probabilities on the entire family of probabilistic mixtures of these. It is possible to find within this family many sets of
six preparations, $P^{\rm s}_{1,0}$, $P^{\rm s}_{1,1}$, $P^{\rm s}_{2,0}$, $P^{\rm s}_{2,1}$, $P^{\rm s}_{3,0}$, 
$P^{\rm s}_{3,1}$, which define mixed preparations $P^{\rm s}_1$, $P^{\rm s}_2$, $P^{\rm s}_3$ that satisfy the operational 
equivalences of Eq.~\eqref{pnc1opeq} \emph{exactly}. We call the $P^{\rm s}_{t,b}$ {\em secondary} preparations.
We can define secondary measurements $M^{\rm s}_1$, $M^{\rm s}_2$, $M^{\rm s}_3$ and their uniform mixture $M^{\rm s}_{*}$ in
a similar fashion. The essence of our approach, then, is to identify such secondary sets of procedures and use {\em these} to 
calculate $A'$.  If quantum theory is correct, then we expect to get a value of $A$ close to 1 if and only if we can find 
suitable secondary procedures that are close to the ideal versions.


\item In Fig.~\ref{fg:convex_combs}, we describe the construction of secondary preparations in a simplified example of six 
density 
operators that deviate from the ideal states only {\em within} the $\hat{x}-\hat{z}$ plane of the Bloch sphere.
In practice, the six density operators realized in the experiment will not quite lie in a plane. We use the same idea to 
contend with this, but with one refinement: we supplement our set of ideal preparations with two additional ones, denoted 
$P^{\rm i}_{4,0}$ and $P^{\rm i}_{4,1}$ corresponding to the two eigenstates of $\vec{\sigma}\cdot \hat{y}$. The two procedures
that are actually realized in the experiment are denoted $P^{\rm p}_{4,0}$ and $P^{\rm p}_{4,1}$ and are considered supplements
to the primary set. We then search for our six secondary preparations among the probabilistic mixtures of this supplemented 
set of primaries rather than among the probabilistic mixtures of the original set. Without this refinement, it can happen that
one cannot find six secondary preparations that are close to the ideal versions. 

\item To see why this refinement is needed, consider the case where the six primary preparations deviate from the ideals within 
the bulk of the Bloch sphere. The fact that our proof only requires that the secondary preparations satisfy 
Eq.~\eqref{qopequiv2}
means that the different pairs, $P^{\rm s}_{t,0}$ and $P^{\rm s}_{t,1}$ for $t\in \{1,2,3\}$, need not all mix to the center of
the Bloch sphere, but only to the {\em same} state. It follows that the three pairs need not be coplanar in the Bloch sphere. 
Note, however, for any {\em two} values, $t$ and $t'$, the four preparations 
$P^{\rm s}_{t,0}, P^{\rm s}_{t,1}, P^{\rm s}_{t',0}, P^{\rm s}_{t',1}$ do need to be coplanar.
Any mixing procedure defines a  map from each of the primary preparations $P^{\rm p}_{t,b}$ to the corresponding secondary 
preparation $P^{\rm s}_{t,b}$, which can be visualized as a motion of
the corresponding point within the Bloch sphere. To ensure that the six secondary preparations approximate well the ideal 
preparations while also defining mixed preparations 
$P^{\rm s}_1$, $P^{\rm s}_2$ and $P^{\rm s}_3$ that satisfy the appropriate operational equivalences, the mixing procedure 
must allow for motion in the $\pm\hat{y}$ direction. Consider what 
happens if one tries to achieve such motion {\em without} supplementing the primary set with the eigenstates 
of $\vec{\sigma}\cdot \hat{y}$.  A given point that is biased towards $-\hat{y}$ can
be moved in the $+\hat{y}$ direction by mixing it with another point that has less bias in the $-\hat{y}$ direction.
However, because the primary preparations are widely separated within the
$\hat{x}-\hat{z}$ plane, achieving a small motion in $+\hat{y}$ direction in
this fashion comes at the price of a large motion within the $\hat{x}-\hat{z}$ plane,
implying a significant motion away from the ideal. This problem is particularly pronounced if the primary
points are very close to coplanar.

\item The best way to move a given point in the $\pm\hat{y}$ direction is to mix it with a point that is at roughly the same location
within the $\hat{x}-\hat{z}$ plane, but displaced in the 
$\pm\hat{y}$ direction. This scheme, however, would require supplementing the primary set with one or two additional 
preparations for every one of its elements. Supplementing the original set 
with just the two  eigenstates of $\vec{\sigma}\cdot \hat{y}$ constitutes a good compromise between keeping the number of 
preparations low and ensuring that the secondary preparations are close
to the ideal. Because the $\vec{\sigma}\cdot \hat{y}$ eigenstates have the greatest possible distance from the 
$\hat{x}-\hat{z}$ plane, they can be used to move any point close to that plane in
the $\pm \hat{y}$ direction while generating only a modest motion within the $\hat{x}-\hat{z}$ plane.
\end{itemize}

\subsubsection{Secondary measurements in quantum theory}
Just as with the case of preparations, we solve the problem of no strict statistical equivalences for measurements by noting
that from the primary set of measurements, $M^{\rm p}_1$, $M^{\rm p}_2$ and $M^{\rm p}_3$, one can infer the statistics of a 
large family of measurements, and one can find three measurements within this family, called the secondary measurements and
denoted $M^{\rm s}_1$, $M^{\rm s}_2$ and $M^{\rm s}_3$, such that their mixture, $M^{\rm s}_*$, satisfies the operational 
equivalence of Eq.~\eqref{mnc1opeq} {\em exactly}. To give the details of our approach, it is again useful to begin with the 
quantum description.

\begin{itemize}
 \item 
Just as a density operator can be written 
$\rho = \frac12(\idn + \vec r \cdot \vec\sigma)$ to define a three-dimensional Bloch vector $\vec r$, an effect can be written 
$E = \frac12(e_0\idn + \vec e\cdot\vec\sigma)$ to define a four-dimensional Bloch-like vector $(e_0, \vec e)$, whose four 
components we will call the $\hat \idn$, $\hat x$, $\hat y$ and $\hat z$ components.  Note that $e_0={\rm tr}(E)$, while 
$e_x = {\rm tr}(\vec{\sigma}\cdot \hat{x} E)$ and so forth. The eigenvalues of $E$ are expressed in terms of these components 
as $\frac{1}{2}(e_o \pm |\vec{e}|)$. Consequently, the constraint that $0 \le E \le \idn$ takes the form of three inequalities
$0 \le e_o \le 2$, $|\vec{e}| \le e_0$ and $|\vec{e}| \le 2-e_0$. This corresponds to the intersection of two cones. For the 
case $e_y=0$, the Bloch representation of the effect space is three-dimensional and is displayed in Fig.~\ref{Blochcone}.
When portraying binary-outcome measurements associated to a POVM $\{ E, \idn - E\}$ in this representation, it is sufficient to
portray the Bloch-like vector $(e_0,\vec{e})$ for outcome $E$ alone, given that the vector for $\idn - E$ is simply 
$(2-e_0, -\vec{e})$. Similarly, to describe any mixture of two such POVMs, it is sufficient to describe the mixture of the 
effects corresponding to the first outcome.

\item
The family of measurements that is defined in terms of the primary set is slightly different than what we had for preparations.  The reason is that each primary measurement on its own generates
a family of measurements by probabilistic post-processing of its outcome. If we denote the outcome of the original measurement by $X$ and that of the processed measurement by $X'$, then the 
probabilistic processing is a conditional probability $p(X'|X)$.  It is sufficient to determine the convexly-extremal post-processings, since all others can be obtained from these by mixing. 
For the case of binary outcome measurements considered here, there are just four extremal post-processings: the identity process, $p(X'|X)=\delta_{X',X}$; the process that flips the outcome,
$p(X'|X)=\delta_{X',X\oplus 1}$; the process that always generates the outcome $X'=0$, $p(X'|X)=\delta_{X',0}$; and the process that always generates the outcome $X'=1$, $p(X'|X)=\delta_{X',1}$.
Applying these to our three primary measurements, we obtain
eight measurements in all: the two that generate a fixed outcome, the three originals, and the three originals with the outcome flipped.   If the set of primary measurements corresponded to the
ideal set, then the eight extremal post-processings would correspond to the observables $0, \idn, \vec{\sigma}\cdot \hat{n_1}, {-}\vec{\sigma}\cdot \hat{n_1}, \vec{\sigma}\cdot \hat{n_2}, {-}\vec{\sigma}\cdot \hat{n_2}, \vec{\sigma}\cdot \hat{n_3}, {-}\vec{\sigma}\cdot \hat{n_3}$.
In practice, the last six measurements will be unsharp.  These eight measurements can then be mixed probabilistically to define the family of measurements from which the secondary measurements 
must be chosen.  We refer to this family as the {\em convex hull of the post-processings} of the primary set.

\item 
Consider a simplified example wherein the primary measurements have Bloch-like vectors with vanishing component
along $\hat{y}$, $e_y = 0$, and unit component along $\idn$, $e_0=1$, so that 
$E = \frac12(\idn + e_x \vec\sigma\cdot \hat{x} + e_z \vec\sigma \cdot \hat{z})$. In this case, the constraint 
$0\le E\le \idn$ reduces to $|\vec{e}|\le 1$, which is the same constraint that applies to density operators confined to the
$\hat{x}-\hat{z}$ plane of the Bloch sphere. Here, the only deviation from the ideal is within this plane, and the construction
is precisely analogous to what is depicted in Fig.~\ref{fg:convex_combs}.

\item
Unlike the case of preparations, however, the primary measurements can deviate from the ideal in the $\hat \idn$ direction, 
that is, $E$ may have a component along $\idn$ that deviates from $1$, which corresponds to introducing a state-independent 
bias on the outcome of the measurement. This is where the extremal post-processings yielding the constant-outcome measurements
corresponding to the observables 0 and $\idn$ come in. They allow one to move in the $\pm\hat \idn$ direction.

Fig.~\ref{Blochcone} presents an example wherein the primary measurements have Bloch-like vectors that deviate from the ideal 
not only within the $\hat{x}-\hat{z}$ plane, but in the $\hat{\idn}$ direction as well (it is still presumed, however, that all
components in the $\hat{y}$ direction are vanishing).

\begin{figure}
 \centering
 \includegraphics[width=0.9\textwidth]{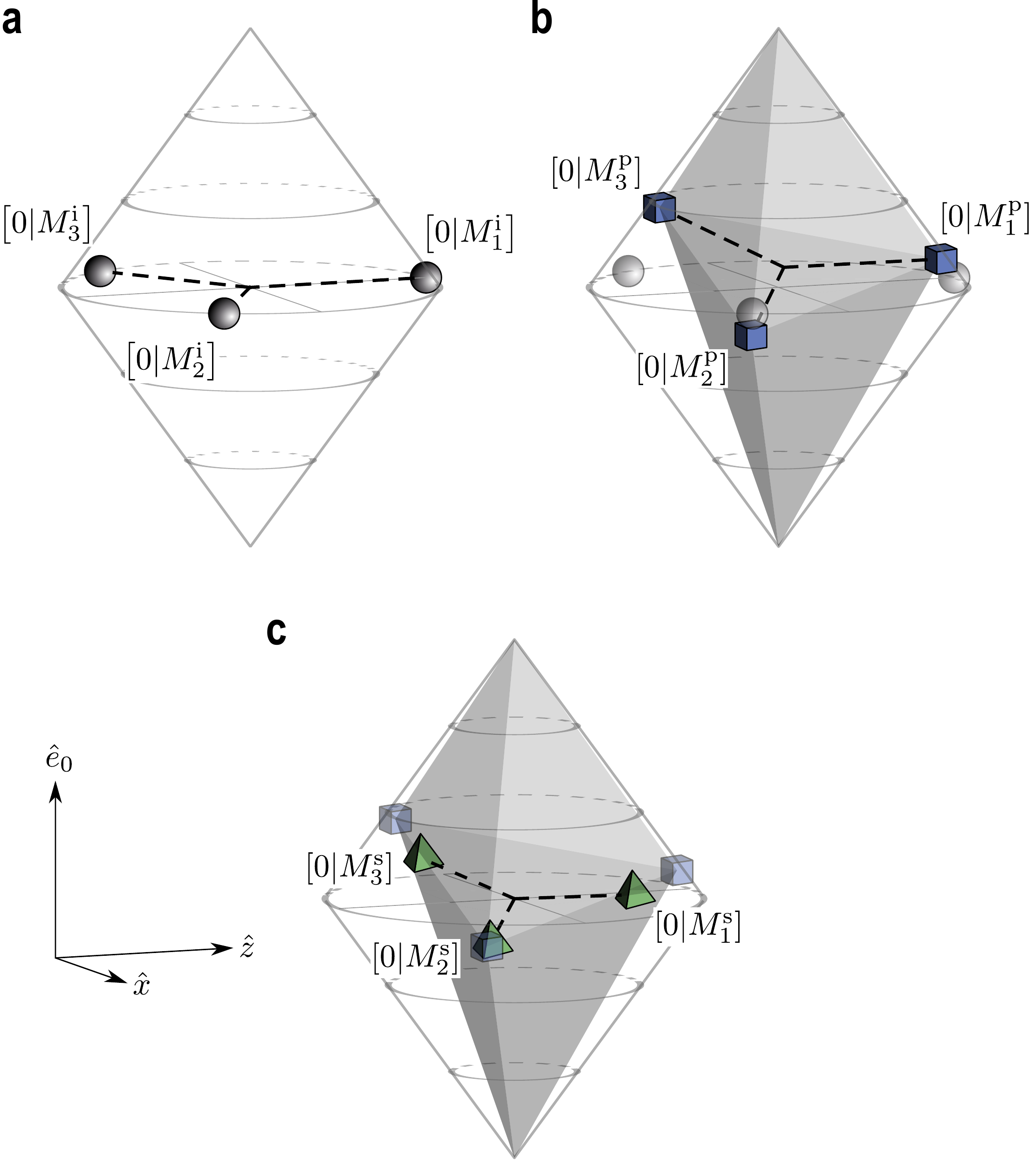}
 \caption{A depiction of the construction of secondary measurements from primary ones in the simplified case where the 
 component along $\hat{y}$ is zero. For each measurement, we specify the point corresponding to the Bloch representation of its
 first outcome. These are labelled $[0|M_1]$, $[0|M_2]$ and $[0|M_3]$. The equal mixture of these three, labelled $[0|M_*]$, 
 is the centroid of these three points, i.e. the point equidistant from all three. \textbf{a}. The ideal measurements 
 $[0|M^{\rm i}_{t}]$ with centroid at $\idn/2$, illustrating that the operational equivalence \eqref{mnc1opeq} is satisfied 
 exactly.
 \textbf{b}. Errors in the experiment (exaggerated) will imply that the realized measurements $[0|M^{\rm p}_{t,}]$ 
 (termed primary) will deviate from the ideal, and their centroid deviates from $\idn/2$. The family of points corresponding to
 probabilistic mixtures of the $[0|M^{\rm p}_{t}]$ and the
 observables $0$ and $\idn$ are depicted by the grey region. (For clarity, we have not depicted the outcome-flipped versions of
 the three primary measurements, and have not included them in the probabilistic mixtures. As we note in the text, such a 
 restriction still allows for a good construction.)
 \textbf{c}. The secondary measurements $M^{\rm s}_{t}$ that have been chosen from this grey region. They are chosen such that
 their centroid is at $\idn/2$, restoring the operational equivalence~\eqref{mnc1opeq}. Figure credit: Michael Mazurek \cite{exptl}.}
\label{Blochcone}
\end{figure}

In practice, of course, the $\hat{y}$ component of our measurements never vanishes precisely either.  We therefore apply the same trick as we did for the preparations.

We supplement the set of primary measurements with an additional measurement, denoted $M^{\rm p}_4$, that ideally corresponds to the observable $\vec{\sigma}\cdot \hat{y}$. The post-processing 
which flips the outcome then corresponds to the observable $-\vec{\sigma}\cdot \hat{y}$.  Mixing the primary measurements with $M^{\rm p}_4$ and its outcome-flipped counterpart
allows motion in the $\pm\hat{y}$ direction within the Bloch cone.

\item
Note that the capacity to move in both the $\hat{y}$ and the $-\hat{y}$ direction is critical for achieving the operational 
equivalence of Eq.~\eqref{mnc1opeq}, because if the secondary measurements had a common bias in the $\hat{y}$ direction, they 
could not mix to the POVM $\{ \idn/2, \idn/2\}$ as Eq.~\eqref{qopequiv1} requires.  For the preparations, by contrast, 
supplementing the primary set by just {\em one} of the eigenstates of $\vec{\sigma}\cdot \hat{y}$ would still work, given that 
the mixed preparations $P^{\rm s}_{t}$ do not need to coincide with the completely mixed state $\idn/2$.

The secondary measurements $M^{\rm s}_1$, $M^{\rm s}_2$ and $M^{\rm s}_3$ are then chosen from the convex hull of the post-processings of the $ M^{\rm p}_1,M^{\rm p}_2,M^{\rm p}_3,M^{\rm p}_4$.
Without this supplementation, it may be impossible to find secondary measurements that define an $M^{\rm s}_*$ that satisfies the operational equivalences while providing a good approximation 
to the ideal measurements.
\item
In all, under the extremal post-processings of the supplemented set of primary measurements, we obtain ten points which ideally
correspond to the observables $0, \idn, \vec{\sigma}\cdot \hat{n_1}, {-}\vec{\sigma}\cdot \hat{n_1}, \vec{\sigma}\cdot \hat{n_2},
{-}\vec{\sigma}\cdot \hat{n_2}, \vec{\sigma}\cdot \hat{n_3}, {-}\vec{\sigma}\cdot \hat{n_3}, \vec{\sigma}\cdot \hat{y}$, and  
${-}\vec{\sigma}\cdot \hat{y}$.
Note that the outcome-flipped versions of the three primary measurements are not critical for defining a good set of secondary
measurements, and indeed we find that we can dispense with them and
still obtain good results. This is illustrated in the example of Fig.~\ref{Blochcone}.
\item
Note that in order to identify which density operators have been realized in an experiment, 
the set of measurements must be complete for state tomography. Similarly, to identify which sets of effects
have been realized, the set of preparations must be complete for measurement tomography. This tomographic completeness
is crucial to be able to explicitly verify the operational equivalences that a test of noncontextuality requires.
However, the original ideal sets fail to be tomographically complete because they are restricted to a plane of the Bloch 
sphere, and an effective way to complete them is to add the observable $\vec{\sigma}\cdot \hat{y}$ to the measurements and its 
eigenstates to the preparations. Therefore, even if we did not already need to supplement these ideal sets for the purpose of 
providing greater leeway in the construction of the secondary procedures, we would be forced to do so in order to ensure that 
one can verify operational equivalences explicitly.
\end{itemize}

\subsubsection{Secondary preparations and measurements in generalized probabilistic theories}
To analyze our data in a manner that does not prejudice which model---noncontextual, quantum, or otherwise---does justice to it, we
must search for representations of the preparations and measurements not amongst density operators and sets of effects, but rather 
their more abstract counterparts in the formalism of generalized probabilistic theories (GPTs) \cite{hardy5axioms,barrettgpt}, 
called generalized states and effects. The assumption that the system is a qubit is replaced by the strictly weaker assumption 
that three two-outcome measurements are tomographically complete. (In generalized probabilistic theories, a set of measurements
is called tomographically complete if their statistics suffice to determine the state.)
We take these states and effects as estimates of our primary preparations and measurements, and we define our estimate of the 
secondary procedures in terms of these, which in turn are used to calculate our estimate for $A'$. We explain how the raw data
is fit to a set of generalized states and effects following the procedure outlined in Section \ref{secondarytrick}.

Since we do not want to presuppose that our experiment is well fit by a quantum description, we work with GPT states and effects
which are inferred from  the matrix $\mathbb{D}^{\rm p}$:
\begin{equation}
\mathbb{D}^{\rm p} =
 \begin{pmatrix}
  p^{1}_{1,0}  & p^{1}_{1,1} & \cdots & p^{1}_{4,0}  & p^{1}_{4,1} \\
  p^{2}_{1,0}  & p^{2}_{1,1} & \cdots & p^{2}_{4,0}  & p^{2}_{4,1} \\
  p^{3}_{1,0}  & p^{3}_{1,1} & \cdots & p^{3}_{4,0}  & p^{3}_{4,1} \\
  p^{4}_{1,0}  & p^{4}_{1,1} & \cdots & p^{4}_{4,0}  & p^{4}_{4,1}
 \end{pmatrix}.
 \label{prim_matrix}
\end{equation}
where
\beq
p^{t'}_{t,b}\equiv p(0|M^{\rm p}_{t'},P^{\rm p}_{t,b})
\eeq
is the probability of obtaining outcome $0$ in the $t'$th measurement that was actually realized in the experiment (recall that
we term this measurement primary and denote it by $M^{\rm p}_{t'}$), when it follows the $(t,b)$th preparation that was actually
realized in the experiment (recall that we term this preparation primary and denote it by $P^{\rm p}_{t,b}$).
These probabilities are estimated by fitting the raw  experimental data (which are merely finite samples of the true probabilities)
to a GPT. We now describe this procedure before moving on to the construction of $\mathbb{D}^{\rm s}$:

{\em Fitting the raw data to a generalized probabilistic theory: }
In our experiment we perform four measurements on each of eight input states. If we define $r^{t'}_{t,b}$ as the fraction of 
`0' outcomes returned by measurement $M_{t'}$ on preparation $P_{t,b}$, the results can be summarized in a $4\times8$ matrix 
of raw data, $\mathbb{D}^{\rm r}$, defined as:
\begin{equation}
\mathbb{D}^{\rm r} =
 \begin{pmatrix}
  r^{1}_{1,0}  & r^{1}_{1,1} & \cdots & r^{1}_{4,0}  & r^{1}_{4,1} \\
  r^{2}_{1,0}  & r^{2}_{1,1} & \cdots & r^{2}_{4,0}  & r^{2}_{4,1} \\
  r^{3}_{1,0}  & r^{3}_{1,1} & \cdots & r^{3}_{4,0}  & r^{3}_{4,1} \\
  r^{4}_{1,0}  & r^{4}_{1,1} & \cdots & r^{4}_{4,0}  & r^{4}_{4,1}
 \end{pmatrix}.
\end{equation}
Each row of $\mathbb{D}^{\rm r}$ corresponds to a measurement, ordered from top to bottom as $M_{1}$, $M_2$, $M_3$, and $M_4$.
Similary, the columns are labelled from left to right as $P_{1,0}$, $P_{1,1}$, $P_{2,0}$, $P_{2,1}$, $P_{3,0}$, $P_{3,1}$,
$P_{4,0}$, and $P_{4,1}$.

In order to test the assumption that three independent binary-outcome measurements are tomographically complete for our system,
we fit the raw data to a matrix, $\mathbb{D}^{\rm p}$, of primary data defined in Eq.~\eqref{prim_matrix}. $\mathbb{D}^{\rm p}$
contains the outcome probabilities of four measurements on eight states in the GPT-of-best-fit to the raw data. We fit to a GPT
in which three 2-outcome measurements are tomographically complete, which we characterize with the following proposition, a 
special case of Theorem \ref{Dptheorem}:

\begin{proposition}\label{rawprop}
A matrix $\mathbb{D}^{\rm p}$ can arise from a GPT in which three two-outcome measurements are tomographically complete if and 
(with a measure zero set of exceptions) only if $a p^1_{t,b}+b p^2_{t,b}+c p^3_{t,b}+ d p^4_{t,b} -1 = 0$ for some real constants
$\{a,b,c,d\}$.\footnote{See Ref.~\cite{exptl} for a proof of this proposition. Our proof of Theorem \ref{Dptheorem} is a generalization of that proof.}
\end{proposition}
Geometrically, the proposition dictates that the eight columns of $\mathbb{D}^{\rm p}$ lie on the 3-dimensional hyperplane 
defined by the constants $\{a,b,c,d\}$.

To find the GPT-of-best-fit we fit a 3-d hyperplane to the eight 4-dimensional points that make up the columns of 
$\mathbb{D}^{\rm r}$. We then map each column of $\mathbb{D}^{\rm r}$ to its closest point on the hyperplane, and these eight 
points will make up the columns of $\mathbb{D}^{\rm p}$. We use a weighted total least-squares procedure~\cite{krystek07,numrec}
to perform this fit. Each element of $\mathbb{D}^{\rm r}$ has an uncertainty, $\Delta r^{t'}_{t,b}$, which is estimated assuming
the dominant source of error is the statistical error arising from Poissonian counting statistics. We define the 
\emph{weighted distance}, $\chi_{t,b}$, between the $(t,b)$ column of $\mathbb{D}^{\rm r}$ and $\mathbb{D}^{\rm p}$ as 
$\chi_{t,b}=\sqrt{\sum_{t'=1}^4 \left(r^{t'}_{t,b}-p^{t'}_{t,b}\right)^2/\left(\Delta r^{t'}_{t,b}\right)^2}$. Finding the 
best-fitting hyperplane can be summarized as the following minimization problem:
\begin{equation}
 \begin{aligned}\label{eq:chisq_comp}
 & \underset{\{p^i_{t,b},a,b,c,d\}}{\text{minimize}}
 & & \chi^2 = \sum_{t=1}^4\sum_{b=0}^1 \chi_{t,b}^2, \\
 & \text{subject to}
 & & a p^1_{t,b} + b p^2_{t,b} + c p^3_{t,b} + d p^4_{t,b} - 1 = 0 & \\
 & & & \forall \,t=1,\ldots,4, b=0,1.
 \end{aligned}
\end{equation}

The optimization problem as currently phrased is a problem in 36 variables---the 32 elements of $\mathbb{D}^{\rm p}$ together 
with the hyperplane parameters $\{a,b,c,d\}$. We can simplify this by first solving the simpler problem of finding the weighted
distance $\chi_{t,b}$ between the $(t,b)$ column of $\mathbb{D}^{\rm r}$ and the hyperplane $\{a,b,c,d\}$. This can be phrased
as the following 8-variable optimization problem:
\begin{equation}
 \begin{aligned}\label{eq:chi_j}
 & \underset{\{p^1_{t,b},p^2_{t,b},p^3_{t,b},p^4_{t,b}\}}{\text{minimize}} & & \chi_{t,b}^2 = \sum_{t'=1}^4 \frac{(r^{t'}_{t,b}
 -p^{t'}_{t,b})^2}{\left(\Delta r^{t'}_{t,b}\right)^2}, \\
 & \text{subject to} & & a p^1_{t,b} + b p^2_{t,b} + c p^3_{t,b} + d p^4_{t,b} - 1 = 0.
 \end{aligned}
\end{equation}
Using the method of Lagrange multipliers~\cite{krystek07}, we define the Lagrange function
$\Gamma = \chi_{t,b}^2 + \gamma(a p^1_{t,b} + b p^2_{t,b} + c p^3_{t,b} + d p^4_{t,b} - 1)$, where $\gamma$ denotes the Lagrange
multiplier, then simultaneously solve
\begin{equation}
\frac{\partial\Gamma}{\partial\gamma}=0
\end{equation}
and
\begin{equation}
\frac{\partial\Gamma}{\partial p^{t'}_{t,b}}=0, \; t' = 1, \ldots, 4
\end{equation}
for the variables $\gamma$, $p^1_{t,b}$, $p^2_{t,b}$, $p^3_{t,b}$, and $p^4_{t,b}$. Substituting the solutions for 
$p^1_{t,b}$, $p^2_{t,b}$, $p^3_{t,b}$ and $p^4_{t,b}$ into Eq.~\eqref{eq:chi_j} we find
\begin{equation}
\chi_{t,b}^2 = \frac{(a r^1_{t,b} + b r^2_{t,b} + c r^3_{t,b} + d r^4_{t,b} - 1)^2}{\left(a\Delta r^1_{t,b}\right)^2 + 
\left(b\Delta r^2_{t,b}\right)^2 + \left(c\Delta r^3_{t,b}\right)^2 + \left(d\Delta r^4_{t,b}\right)^2},
\end{equation}
which now only contains the variables $a$, $b$, $c$, and $d$.

The hyperplane-finding problem can now be stated as the following four-variable optimization problem:
\begin{equation}
 \begin{aligned}
  \underset{\{a,b,c,d\}}{\text{minimize}}
 & & \chi^2 = \sum_{t=1}^4\sum_{b=0}^1 \chi_{t,b}^2
 \end{aligned}
\end{equation}
which we solve numerically.

The $\chi^2$ parameter returned by the fitting procedure is a measure of the goodness-of-fit of the hyperplane to the data. 
Since we are fitting eight datapoints to a hyperplane defined by four fitting parameters $\{a,b,c,d\}$, we expect the $\chi^2$ 
parameter to be drawn from a $\chi^2$ distribution with four degrees of freedom~\cite{numrec}, which has a mean of 4. 
The experiment was run 100 times and 100 independent $\chi^2$ parameters were obtained; these were found to have a mean of 
$3.9\pm0.3$ \cite{exptl}. In addition a more stringent test of the fit of the model to the data was performed 
by summing the counts from all 100 experimental runs before performing a single fit. 
This fit returns a $\chi^2$ of 4.33, which has a $p$-value of 36\%. 
The outcomes of these tests are consistent with our assumption that the raw data can be explained by a GPT in which three 
2-outcome measurements are tomographically complete and which also exhibits Poissonian counting 
statistics.\footnote{Note that the null hypothesis here is that the model (the GPT hyperplane) fits the data well. Since the $p$-value
is 0.36 (close to 0.5), we have no compelling evidence to reject the null hypothesis.}
Had the fitting 
procedure returned $\chi^2$
values that were much higher,\footnote{And consequently the p-value much lower than 0.5.} this might have indicated that the theoretical
description of the preparation and measurement 
procedures required more than three degrees of freedom. On the other hand, had the fitting returned an average $\chi^2$ much 
lower than 4,\footnote{Consequently a p-value much higher than 0.5} this could have indicated that we had overestimated the amount of uncertainty in our data.

After finding the hyperplane-of-best-fit $\left\{a,b,c,d\right\}$, we find the points on the hyperplane that are closest to 
each column of $\mathbb{D}^{\rm r}$. This is done by numerically solving for $p^1_{t,b}$, $p^2_{t,b}$, $p^3_{t,b}$, and 
$p^4_{t,b}$ in \eqref{eq:chi_j} for each value of $(t,b)$. The point on the hyperplane closest to the $(t,b)$ column of 
$\mathbb{D}^{\rm r}$ becomes the $(t,b)$ column of $\mathbb{D}^{\rm p}$. The matrix $\mathbb{D}^{\rm p}$ is then used to 
find the secondary preparations and measurements.

{\em Inferring the secondary data matrix $\mathbb{D}^{\rm s}$: }The rows of the $\mathbb{D}^{\rm p}$ matrix define the GPT effects. We denote the vector 
defined by the $t$th row, which is 
associated to the measurement event $[0|M_t^{\rm p}]$ (obtaining the 0 
outcome in the primary measurement $M_t^{\rm p}$), by $\mathbf{M}_t^{\rm p}$.
Similarly, the columns of this matrix define the GPT states.  We denote the vector associated to the $(t,b)$th column, which is
associated to the primary preparation $P^{\rm p}_{t,b}$, 
by $\mathbf{P}^{\rm p}_{t,b}$.

As we outlined in the previous section, we define the {\em secondary} preparation $P^{\rm s}_{t,b}$ by a probabilistic mixture of 
the primary preparations.  Thus, the GPT state of the secondary 
preparation is a vector $\mathbf{P}^{\rm s}_{t,b}$ that is a probabilistic mixture of the $\mathbf{P}^{\rm p}_{t,b}$,
\beq
\mathbf{P}^{\rm s}_{t,b} = \sum_{t'=1}^4\sum_{b'=0}^1 u_{t',b'}^{t,b} \mathbf{P}_{t',b'}^{\rm p},
\eeq
where the $u_{t',b'}^{t,b}$ are the weights in the mixture.

A secondary measurement $M_{t'}^{\rm s}$ is obtained from the primary measurements in a similar fashion, but in addition to 
probabilistic mixtures, one must allow certain post-processings of the measurements, in analogy to the quantum case described 
above.
The set of all post-processings of the primary outcome-0 measurement events has extremal elements consisting of the outcome-0 
measurement events themselves together with:
the measurement event that {\em always} occurs (i.e. always obtaining outcome `0'), which is represented by the vector of 
probabilities where every entry is 1, denoted $\mathbf{1}$; the measurement event that {\em never} occurs (i.e. never 
obtaining outcome `0' or always obtaining outcome `1'), which is represented by the vector of probabilities where every 
entry is 0, denoted $\mathbf{0}$; and the outcome-1 measurement events, $[1|M_t^{\rm p}]$,
represented by the vector $\mathbf{1}-\mathbf{M}_t^{\rm p}$.

We can therefore define our three secondary outcome-0 measurement events as probabilistic mixtures of the four primary ones as 
well as the extremal post-processings mentioned above, that is
\beq
\mathbf{M}_t^{\rm s} = \sum_{t'=1}^{4} v_{t'}^t\mathbf{M}_{t'}^{\rm p} + v_{\mathbf{0}}^t \mathbf{0} + v_{\mathbf{1}}^t\mathbf{1}
+ \sum_{t''=1}^{4} v_{\lnot t''}^t(\mathbf{1}-\mathbf{M}_{t''}^{\rm p}),
\eeq
where for each $t$, the vector of weights in the mixture is $(v_1^t,v_2^t,v_3^t,v_4^t,v_{\mathbf{0}}^t,v_{\mathbf{1}}^t,
v_{\lnot 1}^t,v_{\lnot 2}^t,v_{\lnot 3}^t,v_{\lnot 4}^t)$. We see that this is a particular type of linear transformation on 
the rows.

Again, as mentioned in the discussion of the quantum case, we can in fact limit the post-processing to exclude the
outcome-1 measurement events for
$M_1$, $M_2$ and $M_3$, keeping only the outcome-1 event for $M_4$, and still obtain good results. Thus we found it sufficient 
to search for secondary outcome-0 measurement events among those of the form
\beq
\mathbf{M}_t^{\rm s} = \sum_{t'=1}^{4} v_{t'}^t\mathbf{M}_{t'}^{\rm p} + v_{\mathbf{0}}^t \mathbf{0} + v_{\mathbf{1}}^t\mathbf{1}
+ v_{\lnot 4}^t(\mathbf{1}-\mathbf{M}_{4}^{\rm p}),
\eeq
where for each $t$, the vector of weights in the mixture is $(v_1^t,v_2^t,v_3^t,v_4^t,v_{\mathbf{0}}^t,v_{\mathbf{1}}^t,
v_{\lnot 4}^t)$.

Returning to the preparations, we choose the weights $u^{t,b}_{t',b'}$ to maximize the function
\beq
C_{\rm P} \equiv \frac{1}{6}\sum_{t=1}^3\sum_{b=0}^1 u_{t,b}^{t,b}
\eeq
subject to the linear constraint
\beq
\frac{1}{2}\sum_b \mathbf{P}^{\rm s}_{1,b}=\frac{1}{2}\sum_b \mathbf{P}^{\rm s}_{2,b}=\frac{1}{2}\sum_b \mathbf{P}^{\rm s}_{3,b}.
\eeq
This optimization ensures that the secondary preparations are as close as possible to the primary ones while ensuring that they
satisfy the relevant operational equivalences {\em exactly}. For the experiment of Ref.~\cite{exptl}, 
the reported $C_{\rm P}=0.9969\pm0.0001$, indicating that the secondary preparations are indeed very close 
to the primary ones.

The scheme for finding the weights $(v_1^t,v_2^t,v_3^t,v_4^t,v_{\mathbf{0}}^t,v_{\mathbf{1}}^t,v_{\lnot 4}^t)$ that define the
secondary measurements is analogous. Using a linear program, we find the vector of such weights that maximizes the function
\beq
C_{\rm M}\equiv \frac{1}{3}\sum_{t=1}^{3}v_t^t,
\eeq
 subject to the constraint that
 \beq
 \mathbf{M}_*^{\rm s}=\frac{1}{2}\mathbf{1},
 \eeq
 where $\mathbf{M}_*^{\rm s} \equiv \frac{1}{3}\sum_{t=1}^{3}\mathbf{M}_t^{\rm s}$.
 A high value of $C_\mathrm{M}$ signals that each of the three secondary measurements is close to the corresponding primary one.
 The experiment of Ref.~\cite{exptl} reported $C_{\rm M} = 0.9976\pm0.0001$, again indicating the closeness of the secondary 
 measurements to the primary ones.

This optimization defines the precise linear transformation of the rows of $\mathbb{D}^{\rm p}$ and the linear transformation of
the columns  of $\mathbb{D}^{\rm p}$ that serve to define the secondary preparations and measurements. By combining the operations
on the rows and on the columns, we obtain from $\mathbb{D}^{\rm p}$ a $3\times 6$ matrix, denoted $\mathbb{D}^{\rm s}$, whose 
entries $s^{t'}_{t,b}$ are
\begin{equation}
\sum_{\tau=1}^{4} \sum_{\beta=0}^{1}  u^{t,b}_{\tau,\beta} \left[ \sum_{\tau'=1}^{4} v^{t'}_{\tau'} p^{\tau'}_{\tau,\beta} + v^{t'}_{\mathbf{0}} 0 + v^{t'}_{\mathbf{1}} 1 +  v^{t'}_{\lnot4} (1-p^{4}_{\tau,\beta})\right]
\end{equation}
where $t',t\in \{1,2,3\}$, $b\in \{0,1\}$.
This matrix describes the secondary preparations $P_{t,b}^{\rm s}$ and measurements $M_{t'}^{\rm s}$.
The component $s^{t'}_{t,b}$ of this matrix describes the probability of obtaining outcome 0 in measurement $M^{\rm s}_{t'}$ on preparation $P^{\rm s}_{t,b}$, that is,
\beq
s^{t'}_{t,b} \equiv p(0|M^{\rm s}_{t'},P^{\rm s}_{t,b}).
\eeq
These probabilities are the ones that are used to calculate the value of $A'$ via Eq.\eqref{A'expr}.
The value of $A'$ reported in the experiment was $A'=0.99709\pm0.00007$, well above the noncontextual bound of $5/6\approx0.833$ \cite{exptl}.

\section{Chapter summary}
We have shown how to contend with the problem of noisy measurements and inexact operational equivalences in tests of noncontextuality.
As we explained, the methods used to derive the noncontextuality inequality motivated by the 18 ray Kochen-Specker construction can be used to derive 
such tests of noncontextuality from other KS-uncolourable hypergraphs as well. We also reported and 
generalized the methods adopted in the experimental test of the FCF inequality in hopes that such a generalization 
will be useful in future experimental tests of noncontextuality. An open challenge that remains is to put the assumption of 
tomographic completeness of a finite set of preparations and measurements on a surer footing, akin to the assumption of 
no-signalling in Bell tests.

\chapter{Back to Specker's scenario: a theory-independent analysis}
In this chapter we return to the problem of contextuality in Specker's scenario that we first discussed in Chapter 2. 
In Chapter 2, we considered the LSW inequality \cite{LSW} which presumes the validity of quantum theory, even though it does not 
assume outcome determinism for unsharp measurements. Following the development towards theory-independent tests of 
contextuality in the previous chapter, we will carry out this exercise for the case of Specker's scenario\footnote{And generalizations thereof to what will be called 
``$n$-cycle scenarios''.} in this chapter. Indeed, while the previous chapter showed how to derive robust noncontextuality inequalities from so-called ``state-independent''
proofs of the Kochen-Specker theorem (based on KS-uncolourability), the present chapter achieves this for ``state-dependent'' proofs of contextuality.

The operational noncontextuality inequalities derived in this chapter do not rely on the assumption that measurement 
outcomes are fixed deterministically 
by the ontic state of the system. They constitute a proper operational generalization of the LSW inequality discussed in Chapter 2,
explicitly taking into account the lack of perfect predictability of measurement outcomes in realistic experiments.
We construct {\em quantum} violations of these inequalities. 

In deriving these inequalities, no assumption of the validity of quantum theory is made. An {\em experimental} violation 
of them would serve as a genuine test of nonclassicality which any 
operationally motivated theory of physics---in particular, any future modification of quantum theory---would have to 
accommodate. The most basic of these
inequalities applies to the case of Specker's scenario which involves three two-outcome measurements, every pair of 
which is jointly measured. Specker's 
scenario is the minimal scenario in which contextuality with respect to joint measurement contexts can be expected to
manifest itself and our analysis 
provides a robust noncontextuality inequality for this scenario before moving on to more general $n$-cycle scenarios.

This chapter is based on joint work with Rob Spekkens.\footnote{This work is unpublished and we present some preliminary 
results in this chapter. A talk based on an earlier version of this work is available on PIRSA: \url{http://pirsa.org/14010102/}.}

\section{Introduction}

Specker's scenario was first described by Ernst Specker in the form of a parable \cite{Spe60, LSW}, a rendition of which (due to Liang, Spekkens, and Wiseman \cite{LSW})
is reproduced below:
\begin{quotation}
At the Assyrian School of Prophets in
Arba'ilu in the time of King Asarhaddon
[(681-669 BCE)], there taught a seer from
Nineva. He was a distinguished representative
of his faculty (eclipses of the sun and
moon) and aside from the heavenly bodies,
his interest was almost exclusively in his
daughter. His teaching success was limited;
the subject proved to be dry and required
a previous knowledge of mathematics which
was scarcely available. If he did not find the
student interest which he desired in class, he
did find it elsewhere in overwhelming measure.
His daughter had hardly reached a marriageable
age when he was flooded with requests
for her hand from students and young
graduates. And though he did not believe
that he would always have her by his side,
she was in any case still too young and her
suitors in no way worthy. In order that the
suitors might convince themselves of their unworthiness,
he promised them that she would
be wed to the one who could solve a prediction
task that was posed to them.
Each suitor was taken before a table on
which three little boxes stood in a row, [each
of which might or might not contain a gem],
and was asked to predict which of the boxes
contained a gem and which did not. But no
matter how many times they tried, it seemed
impossible to succeed in this task. After each
suitor had made his prediction, he was ordered
by the father to open any two boxes
which he had predicted to be both empty or
any two boxes which he had predicted to be
both full [in accordance with whether he had
predicted there to be at most one gem among
the three boxes, or at least two gems, respectively].
But it always turned out that one
contained a gem and the other one did not,
and furthermore the stone was sometimes in
the first and sometimes in the second of the
boxes that were opened. But how can it be
possible, given three boxes, to neither be able
to pick out two as empty nor two as full?
The daughter would have remained unmarried
until the father's death, if not for
the fact that, after the prediction of the son
of a prophet [whom she fancied], she quickly
opened two boxes herself, one of which had
been indicated to be full and the other empty,
and the suitor's prediction [for these two
boxes] was found, in this case, to be correct.
Following the weak protest of her father that
he had wanted two other boxes opened, she
tried to open the third. But this proved impossible
whereupon the father grudgingly admitted
that the prediction, being unfalsified,
was valid. [The daughter and the suitor were
married and lived happily ever after.]
\end{quotation}
In summary, Specker's scenario involves three boxes,
each of which either contains a gem or is empty, such that any two boxes can be opened together but all three boxes can't be opened at once. Each suitor 
seeking to marry the seer's daughter is asked to predict the occupancy of the three boxes, after which the seer asks him to open two boxes predicted to be 
both empty or both full. Upon opening two such boxes it is always found that one contains a gem and the other does not. Furthermore, the gem is sometimes 
found in the first box and sometimes in the second box that was opened. These correlations were described as the `over-protective seer' (OS) correlations in
Ref.~\cite{LSW}:
\begin{eqnarray}
\forall i\neq j: &&p(X_i=0,X_j=1|M_{ij};P_*)=\frac{1}{2}\nonumber\\
&&p(X_i=1,X_j=0|M_{ij};P_*)=\frac{1}{2},
\end{eqnarray}
where $i,j\in\{1,2,3\}$ label the three boxes, $M_{ij}$ denotes the operation of opening two boxes $i$ and $j$ together while $X_i$ and $X_j$
denote their respective occupancy ($0$ for no gem and $1$ for a gem), and $P_*$ denotes the 
preparation of the three boxes which yields OS correlations. Clearly, it is always the case that upon opening two boxes $i$ and $j$ one contains a gem
and the other does not: $p(X_i\neq X_j|M_{ij};P_*)=1, \forall i\neq j$.

As we did in Chapter 2, we will soon formalize this scenario in terms of performing two-outcome measurements (instead of opening boxes) and observing their outcome (instead of 
occupancy of the opened boxes). For now, note that observing OS correlations is surprising only if, besides the assumption that the probability of occupancy of a given box
is independent of which other box it is opened with (this is essentially \emph{measurement noncontextuality}), one has also made the 
assumption of \emph{outcome-determinism}: that the occupancy of each box is fixed deterministically by some mechanism. On relaxing the assumption of 
outcome-determinism, the OS correlations should not be surprising: they merely say that upon opening any two boxes, a fair coin flip decides 
which of them contains a gem and which one doesn't. This is consistent with the fact that the (marginalized) 
occupancy of each box is uniformly random.
On the other hand, the OS correlations \emph{would} be surprising
if there are operational reasons implying that the occupancy of each box was decided deterministically, or at least with some bias, at the ontological level.
Such operational reasons amount to observing high predictability of the individual measurements for an appropriate set of preparations of the boxes.
This intuitive tradeoff between the observed anticorrelation
and the predictability of occupancy of each box is what we formalize in our noncontextuality inequality for Specker's scenario: a high probability
of anticorrelation implies a low predictability of occupancy.

The modern rendition of this parable in Ref.~\cite{LSW} made it precise in terms of joint measurements of POVMs for which such joint measurability relations
---pairwise joint measurements but no triplewise joint measurements---are possible. The LSW inequality formulated for qubit POVMs in Ref.~\cite{LSW} 
provided a necessary criterion for deciding whether the statistics 
of the pairwise joint measurements of the POVMs is consistent with a noncontextual model of quantum theory. It was conjectured that this is 
always the case \cite{LSW}. However, as we showed in Chapter 2, this is not the case and that there exist qubit POVMs for 
which the LSW inequality can be violated. Such a violation has been claimed to be experimentally demonstrated recently 
in an experiment using the polarization of single photons \cite{Xue}. 

A further feature that was appreciated only after the results of
Refs.~\cite{LSW,KG} is the following: there exist
three two-outcome qubit POVMs such that all three \emph{are} jointly measurable, \emph{yet} they admit pairwise joint measurements which cannot be obtained 
from the marginalization of any triplewise joint measurement. This was first noticed in Ref.~\cite{FYuOh} and later 
analyzed in more detail in Ref.~\cite{RK}.\footnote{Note, however,
that Ref.~\cite{FYuOh} does \emph{not} proceed from the assumption of noncontextuality that we use in our approach, and we refer to \cite{FYuOh} only because, as far as we know,
it first pointed out
this property of the joint measurability of POVMs.}
This means that the narrative of Specker's parable changes in the following way to accomodate this quantum feature: it is no longer necessary that the three
boxes can't be opened together. Rather, even if the three boxes could be opened together,
it may still be that opening the boxes pairwise leads to statistics which cannot arise from opening all three of them together.
In view of this, we will relax the requirement that the three qubit POVMs be triplewise incompatible since it turns out not 
to be necessary to witness contextuality.

Much as Refs.~\cite{LSW,KG,RK} clarified the impossibility of explaining statistics of qubit POVMs in a noncontextual model, 
they left open the question of whether it is possible to make a theory-independent claim regarding 
contextuality in this scenario. We address this deficiency of earlier analyses by deriving noncontextuality inequalities that
apply to 
{\em any} operational theory that might govern the experimental statistics. Notice that Specker's scenario is the \emph{minimal} 
scenario in which contextuality with respect to joint measurement
contexts can be manifested: at least three measurements are needed for distinct joint measurement contexts for a given 
measurement to be defined
and every pair of these measurements has to be jointly measurable for a scenario where contextuality can be expected.
Furthermore, two-outcome measurements are the simplest possible ones.

\section{Specker's scenario}
\subsection{Noncontextuality inequalities for Specker's scenario}
We consider three two-outcome measurements, $\{M_1,M_2,M_3\}$, each $M_i$ with outcomes labelled by $X_i\in \{0,1\}$,
such that every pair,
that is, $\{M_i,M_j\}$ for $(ij)\in\{(12),(23),(31)\}$, admits of a joint measurement, denoted by $M_{ij}$. $M_{ij}$ is a 
measurement procedure---with four outcomes denoted by $(X_i,X_j)$---whose measurement statistics
can be coarse-grained to obtain the measurement statistics of both $M_i$ and $M_j$ for any preparation $P\in\mathcal{P}$:
\begin{eqnarray}
p(X_i|M_i,P)&\equiv&\sum_{X_j}p(X_i,X_j|M_{ij},P),\nonumber\\
p(X_j|M_j,P)&\equiv&\sum_{X_i}p(X_i,X_j|M_{ij},P).
\end{eqnarray}

Denoting by $M_i^{(j)}$ $(M_j^{(i)})$ the coarse-graining over $X_j$ $(X_i)$ of $M_{ij}$,
pairwise joint measurability of $M_1$, $M_2$ and $M_3$ implies these operational equivalences: 
\begin{eqnarray}\label{eq:opeqmmtsspeck}
M_1^{(2)}&\simeq& M_1^{(3)}\simeq M_1,\nonumber\\
M_2^{(1)}&\simeq& M_2^{(3)}\simeq M_2,\nonumber\\
M_3^{(1)}&\simeq& M_3^{(2)}\simeq M_3. 
\end{eqnarray}

We now define a measurement $M_*$ as follows: sample $(ij) \in\{ (12),(23),(31)\}$ with probability 1/3 each
and then implement $M_{ij}$ and record $(X_i,X_j)$. We are interested in the probability of recording anticorrelated outcomes,
\begin{align}
p(\text{anti}|M_*,P) &\equiv \frac{1}{3} \sum_{(ij)} p(X_i\ne X_j|M_{ij},P).
\end{align}
Similarly, we consider another set of measurements, $\{M'_{12},M'_{23},M'_{31}\}$, which also achieve a joint measurement of the respective pairs:
\begin{eqnarray}\label{eq:opeqmmtsspeck2}
M_1^{\prime(2)}&\simeq& M_1^{\prime(3)}\simeq M_1,\nonumber\\
M_2^{\prime(1)}&\simeq& M_2^{\prime(3)}\simeq M_2,\nonumber\\
M_3^{\prime(1)}&\simeq& M_3^{\prime(2)}\simeq M_3. 
\end{eqnarray}
We also define a measurement procedure $M'_*$ implementing $M'_{12}$, $M'_{23}$, or $M'_{31}$ with equal probabilities such that $p(\text{anti}|M'_*,P)$ is the probability of obtaining 
anticorrelated outcomes for $M'_*$. 

We define \emph{predictability of $(M,P)$}:
\begin{align}
\eta(M,P) \equiv 2 \max_{X\in\{0,1\}}p(X|M,P)-1,
\end{align}
where $\eta(M,P)$ is a measure of how \emph{predictable}, or far away from uniformly random,
the distribution over outcomes is for a two-outcome measurement $M$ performed following a preparation $P$ on a system.

Let $P_*$, $P_*^{\perp}$, $P_1$, $P_1^{\perp}$, $P_2$, $P_2^{\perp}$, $P_3$, $P_3^{\perp}$ be preparation procedures, and let 
$P^{\rm (ave)}_x$ be the preparation procedure obtained by implementing $P_x$ with probability $1/2$ and $P_x^{\perp}$ with
probability $1/2$ for $x\in \{1,2,3,*\}$.  We suppose that the following operational equivalences among the preparations hold:
\begin{align}\label{eq:opeqprepsspeck}
P^{\rm (ave)}_* \simeq P^{\rm (ave)}_1 \simeq P^{\rm (ave)}_2 \simeq P^{\rm (ave)}_3.
\end{align} 

We can now state our noncontextuality inequalities for Specker's scenario:
\begin{theorem}\label{speckertheorem}
An operational theory which satisfies the operational equivalences of Eqs.~\eqref{eq:opeqmmtsspeck},
\eqref{eq:opeqmmtsspeck2}, and \eqref{eq:opeqprepsspeck},
and admits a noncontextual ontological model must necessarily satisfy the following noncontextuality inequality in Specker's scenario:
\begin{align}\label{eq:maininequality}
&p({\rm anti}|M_*,P_*)+p({\rm anti}|M'_*,P_*^{\perp})\nonumber \\
&\le 2\left(1-\frac{1}{3}\eta_{\rm ave}\right),
\end{align}
where 
\begin{equation}
\eta_{\rm ave}\equiv \frac{1}{6}\sum_{i=1}^3 \left(\eta(M_{i},P_i)+\eta(M_{i},P^{\perp}_{i})\right).
\end{equation}
On the other hand, using only the operational equivalences of Eqs.~\eqref{eq:opeqmmtsspeck} and \eqref{eq:opeqprepsspeck}, such an operational
theory must also satisfy:
\begin{equation}\label{eq:speckineq1}
p({\rm anti}|M_*,P_*)+p({\rm anti}|M_*,P^{\perp}_*)\leq 2\left(1-\frac{1}{3}\eta_{\rm ave}\right),
\end{equation}
and
\begin{align}\label{eq:speckineq2}
&p({\rm anti}|M_*,P_*)\nonumber \\
&\le \frac{2}{3}\left(2-\eta_{\rm ave}\right).
\end{align}
\end{theorem}

The proof can be found in the Appendix to this chapter.
The upper bound in Eqs.~(\ref{eq:maininequality}) and (\ref{eq:speckineq1}) is nontrivial for all values of $\eta_{\rm ave}>0$
while the upper bound in Eq.~(\ref{eq:speckineq2}) is nontrivial only for values of $\eta_{\rm ave}>\frac{1}{2}$. 
When $\eta_{\rm ave}>\frac{1}{2}$,
Eq.~(\ref{eq:speckineq2}) requires fewer measurement and preparation procedures than Eqs.~(\ref{eq:maininequality}) 
and (\ref{eq:speckineq1}) in order to refute noncontextuality.
Eq.~(\ref{eq:speckineq1}) in turn requires fewer measurement procedures to be implemented than Eq.~(\ref{eq:maininequality}).

An analysis of Specker's scenario according to KS-noncontextuality would require that the probability of anticorrelation is bounded above by $2/3$ 
for both $M_*$ and $M'_*$, so that $p({\rm anti}|M_*,P)+p({\rm anti}|M'_*,P^{\perp}) \le 4/3$.
Similarly, in such an analysis, $p({\rm anti}|M_*,P_*)+p({\rm anti}|M_*,P^{\perp}_*) \le 4/3$. But from Theorem \ref{speckertheorem} it is clear that
these inequalities are not warranted by the assumption of noncontextuality alone.
The noncontextual bound of $4/3$ will hold if and only if one has verified that $\eta(M_{i},P_i)=\eta(M_{i},P^{\perp}_{i})=1$
for all $M_i, P_i, P^{\perp}_i$, $i\in\{1,2,3\}$. That is, when each measurement $M_i$ produces deterministic outcomes 
on both preparations $P_i$ and $P^{\perp}_i$. In this case, $\eta_{\rm ave}=1$. At the other extreme, if each $M_i$ has no dependence on the corresponding 
preparation procedures $P_i$ and $P_i^{\perp}$, so that $\eta(M_{i},P_i)=\eta(M_{i},P_i^{\perp})=0$, and therefore $\eta_{\rm ave}=0$,
then a noncontextual model can achieve perfect anticorrelation.
{\em A mere observation of perfect anticorrelation on its own, therefore, is not enough to demonstrate contextuality}: one also needs to check that the average predictability
is sufficiently large: $\eta_{\rm ave}>0$ for Eqs.~(\ref{eq:maininequality}), (\ref{eq:speckineq1}), and $\eta_{\rm ave}>\frac{1}{2}$ for Eq.~(\ref{eq:speckineq2}).
Our noncontextuality inequalities in Eqs.~(\ref{eq:maininequality}), (\ref{eq:speckineq1}), and (\ref{eq:speckineq2}) imply 
a quantitative \emph{tradeoff} between operational quantities:
the \emph{anticorrelation} achievable in an operational theory admitting a noncontextual ontological model and the 
\emph{predictabilities} of the measurements
involved with respect to various preparations on which they are carried out. 

We will see that in operational quantum theory the
noncontextuality inequality of Eq.~(\ref{eq:maininequality}) can be violated.
It follows that if operational quantum theory correctly describes our experiments and one can devise an experiment that is 
sufficiently precise that it can approach the violation predicted by quantum theory, then this experiment should yield
a violation of the noncontextuality inequality. Moreover, if such an experiment is performed and the violation is 
observed, then this observation rules out the existence of a noncontextual ontological model {\em regardless of the
validity of operational quantum theory}. Hence the result is theory-independent. 

We leave open the question of whether a quantum violation exists for the noncontextuality inequalities of Eqs.~(\ref{eq:speckineq1}) and (\ref{eq:speckineq2}).
Our construction for the quantum violation of Eq.~(\ref{eq:maininequality}) does not suggest such a violation.
Our noncontextuality inequalities of Eqs.~(\ref{eq:maininequality}), (\ref{eq:speckineq1}), and (\ref{eq:speckineq2}),
are a proper operational generalization of the LSW inequality which was first obtained in Ref.~\cite{LSW} 
and for which a quantum violation was first theoretically predicted in Ref.~\cite{KG} followed by an experimental demonstration of this in Ref.~\cite{Xue}. 
While the LSW inequality holds for noncontextual
models of operational \emph{quantum} theory, where outcome-determinism for projective measurements can be justified from preparation
noncontextuality, our noncontextuality inequalities apply to \emph{arbitrary} operational theories.

\begin{figure}
\centering
\includegraphics[scale=0.5]{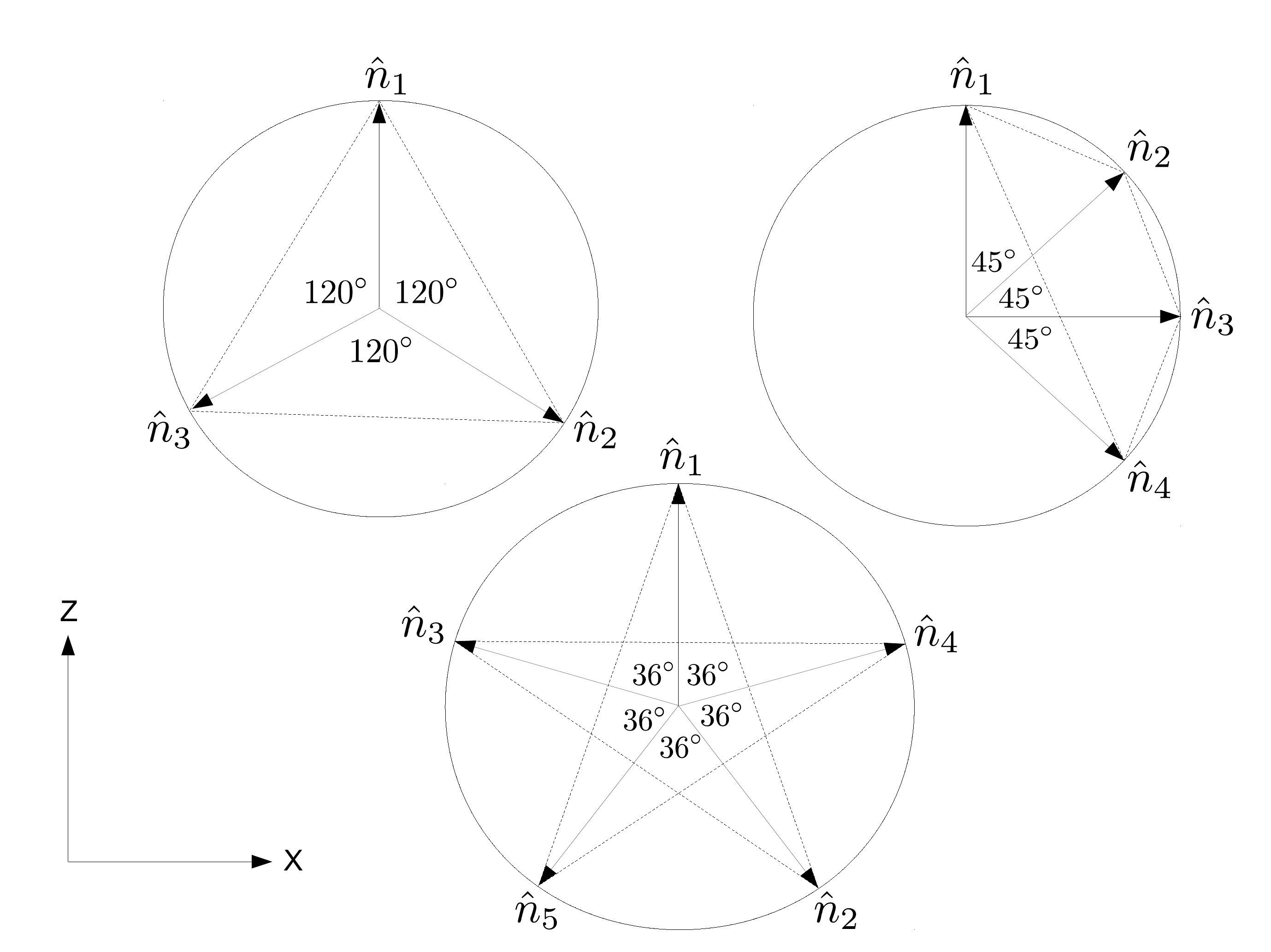}
\caption{Choice of measurement directions for $n=3,4,5$.}
\label{ncyclemmts}
\end{figure}

\begin{figure}
\centering
\includegraphics[scale=0.5]{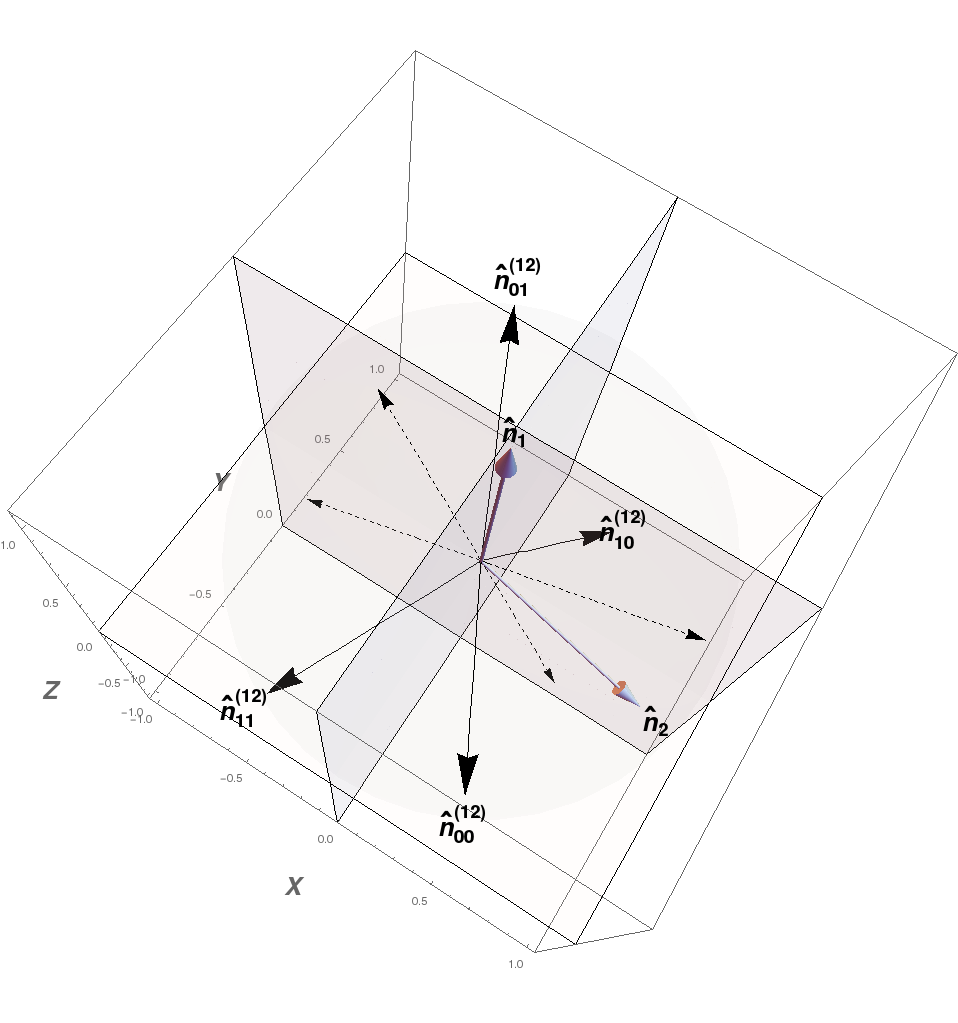}
\caption{Choice of joint POVM directions for $\hat{n}_1$ and $\hat{n}_2$ in Specker's scenario ($n=3$), $\eta_0=0.4566$.}
\label{speckerjoint}
\end{figure}

\subsection{Quantum violation of a noncontextuality inequality in Specker's scenario}
Inspired by the construction we presented in Chapter 2 (following Ref.~\cite{KG}), we show the quantum violation of 
Eq.~(\ref{eq:maininequality}).
We take $\{M_1,M_2,M_3\}$ to be three qubit measurements in an equatorial plane of the Bloch sphere, say the $ZX$ plane,
such that they are pairwise jointly measurable. $M_i$ is associated with the qubit POVM $\{E^{(i)}_0,E^{(i)}_1\}$, given by:
\begin{equation}
E^{(i)}_{X_i}\equiv \frac{1}{2}I+(-1)^{X_i}\frac{1}{2}\eta_{0}\vec{\sigma}\cdot \hat{n}_i,
\end{equation}
where $\hat{n}_i$ is the measurement direction, $\vec{\sigma}\equiv(\sigma_x,\sigma_y,\sigma_z)$ is the vector of qubit Pauli 
matrices, and $I$ is the identity matrix: 
\begin{eqnarray}\nonumber
\sigma_x=\left(\begin{matrix} 0&1\\1&0\end{matrix}\right),
\sigma_y=\left(\begin{matrix} 0&-i\\i&0\end{matrix}\right),
\sigma_z=\left(\begin{matrix} 1&0\\0&-1\end{matrix}\right),
I=\left(\begin{matrix} 1&0\\0&1\end{matrix}\right).
\end{eqnarray}
The necessary and sufficient condition for pairwise joint measurability of $\{M_1,M_2,M_3\}$ \cite{KG,heinosaari,YuOh} is 
\beq\label{etabound}
\eta_0 \le \min_{(i,j)} \frac{1}{\sqrt{1+ \sqrt{ 1- (\hat{n}_i \cdot \hat{n}_j )^2 }}}.
\eeq
We consider two sets of pairwise joint measurements denoted by $\{M_{12},M_{23},M_{31}\}$ and $\{M'_{12},M'_{23},M'_{31}\}$, where 
$M_{ij}$ is associated with the qubit POVM $\{E^{(ij)}_{00},E^{(ij)}_{01},E^{(ij)}_{10},E^{(ij)}_{11}\}$ and $M'_{ij}$ is 
associated with the qubit POVM $\{E^{\prime (ij)}_{00},E^{\prime (ij)}_{01},E^{\prime (ij)}_{10},E^{\prime (ij)}_{11}\}$. These
are given by:
\begin{equation}
 E^{(ij)}_{X_iX_j}\equiv\frac{1}{2}\left(1+(-1)^{X_i+X_j} \eta_{0}^2 \hat{n}_i.\hat{n}_j\right)\Pi_{\hat{n}^{ij}_{X_iX_j}}
\end{equation}
where 
\begin{eqnarray}
&&\Pi_{\hat{n}^{ij}_{X_iX_j}}\equiv\frac{1}{2}(I+\vec{\sigma}.\hat{n}^{ij}_{X_iX_j}),\nonumber\\
&&\hat{n}^{ij}_{X_iX_j}\equiv\frac{\eta_{0}((-1)^{X_i}\hat{n}_i+(-1)^{X_j}\hat{n}_j)-(-1)^{X_i+X_j}\vec{a}_{ij}}{1+(-1)^{X_i+X_j} \eta_{0}^2 \hat{n}_i.\hat{n}_j},\nonumber
\end{eqnarray}
and
\begin{equation}
 E^{\prime(ij)}_{X_iX_j}\equiv\frac{1}{2}\left(1+(-1)^{X_i+X_j} \eta_{0}^2 \hat{n}_i.\hat{n}_j\right)\Pi_{\hat{n}^{\prime ij}_{X_iX_j}}
\end{equation}
where
\begin{eqnarray}
&&\Pi_{\hat{n}^{\prime ij}_{X_iX_j}}\equiv\frac{1}{2}(I+\vec{\sigma}.\hat{n}^{\prime ij}_{X_iX_j}),\nonumber\\
&&\hat{n}^{\prime ij}_{X_iX_j}\equiv\frac{\eta_{0}((-1)^{X_i}\hat{n}_i+(-1)^{X_j}\hat{n}_j)+(-1)^{X_i+X_j}\vec{a}_{ij}}{1+(-1)^{X_i+X_j} \eta_{0}^2 \hat{n}_i.\hat{n}_j},\nonumber
\end{eqnarray}
and
\begin{equation}
\vec{a}_{ij}\equiv(0,\sqrt{1+\eta_{0}^4(\hat{n}_i.\hat{n}_j)^2-2\eta_{0}^2},0).\nonumber
\end{equation}

It is easy to verify that the operational equivalences of Eqs.~(\ref{eq:opeqmmtsspeck}) and (\ref{eq:opeqmmtsspeck2}) hold for this choice of pairwise joint measurements.
Let $P_x$ and $P^{\perp}_x$ be the preparation procedures associated respectively with
$$\rho_x  \equiv \frac{1}{2}I+ \frac{1}{2} \vec{\sigma}\cdot\hat{n}_x=|+\hat{n}_x\rangle\langle+\hat{n}_x|,$$
and $$\rho^{\perp}_x  \equiv \frac{1}{2}I- \frac{1}{2} \vec{\sigma}\cdot \hat{n}_x=|-\hat{n}_x\rangle\langle-\hat{n}_x|,$$
where $x\in \{ *,1,2,3\}$. 
Because $\frac{1}{2}\rho_x+\frac{1}{2}\rho^{\perp}_x = \frac{1}{2}I$ for all $x\in \{ *,1,2,3\}$, 
the preparation procedures $P^{({\rm ave})}_1$, $P^{({\rm ave})}_2$, $P^{({\rm ave})}_3$ and $P^{({\rm ave})}_*$ are all associated to the same mixed state and consequently the 
operational equivalences of Eq.~\eqref{eq:opeqprepsspeck} are satisfied.

The preparation $P_*$ is associated with a vector $\hat{n}_*$ perpendicular to the $ZX$ plane. We choose:
\begin{eqnarray}
\hat{n}_1 &\equiv&(0,0,1),\nonumber\\
\hat{n}_2 &\equiv&(\frac{\sqrt{3}}{2},0,-\frac{1}{2}),\nonumber\\
\hat{n}_3 &\equiv&(-\frac{\sqrt{3}}{2},0,-\frac{1}{2}),\nonumber\\
\hat{n}_* &\equiv& (0,1,0),
\end{eqnarray}
so that $\hat{n}_1,\hat{n}_2,\hat{n}_3,$ lie in the $ZX$ plane and $\hat{n}_*$ is perpendicular to it (Figs. \ref{ncyclemmts} and \ref{speckerjoint}).
Our choice requires $\eta_{0}\in \left[0,\sqrt{3}-1\right]$, as can be inferred from Eq.~(\ref{etabound}).
Clearly, $\forall i\in\{1,2,3\}: \eta(M_{i},P_{i})= \eta(M_{i},P^{\perp}_{i})=\eta_0$. The noncontextual bound in Eq.~\eqref{eq:maininequality} is therefore 
\begin{equation}
2-\frac{2}{3}\eta_0.
\end{equation}
We have:
\begin{align}
&p({\rm anti}|M_*,P_*)+ p({\rm anti}|M'_*,P^{\perp}_*)\nonumber\\
&=1+\frac{\eta_{0}^2}{2}+\sqrt{1+\frac{\eta_{0}^4}{4}-2\eta_{0}^2}.
\end{align}
The difference has the value
\begin{equation}
\left(\sqrt{1+\frac{\eta_{0}^4}{4}-2\eta_{0}^2} + \frac{\eta_{0}^2}{2} + \frac{2}{3}\eta_0 - 1\right).
\end{equation}
The largest violation of the inequality for our choice of preparations and measurements occurs when $\eta_0 \approx 0.4566$ so that the violation is $0.1793$:
in this case the noncontextual bound on the anticorrelation is $1.6956$ and the quantum value is $1.8749$. In the next section,
we generalize our analysis of Specker's scenario to the case of $n$-cycle scenarios.

\section{\lowercase{$n$}-cycle scenarios}
\subsection{Noncontextuality inequalities for $n$-cycle scenarios}
An $n$-cycle scenario consists of $n$ binary-outcome measurements, $\{M_1,\dots,M_n\}$, such that the pairs $\{M_i,M_{i+1\mod n}\}$ are jointly 
measurable for all $i\in \{1,\dots,n\}$. They fall into two categories: even n-cycle scenarios and odd
$n$-cycle scenarios. We will now generalize our operational noncontextuality inequality for Specker's scenario to odd $n$-cycle scenarios for all odd $n\geq3$.
We also prove noncontextuality inequalities for all even $n$-cycle scenarios, $n\geq4$. 
The joint measurability of pairs $\{M_i,M_j\}$, $j=i+1 \mod n$, with joint measurement denoted by $M_{ij}$ requires the operational equivalences: 
\begin{eqnarray}\label{eq:optlequivalencesMgen1}
M_1^{(2)}&\simeq&M_1^{(n)}\simeq M_1,\nonumber\\
M_2^{(1)}&\simeq&M_2^{(3)}\simeq M_2,\nonumber\\
M_3^{(2)}&\simeq&M_3^{(4)}\simeq M_3,\nonumber\\
&\vdots&\nonumber\\
M_n^{(n-1)}&\simeq&M_n^{(1)}\simeq M_n. 
\end{eqnarray}

Similarly, another set of joint measurements, $M'_{ij}$, 
require the operational equivalences:
\begin{eqnarray}\label{eq:optlequivalencesMgen2}
M_1^{\prime(2)}&\simeq&M_1^{\prime(n)}\simeq M_1,\nonumber\\
M_2^{\prime(1)}&\simeq&M_2^{\prime(3)}\simeq M_2,\nonumber\\
M_3^{\prime(2)}&\simeq&M_3^{\prime(4)}\simeq M_3,\nonumber\\
&\vdots&\nonumber\\
M_n^{\prime(n-1)}&\simeq&M_n^{\prime(1)}\simeq M_n. 
\end{eqnarray}

Finally, we define the measurement $M_*$ (respectively, $M_*^{\prime}$): 
sample $(i,j)$, where $i\in\{1,\dots,n\}$ and $j=i+1 \mod n$, with probability $1/n$ each (i.e. uniformly at random) and then implement $M_{ij}$ (respectively, $M'_{ij}$) and record the outcome $(X_i,X_j)$.

\emph{Odd n.}---For odd $n\geq3$, we compute the probability of recording anticorrelated outcomes,
\begin{align}
p(\text{anti}|M_*,P) \equiv& \frac{1}{n} \sum_{i=1}^n p(X_i\ne X_j|M_{ij},P)
\end{align}
where $j=i+1\mod n$ for a given $i$. We are also interested in $p(\text{anti}|M'_*,P)$.

\emph{Even n.}---For even $n\geq4$, we compute the probability of recording positively correlated outcomes for all pairs except $\{M_1,M_n\}$ for which we compute the probability of recording anticorrelated outcomes. This is the pattern of (anti)correlations 
in the ``chained Bell inequalities'' of Ref.~\cite{BCaves}, so we denote it by ``chained'':
\begin{align}
p(\text{chained}|M_*,P) &\equiv \frac{1}{n}\sum_{i=1}^{n-1} p(X_i=X_j|M_{ij},P)\nonumber\\
&+  \frac{1}{n}p(X_n\neq X_1|M_{n1},P).
\end{align}
We are also interested in the corresponding quantity for $M'_*$, $p(\text{chained}|M'_*,P)$.

Furthermore, let $P_*$, $P_*^{\perp}$, $P_i$, $P_i^{\perp}$, $i\in\{1,\dots,n\}$, be preparation procedures and let  $P^{\rm (ave)}_x$ be the preparation procedure obtained by implementing $P_x$ with probability $1/2$ and $P_x^{\perp}$ with probability $1/2$ for $x\in \{1,\dots,n,*\}$.  
We suppose that the following operational equivalences hold:
\begin{align}\label{eq:optlequivalencesPgen}
P^{\rm (ave)}_* \simeq P^{\rm (ave)}_1 \simeq P^{\rm (ave)}_2 \simeq \dots \simeq P^{\rm (ave)}_n.
\end{align}

We can now state our second theorem which recovers Theorem \ref{speckertheorem} as a special case for $n=3$.
\begin{theorem}\label{ncycletheorem}
An operational theory which satisfies the operational equivalences of Eqs.~(\ref{eq:optlequivalencesMgen1}), (\ref{eq:optlequivalencesMgen2}), and (\ref{eq:optlequivalencesPgen}),
and admits a noncontextual ontological model must necessarily satisfy the following noncontextuality inequalities:
\begin{equation}\label{eq:maininequalitygen}
p({\rm anti}|M_*,P_*)+p({\rm anti}|M'_*,P_*^{\perp})\le 2\left(1-\frac{1}{n}\eta_{\rm ave}\right)
\end{equation}
for odd $n\geq3$, and
\begin{equation}\label{eq:maininequality_chsh}
p({\rm chained}|M_*,P_*)+p({\rm chained}|M'_*,P_*^{\perp})\le 2\left(1-\frac{1}{n}\eta_{\rm ave}\right)
\end{equation}
for even $n\geq4$, where
\begin{equation}
\eta_{\rm ave}\equiv\frac{1}{2n}\sum_{i=1}^n \left(\eta(M_{i},P_i)+\eta(M_{i},P^{\perp}_{i})\right).
\end{equation}
On the other hand, using only the operational equivalences of Eqs.~(\ref{eq:optlequivalencesMgen1}) and (\ref{eq:optlequivalencesPgen}), 
such an operational theory must satisfy:
\begin{equation}\label{eq:oddncycleineq1}
p({\rm anti}|M_*,P_*)+p({\rm anti}|M_*,P^{\perp}_*)\leq 2\left(1-\frac{1}{n}\eta_{\rm ave}\right),
\end{equation}
and
\begin{align}\label{eq:oddncycleineq2}
&p({\rm anti}|M_*,P_*)\nonumber \\
&\le \frac{n-1}{n}+2\frac{(1-\eta_{\rm ave})}{n},
\end{align}
for odd $n$-cycle scenarios, besides 
\begin{equation}\label{eq:evenncycleineq1}
p({\rm chained}|M_*,P_*)+p({\rm chained}|M_*,P^{\perp}_*)\leq 2\left(1-\frac{1}{n}\eta_{\rm ave}\right)
\end{equation}
and
\begin{align}\label{eq:evenncycleineq2}
&p({\rm chained}|M_*,P_*)\nonumber \\
&\le \frac{n-1}{n}+2\frac{(1-\eta_{\rm ave})}{n}
\end{align}
for even $n$-cycle scenarios.
\end{theorem}
The proof of these inequalities generalizes the proof of Theorem \ref{speckertheorem} and is provided in the Appendix to this chapter.
These inequalities quantify the tradeoff between
the degree to which a pattern of correlations---``anti'' or ``chained''---is achievable in an operational theory which admits a
noncontextual ontological model versus the operational predictabilities of the measurements involved. We will now 
show a quantum violation of Eqs.~(\ref{eq:maininequalitygen}) and (\ref{eq:maininequality_chsh}), while we leave open the question of whether quantum violation of Eqs.~(\ref{eq:oddncycleineq1}),
(\ref{eq:oddncycleineq2}), (\ref{eq:evenncycleineq1}), (\ref{eq:evenncycleineq2}) is possible.

\subsection{Quantum violation of $n$-cycle noncontextuality inequalities}
\emph{Quantum realization for odd $n\geq3$.}---We can violate the operational inequality of Eq.~\eqref{eq:maininequalitygen} in quantum theory using a 
generalization of the $n=3$ construction in Specker's scenario. The value of $\hat{n}_i.\hat{n}_j$ is given by $\hat{n}_i.\hat{n}_j=\cos \frac{n-1}{n}\pi$, where $i\in\{1,\dots,n\}$
and $j=(i+1)\mod n$. That is, our measurements are in an equatorial plane of the Bloch sphere, say the $ZX$ plane, such that $\hat{n}_i$ and $\hat{n}_j$ 
are at an angle of $\frac{n-1}{n}\pi$ relative to each other: $\hat{n}_k \equiv \left(\sin \frac{(k-1)(n-1)}{n}\pi,0,\cos \frac{(k-1)(n-1)}{n}\pi\right)$,
for all $k\in\{1,2,\dots,n\}$, and, as before, $\hat{n}_* \equiv (0,1,0)$. Our construction of the pairwise joint measurements proceeds exactly as 
in the $n=3$ case described earlier, the joint POVMs given by $M_{ij}=\{E^{(ij)}_{X_iX_j}\}$ and $M^{\prime}_{ij}=\{E^{\prime(ij)}_{X_iX_j}\}$. That this construction leads to
a quantum violation will be shown below.

\emph{Quantum realization for even $n\geq4$.}---We can violate the operational inequality of Eq.~\eqref{eq:maininequality_chsh} in quantum theory
using the following construction: our choice of measurements is given by $\hat{n}_i.\hat{n}_j=\cos \frac{\pi}{n}$, where $i\in\{1,\dots,n-1\}$ and $j=i+1$,
and $\hat{n}_n.\hat{n}_1=\cos \frac{(n-1)\pi}{n}$: $\hat{n}_k \equiv (\sin \frac{(k-1)\pi}{n},0,\cos \frac{(k-1)\pi}{n})$ for all $k\in\{1,2,\dots,n\}$. Also, $\hat{n}_* \equiv (0,1,0)$.
The joint POVMs are given by $M_{ij}=\{F^{(ij)}_{X_iX_j}\}$ and $M^{\prime}_{ij}=\{F^{\prime(ij)}_{X_iX_j}\}$, where $F^{(n1)}_{X_nX_1}=E^{(n1)}_{X_nX_1}$
and $F^{\prime (n1)}_{X_nX_1}=E^{\prime (n1)}_{X_nX_1}$, while for $i\in\{1,\dots,n-1\}, j=i+1$:
\begin{equation}
F^{(ij)}_{X_iX_j}\equiv\frac{1}{2}\left(1-(-1)^{X_i+X_j} \eta_{0}^2 \hat{n}_i.\hat{n}_j\right)\Pi_{\hat{n}^{ij}_{X_iX_j}}
\end{equation}
where 
\begin{eqnarray}
&&\Pi_{\hat{n}^{ij}_{X_iX_j}}\equiv\frac{1}{2}(I+\vec{\sigma}.\hat{n}^{ij}_{X_iX_j}),\nonumber\\
&&\hat{n}^{ij}_{X_iX_j}\equiv\frac{\eta_{0}((-1)^{X_i}\hat{n}_i+(-1)^{X_j}\hat{n}_j)+(-1)^{X_i+X_j}\vec{a}_{ij}}{1-(-1)^{X_i+X_j} \eta_{0}^2 \hat{n}_i.\hat{n}_j},\nonumber
\end{eqnarray}
\begin{equation}
F^{\prime(ij)}_{X_iX_j}\equiv\frac{1}{2}\left(1-(-1)^{X_i+X_j} \eta_{0}^2 \hat{n}_i.\hat{n}_j\right)\Pi_{\hat{n}^{\prime ij}_{X_iX_j}}
\end{equation}
where
\begin{eqnarray}
&&\Pi_{\hat{n}^{\prime ij}_{X_iX_j}}\equiv\frac{1}{2}(I+\vec{\sigma}.\hat{n}^{\prime ij}_{X_iX_j}),\nonumber\\
&&\hat{n}^{\prime ij}_{X_iX_j}\equiv\frac{\eta_{0}((-1)^{X_i}\hat{n}_i+(-1)^{X_j}\hat{n}_j)-(-1)^{X_i+X_j}\vec{a}_{ij}}{1-(-1)^{X_i+X_j} \eta_{0}^2 \hat{n}_i.\hat{n}_j},\nonumber
\end{eqnarray}
and
\begin{equation}
\vec{a}_{ij}\equiv(0,\sqrt{1+\eta_{0}^4(\hat{n}_i.\hat{n}_j)^2-2\eta_{0}^2},0).\nonumber
\end{equation}

\emph{Quantum violation for all $n\geq3$.}---For all $n\geq3$, the noncontextual upper bound is
$$2 - 2\frac{\eta_0}{n}.$$

The quantum value for both odd and even $n$ takes the same form given our construction. For odd $n\geq 3$:
\begin{eqnarray}
&&p({\rm anti}|M_*,P_*)+ p({\rm anti}|M'_*,P^{\perp}_*)\nonumber\\
&=&1+\eta_0^2 \cos \frac{\pi}{n}+\sqrt{1+\eta_{0}^4\left(\cos \frac{\pi}{n}\right)^2-2\eta_{0}^2}.
\end{eqnarray}
For even $n\geq 4$:
\begin{eqnarray}
&&p({\rm chained}|M_*,P_*)+ p({\rm chained}|M'_*,P^{\perp}_*)\nonumber\\
&=&1+\eta_0^2 \cos \frac{\pi}{n}+\sqrt{1+\eta_{0}^4\left(\cos \frac{\pi}{n}\right)^2-2\eta_{0}^2}.
\end{eqnarray}
The quantity on the left-hand-side of the noncontextuality inequality,
$p({\rm anti}|M_*,P_*)+ p({\rm anti}|M'_*,P^{\perp}_*)$ (for odd $n$) or $p({\rm chained}|M_*,P_*)+ p({\rm chained}|M'_*,P^{\perp}_*)$
(for even n), will be called a \emph{contextuality witness}.

The quantum violation is therefore given by
\beq
Q_{\rm viol}\equiv\sqrt{1+\eta_{0}^4\left(\cos \frac{\pi}{n}\right)^2-2\eta_{0}^2} + \eta_0^2 \cos \frac{\pi}{n} + 2\frac{\eta_0}{n} - 1
\eeq
Figures \ref{n3} - \ref{n200} depict the variation in the quantum and classical values as well as their difference 
(marked along the vertical axis) over the whole range of values of
$\eta_0$ (marked along the horizontal axis). For each $n$, beyond some {\em critical} value of $\eta_0$ the quantum violation disappears ($Q_{\rm viol}=0$)
and 
the statistics from the joint measurements then cannot rule out noncontextuality ($Q_{\rm viol}\leq0$). The upper bound on $\eta_0$
is given by the joint measurability condition of Eq.~\eqref{etabound}. Table \ref{qviol} lists the maximum $Q_{\rm viol}$, the corresponding
optimal $\eta_0$, critical $\eta_0$, and the upper bound on $\eta_0$ for a few values of $n$.
Note also that the 
noncontextual bound scales as $\sim 1-\frac{1}{n}$ and the quantum value scales as $\sim 1-\frac{1}{n^2}$ in the limit 
$n\rightarrow\infty$. It is an open question whether this quantum violation is optimal.

\begin{table}[htb!]
\centering
\begin{tabular}{ | c | c | c | c | c |}
 \hline
      $n$ & $Q_{\rm viol}$ & Optimal $\eta_0$ & Critical $\eta_0$ & Upper bound on $\eta_0$\\ \hline
      3 & 0.1793 & 0.4566 & 0.6981 & 0.7320\\ \hline
      4 & 0.1557 & 0.5029 & 0.7369 & 0.7653\\ \hline
      5 & 0.1393 & 0.5412 & 0.7693 & 0.7936\\ \hline
      6 & 0.1266 & 0.5727 & 0.7953 & 0.8164\\ \hline
      7 & 0.1164 & 0.5990 & 0.8163 & 0.8351\\ \hline
      8 & 0.1079 & 0.6213 & 0.8336 & 0.8504\\ \hline
      9 & 0.1007 & 0.6403 & 0.8479 & 0.8632\\ \hline
     10 & 0.0944 & 0.6569 & 0.8601 & 0.8740\\ \hline
     11 & 0.0889 & 0.6715 & 0.8704 & 0.8832\\ \hline
     12 & 0.0841 & 0.6822 & 0.8794 & 0.8912\\ \hline
     13 & 0.0798 & 0.6960 & 0.8872 & 0.8982\\ \hline
     14 & 0.0759 & 0.7064 & 0.8940 & 0.9044\\ \hline
     99 & 0.0160 & 0.8881 & 0.9829 & 0.9845\\ \hline
    100 & 0.0159 & 0.8887 & 0.9831 & 0.9846\\ \hline
    199 & 0.0086 & 0.9211 & 0.9914 & 0.9921\\\hline
    200 & 0.0085 & 0.9213 & 0.9914 & 0.9922\\\hline
\end{tabular}
\caption{Quantum violation for n-cyle scenarios}
\label{qviol}
\end{table}

\begin{figure}[htb!]
\centering
 \includegraphics[scale=0.52]{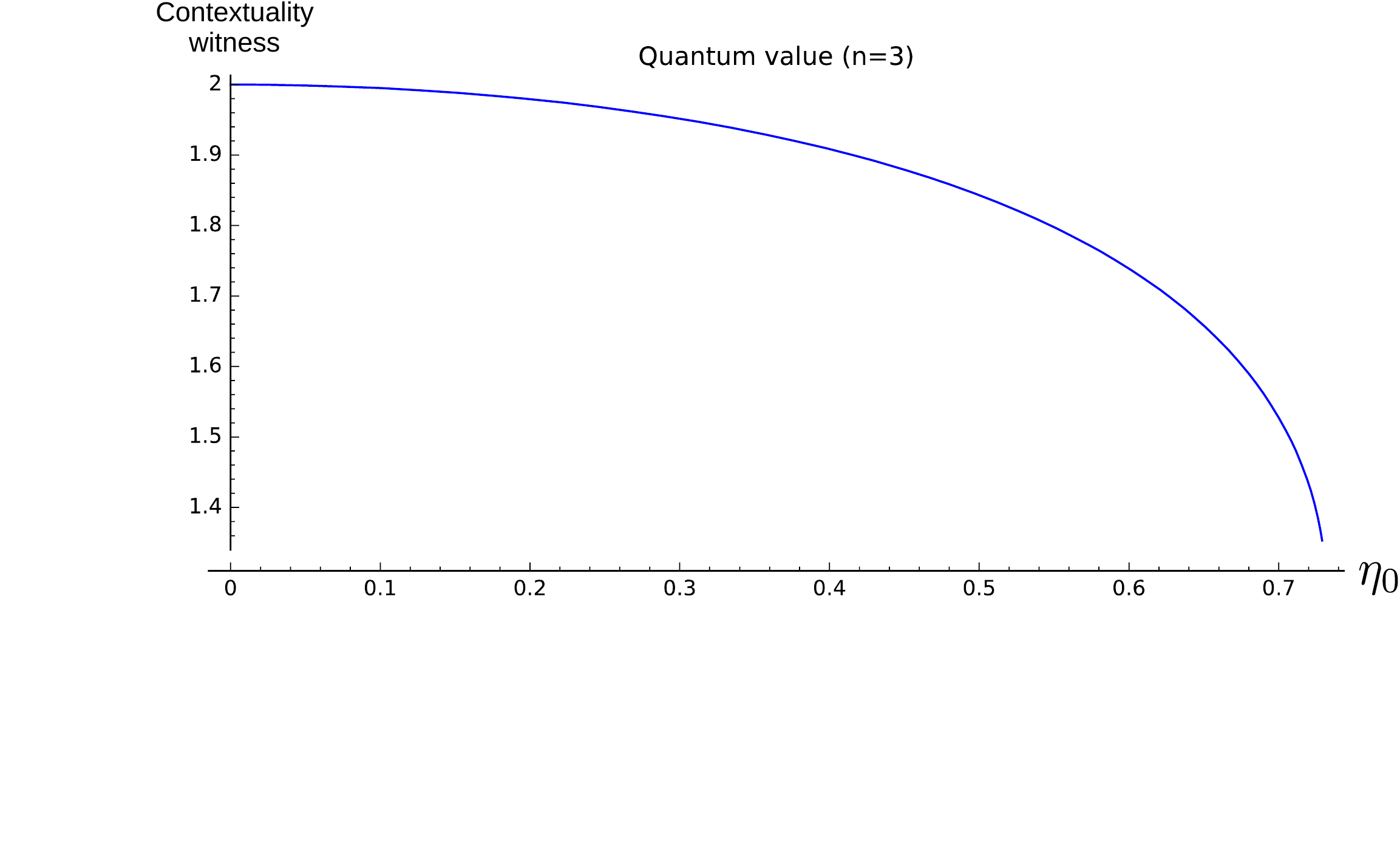}
 \includegraphics[scale=0.52]{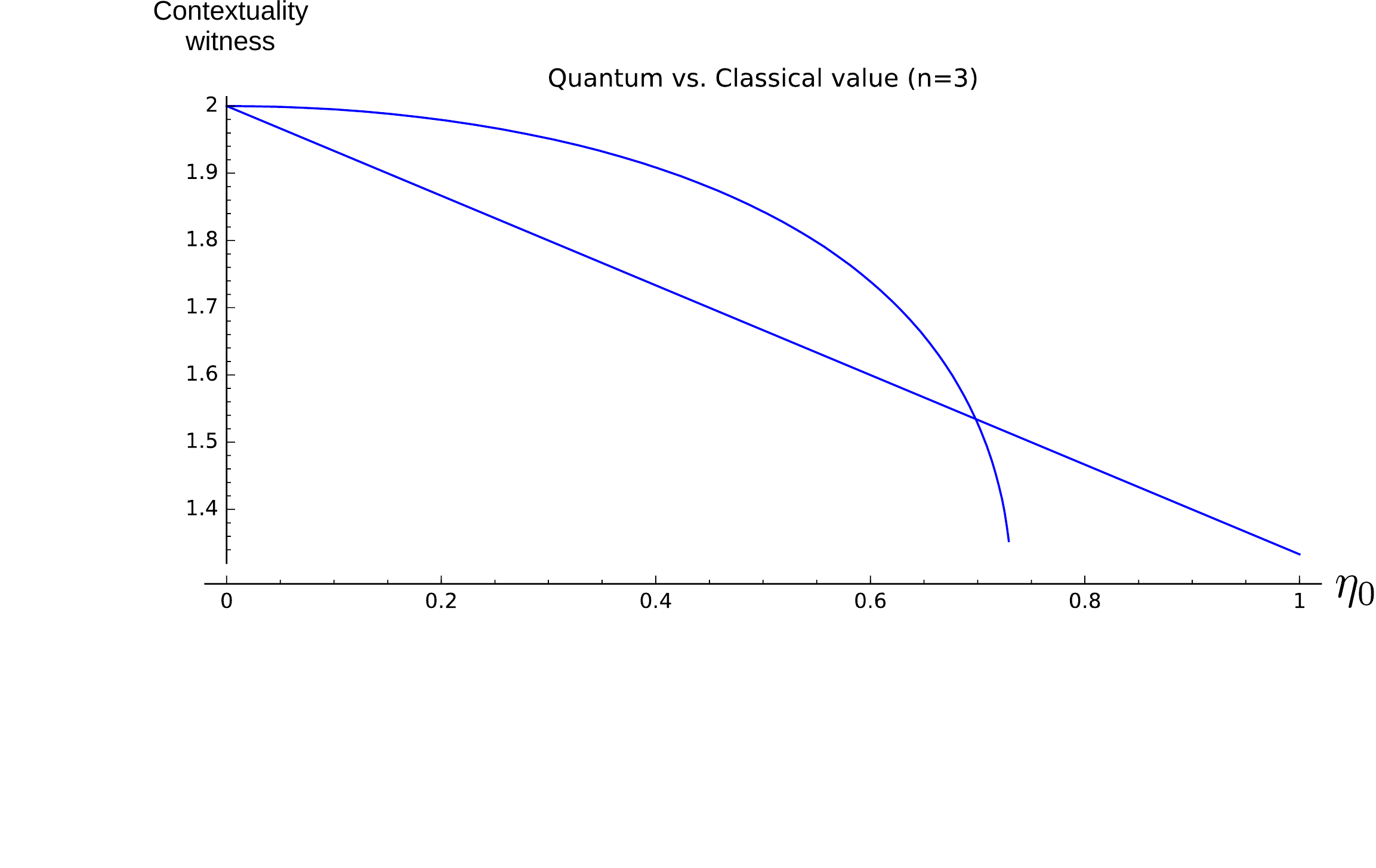}
 \includegraphics[scale=0.52]{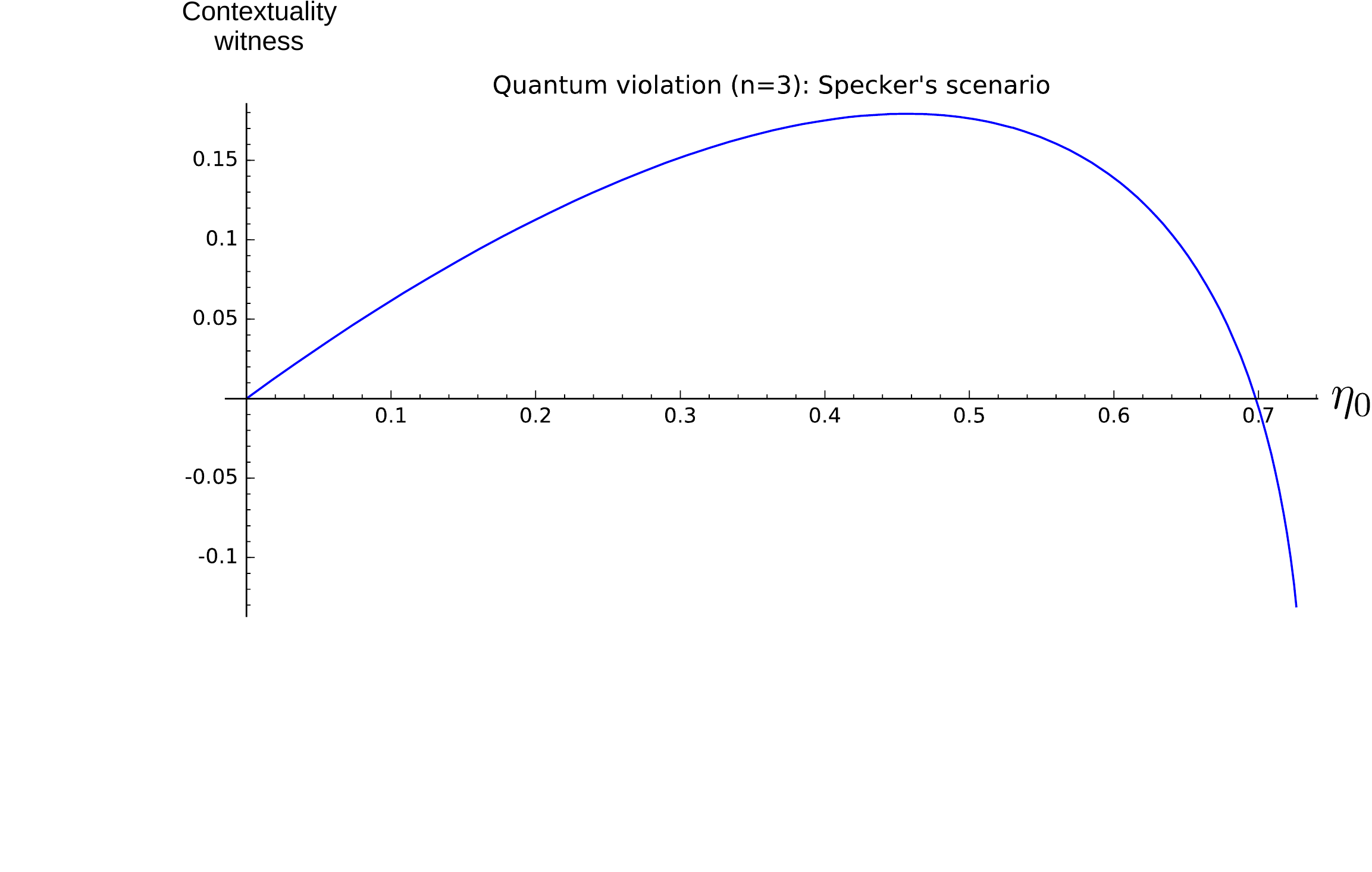}
\caption{Specker's scenario. In the middle figure, straight line denotes classical value.}
\label{n3}
\end{figure}

\begin{figure}
\centering
 \includegraphics[scale=0.52]{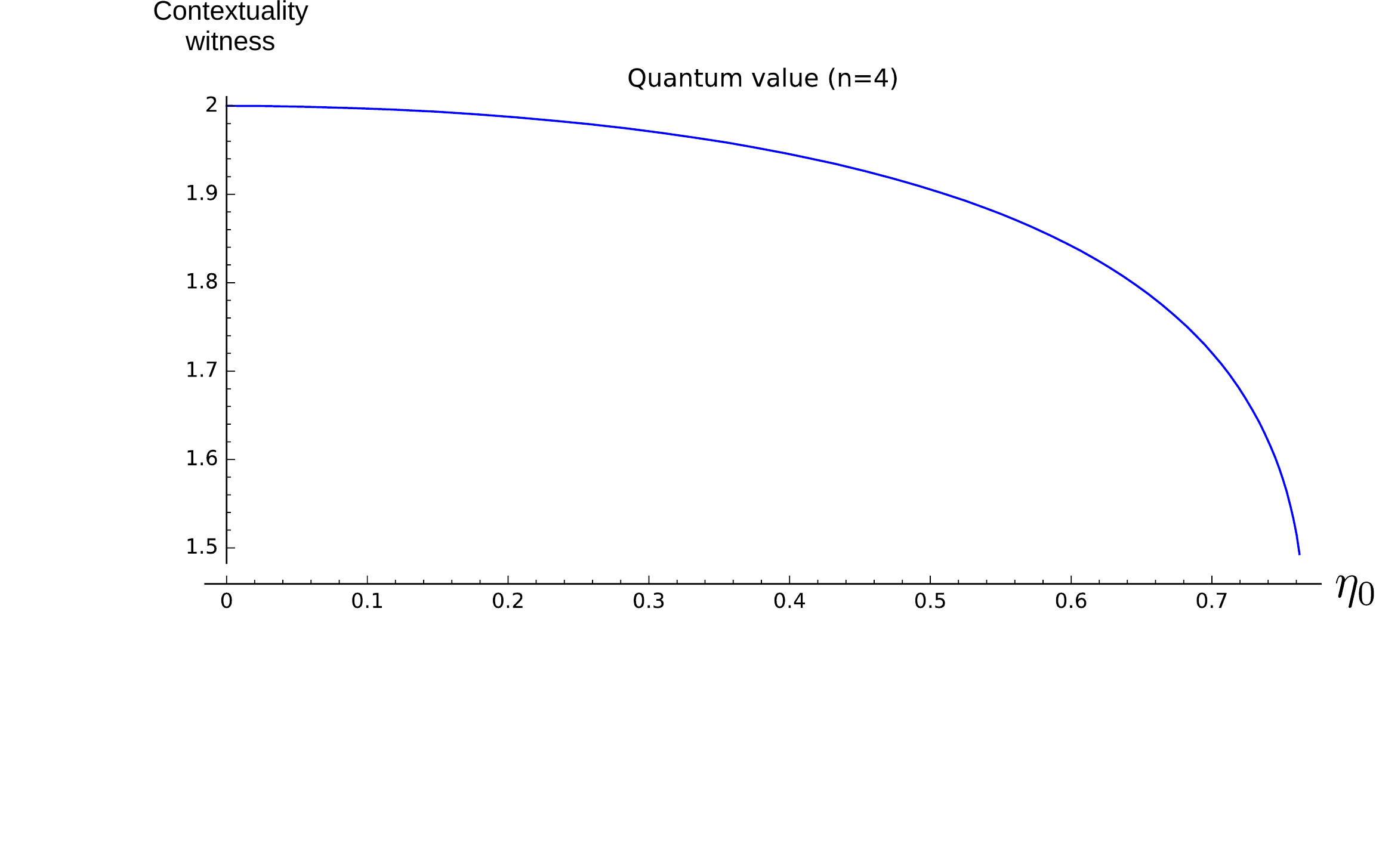}
 \includegraphics[scale=0.52]{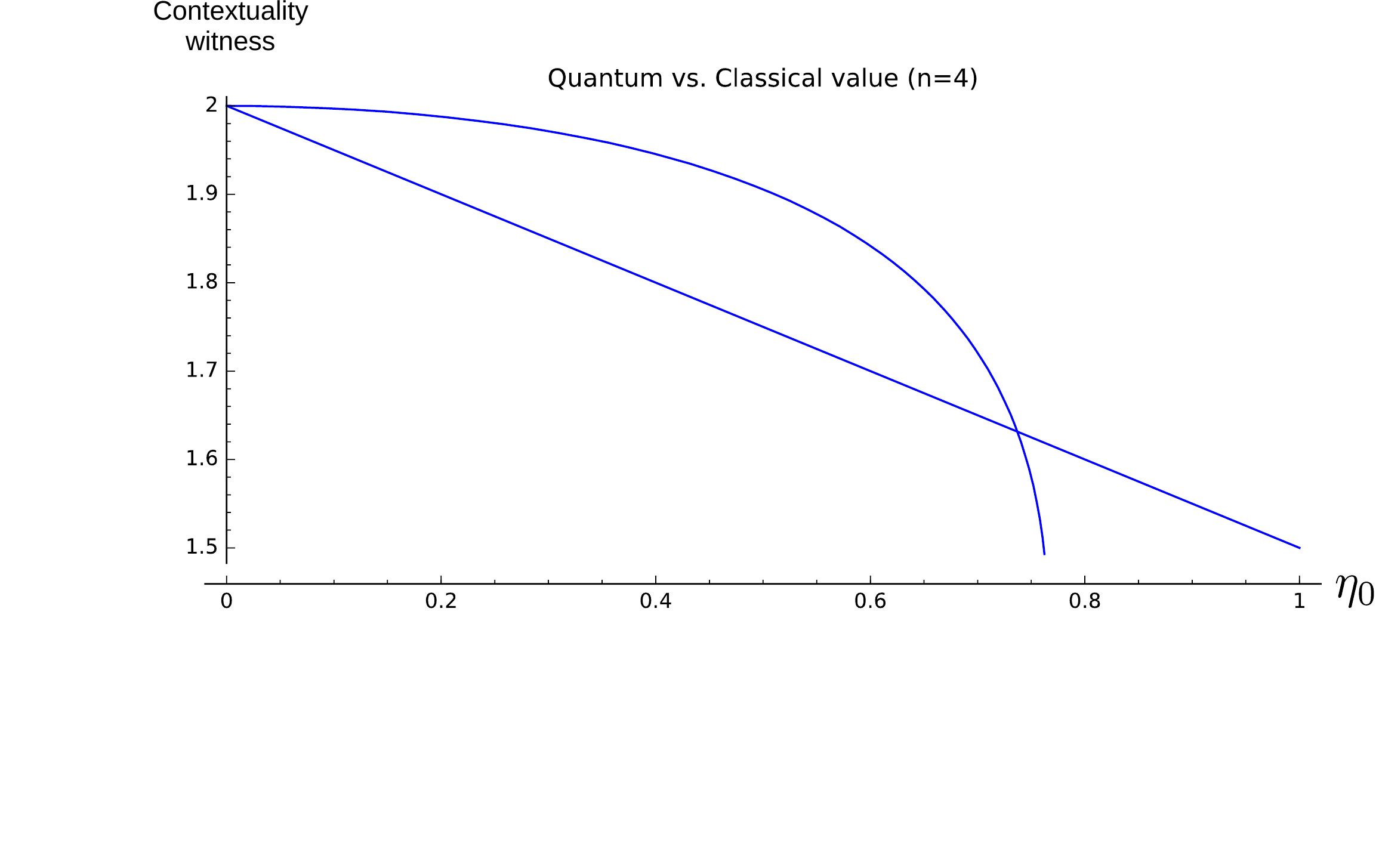}
 \includegraphics[scale=0.52]{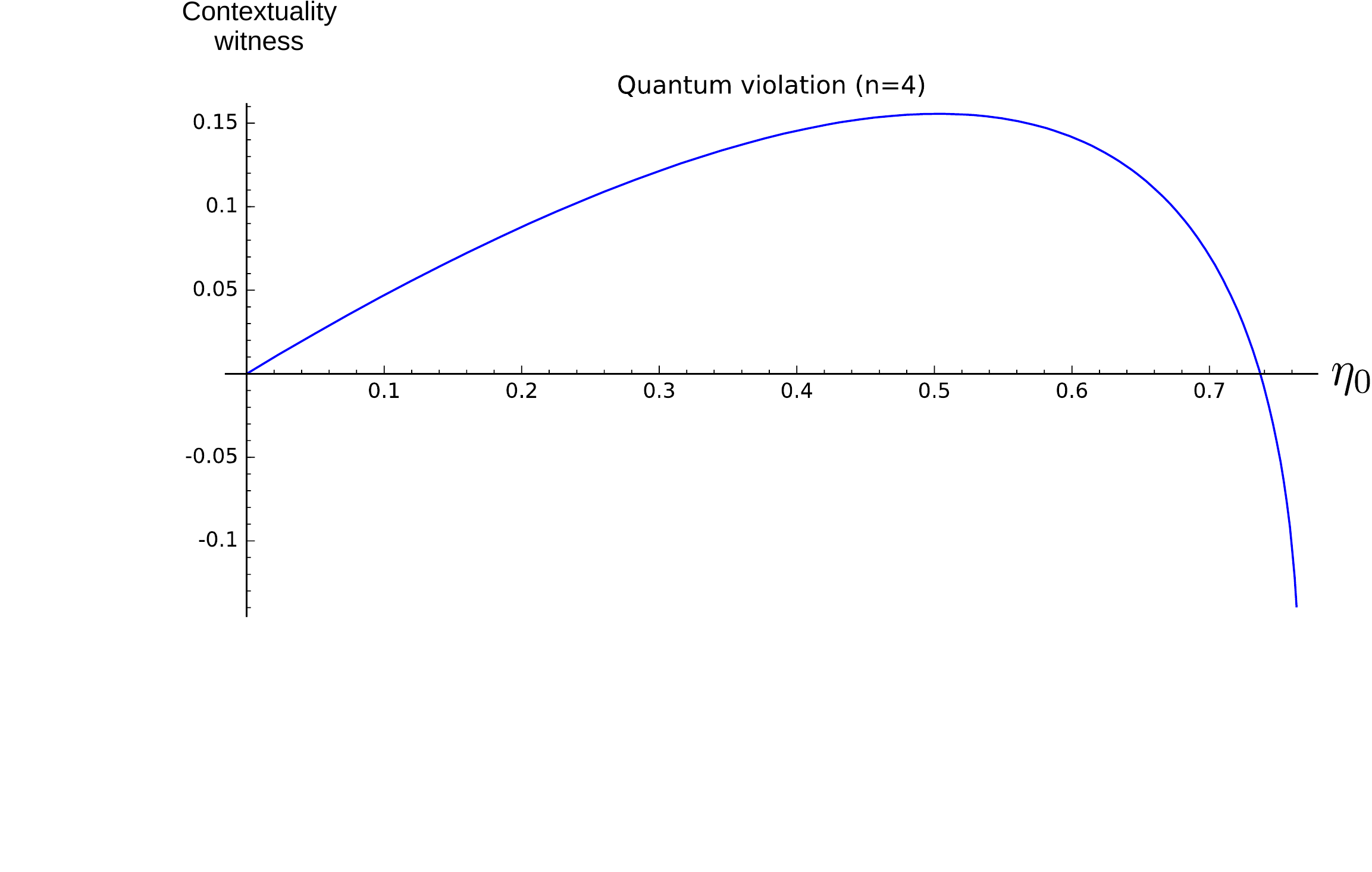}
\caption{n=4 scenario. In the middle figure, straight line denotes classical value.}
\label{n4}
\end{figure}

\begin{figure}
\centering
 \includegraphics[scale=0.52]{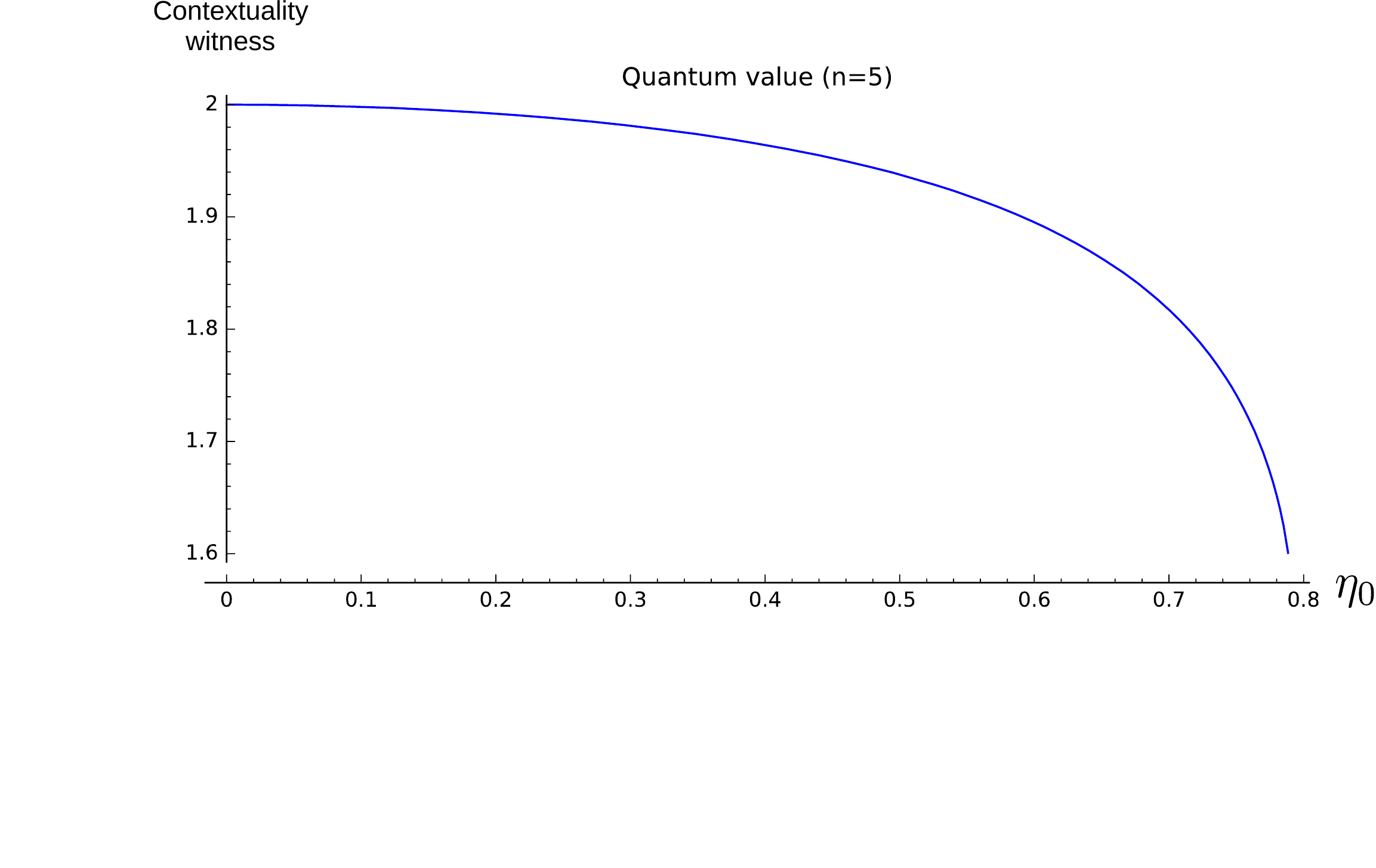}
 \includegraphics[scale=0.52]{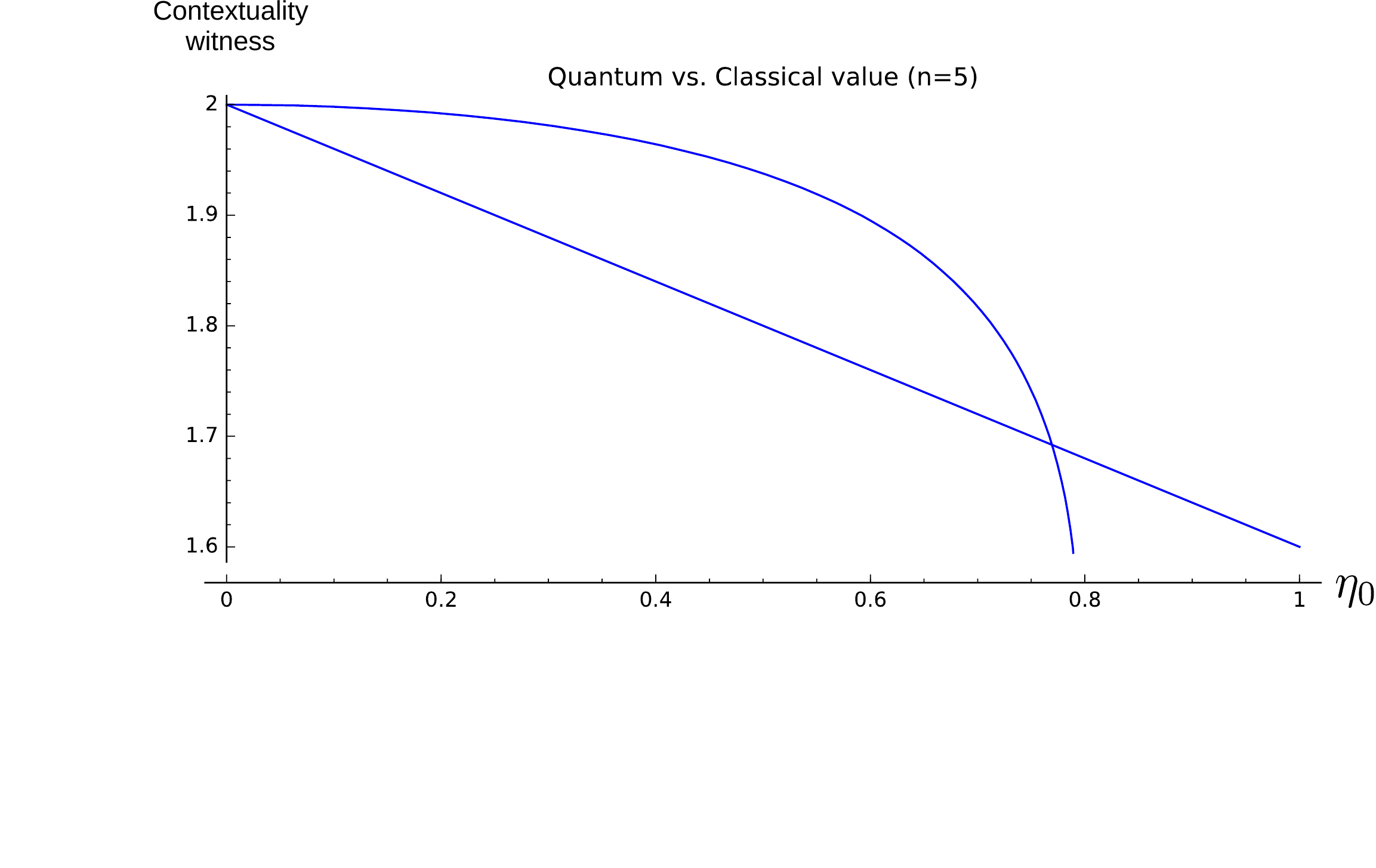}
 \includegraphics[scale=0.52]{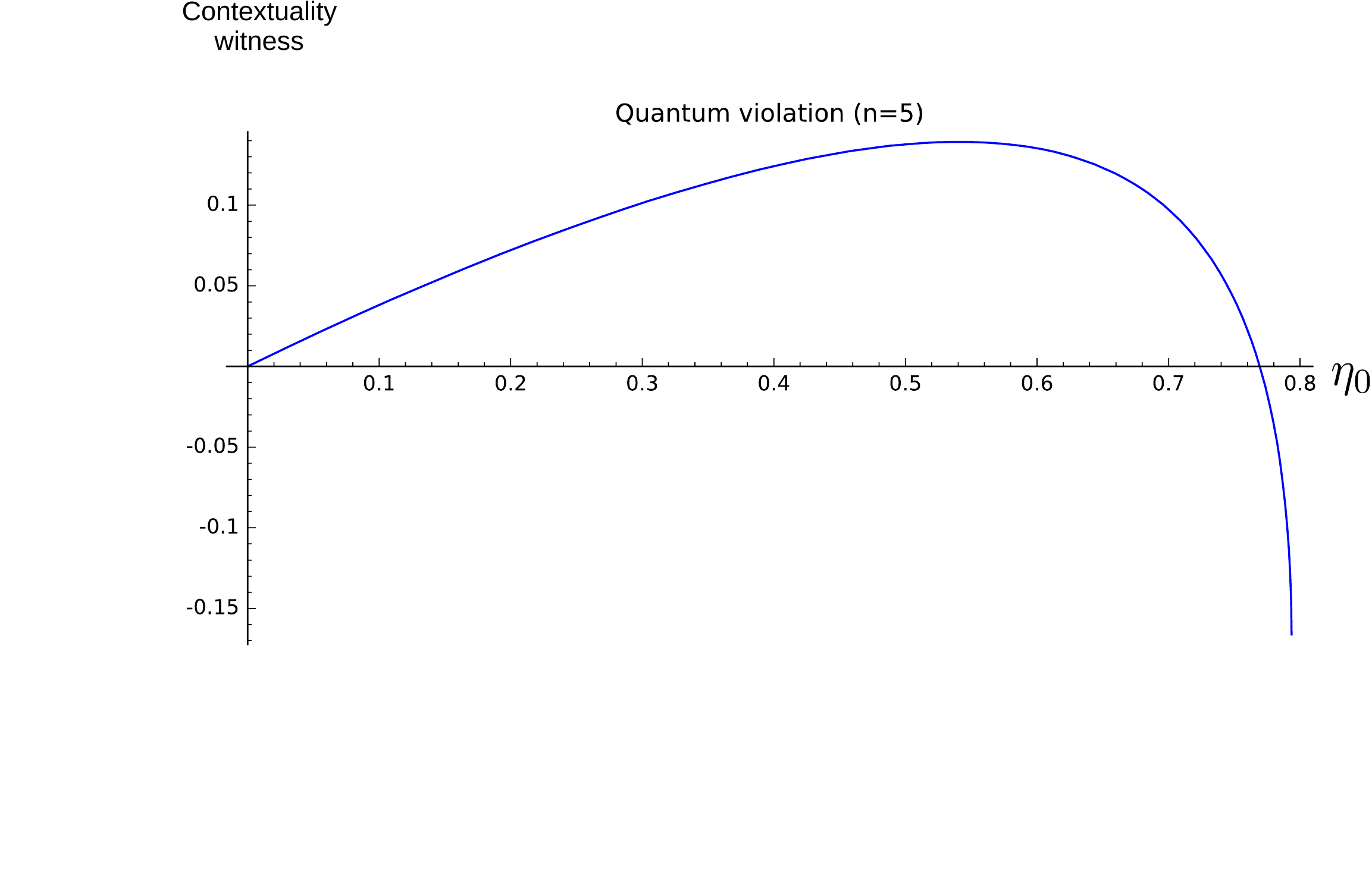}
\caption{n=5 scenario. In the middle figure, straight line denotes classical value.}
\label{n5}
\end{figure}

\begin{figure}
\centering
 \includegraphics[scale=0.52]{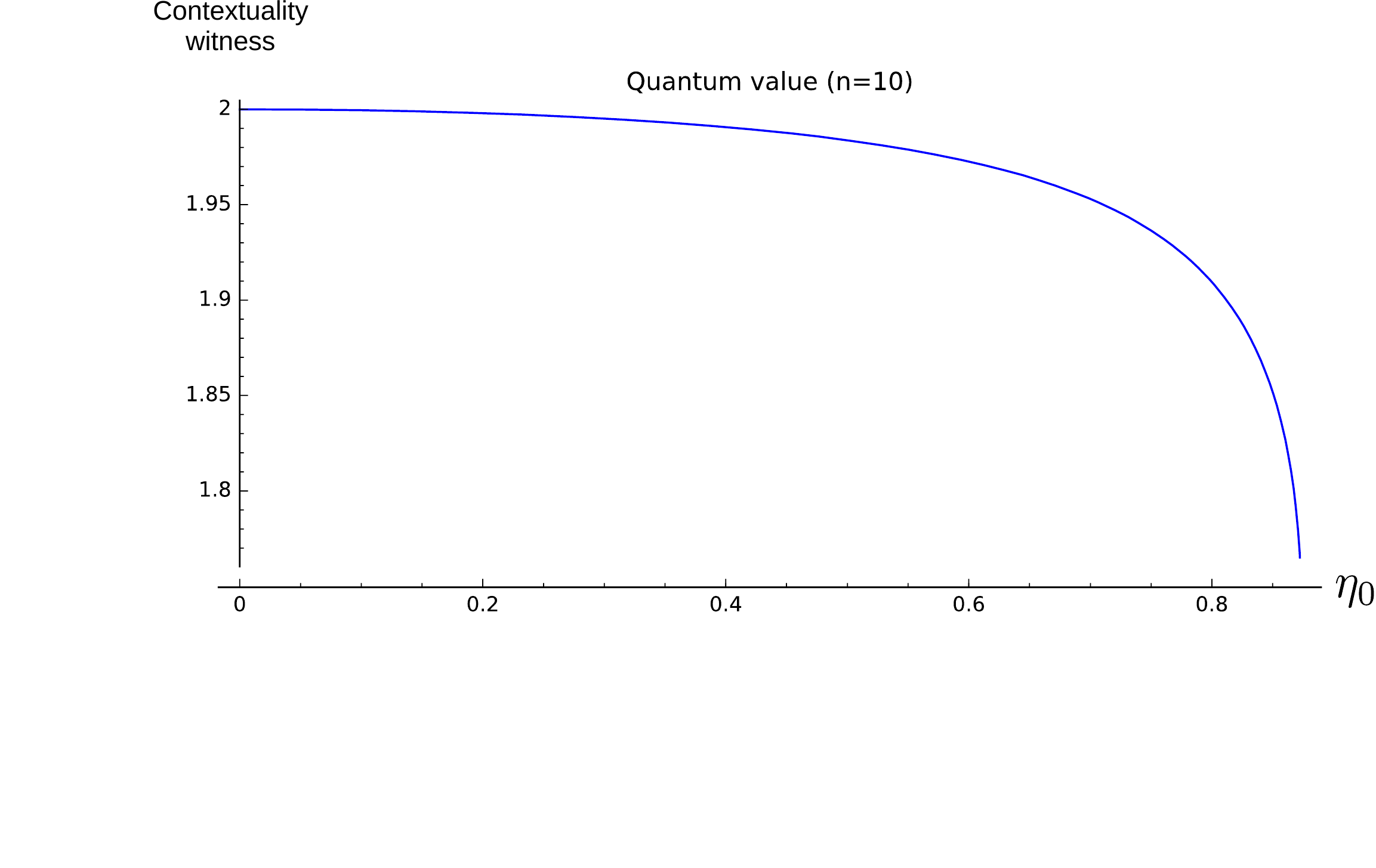}
 \includegraphics[scale=0.52]{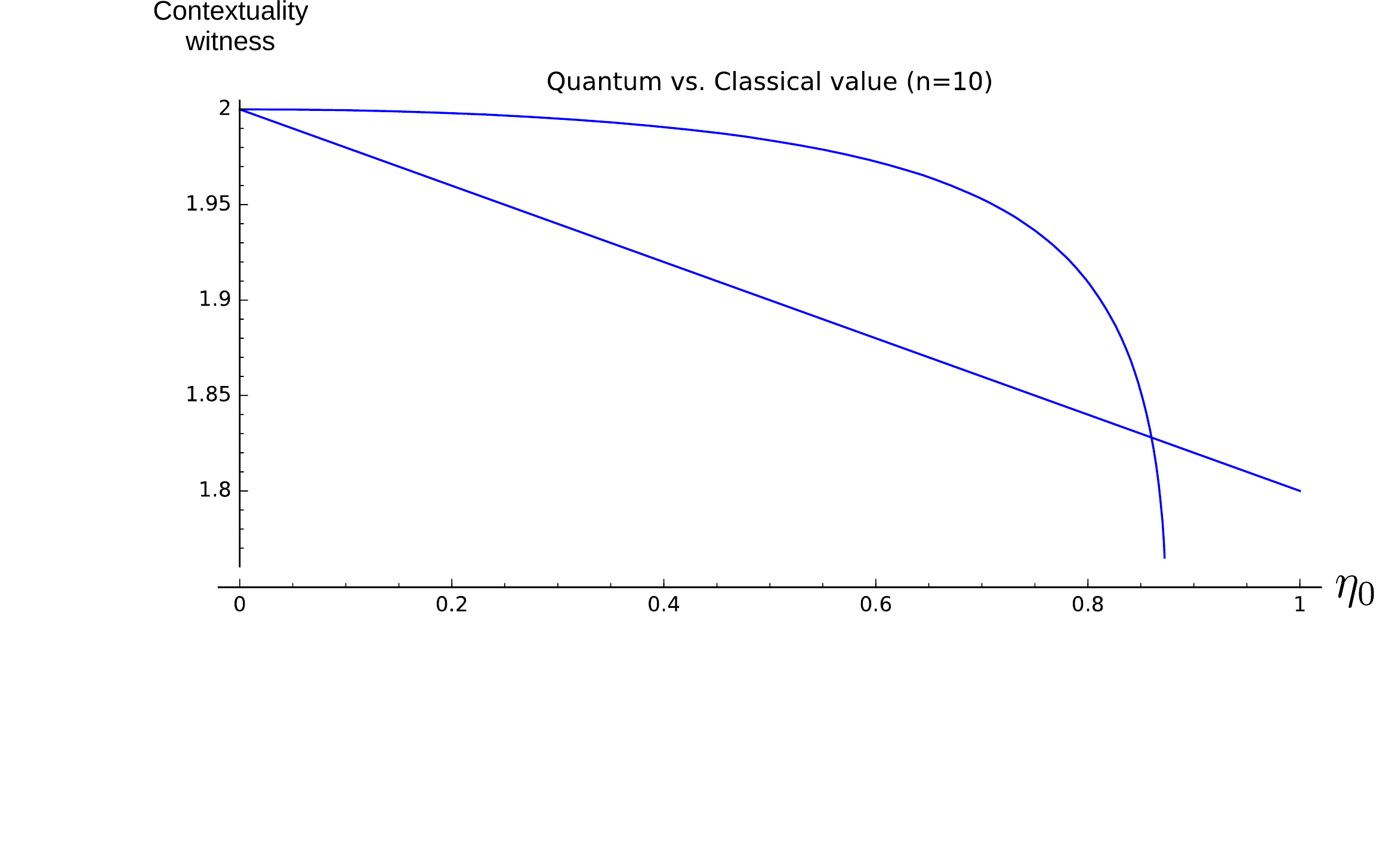}
 \includegraphics[scale=0.52]{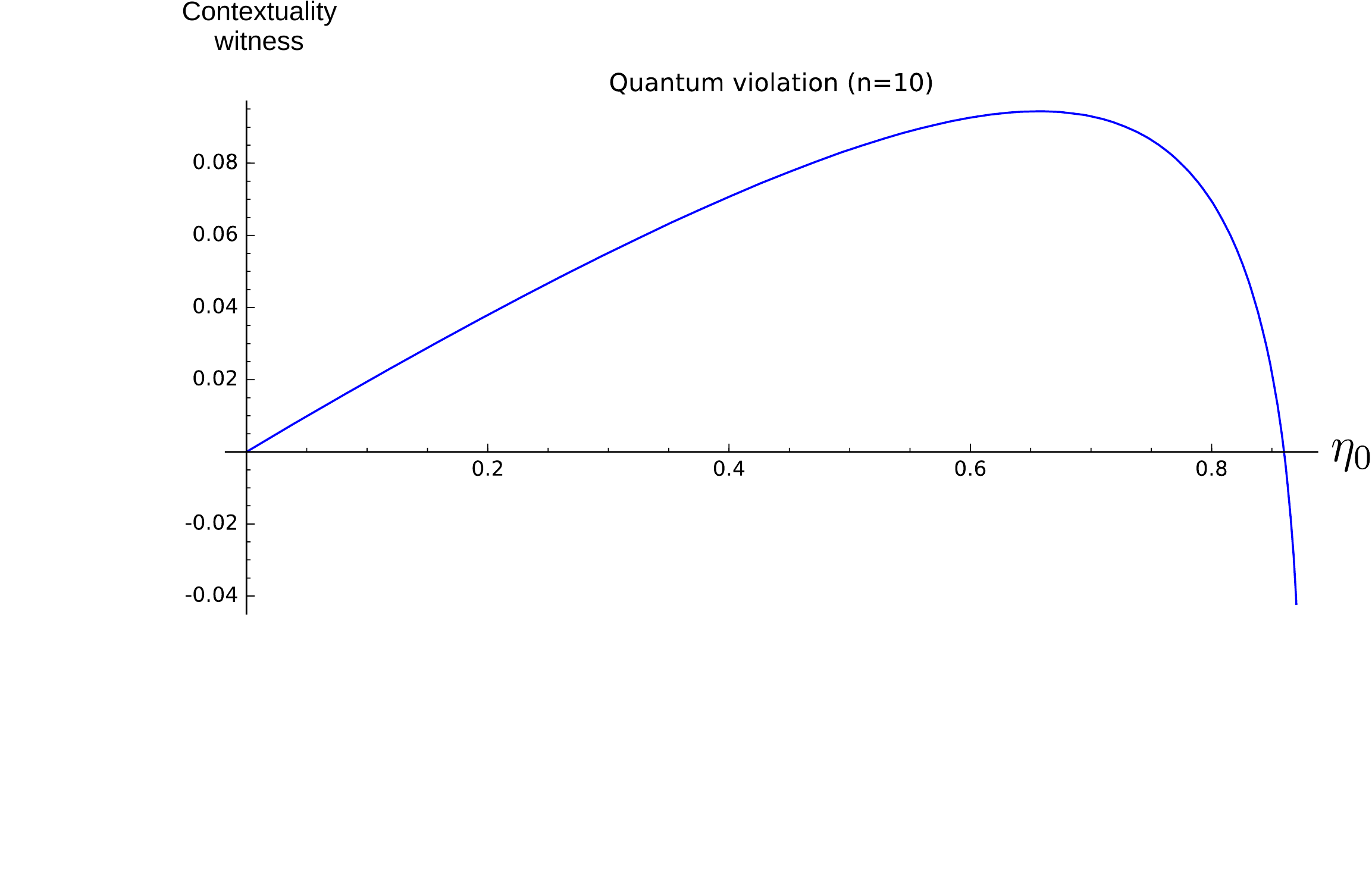}
\caption{n=10 scenario. In the middle figure, straight line denotes classical value.}
\label{n10}
\end{figure}

\begin{figure}
\centering
 \includegraphics[scale=0.52]{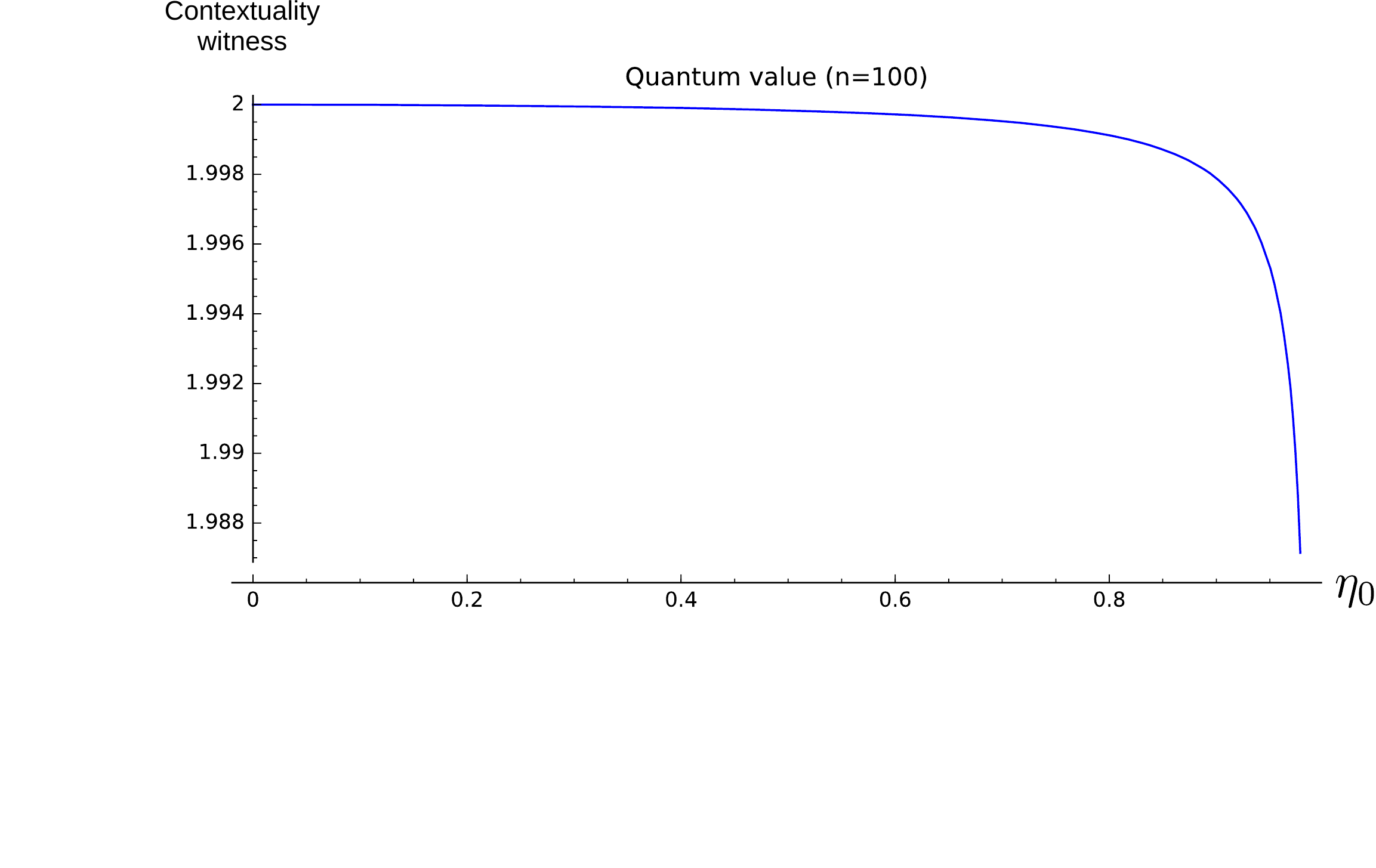}
 \includegraphics[scale=0.52]{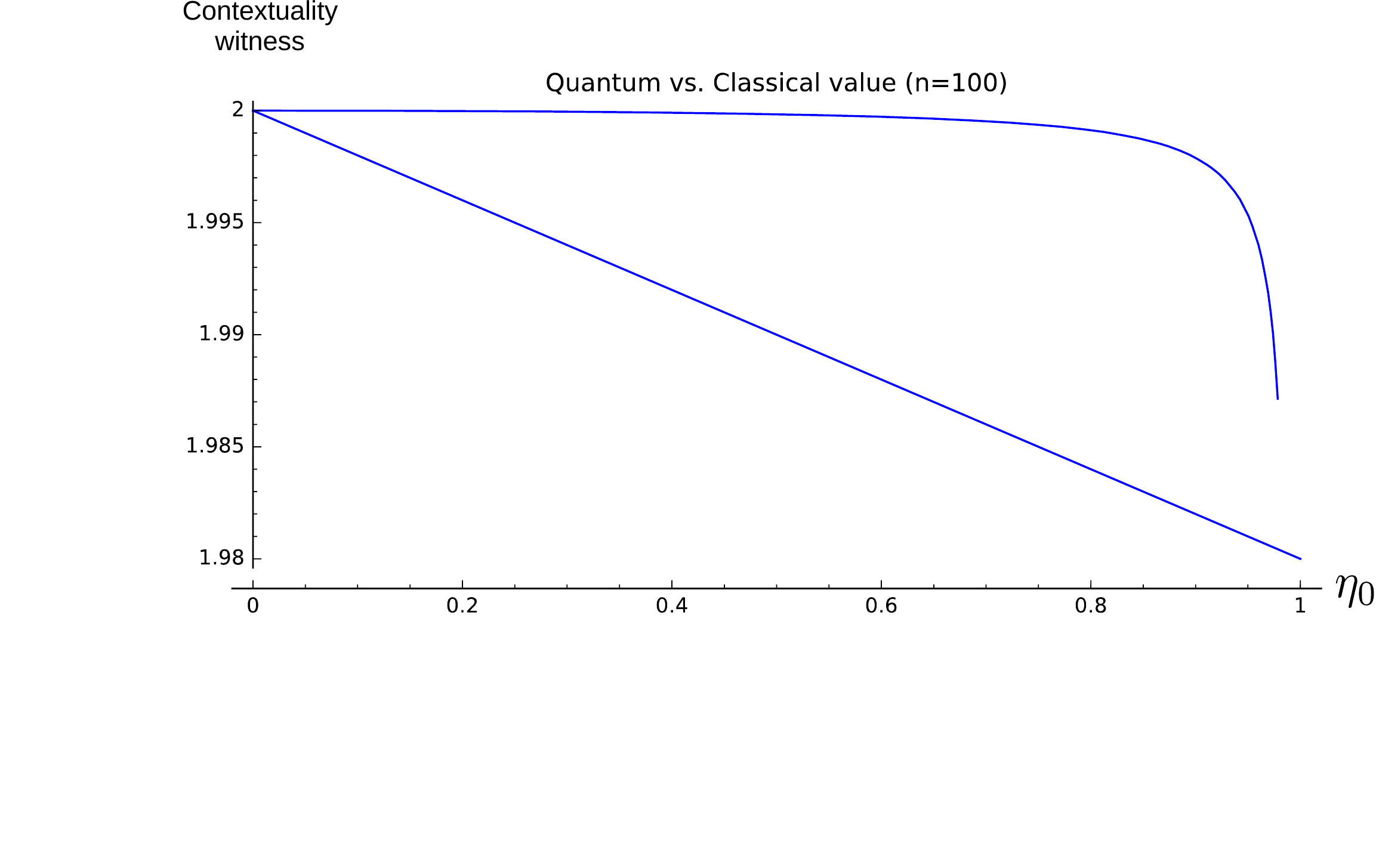}
 \includegraphics[scale=0.52]{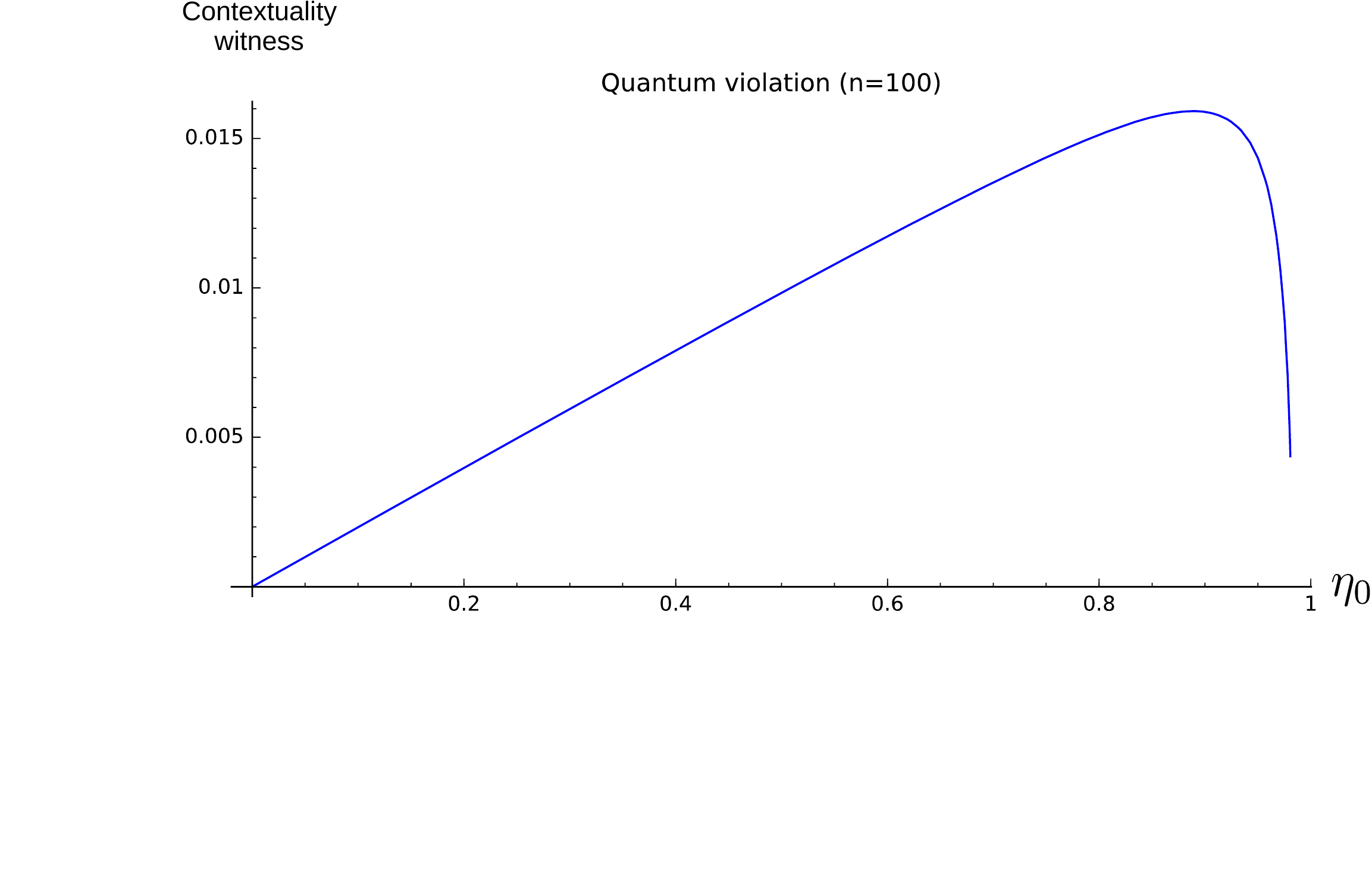}
\caption{n=100 scenario. In the middle figure, straight line denotes classical value.}
\label{n100}
\end{figure}

\begin{figure}
\centering
 \includegraphics[scale=0.52]{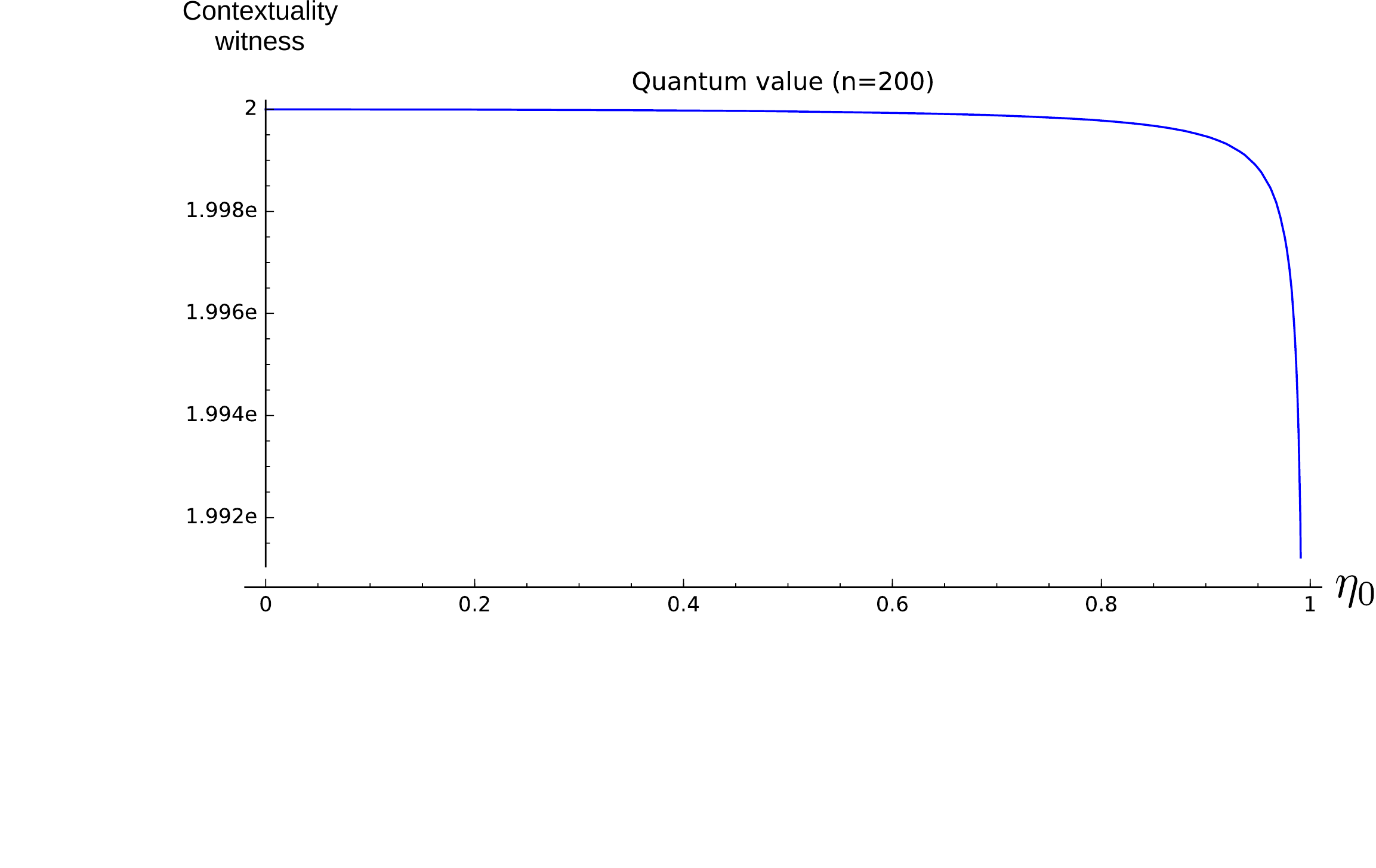}
 \includegraphics[scale=0.52]{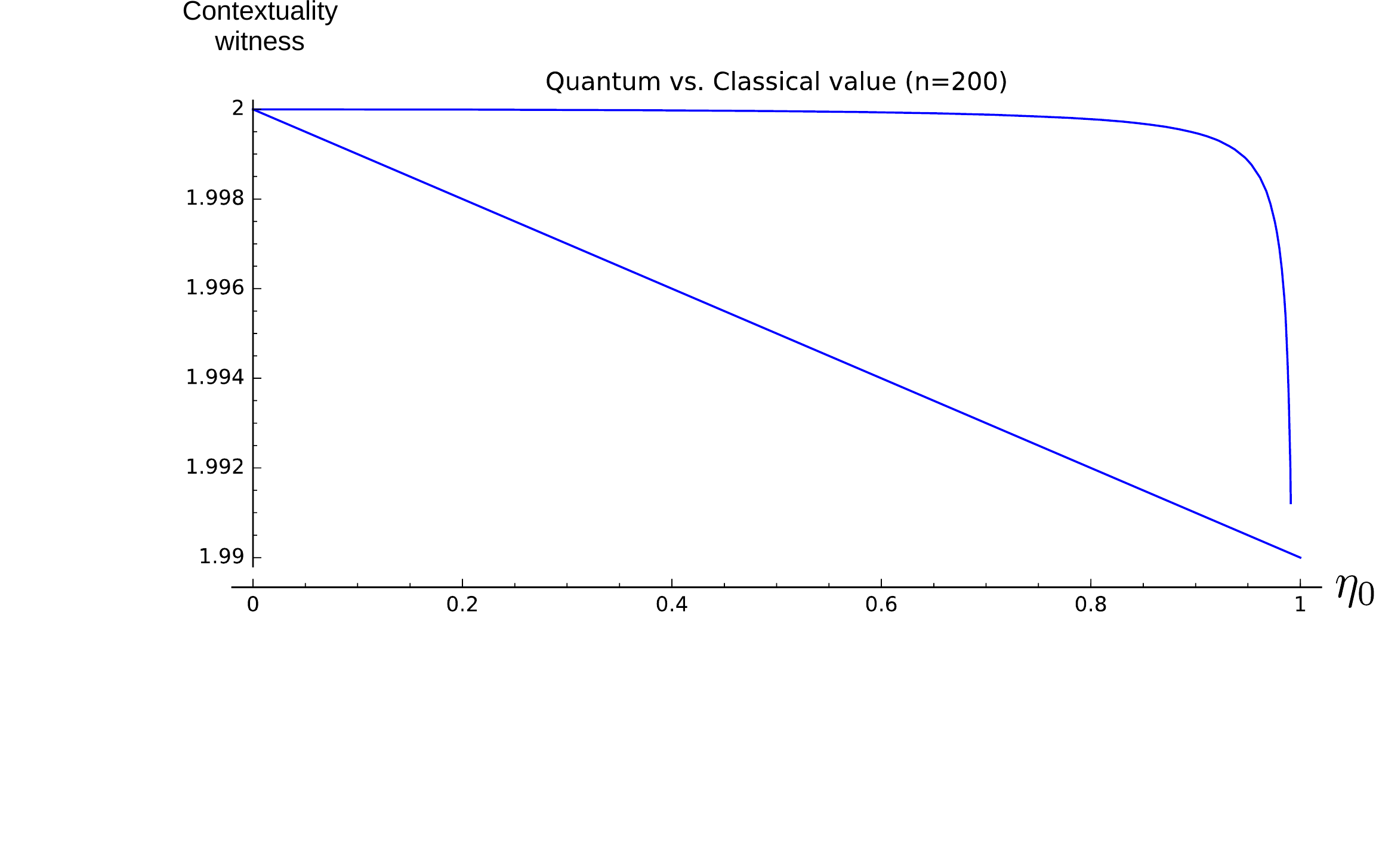}
 \includegraphics[scale=0.52]{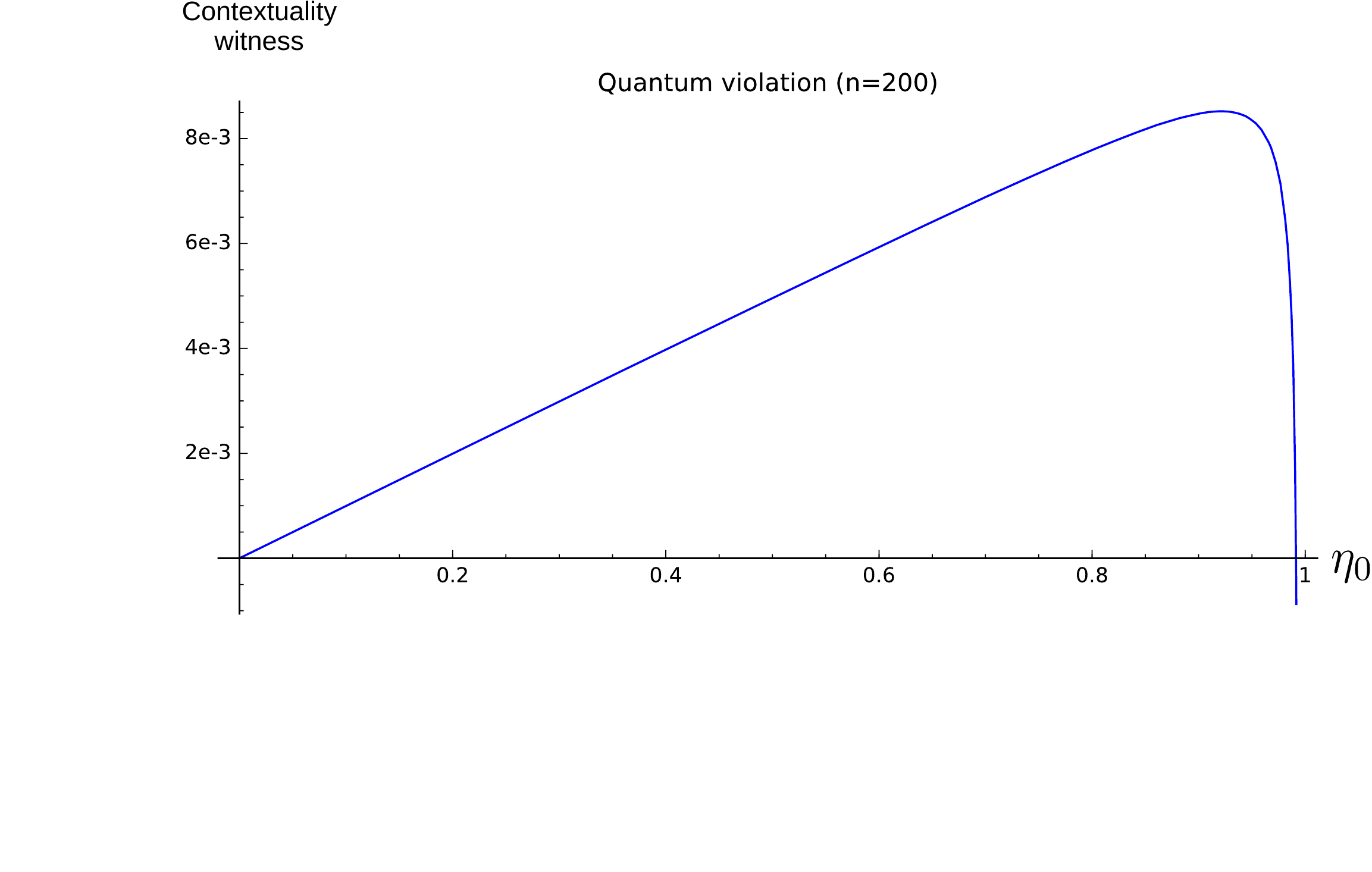}
\caption{n=200 scenario. In the middle figure, straight line denotes classical value.}
\label{n200}
\end{figure}
\clearpage
\section{Chapter summary}
As with the results of the previous chapter, the operational noncontextuality inequalities obtained in this chapter are 
also a significant improvement over previous proposals:
they do not need the assumption of outcome-determinism, only the assumption of noncontextuality for preparations and measurements.
Experimental testability, though, may still be a challenge for the two reasons we outlined in Chapter 6: firstly, the
operational equivalences that we need to verify (in order for the assumption of noncontextuality to have nontrivial consequences) may only be approximately satisfied in a real experiment, 
and secondly, operational equivalence for preparation procedures
is defined relative to \emph{all} measurement procedures and that for measurement procedures relative to \emph{all} preparation procedures,
so an assumption of tomographic completeness is required in order that a finite set of preparations or measurements may be enough to verify the required
operational equivalences. 
Testing the tomographic completeness of a finite set of experimental procedures when quantum theory is not assumed is another challenge.
An approach to handling these issues---deviation from exact operational equivalence and testing the assumption of tomographic completeness---has been demonstrated in 
Ref.~\cite{exptl} (as also illustrated in Chapter 6). The same techniques can potentially
be used to test the noncontextuality inequalities of this chapter in an experimentally robust manner.

In the previous chapter, we demonstrated how the so-called ``state-independent'' proofs of KS-contextuality, based on KS-uncolourability,
can be turned into robust noncontextuality inequalities. In the present chapter we have turned so-called ``state-dependent'' proofs of KS-contextuality into robust noncontextuality inequalities.
Hence, in line with previous work \cite{exptl,KunjSpek}, we have demonstrated that theory-independent tests of noncontextuality are possible, just like tests of 
local causality. The price we need to pay for such a theory-independent test, in the absence of a physical principle like no-faster-than-light-signalling coming to our rescue,
is the requirement that the operational equivalence of experimental procedures must be verified explicitly rather than coming for ``free'' on account of 
spacelike separation (as in Bell experiments).

\section*{APPENDIX}
\subsection*{Constraints from measurement noncontextuality -- the n-cycle polytope}
Once the operational equivalences of Eqs.~(\ref{eq:optlequivalencesMgen1}) and (\ref{eq:optlequivalencesMgen2}) have been verified, the assumption of measurement noncontextuality
places nontrivial bounds on the response functions in an ontological model, characterized by the $n$-cycle polytope (the so-called ``no-disturbance'' polytope of Ref.~\cite{cabelloncycle}).
The $n$-cycle polytope, constrained only by normalization and positivity of probabilities, lives in $\mathbb{R}^{2n}$, has $2^n$ deterministic vertices and $2^{n-1}$ indeterministic vertices 
(see Ref.~\cite{cabelloncycle}, in particular Theorem 2). Relabelling the outcome of each $M_i$ by $S_i\equiv(-1)^{X_i}$, where $X_i\in\{0,1\}$, the $2^n$ deterministic vertices of the 
n-cycle polytope are given by the vectors,
\begin{equation}
 (\langle S_1\rangle,\dots,\langle S_n\rangle,\langle S_1\rangle\langle S_2\rangle,\dots,\langle S_n\rangle\langle S_1\rangle),
\end{equation}
where $\langle S_i\rangle\in\{+1,-1\}$. The $2^{n-1}$ indeterministic vertices are given by the vectors
\begin{equation}
 (0,\dots,0,\langle S_1S_2\rangle,\dots,\langle S_nS_1\rangle),
\end{equation}
where $\langle S_iS_{i+1}\rangle\in\{+1,-1\}$ and the number of entries with $\langle S_iS_{i+1}\rangle=-1$ is odd. (Note that the addition $i+1$ is modulo $n$ so that 
for $i=n$, $n+1=1$.)

We denote the $2^n$ deterministic vertices by 
\begin{eqnarray}
\kappa_{\rm d}\in\big\{(\langle S_1\rangle,\dots,\langle S_n\rangle,\langle S_1\rangle\langle S_2\rangle,\dots,\langle S_n\rangle\langle S_1\rangle)|\nonumber\\
\langle S_i\rangle\in\{-1,+1\}\forall i\in\{1,\dots,n\}\big\},
\end{eqnarray}
and the $2^{n-1}$ indeterministic vertices by 
\begin{eqnarray}
\kappa_{\rm in}\in\{(0,\dots,0,\langle S_1S_2\rangle,\dots,\langle S_nS_1\rangle)|\langle S_iS_{i+1}\rangle\in\{-1,+1\},\nonumber\\
\text{odd number of }\langle S_iS_{i+1}\rangle=-1, i\in\{1,\dots,n\}\}.\nonumber\\
\end{eqnarray}
Also, we let $\kappa$ denote any vertex (deterministic or indeterministic) of the $n$-cycle polytope.

For measurement $M_i$, $i\in \{1,\dots,n\}$, and a given ontic state $\lambda\in \Lambda$, we can define:
\beq
\zeta(M_i,\lambda)\equiv\max_{S_i\in\{+1,-1\}}\xi(S_i|M_i,\lambda).
\eeq
Now:
\beqa
\zeta(M_i,\lambda)&=&\max_{S_i\in\{+1,-1\}}\sum_{\kappa}\xi(S_i|M_i,\kappa)\mu(\kappa|\lambda)\\
&\le& \sum_{\kappa}\max_{S_i\in\{+1,-1\}}\xi(S_i|M_i,\kappa)\mu(\kappa|\lambda)\\
&\equiv& \sum_{\kappa}\zeta(M_i,\kappa)\mu(\kappa|\lambda),
\eeqa
so that $\zeta(M_i,\kappa)=\max_{S_i\in\{+1,-1\}}\xi(S_i|M_i,\kappa)$.

We note that:
\beqa
&&\zeta(M_i,\kappa_{\rm d})=1,\\
&&\zeta(M_i,\kappa_{\rm in})=\frac{1}{2},\\
\eeqa
for all deterministic and indeterministic vertices, respectively.

We then have:
\beqa
&&\zeta(M_i,\lambda)\\
&\le&\frac{1}{2}\mu(\{\kappa_{\rm in}\}|\lambda)+\mu(\{\kappa_{\rm d}\}|\lambda)\\
&=&\frac{1}{2}(1-\mu(\{\kappa_{\rm d}\}|\lambda))+\mu(\{\kappa_{\rm d}\}|\lambda)\\
&=&\frac{1}{2}(1+\mu(\{\kappa_{\rm d}\}|\lambda)),
\eeqa
where $\{\kappa_{\rm in}\}$ denotes the set of all indeterministic vertices and $\{\kappa_{\rm d}\}$ denotes the set of all deterministic vertices.

We have used the fact that $\forall\kappa:\mu(\kappa|\lambda)\geq0,$ $\sum_{\kappa}\mu(\kappa|\lambda)=\mu(\{\kappa_{\rm in}\}|\lambda)+\mu(\{\kappa_{\rm d}\}|\lambda)=1$.
We can rewrite the last inequality on $\zeta(M_i,\lambda)$ as:
\beq
\mu(\{\kappa_{\rm d}\}|\lambda)\ge\eta(M_i,\lambda),
\eeq
where $\eta(M_i,\lambda)\equiv2\zeta(M_i,\lambda)-1$ and $i\in\{1,\dots,n\}$, so that $\mu(\{\kappa_{\rm d}\}|\lambda)\ge\max_i\eta(M_i,\lambda)$. Further:
\beq
\max_{i\in\{1,\dots,n\}} \eta(M_i,\lambda)\ge \frac{1}{n}\sum_{i=1}^n\eta(M_i,\lambda),
\eeq
which implies:
\beq\label{predineq}
\mu(\{\kappa_{\rm d}\}|\lambda)\ge\frac{1}{n}\sum_{i=1}^n\eta(M_i,\lambda).
\eeq

\subsubsection*{Odd n-cycle}
For odd $n$, the quantity of interest is 
\begin{equation}
\xi({\rm anti}|M_*,\lambda)\equiv\frac{1}{n}\sum_{i=1}^n \xi(S_iS_j=-1|M_{ij},\lambda),
\end{equation}
where for each $i$, $j=i+1\mod n$, and, of course,
\begin{eqnarray}
\xi(S_iS_j=-1|M_{ij},\lambda)&\equiv&\xi(S_i=1,S_j=-1|M_{ij},\lambda)\nonumber\\
&+&\xi(S_i=-1,S_j=1|M_{ij},\lambda).\nonumber\\
\end{eqnarray}
It is easy to verify the following: 
\begin{enumerate}
 \item Every deterministic vertex $\kappa_{\rm d}$ is such that $\xi({\rm anti}|M_*,\kappa_{\rm d})\leq\frac{n-1}{n}$.
 \item The unique indeterministic vertex corresponding to perfect anticorrelation, denoted by $\kappa_*\in\{\kappa_{\rm in}\}$, 
 satisfies $\xi({\rm anti}|M_*,\kappa_*)=1$, since $\langle S_iS_j\rangle=-1$ for all $i\in\{1,\dots,n\}$ and $j=i+1\mod n$.
 \item The remaining indeterministic vertices in $\{\kappa_{\rm in}\}\backslash\kappa_*$ satisfy $\xi({\rm anti}|M_*,\kappa_{\rm in})\leq \frac{n-2}{n}$.
 \end{enumerate}

Now,
\begin{equation}
\xi({\rm anti}|M_*,\lambda)=\sum_{\kappa}\xi({\rm anti}|M_*,\kappa)\mu(\kappa|\lambda),
\end{equation}
therefore
\begin{eqnarray}
&&\xi({\rm anti}|M_*,\lambda)\nonumber\\
&\leq&\mu(\kappa_*|\lambda)+\frac{n-1}{n}\mu(\{\kappa_{\rm d}\}|\lambda)+\frac{n-2}{n}\mu(\{\kappa_{\rm in}\}\backslash\kappa_*|\lambda),\nonumber\\
\end{eqnarray}
where $\mu(\kappa_*|\lambda)+\mu(\{\kappa_{\rm d}\}|\lambda)+\mu(\{\kappa_{\rm in}\}\backslash\kappa_*|\lambda)=1$ for each $\lambda\in\Lambda$. Writing 
$\mu(\kappa_*|\lambda)=1-\mu(\{\kappa_{\rm d}\}|\lambda)-\mu(\{\kappa_{\rm in}\}\backslash\kappa_*|\lambda)$, we obtain
\begin{equation}
\xi({\rm anti}|M_*,\lambda)\leq1-\frac{1}{n}\mu(\kappa_{\rm d}|\lambda).
\end{equation}
Using Eq.~(\ref{predineq}), we therefore have
\begin{equation}\label{antibound}
\xi({\rm anti}|M_*,\lambda)\leq1-\frac{1}{n^2}\sum_{i=1}^n\eta(M_i,\lambda).
\end{equation}

\subsubsection*{Even n-cycle}
For even $n$, the quantity of interest is 
\begin{eqnarray}
\xi({\rm chained}|M_*,\lambda)&\equiv&\frac{1}{n}\sum_{i=1}^{n-1} \xi(S_iS_j=1|M_{ij},\lambda)\nonumber\\
&+&\frac{1}{n}\xi(S_nS_1=-1|M_{n1},\lambda),
\end{eqnarray}
where for each $i$, $j=i+1$, and, of course,
\begin{eqnarray}
\xi(S_iS_j=1|M_{ij},\lambda)&\equiv&\xi(S_i=1,S_j=1|M_{ij},\lambda)\nonumber\\
&+&\xi(S_i=-1,S_j=-1|M_{ij},\lambda).\nonumber\\
\end{eqnarray}
It is easy to verify the following: 
\begin{enumerate}
 \item Every deterministic vertex $\kappa_{\rm d}$ is such that $\xi({\rm chained}|M_*,\kappa_{\rm d})\leq\frac{n-1}{n}$.
 \item The unique indeterministic vertex corresponding to perfect chained correlation, denoted by $\kappa_*\in\{\kappa_{\rm in}\}$,
 satisfies $\xi({\rm chained}|M_*,\kappa_*)=1$, since $\langle S_iS_j\rangle=1$ for all $i\in\{1,\dots,n-1\}, j=i+1$ and $\langle S_nS_1\rangle=-1$.
 \item The remaining indeterministic vertices in $\{\kappa_{\rm in}\}\backslash\kappa_*$ satisfy $\xi({\rm chained}|M_*,\kappa_{\rm in})\leq \frac{n-2}{n}$.
 \end{enumerate}
Now,
\begin{equation}
\xi({\rm chained}|M_*,\lambda)=\sum_{\kappa}\xi({\rm chained}|M_*,\kappa)\mu(\kappa|\lambda),
\end{equation}
therefore
\begin{eqnarray}
&&\xi({\rm chained}|M_*,\lambda)\nonumber\\
&\leq&\mu(\kappa_*|\lambda)+\frac{n-1}{n}\mu(\{\kappa_{\rm d}\}|\lambda)+\frac{n-2}{n}\mu(\{\kappa_{\rm in}\}\backslash\kappa_*|\lambda),\nonumber\\
\end{eqnarray}
where $\mu(\kappa_*|\lambda)+\mu(\{\kappa_{\rm d}\}|\lambda)+\mu(\{\kappa_{\rm in}\}\backslash\kappa_*|\lambda)=1$ for each $\lambda\in\Lambda$. Writing 
$\mu(\kappa_*|\lambda)=1-\mu(\{\kappa_{\rm d}\}|\lambda)-\mu(\{\kappa_{\rm in}\}\backslash\kappa_*|\lambda)$, we obtain
\begin{equation}
\xi({\rm chained}|M_*,\lambda)\leq1-\frac{1}{n}\mu(\kappa_{\rm d}|\lambda).
\end{equation}
Using Eq.~(\ref{predineq}), we therefore have
\begin{equation}\label{chainedbound}
\xi({\rm chained}|M_*,\lambda)\leq1-\frac{1}{n^2}\sum_{i=1}^n\eta(M_i,\lambda).
\end{equation}

Based on Eqs.~(\ref{antibound}) and (\ref{chainedbound}), we can now state a lemma---derived, as above, from the assumption of measurement noncontextuality---that will be useful in our proof of the noncontextuality inequalities:
\begin{lemma}\label{lemma:bound}
For odd $n\geq3$,
\begin{equation}
\xi({\rm anti}|M_*,\lambda)\le 1 - \frac{1}{n^2}\sum_{i=1}^n\eta(M_{i},\lambda).
\end{equation}
The same bound holds for $\xi({\rm anti}|M'_*,\lambda)$. Also, for even $n\geq4$,
\begin{equation}
\xi({\rm chained}|M_*,\lambda)\le 1 - \frac{1}{n^2}\sum_{i=1}^n\eta(M_{i},\lambda).
\end{equation}
The same bound holds for $\xi({\rm chained}|M'_*,\lambda)$.
\end{lemma}

Remember, there are two $n$-cycle polytopes, one associated with joint measurements $\{M_{ij}\}$ and the other associated with joint measurements $\{M^{\prime}_{ij}\}$. Inequalities analogous to
Eqs.~(\ref{antibound}) and (\ref{chainedbound}), derived for the $\{M_{ij}\}$ polytope, also hold for the polytope associated with primed measurements $\{M^{\prime}_{ij}\}$. This is the content
of Lemma \ref{lemma:bound}. 

Note that if, at a particular value of $\lambda$, all the measurements depend deterministically on $\lambda$, so that 
$\forall i, \eta(M_{i},\lambda)=1$, then we recover the bounds for an outcome-deterministic 
noncontextual ontological model, given by $\xi({\rm anti}|M_*,\lambda) \le \frac{n-1}{n}$ (similarly for $M^{\prime}_*$) and $\xi({\rm chained}|M_*,\lambda) \le \frac{n-1}{n}$ (similarly for $M^{\prime}_*$)
for odd and even $n$, respectively. If, on the other extreme, at a particular value of $\lambda$,
all the measurements are independent of $\lambda$, so that $\forall i, \eta(M_{i},\lambda)=0$, then the bounds become
trivial, $\xi({\rm anti}|M_*,\lambda) \le 1$ and $\xi({\rm chained}|M_*,\lambda) \le 1$ (similarly for $M^{\prime}_*$): there is no obstacle to seeing perfect anticorrelation (for odd $n$) or 
perfect chained correlation (for even $n$) in a measurement noncontextual ontological model in this case.

It now remains to use the assumption of preparation noncontextuality to obtain the operational noncontextuality inequalities of Theorems \ref{speckertheorem} and \ref{ncycletheorem}.

\subsection*{Proof of Theorems \ref{speckertheorem} and \ref{ncycletheorem}}\label{prooftheorem1}
In this appendix we prove the inequalities in Theorem \ref{ncycletheorem}. The inequalities in Theorem \ref{speckertheorem} are a special case of these inequalities when $n=3$.

We define the operational \emph{predictability} of measurement $M$, implemented following a preparation $P$, as
\begin{equation}
\eta(M,P) \equiv 2 \max_{X\in\{0,1\}}p(X|M,P)-1.
\end{equation}
Recall that the ontological predictability, defined earlier, is given by $\eta(M,\lambda)=2 \max_{X\in\{0,1\}}\xi(X|M,P)-1$.

\begin{lemma}\label{lemma:predictability}
The ontological predictability of $M$ given $\lambda$, $\eta(M,\lambda)$, averaged over $\mu(\lambda|P)$, must be at least as great as the operational predictability of $M$ given $P$,
\[
\sum_{\lambda\in\Lambda} \eta(M,\lambda) \mu(\lambda|P) \ge \eta(M,P).
\]
\end{lemma}
\begin{proof}
\begin{eqnarray}
&&\eta(M,P)=2\max_{X\in\{0,1\}}p(X|M,P)-1\nonumber\\
&=&2\max_{X\in\{0,1\}}\sum_{\lambda\in\Lambda} \xi(X|M,\lambda)\mu(\lambda|P)-1\nonumber\\
&\leq&2\sum_{\lambda\in\Lambda}\max_{X\in\{0,1\}}\xi(X|M,\lambda)\mu(\lambda|P)-1\nonumber\\
&=&\sum_{\lambda\in\Lambda}\eta(M,\lambda)\mu(\lambda|P).\nonumber
\end{eqnarray}
\end{proof}
From the requirement that the ontological model reproduces the operational predictions, it follows that:
for odd $n\geq3$
\begin{align}
&p({\rm anti}|M_*,P_*)+p({\rm anti}|M'_*,P^{\perp}_*)\nonumber \\
&=\sum_{\lambda\in\Lambda} \xi({\rm anti}|M_*,\lambda) \mu(\lambda|P_*) + \sum_{\lambda\in\Lambda} \xi({\rm anti}|M'_*,\lambda) \mu(\lambda|P^{\perp}_*) \nonumber \\
&\le 2- \frac{1}{n^2}\sum_{\lambda\in\Lambda}  \sum_{i=1}^n\eta(M_i,\lambda)\left(\mu(\lambda|P_*)+ \mu(\lambda|P^{\perp}_*)\right),
\end{align}
where the last inequality is a consequence of Lemma \ref{lemma:bound}. Similarly, for even $n\geq4$
\begin{align}
&p({\rm chained}|M_*,P_*)+p({\rm chained}|M'_*,P^{\perp}_*)\nonumber \\
&=\sum_{\lambda\in\Lambda} \xi({\rm chained}|M_*,\lambda) \mu(\lambda|P_*) + \sum_{\lambda\in\Lambda} \xi({\rm chained}|M'_*,\lambda) \mu(\lambda|P^{\perp}_*) \nonumber \\
&\le 2- \frac{1}{n^2}\sum_{\lambda\in\Lambda}  \sum_{i=1}^n\eta(M_i,\lambda)\left(\mu(\lambda|P_*)+ \mu(\lambda|P^{\perp}_*)\right),
\end{align}
where we have again used Lemma \ref{lemma:bound}.

{\bf We now make use of the assumption of preparation noncontextuality.}  Recalling the assumed operational equivalences among preparations (which have to be experimentally verified),
Eq.~\eqref{eq:optlequivalencesPgen}, the definition of preparation noncontextuality, and how mixtures of preparation procedures are represented in an 
ontological model, we infer that 
\begin{align}
&\forall i: \frac{1}{2}\mu(\lambda|P_*)+\frac{1}{2} \mu(\lambda|P^{\perp}_*)\nonumber\\
&=\frac{1}{2}\mu(\lambda|P_i)+\frac{1}{2} \mu(\lambda|P^{\perp}_{i}).
\end{align}
It then follows that
\begin{align}
&p({\rm anti}|M_*,P_*)+p({\rm anti}|M'_*,P^{\perp}_*)\nonumber \\
&\le 2-\frac{1}{n^2}\sum_{\lambda\in\Lambda}  \sum_{i=1}^n\eta(M_i,\lambda) (\mu(\lambda|P_i)+ \mu(\lambda|P^{\perp}_{i}))\nonumber\\
\end{align}
for odd $n\geq3$ and 
\begin{align}
&p({\rm chained}|M_*,P_*)+p({\rm chained}|M'_*,P^{\perp}_*)\nonumber \\
&\le 2-\frac{1}{n^2}\sum_{\lambda\in\Lambda}  \sum_{i=1}^n\eta(M_i,\lambda) (\mu(\lambda|P_i)+ \mu(\lambda|P^{\perp}_{i}))\nonumber,\\
\end{align}
for even $n\geq4$.

Finally, {\bf making use of the bound derived in lemma~\ref{lemma:predictability}}, we obtain the operational inequality
\begin{align}
&p({\rm anti}|M_*,P_*)+p({\rm anti}|M'_*,P^{\perp}_*)\nonumber \\
&\le 2- \frac{1}{n^2} \sum_{i=1}^n\left(\eta(M_i,P_i)+\eta(M_i,P^{\perp}_i)\right)\nonumber\\
&=2\left(1-\frac{1}{n}\eta_{\rm ave}\right),
\end{align}
where 
\begin{equation}
\eta_{\rm ave}\equiv \frac{1}{2n}\sum_{i=1}^n \left(\eta(M_{i},P_i)+\eta(M_{i},P^{\perp}_{i})\right).
\end{equation}
which completes the proof of the inequality in Eq.~(\ref{eq:maininequalitygen}) for odd $n\geq3$. For even $n\geq4$
\begin{align}
&p({\rm chained}|M_*,P_*)+p({\rm chained}|M'_*,P^{\perp}_*)\nonumber \\
&\le 2- \frac{1}{n^2} \sum_{i=1}^n\left(\eta(M_i,P_i)+\eta(M_i,P^{\perp}_i)\right)\nonumber\\
&=2\left(1-\frac{1}{n}\eta_{\rm ave}\right),
\end{align}
which is the inequality in Eq.\eqref{eq:maininequality_chsh} for even $n\geq4$.

Now, for the other four inequalities in Eqs.~\eqref{eq:oddncycleineq1}-\eqref{eq:oddncycleineq2} and Eqs.~\eqref{eq:evenncycleineq1}-\eqref{eq:evenncycleineq2}:
\begin{align}
&p({\rm anti}|M_*,P_*)+p({\rm anti}|M_*,P^{\perp}_*)\nonumber\\
&=\sum_{\lambda\in\Lambda} \xi({\rm anti}|M_*,\lambda)(\mu(\lambda|P_*)+\mu(\lambda|P^{\perp}_*))\nonumber\\
&\le\sum_{\lambda\in\Lambda} \left(1 - \frac{1}{n^2}\sum_{i=1}^n\eta(M_{i},\lambda)\right)\left(\mu(\lambda|P_*)+\mu(\lambda|P^{\perp}_*)\right)\nonumber\\
&\le 2- \frac{1}{n^2} \sum_{i=1}^n\left(\eta(M_i,P_i)+\eta(M_i,P^{\perp}_i)\right)\nonumber\\
&=2\left(1-\frac{1}{n}\eta_{\rm ave}\right),
\end{align}
which is the inequality of Eq.\eqref{eq:oddncycleineq1}. Similarly, the inequality of Eq.\eqref{eq:evenncycleineq1} can be shown to hold. To obtain Eqs.\eqref{eq:oddncycleineq2} and \eqref{eq:evenncycleineq2},
note that
\begin{eqnarray}
&&\eta_{\rm ave}=\frac{1}{2n}\sum_{i=1}^n \left(\eta(M_{i},P_i)+\eta(M_{i},P^{\perp}_{i})\right)\nonumber\\
&\leq&\frac{1}{2n}\sum_{i=1}^n\left(\sum_{\lambda\in\Lambda} \eta(M_i,\lambda)\mu(\lambda|P_i)+\sum_{\lambda\in\Lambda} \eta(M_i,\lambda)\mu(\lambda|P^{\perp}_i)\right)\nonumber\\
&&\text{(using lemma }\ref{lemma:predictability}\text{)},\nonumber\\
&=&\frac{1}{2n}\sum_{i=1}^n\sum_{\lambda\in\Lambda} \eta(M_i,\lambda)\left(\mu(\lambda|P_i)+\mu(\lambda|P^{\perp}_i)\right),\nonumber\\
&=&\frac{1}{2n}\sum_{\lambda\in\Lambda} \sum_{i=1}^n\eta(M_i,\lambda)\left(\mu(\lambda|P_*)+\mu(\lambda|P^{\perp}_*)\right)\nonumber\\
&&\text{(using preparation noncontextuality)}\nonumber\\
&\leq&\frac{1}{2n}\sum_{\lambda\in\Lambda} \sum_{i=1}^n\eta(M_i,\lambda)\mu(\lambda|P_*)+\frac{1}{2},\nonumber\\
&&\text{(using the fact that }\sum_{i=1}^n\eta(M_i,\lambda)\leq n \text{ in the second term)},\nonumber
\end{eqnarray}
so that 
\begin{equation}
\frac{1}{n^2}\sum_{\lambda\in\Lambda} \sum_{i=1}^n\eta(M_i,\lambda)\mu(\lambda|P_*)\geq \frac{1}{n}(2\eta_{\rm ave}-1)
\end{equation}
and 
\begin{eqnarray}
p({\rm anti}|M_*,P_*)&=&\sum_{\lambda\in\Lambda} \xi({\rm anti}|M_*,\lambda)\mu(\lambda|P_*)\nonumber\\
&\leq& \sum_{\lambda\in\Lambda} \left(1 - \frac{1}{n^2}\sum_{i=1}^n\eta(M_{i},\lambda)\right)\mu(\lambda|P_*)\nonumber\\
&&\text{(using measurement noncontextuality, lemma \ref{lemma:bound})}\nonumber\\
&\leq&\frac{n-1}{n}+\frac{2}{n}(1-\eta_{\rm ave}),
\end{eqnarray}
which is Eq.~\eqref{eq:oddncycleineq2}. Eq.~\eqref{eq:evenncycleineq2} can be obtained similarly. 

\chapter{Future directions}
In this thesis, I have reported progress towards making contextuality an experimentally testable notion of nonclassicality, based on
work I have jointly carried out with various collaborators. In the first chapter, we introduced the framework of operational
theories and ontological models within which discussions of contextuality are carried out in this thesis. 
In the second chapter, we showed how Specker's scenario admits contextuality with respect to the LSW inequality. 
The third chapter examined the relationship between joint measurability and contextuality in quantum theory. 
The fourth chapter demonstrated how any joint measurability structure is realizable in quantum theory, opening up the 
possibility of contextuality in scenarios not envisaged by the Kochen-Specker theorem. In the fifth chapter, we argued why 
Fine's theorem does not absolve one of the need to justify outcome determinism in ontological models of quantum theory. 
In chapter six we showed how to obtain robust noncontextuality inequalities directly inspired by the Kochen-Specker theorem
in arbitrary operational theories. We also outlined a procedure for handling the problem of inexact operational equivalences 
in tests of noncontextuality, exemplified by the adoption of this procedure in the experiment of Ref.~\cite{exptl}.
In the seventh chapter, we returned to Specker's scenario and provided an analysis independent of quantum theory, unlike the 
analysis in the second chapter. We also generalized the insights from Specker's scenario to arbitrary $n$-cycle scenarios 
and provided noncontextuality inequalities for these scenarios.

All this work opens up the opportunity to pursue a full development of contextuality into a comprehensive tool that can be used to
detect the ``nonclassicality'' or ``quantumness'' of phenomena in a wide variety of physical scenarios. 
An important challenge is to build a quantitative bridge connecting our work to previous work in Kochen-Specker contextuality, 
in particular the graph-theoretic characterization of contextuality scenarios \cite{CSW, AFLS},
and bring this work within the ambit of the operational approach to contextuality due to Spekkens \cite{genNC}. This would 
unify these approaches quantitatively and open the door to harnessing this nonclassicality for concrete information-theoretic tasks
not relying on Bell tests. So far, Bell tests -- experiments requiring (at least) two distant nonsignalling parties -- are the gold standard
for tests of nonclassicality. Work on contextuality along these lines could potentially yield ways of certifying nonclassicality
that can be implemented within the same laboratory, covering a whole range of experimental scenarios that are not of the Bell 
type. Akin to Bell tests, the ``no-go'' conclusions drawn from these experiments would not rely on the validity of quantum 
theory. One also needs to \emph{extend} the scope of the graph-theoretic framework to include
contextuality scenarios that, in quantum theory, are only realizable with nonprojective observables \cite{KHF}.
The precise nature of the relation between Bell scenarios and tests of contextuality, beyond the folklore that
considers Bell nonlocality a special case of contextuality, is worth exploring in view of 
recent results on contextuality \`a la Spekkens \cite{MattP}. Indeed,
following recent progress in understanding Bell experiments in terms of causal structure \cite{HLP}, I would also
like to investigate whether, and to what extent, contextuality could be formulated in an appropriate causal 
framework that generalizes the framework adopted for Bell scenarios \cite{robPIRSA}.

It would also be important to investigate contextuality as a resource in quantum information theory and study the robustness of 
this resource, perhaps in terms of channels that ``break'' it: with an appropriate definition of ``contextuality breaking'' 
channels in hand, one could expect to draw connections with the previously studied cases of channels that break entanglement, 
incompatibility, nonlocality, and nonclassicality\cite{EBC}. 
Indeed, since noise is expected to destroy contextuality, part of the 
motivation here is to also study the robustness of the ``magic for quantum computation'' that contextuality (in the Kochen-Specker
framework) supplies \cite{magic, magic2}. The Spekkens' framework for contextuality allows a robustness analysis of this sort and proving 
noise thresholds beyond which the ``magic'' is lost is a necessary exercise in order to understand the precise relationship between 
contextuality and practical quantum computation. Evidence of the role that (Kochen-Specker) contextuality plays in
measurement based quantum computing \cite{magic2} also needs to be subjected to such a robustness analysis in the Spekkens' framework.

Finally, an interesting avenue is to investigate whether there's a concrete connection between contextuality and other examples
of quantum ``weirdness''\cite{Spek08}. The motivation here is to ask whether contextuality is indeed \emph{the} nonclassical 
feature of quantum theory that implies (or is implied by) all its other ``weird'' features, in particular features responsible
for claimed quantum-over-classical advantages. If so, then we can claim to have found a characterization of nonclassicality that
clearly separates quantum theory from attempts at its classical simulations, particularly in scenarios not of the Bell type. If
not, then the challenge will be to identify what noncontextual part of quantum theory is still nonclassical, where this 
nonclassicality cannot be attributed to a restriction on how much any agent can know about the physical state of a system
(that is, an ``epistemic restriction''). We already know that a lot of the features of quantum theory can be simulated in 
classical theories imposing an epistemic restriction\cite{quasiquant}. These theories, however, are local and noncontextual, 
so they cannot reproduce \emph{all} of quantum theory: they are only ``weakly nonclassical''.

In conclusion, let me again emphasize that the work presented in this thesis does not fall in the same category as much of the work in the 
existing literature on contextuality in the traditional
Kochen-Specker paradigm (such as \cite{CSW, AFLS}). In going beyond the Kochen-Specker theorem, this thesis 
contributes towards the project of developing a strictly operational understanding of contextuality.


\begin{thebibliography}{}
\bibitem{merminquote} Indeed, in Ref.~\cite{Cabello98}, David Mermin is quoted
as having said: 
``the whole notion of an experimental test of KS misses the point,'' a view
that was held by many researchers at the time.
\end{thebibliography}
\end{document}